\documentclass[a4paper,onecolumn,10pt,accepted=2023-05-12]{quantumarticle}
\pdfoutput=1

\usepackage[utf8]{inputenc}
\usepackage[english]{babel}
\usepackage[T1]{fontenc}
\usepackage{tikz}
\usepackage{lipsum}
\usepackage{amsmath,amsthm}
\usepackage{amssymb}
\usepackage{hyperref,natbib}
\usepackage{color, soul}
\usepackage{graphicx}
\usepackage{verbatim}

\newtheorem{theorem}{Theorem}[section]
\newtheorem{lemma}{Lemma}[section]
\newtheorem{definition}{Definition}[section]

\newtheorem{corollary}{Corollary}[section]
\newtheorem{proposition}{Proposition}[section]
\newtheorem{assumption}{Assumption}[section]
\newtheorem{problem}{Problem}[section]
\newtheorem{standard}{Standard}[section]
\newtheorem{question}{Question}[section]
\newtheorem{example}{Example}[section]
\newtheorem{remark}{Remark}[section]
\newtheorem*{mainmessage}{Main Message}

\newtheorem*{namedthm}{\namedthmname}
\newcounter{namedthm}
\makeatletter
\newenvironment{named}[1]
  {\def\namedthmname{#1}%
   \refstepcounter{namedthm}%
   \namedthm\def\@currentlabel{#1}}
  {\endnamedthm}
\makeatother

\newcommand{\eqn}[1]{(\ref{eqn:#1})}
\newcommand{\eq}[1]{(\ref{eq:#1})}
\newcommand{\prb}[1]{\hyperref[prb:#1]{Problem~\ref*{prb:#1}}}
\newcommand{\thm}[1]{\hyperref[thm:#1]{Theorem~\ref*{thm:#1}}}
\newcommand{\cor}[1]{\hyperref[cor:#1]{Corollary~\ref*{cor:#1}}}
\newcommand{\defn}[1]{\hyperref[defn:#1]{Definition~\ref*{defn:#1}}}
\newcommand{\hypo}[1]{\hyperref[hypo:#1]{Hypothesis~\ref*{hypo:#1}}}
\newcommand{\lem}[1]{\hyperref[lem:#1]{Lemma~\ref*{lem:#1}}}
\newcommand{\prop}[1]{\hyperref[prop:#1]{Proposition~\ref*{prop:#1}}}
\newcommand{\assum}[1]{\hyperref[assum:#1]{Assumption~\ref*{assum:#1}}}
\newcommand{\stand}[1]{\hyperref[stand:#1]{Standard~\ref*{stand:#1}}}
\newcommand{\fig}[1]{\hyperref[fig:#1]{Figure~\ref*{fig:#1}}}
\newcommand{\tab}[1]{\hyperref[tab:#1]{Table~\ref*{tab:#1}}}
\newcommand{\algo}[1]{\hyperref[algo:#1]{Algorithm~\ref*{algo:#1}}}
\renewcommand{\sec}[1]{\hyperref[sec:#1]{Section~\ref*{sec:#1}}}

\newcommand{\append}[1]{\hyperref[append:#1]{Appendix~\ref*{append:#1}}}

\newcommand{\fac}[1]{\hyperref[fac:#1]{Fact~\ref*{fac:#1}}}
\newcommand{\lin}[1]{\hyperref[lin:#1]{Line~\ref*{lin:#1}}}
\newcommand{\fnote}[1]{\hyperref[fnote:#1]{Footnote~\ref*{fnote:#1}}}

\newcommand{\rmk}[1]{\hyperref[rmk:#1]{Remark~\ref*{rmk:#1}}}
\newcommand{\ques}[1]{\hyperref[ques:#1]{Question~\ref*{ques:#1}}}
\newcommand{\examp}[1]{\hyperref[examp:#1]{Example~\ref*{examp:#1}}}
\newcommand{\slog}{\hyperref[slog:1]{Main Message}}

\newcommand{\arxiv}[1]{arXiv:\href{https://arxiv.org/abs/#1}{\ttfamily{#1}}}

\def\>{\rangle}
\def\<{\langle}
\def\trans{^{\top}}

\newcommand{\R}{\mathbb{R}}
\newcommand{\C}{\mathbb{C}}
\newcommand{\N}{\mathbb{N}}
\newcommand{\y}{\rangle}
\newcommand{\z}{\langle}
\newcommand{\E}{\mathcal{E}}
\newcommand{\F}{\mathcal{F}}
\newcommand{\bra}[1]{\left\langle #1 \right|}
\newcommand{\ket}[1]{\left| #1 \right\rangle}
\newcommand{\ip}[2]{\left\langle{#1}\right|\left.\!{#2}\right\rangle}

\renewcommand{\d}{\mathrm{d}}
\newcommand{\tr}{\mathrm{tr}}
\renewcommand{\bf}{\mathbf}

\begin{document}
\title{On Quantum Speedups for Nonconvex Optimization via Quantum Tunneling Walks}

\author{Yizhou Liu}
\affiliation{Department of Engineering Mechanics, Tsinghua University, 100084 Beijing, China}
\email{liuyz18@tsinghua.org.cn}
\orcid{0000-0003-2105-0894}

\author{Weijie J. Su}
\affiliation{Department of Statistics and Data Science, University of Pennsylvania}
\email{suw@wharton.upenn.edu}
\homepage{http://stat.wharton.upenn.edu/~suw}

\author{Tongyang Li}
\email{tongyangli@pku.edu.cn}
\homepage{https://www.tongyangli.com}
\affiliation{Center on Frontiers of Computing Studies, Peking University, 100871 Beijing, China}
\affiliation{School of Computer Science, Peking University, 100871 Beijing, China}

\maketitle


\begin{abstract}
Classical algorithms are often not effective for solving nonconvex optimization problems where local minima are separated by high barriers. In this paper, we explore possible quantum speedups for nonconvex optimization by leveraging the \emph{global} effect of quantum tunneling. Specifically, we introduce a quantum algorithm termed the quantum tunneling walk (QTW) and apply it to nonconvex problems where local minima are approximately global minima.
We show that QTW achieves quantum speedup over classical stochastic gradient descents (SGD) when the barriers between different local minima are high but thin and the minima are flat. 
Based on this observation, we construct a specific double-well landscape, where classical algorithms cannot efficiently hit one target well knowing the other well but QTW can when given proper initial states near the known well.
Finally, we corroborate our findings with numerical experiments.
\end{abstract}




\section{Introduction}\label{sec:intro}
Nonconvex optimization plays a central role in machine learning because the training of many modern machine learning models, especially those from deep learning, requires optimization of nonconvex loss functions. Among algorithms for solving nonconvex optimization problems, stochastic gradient descent (SGD) and its variants, such as Adam~\cite{ADAM}, Adagrad~\cite{duchi2011adagrad}, etc., are widely used in practice. In theory, their provable guarantee has been studied from various perspectives.

In this paper, we adopt the perspective of studying gradient descents via the analysis of their behavior in continuous-time limits as differential equations, following a recent line of work in~\cite{su2016differential,wibisono2016variational,jordan2018dynamical,shi2021understanding}. In particular, let $f\colon\R^{d}\to\R$ be the objective function constructed via all data. The SGD $x_{k+1} = x_k - s\tilde{\nabla} f(x_k)$
with learning rate $s$ and estimated gradient, $\tilde{\nabla} f$, evaluated from a mini-batch can be modeled by $x_{k+1} = x_k - s\nabla f(x_k) - s\xi_{k}$ 
with normally distributed noise $\xi_{k}$ (this is also known as the unadjusted Langevin dynamics). In the continuous-time limit, we can obtain a learning-rate-dependent stochastic differential equation (SDE), approximating the discrete algorithm:
\begin{align}\label{eq:SDE}
    \d x = -\nabla f(x)\d t + \sqrt{s}\d W,
\end{align}
where $W$ is a standard Brownian motion. Such approach enjoys clear intuition from physics. In particular, Eq.~\eq{SDE} is essentially a non-equilibrium thermodynamic process: gradient descent provides driving forces, the stochastic term serves as thermal motions, and a combination of these two ingredients enables convergence to the thermal distribution, also known as the Gibbs distribution. A systematic study of Eq.~\eq{SDE} was conducted in a recent work by~\cite{SSJ20}. See more details in \sec{classicalpre}.

Nevertheless, algorithms based on gradient descents also have limitations because they only have access to local information about the function, which suffers from fundamental difficulties when facing landscapes with intricate local structures such as vanishing gradient~\cite{hochreiter1998vanishing}, nonsmoothness~\cite{KS21}, negative curvature~\cite{CB21}, etc. In terms of optimization, we are mostly interested in points with zero gradients, and they can be categorized as \emph{saddle points}, \emph{local optima}, and \emph{global optima}. It is known that variants of SGD can escape from saddle points \cite{ge2015escaping,jin2017escape,allen2018neon2,fang2018spider,fang2019sharp,jin2019stochastic,zhang2021escape}, but one of the most prominent issues in nonconvex optimization is to escape from local minima and reach global minima. Up to now, theoretical guarantee of escaping from local minima by SGD has only been known for some special nonconvex functions~\cite{kleinberg2018alternative}. In general, SGD has to climb through high barriers in landscapes to reach global minima, and this is typically intractable using only gradients that descend the function. In all, fundamentally different ideas, especially those that explores beyond local information, are expected to derive better algorithms for nonconvex optimization in general.

This paper aims to study nonconvex optimization via dynamics from \emph{quantum mechanics}, which can leverage global information about a function $f\colon\R^{d}\to\R$. The fundamental rule in quantum mechanics is the \emph{Schr{\"o}dinger Equation:}\footnote{The standard Schr{\"o}dinger Equation in quantum mechanics is typically written as $i \hbar\frac{\partial}{\partial t} \Phi = \left(-\frac{\hbar^2}{2m}\Delta + f(x)\right)\Phi$. In this paper, we use the form in \eq{Schrodingereq} by setting the Planck constant $\hbar = 1$ and $h = \hbar/\sqrt{2m}$ which is a variable. See also \sec{quantumpre}.}
\begin{align}
    i \frac{\partial}{\partial t} \Phi = \left(-h^2\Delta + f(x)\right) \Phi,
    \label{eq:Schrodingereq}
\end{align}
where $i$ is the imaginary unit, $h$ is defined as the \emph{quantum learning rate}, $\Delta=\sum_{j=1}^{d}\frac{\partial^2}{\partial x_{j}^2}$ is the Laplacian, and $\Phi(t,x)\colon\R\times\R^{d}\to\C$ is a quantum wave function satisfying $\int_{\R^{d}}|\Phi(t,x)|^2\d x=1$ for any $t$.
Measuring the wave function at time $t$, $|\Phi(t,x)|^2$ is the probability density of finding the particle at position $x$.
In Eq.~\eq{Schrodingereq}, the time evolution of wave functions is governed by the Hamiltonian\footnote{In this paper, we refer Hamiltonian to either the total energy of a system or the operator corresponding to the total energy of the system, depending on the context.} $H:=-h^2\Delta + f$, where $-h^2\Delta$ corresponds to the classical kinetic energy and $f$ the potential energy.

In sharp contrast to classical particles, quantum wave functions can tunnel through high potential barriers with significant probability, and this is formally known as \emph{quantum tunneling.} Take a one-dimensional double-well potential $f\colon\R\to\R$ in \fig{1-dimexample} as an example, the goal is to move from the local minimum $x_{-}$ in the left region to the local minimum $x_{+}$ in the right region. Classically, the SDE in \eq{SDE} has to climb through the barrier with height $H_{f}$, and it can take $\exp(\Theta(H_{f}/s))$ time to reach $x_{+}$ (see Section 3.4 of \cite{SSJ20}).

Quantumly, we denote $\Phi_-(x)$ and $\Phi_+(x)$ to be the ground state (i.e., the eigenstate corresponds to the smallest eigenvalue) of the left and the right region, respectively. These states $\Phi_\pm(x)$ are localized near $x_{\pm}$, respectively. We let the wave function be initialized at $\Phi_-(x)$, i.e., $\Phi(0,x)=\Phi_-(x)$.
Under proper conditions, two eigenfunctions with eigenvalues $E_0$ and $E_1$ of $H=-h^2\Delta + f$ can be represented by superposition states
\begin{align}
\Phi_0(x):=(\Phi_+(x) + \Phi_-(x))/\sqrt{2},\\
\Phi_1(x):=(\Phi_+(x) - \Phi_-(x))/\sqrt{2},
\end{align}
respectively. Note that $\Phi_0(x)$ and $\Phi_1(x)$ are \emph{not} localized because they have probability $1/2$ of reaching both $x_{+}$ and $x_{-}$. Specifically, given $\Phi(0,x) = \Phi_-(x) = (\Phi_0(x) + \Phi_1(x))/\sqrt{2}$,
and because the dynamics of the Schr\"odinger equation \eq{Schrodingereq} is $\Phi(t,x) = e^{-iHt}\Phi(0,x)$, we have
\begin{align}
\Phi(t,x) = (e^{-iE_0t}\Phi_0(x) + e^{-iE_1t}\Phi_1(x))/\sqrt{2}.
\end{align}
As a result, after time $t$ where $|E_0-E_1|t =\pi$, we have $\Phi(t,x)\propto \Phi_+(x)$ localized near $x_{+}$. Intuitively, this can be viewed as global evolution and superposition of quantum states, which is capable of acquiring global information of the function $f$ and explains why for various choices of $f$, this quantum evolution time $t$ is much shorter than the classical counterpart by SDE which only takes gradients locally.

\begin{figure}
\centering
\includegraphics[width=0.5\textwidth]{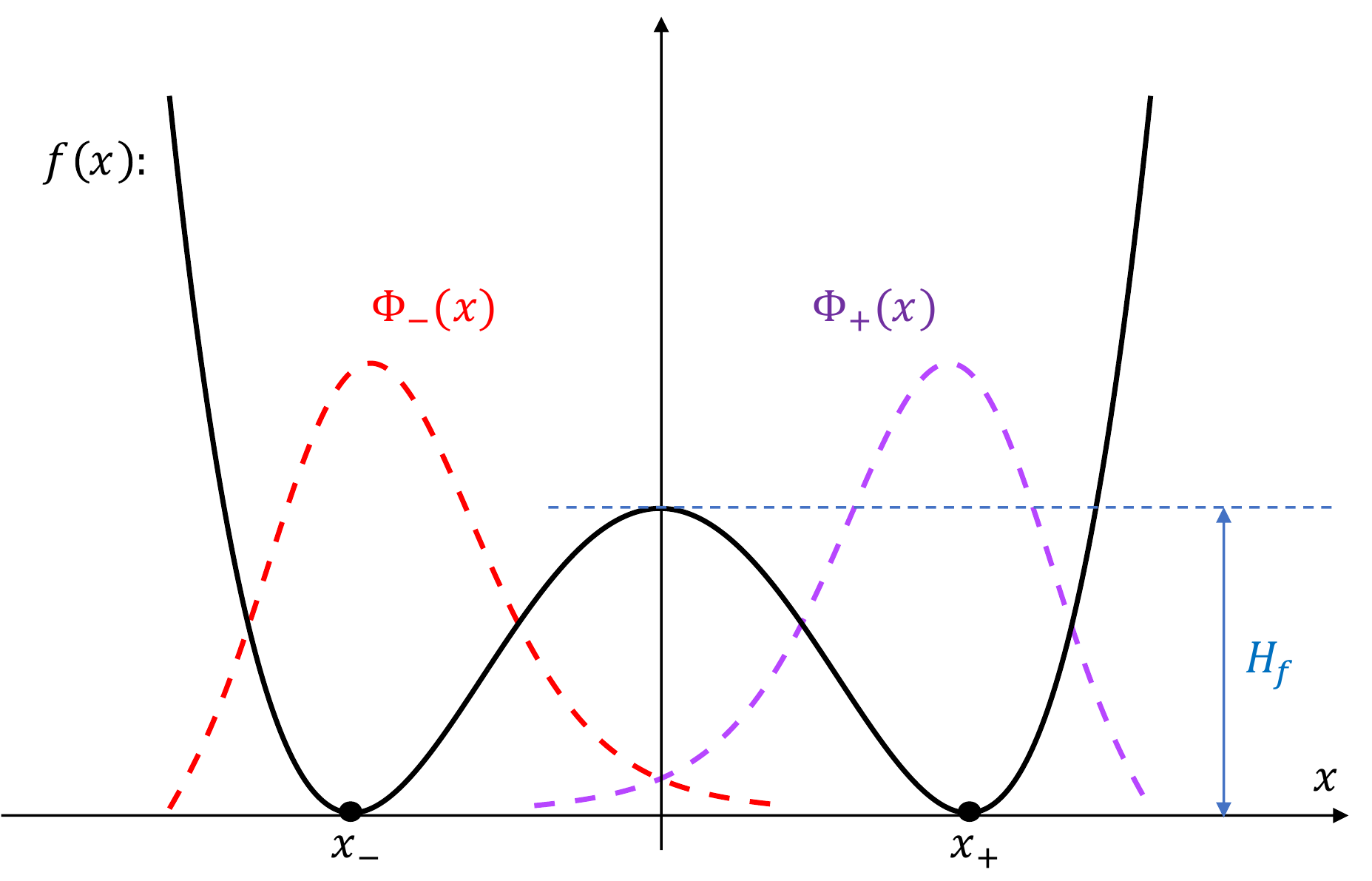}
\caption{An example of quantum tunneling in a double-well optimization function.}
\label{fig:1-dimexample}
\end{figure}

It is a natural intuition to design quantum algorithms using quantum tunneling. Previously, \cite{FGS94,MAL16,CH16,BZ18} studied the phenomenon of quantum tunneling in quantum annealing algorithms \cite{FGS94,FGG01}. However, most of these results studied Boolean functions, which is essentially different from continuous optimization. In addition, quantum annealing focused on ground state preparation instead of the dynamics for quantum tunneling. Up to now, it is in general unclear when we can design  quantum algorithms for optimization by adopting quantum tunneling. Therefore, we ask:

\begin{question}\label{ques:configuration}
On what kind of landscapes can we design algorithms efficiently using quantum tunneling?
\end{question}

To answer this question, we need to figure out specifications of the quantum algorithm, such as the initialization of the quantum wave packet, the landscape's parameters, the measurement strategy, etc.

The next question is to understand the advantage of quantum algorithms based on quantum tunneling. A main reason of studying quantum computing is because it can solve various problems with significant speedup compared to classical state-of-the-art algorithms. In optimization, prior quantum algorithms have been devoted to semidefinite programs~\cite{brandao2016quantum,vanApeldoorn2017quantum,vanApeldoorn2018SDP,brandao2017SDP}, convex optimization~\cite{vanApeldoorn2020optimization,chakrabarti2020optimization}, escaping from saddle points~\cite{zhang2021quantum}, polynomial optimization~\cite{rebentrost2019quantum,li2021optimizing}, finding negative curvature directions~\cite{zhang2019quantum}, etc., but quantum algorithms for nonconvex optimization with provable guarantee in general is widely open as far as we know. Here we ask:

\begin{question}\label{ques:accelaration}
When do algorithms based on quantum tunneling give rise to quantum speedups?
\end{question}

\paragraph{Contributions.}
We systematically study quantum algorithms based on quantum tunneling for a wide range of nonconvex optimization problems. Throughout the paper, we consider benign nonconvex landscapes where \emph{local minima are (approximately) global minima.} We point out that many common nonconvex optimization problems indeed yield objective functions satisfying such benign behaviors, such as tensor decomposition~\cite{ge2015escaping,GM20}, matrix completion~\cite{ge2016matrix,ma2018implicit}, and dictionary learning~\cite{qu2019analysis}, etc. In general, nonconvex problems with discrete symmetry satisfy this assumption, see the surveys by~\cite{Ma21,ZQW21}.

In this paper, we demonstrate the power of quantum computing for the following main problem:
\begin{named}{Main Problem}\label{prb:main}
On a landscape whose local minima are (approximately) global minima, starting from one local minimum, find all local minima with similar function values or find a certain target minimum.
\end{named}
Such a problem is crucial for understanding the \emph{generalization} property of nonconvex landscapes, and in general it also sheds light on nonconvex optimization. First, local minima with similar function values can have dramatically different generalization performance (see Section 6.2.3 of \cite{Sun19}), and solving this \ref{prb:main} can be viewed as a subsequent step of optimization for finding the minimum which generalizes the best. Second, \ref{prb:main} implies the mode connectivity of landscapes, which has been applied to understanding the loss surfaces of various machine learning models including neural networks both empirically~\cite{draxler2018essentially,garipov2018loss} and theoretically~\cite{kuditipudi2019explaining,nguyen2019connected,shevchenko2020landscape}. Third, nonconvex landscapes where the \ref{prb:main} can be efficiently solved can also lead to efficient Monte Carlo sampling, which can be even faster than optimization~\cite{ma2019sampling,talwar2019computational}.

Landscapes whose local minima are (approximately) global significantly facilitate quantum tunneling.
Roughly speaking, since the total energy during our quantum evolution \eq{Schrodingereq} is conserved, quantum tunneling can only efficiently send a state from one minimum to another minimum with similar values. As a conclusion, if the quantum wave function is initialized near a local minimum, we can focus on quantum tunneling between the local ground state of each well, i.e., the tunneling of the particle from the bottom of a well to that of another well. To avoid complicated discussions on the value of the quantum learning rate $h$, we further restrict ourselves to functions whose local minima are global, which would not provide less intuition. Now, an answer to \ques{configuration} can be given as follows:

\begin{theorem}[Quantum tunneling walks, informal]\label{thm:informalQTW}
On landscapes whose local minima are global minima, we have an algorithm called quantum tunneling walks (QTW) which 
initiates the simulation of Eq.~\eq{Schrodingereq} from the local ground state at a minimum, and
measures the position at a time which is chosen uniformly from $[0,\tau]$.
To solve the \ref{prb:main} we can take
\begin{align}\label{eq:tau-informal}
\tau=O(\mathrm{poly}(N)/\Delta E),
\end{align}
where $N$ is the number of
global minima and $\Delta E$ is the minimal spectral gap of the Hamiltonian restricted in a low-energy subspace.
For sufficiently small $h$, we have
\begin{align}\label{eq:QTW}
    \Delta E = \sqrt{h}(b + O(h))e^{-\frac{S_0}{h}},
\end{align}
where $b, S_0>0$ are constants that depend only on $f$.
\end{theorem}

Formal description of the QTW can be found in \sec{QTW}. Here we highlight two important properties of QTW: Quantum mixing time and quantum hitting time.

\emph{Quantum mixing time (\lem{qtwmixingtime} in \sec{Qmixing}).}
Since quantum evolutions are unitary, QTW never converges, a fundamental distinction from SGD. Therefore, to study the mixing properties of QTW, we follow quantum walk literature~\cite{CCD+03} by employing the measurement strategy, where we measure at $t$ uniformly chosen
from $[0,\tau]$. The measured results obey a distribution which is a function of $\tau$, and when $\tau\to +\infty$, the distribution tends to its limit, $\mu_{\rm QTW}$. 
Quantum mixing time is the minimal $\tau$ enabling us to sample from $\mu_{\rm QTW}$ up to some small error. Alternatively speaking, the mixing time evaluates how fast the distribution yielded by QTW converges.
We prove that $\mu_{\rm QTW}$ concentrates near minima, so that
sampling from $\mu_{\rm QTW}$ repeatedly can give positions of all minima. In addition, $\mu_{\rm QTW}$ gives the upper bound on $\tau$ in \eq{tau-informal}.

\emph{Quantum hitting time (\lem{qtwhitting} in \sec{Qhitting})}.
Hitting time is the duration it takes to hit a target region (usually a neighborhood of some minimum).
Quantum hitting time is the minimum evolution time needed for hitting the region of interest once. Despite this straightforward intuition, the formal definition of
quantum hitting time is very different from that of classical hitting time.
Intuitively, repeatedly sampling from $\mu_{\rm QTW}$ can ensure the hitting
of neighborhoods of particular minima, and thus we can use the mixing time to bound the hitting time.
In short, to solve the \ref{prb:main}, we bound the quantum mixing and hitting time to obtain \thm{informalQTW}.

The minimal spectral gap $\Delta E$ in \thm{informalQTW} is calculated in \append{interactionmatrix}.
The quantity $S_0$ is called the \emph{minimal Agmon distance} between different wells, formally defined in \defn{Agmon-distance},
which is related to both the height and width of potential barriers.
The smaller $h$ is, the closer the measured results are to the minima (i.e., the more accurate QTW is), but the longer evolution time the Schr{\"o}dinger equation takes.

As an application of \thm{informalQTW} and a justification of the practicability of QTW, we show how to use QTW to solve the orthogonal tensor decomposition problem.
This problem asks to find all orthogonal components of a tensor. After transforming into a single optimization problem \cite{CLX+09,Hyv99},
the aim is to find all global minima. We present below a bound on the time cost of QTW on decomposing fourth-order tensors and details can be found in \sec{tendecom}.
\begin{proposition}[Tensor decomposition, informal version of \prop{TenDecTtot}]\label{prop:TenDecTtot-main}
Let $d$ be the dimension of the components of the fourth-order tensor $T \in \mathbb{R}^{d^4}$ satisfying \eq{4-tensor}, $\delta$ be the expected risk yielded by the limit distribution $\mu_{\rm QTW}$, and $\epsilon$ be the
maximum error between $\mu_{\rm QTW}$ and the actual obtained distribution (quantified by $L^1$ norm).
For sufficiently small $\epsilon$ and sufficiently small $\delta$, the
total time $T_{\rm tot}$ for finding all orthogonal components of $T$ by QTW satisfies
\begin{align}
    T_{\rm tot} = O(\mathrm{poly}(1/\delta, e^d, 1/\epsilon)) e^{\frac{(d-1) +o_{\delta}(1)}{2\delta}}.
\end{align}
\end{proposition}

Next, we explore the advantages of the quantum tunneling mechanism comparing QTW with SGD and shown by describing landscapes where QTW outperforms SGD.
The time cost for SGD to converge to global minima is loosely $O(1/\lambda_{s})$
and
\begin{align}\label{eq:SGD}
    \lambda_{s} = (a + o(s))e^{-\frac{2H_f}{s}}
\end{align}
by~\cite{SSJ20}. Here, $s$ is the step size or learning rate of SGD. The constants $a>0$ and $H_f>0$ depend only on $f$. Interestingly, running time of QTW and that of SGD have similar form. In \eq{QTW} and \eq{SGD}, there are exponential terms $e^{S_0/h}$ and $e^{2H_f/s}$, respectively. Intuitively, the quantity $H_f$ is the characteristic height of potential barriers, and the quantity $S_0$ depends on not only the height but also the width of potential barriers. For the one-dimensional example in \fig{1-dimexample},
\begin{align}
H_{f}=\max_{\xi\in[x_{-},x_{+}]}f(\xi),\qquad S_{0}=\int_{x_{-}}^{x_{+}} \sqrt{f(\xi)} \d\xi.
\end{align}
(Proof details are given in \sec{onedimexp}.)
Other terms in the bounds, $\mathrm{poly}(N)/\Delta E$ and $1/\lambda_{s}$, are referred to as polynomial coefficients. We make the following comparisons:
\begin{figure}
  \centerline{
  \includegraphics[width=0.94\textwidth]{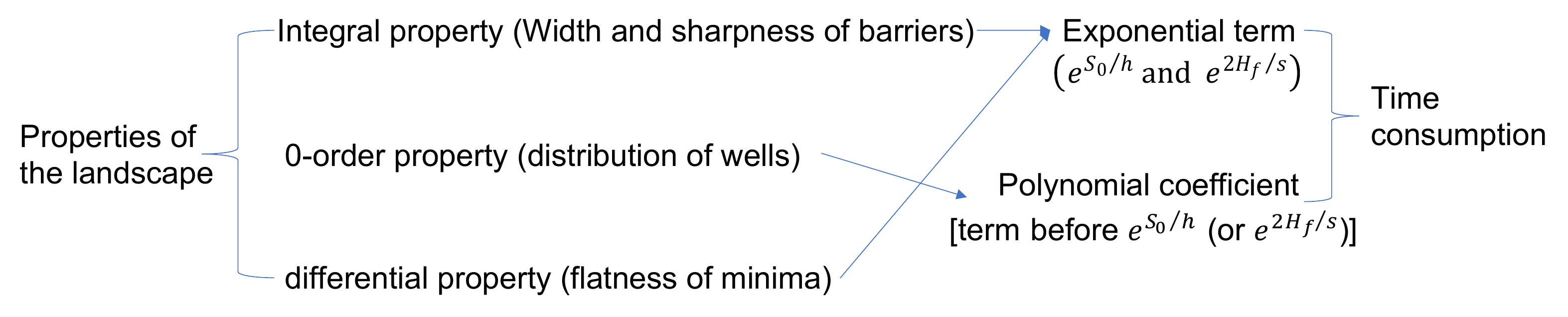}}
    \caption{Flowchart of the \slog.
    }
\label{fig:abstract1}
\end{figure}

\begin{itemize}
\item Regarding the exponential terms $S_0$ and $H_f$, tall barriers means that $H_f$ is large, whereas if the barriers are thin enough, $S_0$ can still be small. This is consistent with the long-standing intuition that tall and thin barriers are easy for tunneling but difficult for climbing \cite{CH16}.
\item Regarding the polynomial coefficients, they are mainly influenced by the distribution or relative positions of the wells. We observe that a symmetric distribution of wells, which can make (the local ground state in) any one well interacts with (the local ground states in) other wells, may reduce the running time of QTW but has no explicit impact on SGD.
\item Flatness of wells is another important factor that influences the running time of both QTW and SGD. We propose standards for comparison (see \sec{standard}), which studies their running time when reaching the same accuracy $\delta$. Same to the effect of $h$, a smaller learning rate $s$ permits more accurate outputs but makes SGD more time consuming. For sufficiently flat minima, $h$ is larger than $s$, leading to a smaller running time for QTW.
\end{itemize}

In summary, we illustrate above observations in \fig{abstract1} and conclude the following:
\begin{mainmessage}[Advantages of the quantum tunneling mechanism, a summary of \sec{standard} and \sec{illustration}]\label{slog:1}
\emph{On landscapes whose local minima are global minima, QTW outperforms SGD on solving the \ref{prb:main} if
barriers of the landscape $f$ is high but thin, wells are distributed symmetrically, and global minima are flat.}
\end{mainmessage}
\begin{remark}
As is indicated above, we compare the costs of QTW and SGD under the same accuracy $\delta$. We introduce two definitions of  accuracy in \sec{standard}: \stand{risk} concerns the expected risk, and \stand{distance} concerns the expected distance to some minima. Mathematically, \stand{risk} and \stand{distance} establish a relationship between the quantum and classical learning rates $h$ and $s$, respectively, enabling direct comparisons.
\end{remark}

Having introduced the general performance of the quantum tunneling walk, we further investigate \ques{accelaration} on some specific scenarios of the \ref{prb:main}. We focus on comparison between query complexities, namely the classical query complexity to local information and the quantum query complexity to the evaluation oracle\footnote{Query complexity of QTW is directly linked to the evolution time, in particular, evolving QTW for time $t$ needs $\tilde{O}(t)$ queries to $U_f$ (see details in \append{quantumsimulation}). As a result, it suffices to analyze the evolution time of QTW. Nevertheless, we state the query complexities for direct comparison.}
\begin{align}\label{eqn:quantumquery}
\hspace{-1.5mm}U_f (\ket{x}\otimes \ket{z}) = \ket{x}\otimes\ket{f(x)+z}\ \ \forall x\in\R^{d},z\in\R.
\end{align}
This is the standard assumption in existing literature on quantum optimization algorithms~\cite{vanApeldoorn2020optimization,chakrabarti2020optimization,zhang2021quantum}.
Different from classical queries that only learn local information of the landscape of $f$, quantum evaluation queries are essentially nonlocal as they can extract information of $f$ at different locations in superposition. Based on this fundamental difference, we are able to prove that QTW can solve a variant of the \ref{prb:main} with exponentially fewer queries than any classical counterparts:

\begin{theorem}\label{thm:provableinformal}
For any dimension $d$, there exists a landscape $f\colon \mathbb{R}^d \to \mathbb{R}$ such that its local minima are global minima, and on which, with high probability, QTW can hit the neighborhood of an unknown global minimum
from the local ground state associated to a known minimum using queries polynomial in $d$, while no
classical algorithm knowing the same minimum can hit the same target region with queries subexponential in $d$.
\end{theorem}
\begin{figure}
  \centerline{
  \includegraphics[width=0.83\textwidth]{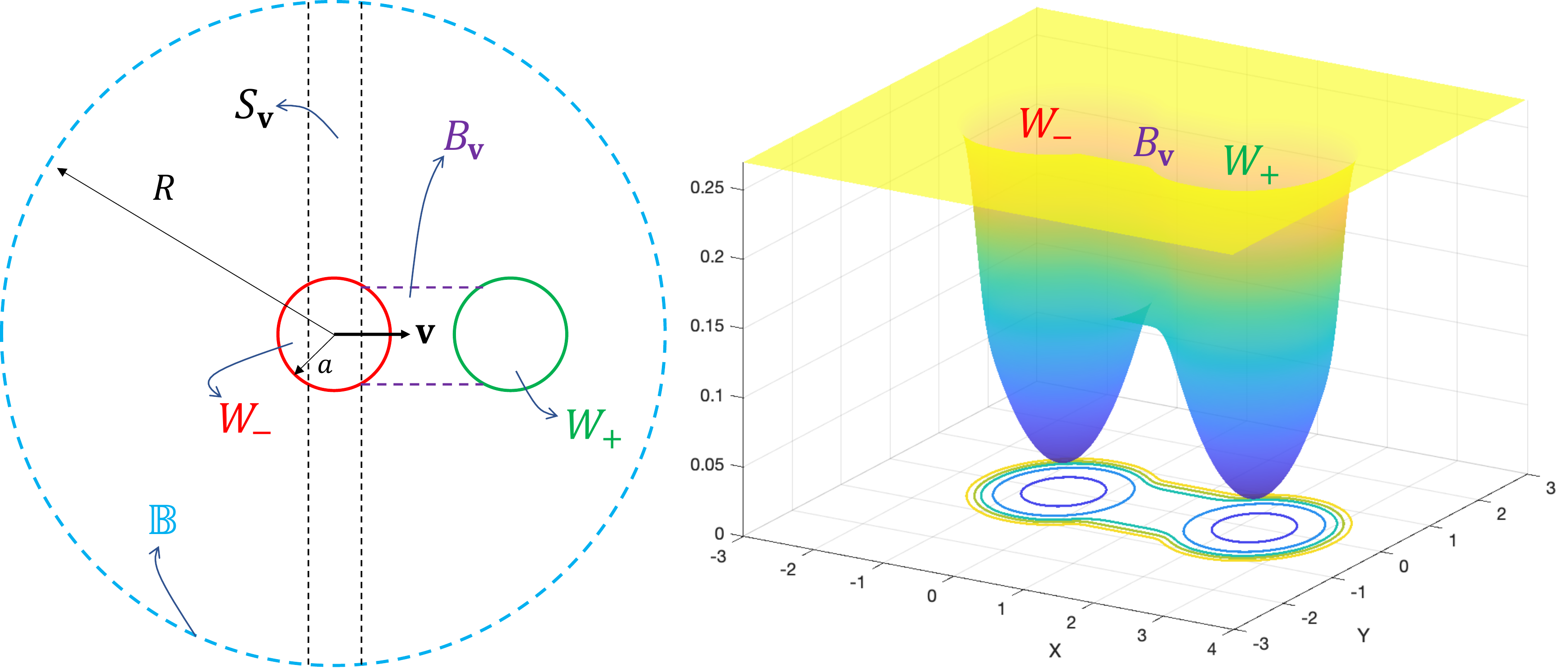}}
    \caption{Sketch of the function in \thm{provableinformal}. The left figure explains the construction in the domain and the right figure plots the landscape of a two-dimensional example.
    }
\label{fig:provable-acceleration}
\end{figure}

Details of \thm{provableinformal} are presented in \sec{separation}.
Following similar idea to~\cite{JLG+18}, our construction
relies on locally non-informative regions.
Main structures of the constructed landscape are illustrated
in \fig{provable-acceleration}, which has two global minima. $W_-$ and $W_+$ are two symmetric wells containing one global minimum respectively, $B_{\mathbf{v}}$ is a plateau connecting $W_-$ and $W_+$, and other places form a much higher plateau. The region $W_+$ is our target.
We show that the landscape satisfies the following properties:
\begin{itemize}
    \item $S_{\mathbf{v}}$, which is a band $\{\bf{x}\mid\bf{x}\in \mathbb{B}, |\bf{x}\cdot \bf{v}|\leq w\}$ with $w$ a constant and $\bf{v}$ a unit vector, occupies dominating measure in the ball $\mathbb{B}$.
    \item In $S_{\mathbf{v}}$, local queries (see \defn{localquery}) do not reveal information about the direction $\mathbf{v}$.
    \item Local queries outside $\mathbb{B}$ do not reveal information about the region inside $\mathbb{B}$.
\end{itemize}
Restricted in the ball $\mathbb{B}$, the first two properties make classical algorithms intractable to escape from $S_{\mathbf{v}}$ and thus cannot hit $W_+$ efficiently.
The last property ensures that, without being restricted in $\mathbb{B}$, classical algorithms are still unable to hit $W_+$ efficiently. See \sec{clb} for details.

Nevertheless, quantum tunneling can be efficient if we carefully design the function values and the parameter $h$.
The design of the parameters should establish the following
main conditions (See \sec{qub} for details):
\begin{itemize}
    \item The wave function always concentrates in $W_-$ or $W_+$.
    \item The quantum learning rate $h$ is small such that our theory based on semi-classical analysis is valid.
    \item Quantum tunneling from $W_-$ to $W_+$ is always easy (can happen within time polynomial in $d$).
\end{itemize}

\paragraph{Organization.}
\sec{prelim} introduces our assumptions and problem settings, both classical and quantum.
In \sec{quantumal}, we explore QTW in details and state the formal version of \thm{informalQTW}. This includes a one-dimensional example, the formal definition of QTW, the mixing and hitting time of QTW, and the example on tensor decomposition (\prop{TenDecTtot-main}).
\sec{comparison} covers detailed quantum-classical comparisons. First, we introduce fair criteria of the comparison. Second, we illustrate the advantages of quantum tunneling and give a detailed view of our \slog. Third, we prove \thm{provableinformal}.
We corroborate our findings with numerical experiments
in \sec{num}. At last, the paper is concluded with discussions in
\sec{discussion}.


\section{Preliminaries}\label{sec:prelim}
\subsection{Notations}
Throughout this paper, the space we consider is either $\mathbb{R}^d$ or a $d$-dimensional smooth compact Riemannian manifold denoted $M$.
Bold lower-case letters $\mathbf{x}$, $\mathbf{y}$,\ldots, are used to denote vectors.
If there is no ambiguity, we use normal lower-case letters, $x$, $y$,\ldots, to denote these vectors for simplicity.
Depending on the context, $\d x$ may refer to either line differential or volume differential.
We use $A_{jj'}$ to denote the element of the matrix $A$ at of row $j$ and column $j'$. Conversely, given all matrix elements $A_{jj'}$, we use the notation $(A_{jj'})$ to denote the matrix.
Unless otherwise specified,
$\| \cdot\|$ is used to denote the $\ell^2$ norm of vectors, spectral norm of matrices, and $L^2$ norm of functions.
Similarly, $\| \cdot\|_1$ is used to denote the $\ell^1$ norm of vectors and $L^1$ norm of functions.

For a function $f$, $\nabla f$ and $\nabla^2 f$ denote the gradient vector and
Hessian matrix, respectively. $\Delta:=\sum_{j=1}^{d}\frac{\partial^2}{\partial x_{j}^2}$ is the Laplacian operator. $C^{\infty}(\mathbb{R}^d)$ is the set of all functions $f\colon\R^{d}\to\R$ that are continuous and differentiable up to any order.
Notations about upper and lower bounds, $O(\cdot)$, $o(\cdot)$, $\Omega(\cdot)$, and $\Theta(\cdot)$, follow common definitions. We also write $f\ll g$ if $f=o(g)$, and $f\sim g$ if $f=\Theta(g)$.
The $\tilde{O}$ notation omits poly-logarithmic terms, namely, $\tilde{O}(f):=O(f\mathrm{poly}(\log f))$ (in this paper, $\log$ denotes the logarithm with base $2$ and $\ln$ denotes the natural logarithm with base $e$).
We write $f=O(g^{\infty})$ if
\begin{align}\label{eqn:infty-power}
\forall N>0,\quad f/g^{N} \to 0~(g\to 0).
\end{align}
Throughout the paper, we write $f\approx g$ if
\begin{align}\label{eqn:approx-definition}
f(x)=g+o(g)
\end{align}
when $g\neq 0$. When $g\to 0$, $f\to 0$ means $f=o(1)$ or, to stress the dependence, $f=o_g(1)$.

For quantum mechanics,
we use the \emph{Dirac notation} throughout the paper. Quantum states are vectors from a Hilbert space with unit norm.  Let $\ket{\phi}$ denote a state vector, and $\bra{\phi}=(\ket{\phi})^{\dagger}$ denote the dual vector that equals to its conjugate transpose. The inner product of two states can be written as $\ip{\psi}{\phi}$.
In the coordinate representation, for each $x\in\R^{d}$ 
we have the wave function $\phi(x):=\ip{x}{\phi}$, where $\ket{x}$ denotes the state localized at $x$.
More basics on quantum mechanics and quantum computing can be found in standard textbooks, for instance~\cite{NC10}.

\subsection{Classical preparations}\label{sec:classicalpre}
Classical algorithms only have access to local information about the objective function at different sites, which is formalized as follows:

\begin{definition}[Algorithms based on local queries]\label{defn:localquery}
Denote a sequence of points and corresponding queries with size $T$
by $\{x_i, q(x_i)\}_{i=1}^T$, where each $q(x_i)$ can include the function value and arbitrary order derivatives (if exist). 
Algorithms based on local queries are those which determine the $j$th point $x_j$ by $\{x_i, q(x_i)\}_{i=1}^{j-1}$.
\end{definition}

As an example, the classical algorithm SGD can be mathematically described by
\begin{definition}[Discrete model of SGD]
Given a function $f(x)$, starting from an initial point $x_0$, discrete SGD iterations are modeled by the following:
\begin{align}
    x_{k+1} = x_k - s\nabla f(x_k) - s\xi_{k},
\end{align}
where $s$ is the \emph{learning rate} and $\xi_k$ is the noise term at the $k$th step.
\end{definition}
The local information in SGD is gradients.
Since $s$ is small, define time $t_k = ks$, the points $\{x_k\}$ can be approximated by points on a smooth curve $\{X(t_k)\}$.
The curve, which can be regarded as the continuous-time limit of discrete SGD,
is determined by a learning-rate-dependent stochastic differential equation
(lr-dependent SDE):
\begin{definition}[SDE approximation of SGD]
\begin{align}
    \d x = -\nabla f(x)\d t + \sqrt{s}\d W,
    \label{eq:lrsde}
\end{align}
where $W$ is a standard Brownian motion.
\end{definition}
The solution of \eq{lrsde}, $X(t)$, is a stochastic process whose probability density $\rho_{\rm SGD}(t,\cdot)$ evolves according to the Fokker–Planck–Smoluchowski equation
\begin{align}
    \frac{\partial \rho_{\rm SGD}}{ \partial t} = \nabla\cdot (\rho_{\rm SGD} \nabla f) + \frac{s}{2} \Delta \rho_{\rm SGD}.
    \label{eq:FPS}
\end{align}
The validity of this SDE approximation has been discussed and verified in previous literature \cite{KY03,CS18,SSJ20,LMA21}.

The results used in the present paper about SGD are based on analyses on \eq{lrsde}.
For SGD, we consider an objective function $f$ in $\mathbb{R}^d$ and assume the following:
\begin{assumption}[Confining condition~\cite{markowich1999trend,pavliotis2014stochastic}]\label{assum:confining}
The objective function $f \in C^{\infty}(\mathbb{R}^d)$ should satisfy
$\lim_{\|x\|\to +\infty } f(x) = +\infty$, and $\forall s>0,\exp(-2f/s)$ is integrable:
\begin{align}
    \int_{\mathbb{R}^d} e^{-\frac{2f(x)}{s}} \d x < \infty.
\end{align}
\end{assumption}
\begin{assumption}[Villani condition~\cite{villani2009hypocoercivity}]\label{assum:Villani}
The following equation holds for all $s>0$:
\begin{align}
    \|\nabla f(x) \|^2 - s\Delta f(x) \to \infty \quad(\|x\|\to \infty).
\end{align}
\end{assumption}
\begin{assumption}[Morse function]\label{assum:morse}
For any critical point $x$ of $f$ (i.e., $\nabla f(x) = 0$), the Hessian matrix $\nabla^2 f(x)$  is nondegenerate (i.e., all the eigenvalues of the Hessian are nonzero).
\end{assumption}
\begin{remark}[Justification of assumptions]
    \assum{confining} is mild, essentially requiring that the function grows sufficiently rapidly when $x$ is far from the origin. Adding $\ell_2$ regularization term to the objective function, or equivalently, employing weight decay in the SGD update can make the condition satisfied.
    \assum{Villani}, giving discrete spectrum of Witten-Laplacian \eq{defnWL}, demands that the gradient has a sufficiently large squared norm compared with the Laplacian of the function. Some loss functions used for training neural networks might not satisfy this condition. However, the Villani condition is not stringent in practice since the SGD iterations are bounded while this condition essentially concerns about the function at infinity. \assum{morse} is made for calculations in spectrum analysis, which is generally not restrictive.
\end{remark}

Under \assum{confining}, Eq.~\eq{FPS} admits a unique invariant Gibbs distribution
\begin{align}
    \mu_{\rm SGD}(x) := \frac{e^{-\frac{2f(x)}{s}}}{\int_{\mathbb{R}^d}e^{-\frac{2f(x)}{s}} \d x}.
    \label{eq:Gibbs}
\end{align}
\begin{definition}
A measurable function $g$ belongs to $L^2(\mu_{\rm SGD}^{-1})$, if
\begin{align}
    \| g\|_{\mu_{\rm SGD}^{-1}} := \Big(\int_{\mathbb{R}^d} g^2 \mu_{\rm SGD}^{-1}(x) \d x\Big)^{\frac{1}{2}} < +\infty.
\end{align}
\end{definition}

Under such measure, we have:
\begin{lemma}[Lemma 2.2 and 5.2 of \cite{SSJ20}]
Under \assum{confining}, if the initial distribution $\rho_{\rm SGD}(0,\cdot)\in L^2(\mu_{\rm SGD}^{-1})$, the lr-dependent SDE \eq{lrsde} admits a weak solution 
whose probability density
\begin{align}
    \rho_{\rm SGD}(t,\cdot)
\in C^1([0,+\infty),
L^2(\mu_{\rm SGD}^{-1})),
\end{align}
is the unique solution to \eq{FPS} and $\rho_{\rm SGD}(t,\cdot) \to \mu_{\rm SGD}~(t \to \infty)$.
\end{lemma}

If we set $\psi_{\rm SGD}(t,\cdot) := \rho_{\rm SGD}(t,\cdot)/\sqrt{\mu_{\rm SGD}} \in L^2(\mathbb{R}^d)$, Eq.~\eq{FPS} is equivalent to
\begin{align}
    s\frac{\partial \psi_{\rm SGD}}{ \partial t} = -\frac{\Delta_{f}^s}{2}\psi_{\rm SGD},
    \label{eq:pseudo}
\end{align}
where $\Delta_{f}^s$ is called the \emph{Witten-Laplacian}, more specifically,
\begin{align}
    \Delta_{f}^s:= -s^2 \Delta  + \| \nabla f\|^2 - s\Delta f.
    \label{eq:defnWL}
\end{align}
Let $\delta_{s,1}$ be the smallest non-zero eigenvalue of $\Delta_{f}^s$,
the following convergence guarantee for SGD holds:
\begin{proposition}[Part of Theorem 2.8 in \cite{Mic19} and Lemma 5.5 in \cite{SSJ20}]\label{prop:mic19}
Under \assum{confining}, \ref{assum:Villani}, and \ref{assum:morse}, for sufficiently small $s$,
\begin{align}
\|\rho_{\rm SGD} (t,\cdot) - \mu_{\rm SGD}\|_{\mu_{\rm SGD}^{-1}} \leq
e^{-\frac{\delta_{s,1}}{2s}t}\|\rho_{\rm SGD} (0,\cdot) - \mu_{\rm SGD}\|_{\mu_{\rm SGD}^{-1}},
\end{align}
where
the smallest positive eigenvalues of the Witten-Laplacian $\Delta_{f}^s$ associated with $f$ satisfies
\begin{align}
    \delta_{s,1} = s(\gamma_1 + o(s))e^{-\frac{2H_f}{s}}.
\end{align}
Here, $H_f$ and $\gamma_1$ are constants depending only on the function $f$.
\end{proposition}
\begin{corollary}\label{cor:sgdmixing}
Assume the assumptions of \prop{mic19} are satisfied,
for sufficiently small $s$ and any $\epsilon>0$,
if
\begin{align}
    t> \frac{2s}{\delta_{s,1}}\ln \frac{\|\rho_{\rm SGD} (0,\cdot) - \mu_{\rm SGD}\|_{\mu_{\rm SGD}^{-1}}}{\epsilon},
\end{align}
then
\begin{align}
    \|\rho_{\rm SGD} (t,\cdot) - \mu_{\rm SGD}\|_{\mu_{\rm SGD}^{-1}}<\epsilon.
\end{align}
\end{corollary}
That is, the convergence time of SGD is loosely $O(s/\delta_{s,1})$ whose magnitude is largely related to $H_f$.
The constant $H_f$ is called the \emph{Morse saddle barrier}, characterizing the largest height of barriers.
Rigorous results about eigenvalues of the Witten-Laplacian are reviewed in \append{W-L}.

\subsection{Quantum preparations}\label{sec:quantumpre}


Our quantum algorithmic method essentially relies on the simulation of the Schr\"odinger equation:
\begin{align}
    i \frac{\partial}{\partial t} \ket{\Phi(t)}= H  \ket{\Phi(t)},
    \label{eq:true}
\end{align}
where $i$ is the imaginary unit, $\hbar$ is the Planck constant, and $H$ is the Hamiltonian.
Physically, we simulate a particle moving under a potential function $f$.
Then, in the coordinate representation, the Schr\"odinger equation is specified:
\begin{align}
    i \hbar\frac{\partial}{\partial t} \Phi = \Big(-\frac{\hbar^2}{2m}\Delta + f(x)\Big) \Phi.
\end{align}
Throughout this paper, we set $h = \hbar/\sqrt{2m}$.
The spectrum of the Hamiltonian will highly depends on the variable $h$. More interestingly, by comparing \eq{pseudo} and \eq{true}, $h$ plays a similar role to that of the learning rate $s$.
Thus we refer to $h$ as the quantum learning rate.
In the reality, $\hbar$ is a fixed constant.
However, since we are simulating quantum evolution by quantum computers, proper rescaling the simulation can equivalently be seen as varying the value of $\hbar$.
The value of $\hbar$ affects the evolution time needed.
However, rescaling $\hbar$ has no impact on quantum query complexity.
Therefore, in this paper, $\hbar$ is an unimportant constant, i.e., $\hbar = 1$.

As is introduced, we consider quantum tunneling from the bottom of a well
to that of another well, in other words, tunneling between \emph{local ground states}.
A local ground state of a well is the local eigenstate of the well with minimum eigenvalue.
Technically, several kinds of local eigenstates are defined (see \defn{localeigen}, \defn{localcutoff}, and \defn{ortholocal} in \append{interactionmatrix}).
Despite of the number of definitions,
different kinds of local eigenstates are close to each other and share the same intuition: eigenstates of the Hamiltonian restricted in regions only contain one well.
For convenience, if no otherwise specified, local eigenstates stand for
orthonormalized eigenstates defined by \defn{ortholocal}.
Actually, there are also tunneling effects between local excited states. However, excited states are difficult to approximate accurately for general landscapes. Besides, due to interference, tunneling effects between different local excited states may cancel each other out. We restrict our attention to tunneling between local ground states in order to obtain explicit results along with a clear physical picture.

Two local ground states can interact strongly with each other only if the difference between their energies is small relative to $h$ \cite{Hel88} (see also discussions after \prop{genPF} in \append{interactionmatrix}). In other words, this requires the function values between two local minima to be close and there is little resonance
between the first (local) excited state in one well
and the (local) ground state of the other \cite{Ras12,SCC91}.
Therefore, our algorithms based on tunneling between local ground states are essentially restricted on landscapes where local minima are approximately global minima.
Note that we can always find small enough $h$ to make two local ground states nonresonant, if the corresponding local minima are not exactly equal. As a result, to avoid more complicated restrictions on $h$, without loss of generality we assume that local minima are global minima, and they all have function value 0. More precisely: 
\begin{assumption}\label{assum:quantum1}
The smooth objective function $f\colon \mathbb{R}^{d}\to \mathbb{R}$ satisfies
\begin{equation}
    0 = \min f < \liminf_{\|x\|\to \infty} f,
\end{equation}
namely, there exists a radius $r$ such that $\inf_{\|x\|>r}f > \min f$.
In addition, $f$ has finite number of local minima, and they can be decomposed as follows:
\begin{equation}
    f^{-1}(0) = U_1 \cup U_2 \ldots \cup U_N,
\end{equation}
\begin{equation}
    U_j = \{x_j\}~\mathrm{is~a~point,}~\nabla f(x_j) = 0,~\mathrm{and}~\nabla^2 f(x_j) > 0~\mathrm{for}~j=1,\ldots,N.
\end{equation}
Each $U_j$ is called a \emph{well}.
\end{assumption}
This assumption will not affect the explicit forms of convergence time or present less physical insights. To further characterize the distance on such landscapes, an important geometric tool we use is the Agmon distance.
\begin{definition}[Agmon distance]\label{defn:Agmon-distance}
Under \assum{quantum1},
the Agmon distance $d(x,y)$ is defined as 
\begin{align}
    d(x,y) := \inf_{\gamma}\int_{\gamma} \sqrt{f(x)}\d x,
\end{align}
where $\gamma$ denotes pairwise $C^1$ paths connecting $x$ and $y$.
For a set $U$, $d(x,U) = d(U,x) := \inf_{y\in U}d(x,y)$.
And for two sets $U_1$ and $U_2$, $d(U_1,U_2) = \inf_{x\in U_1,y\in U_2}d(x,y)$.
\end{definition}
The minimal Agmon distance between wells are defined as
\begin{align}
S_0 :=\min_{j\neq k}d(U_j,U_k).
\end{align}
We only consider resonant wells by assuming the following for simplicity:
\begin{assumption}\label{assum:quantum2}
The energy (eigenvalue) difference between any two local ground states are of the order $O(h^{\infty})$. In addition,
for any well $U_j$, there exists another well $U_k$ ($k\neq j$) such that $d(U_j,U_k) = S_0$.
\end{assumption}
At last, to obtain explicit results we demand that
\begin{assumption}\label{assum:quantum3}
There are a finite number of paths of the Agmon length $S_0$ connecting $U_j$ and $U_k$ if $d(U_j,U_k) = S_0$.
\end{assumption}
\begin{remark}[Physical meaning of \assum{quantum3}]
    Agmon distance serves as an ``action" in dynamics. By the principle of least action, the paths with Agmon length $S_0$ will be real trajectories of the particle (wave) in classical limit. If there are infinite possible trajectories between two points, either the two points will be indeed the same, or the problem can be reduced to one in a lower-dimensional space.
\end{remark}

\assum{quantum1}, \ref{assum:quantum2}, and \ref{assum:quantum3} informally present \assum{TunEnCon}, \ref{assum:3}, and \ref{assum:4} in \append{theory-Schrodinger} which pay more efforts in rigorous descriptions of ``wells", ``local ground states", etc.
Under \assum{quantum1}, \ref{assum:quantum2}, and \ref{assum:quantum3},
we state the main results of \append{theory-Schrodinger} as follows.
For sufficiently small $h$,
the orthonormalized local ground states $\ket{e_j},~j=1,\ldots,N$ almost localize near the wells $U_j,~j=1,\ldots,N$, respectively.
The space $\F$ spanned by $\{\ket{e_j}:j=1,\ldots,N\}$ is exactly a low-energy invariant subspace of the Hamiltonian $H$.
In other words, in the low-energy space $\F$, the particle walks between wells by quantum tunneling.
The Hamiltonian restricted in $\F$, i.e., $H_{|\F}$, determines the strength of the quantum tunneling effect and is called the \emph{interaction matrix}.
To explore $H_{|\F}$, we use the WKB method to estimate local ground states (\append{WKB}).
Any local ground state function decays exponentially with respect to the Agmon distance to its corresponding well (\append{agmon}).
Consequently, the tunneling effects would decay exponentially with respect to $S_0$.
Having captured theses properties, \append{interactionmatrix} can give explicit estimations about $H_{|\F}$, namely, \prop{spePF}, \prop{speW} (with $N^+ = N$), and \thm{energygap}.

Finally, we restate a more formal version of our \ref{prb:main}:

\begin{named}{Main Problem (restated)}\label{prb:mainformal}
Given an objective function $f$ that satisfies \assum{quantum1}, \ref{assum:quantum2}, and \ref{assum:quantum3}, starting from one local minimum, find all local minima or find a certain target minimum.
\end{named}
We make the following remarks for clarification:
\begin{remark}
Assumptions in \sec{classicalpre} and \sec{quantumpre} are not contradictory. When considering SGD, we naturally add to the \ref{prb:mainformal} that the assumptions in \sec{classicalpre} should also be satisfied.
\end{remark}
\begin{remark}
The assumption of starting from one local minimum enables quantum algorithms to prepare a local ground state, or more generally, a state largely in the aforementioned subspace $\F$.
\end{remark}
\begin{remark}
Because finding a precise global minimum is impractical in general, it suffices to find points sufficiently close to the minima of interest. Later in \sec{standard}, we use two different measures of accuracy: 1. the function value difference; and 2. the distance to one of the minimum.
\end{remark}


\section{Quantum Tunneling Walks}\label{sec:quantumal}
In this section, we present full details of the quantum tunneling walk (QTW). We begin with a one-dimensional example in \sec{onedimexp}, and then in \sec{QTW} we formally define QTW. In \sec{Qmixing} and \sec{Qhitting}, we study the mixing and hitting time of QTW, respectively. As an example, we give full details of applying QTW to tensor decomposition in \sec{tendecom}.

\subsection{A one-dimensional example}\label{sec:onedimexp}
We start the introduction of the quantum algorithm QTW with a one-dimensional example which quantifies the
intuitions provided in \sec{intro}.
\sec{onedimexp} also serves as a map connecting each step
of the analysis to the needed mathematical tools.
Later sections can be seen as generalizing results here for high dimensional and multi-well cases.
General descriptions and results begin at \sec{QTW}.

Consider the potential $f(x)$ in \fig{1-dimexample}, which has two global minima,
$x_{\pm} = \pm a$. For simplicity, we take
\begin{equation}
    f(x) = \frac{1}{2} \omega^2 (x + a)^2,\quad x \leq -\epsilon,
    \label{eq:eq01}
\end{equation}
where $\epsilon$ is a small number. In this way, the potential satisfies that $\min f =0$, and $f(x)$ is quadratic near minima. Besides, the symmetry of wells demands $f(x) = f(-x)$. The $f(x)$ for $x\in [-\epsilon,\epsilon]$ can always be made to be smooth.

Near the two minima, $\pm a$, whose local harmonic frequencies are $\omega$, we can solve the Schr\"odinger equation locally and get two local ground states, $\Phi_{\pm}(x)$. For instance, around $-a$, if we set $y = x+ a$, the local ground state is determined by
\begin{equation}
    H \Phi_{-}(x) = \varepsilon_0 \Phi_{-}(x),\quad H = - h^2 \frac{\d^2}{\d^2 y} + \frac{1}{2} \omega^2 y^2,
\end{equation}
where $\varepsilon_0 = h\omega/\sqrt{2}$.
Physically, the demand of localization is equivalent to $\varepsilon_0 \ll f(0)$, indicating that the particle nearly cannot pass through the energy barrier. Concrete mathematical definitions and discussions on local ground states can be found in \append{agmon}.

From a high-level perspective, the main idea of the present paper is to unite the interaction or tunneling between local ground states to realize algorithmic speedups.
As we want to investigate the evolution of states, we need to determine the relationship between local ground states and the eigenstates of the Hamiltonian $H$.
We set the eigenstates of $H$ as $\ket{n},~n=0,1,\ldots$ with energies $E_0 \leq E_1 \leq \cdots$, respectively (i.e., $\ket{0}$ is the global ground state, $\ket{1}$ is the first excited state, etc.).
The overlap of states $\Phi_{\pm}$ is small (namely, $\ip{\Phi_{+}}{\Phi_{-}}\approx 0$), as they are local and separated by a high barrier.
Denote the subspace spanned by $\Phi_{-}$ and $\Phi_{+}$ as $\E$, and that spanned by $\ket{0}$ and $\ket{1}$ as $\F$. Both $\E$ and $\F$ are 2-dimensional and contains states with low energies. It is intuitive that $\E \approx \F$ (which is guaranteed by \prop{dimEdimF}). For the one-dimensional case, we just take $\E = \F$ for simplicity.
In this way, we can represent $\ket{0}$ and $\ket{1}$ by $\ket{\Phi_{\pm}}$ in the following general way:
\begin{align}
    \ket{0} &= \cos\theta \ket{\Phi_{-}} + \sin\theta \ket{\Phi_{+}},\\
    \ket{1} &= \sin\theta \ket{\Phi_{-}} - \cos\theta \ket{\Phi_{+}}.
\end{align}
Restricted in the subspace $\F$, the two-level system Hamiltonian can be written as
\begin{align}
    H|_{\F} &= \left(
    \begin{array}{cc}
        \varepsilon_{-} & -\nu  \\
         -\nu & \varepsilon_{+}
    \end{array}
    \right),\quad \mathrm{under~basis}~\{\ket{\Phi_{-}},\ket{\Phi_{+}}\},\\
     H|_{\F} &= \left(
    \begin{array}{cc}
        E_0 & 0  \\
         0 & E_1
    \end{array}
    \right),\quad\ \mathrm{under~basis}~\{\ket{0},\ket{1}\},
\end{align}
where $\nu$ is called the tunneling amplitude, measuring the interaction between wells.
Because the $f(x)$ we choose is symmetric, $\varepsilon_{-} = \varepsilon_{+} = \varepsilon_{0} = h\omega/\sqrt{2}$. Therefore, we have
$\theta  = \pi/4$ and the energy gap $\Delta E := E_1 - E_0 = 2\nu$.
We will refer to this Hamiltonian restricted in subspace $\F$ as the \emph{interaction matrix}, indicating that $H_{|\F}$ characterizes the interaction between wells.

In our setting, we can begin at a local minimum, where the local ground state is easy to prepare (see justifications in \append{initial}). Without loss of generality, let us begin the quantum simulation at the state $\Phi_{-}$, namely, setting $\ket{\Phi(0)} = \ket{\Phi_{-}}$. After evolution of time $t$, the state becomes
\begin{equation}
    \ket{\Phi(t)} = e^{-iHt}\ket{\Phi(0)} =
    \frac{1}{\sqrt{2}}(e^{-iE_0t}\ket{0} + e^{-iE_1t}\ket{1}) \propto \cos(\Delta E t/2)\ket{\Phi_{-}} + i\sin(\Delta E t/2)\ket{\Phi_{+}}.
\end{equation}
And the probabilities of finding the particle
in the right and left wells are give by
\begin{equation}
    P_{\pm}(t) = |\ip{\Phi_{\pm}}{\Phi(t)}|^2 = \frac{1\mp \cos(\Delta E t)}{2}.
\end{equation}
The energy gap $\Delta E$ is also called the Rabi oscillation frequency, suggesting that the particle oscillates between the two wells periodically. Our aim is to pass through the barrier and find other local minima (for the case here, is to find the other local minimum). Since for small $h$, the local
state $\ket{\Phi_{+}}$ distributes in a very convex region near
the right local minimum, it suffice to solve our problem by
measuring the position of state $\ket{\Phi(t)}$ when $P_{+}(t)$ is large.
However, we may not be able to know $\Delta E$ precisely in advance.
So, we will apply the method of quantum walks: evolving
system for time $t$ which is chosen randomly from $[0,\tau]$, and then measuring the position \cite{CCD+03}.
The resulted distribution is
\begin{align}
    \rho_{\rm QTW}(\tau,x) \approx |\ip{x}{\Phi_{-}}|^2\int_{0}^{\tau}P_{-}(t)\frac{\d t}{\tau} + |\ip{x}{\Phi_{+}}|^2\int_{0}^{\tau}P_{+}(t)\frac{\d t}{\tau},
\end{align}
where QTW denotes quantum tunneling walk. Define $p_{-\to \pm}(\tau) = \int_{0}^{\tau}P_{\pm}(t)\frac{\d t}{\tau}$, which is the probabilities of finding the right and left local ground states, respectively.
Since $|\ip{x}{\Phi_{-}}|^2$ is small for $x$ near $+a$, the probability of finding the particle near $+a$ is determined by
\begin{align}
    \int_{\rm right~well} \rho_{\rm QTW}(\tau,x) \d x \approx  \int_{\rm right~ well}\d x |\ip{x}{\Phi_{+}}|^2p_{-\to +}
    \approx p_{-\to +}.
\end{align}
Therefore, it suffice to study properties of $p_{-\to \pm}$.
In the present case,
\begin{align}
    p_{-\to \pm} = \frac{1}{2}\left(1 \mp \frac{\sin(\Delta E \tau)}{\Delta E \tau} \right),
\end{align}
which will converge when $\tau \to \infty$. This fact ensures that
we can find the right local minimum $+a$ with a probability larger than some constant after evolving the system for enough long time.

Starting from $\ket{\Phi_{-}}$, the hitting time for the right well is
\begin{align}
    T_{\rm hit}(\Phi_{+}|\Phi_{-}) = \inf_{\tau > 0} \frac{\tau}{p_{-\to +}(\tau)}.
\end{align}
Since the probability for successful tunneling in one trial is $p_{-\to +}(\tau)$, we can repeat trials for $1/p_{-\to +}(\tau)$ times to secure one success and the total evolution time is $\tau/p_{-\to +}(\tau)$.
For sufficiently small $\epsilon \ll 1$, if $\frac{1}{\Delta E \tau} \leq \epsilon$, we can get $p_{-\to +}(\tau) \geq \frac{1}{2}(1-\epsilon)$. Therefore, the hitting time can be bounded by $O(\frac{1}{\Delta E\epsilon})$.

As is going to be shown later, if we want to find \emph{all} local minima, the mixing time would be a better quantifier. The limiting distribution is
\begin{align}
    \mu_{\rm QTW}(x) := \lim_{\tau \to \infty}\rho_{\rm QTW}(\tau,x) \approx |\ip{x}{\Phi_{-}}|^2 p_{-\to -} + |\ip{x}{\Phi_{+}}|^2p_{-\to +}.
\end{align}
The mixing time measures how fast the distribution $\rho_{\rm QTW}(\tau,x)$ converges to $\mu_{\rm QTW}(x)$.
We define $T_{\rm mix}$ as the $\epsilon$-close mixing time which satisfies
\begin{align}
   T_{\rm mix} =\inf_{\|\rho_{\rm QTW}(\tau,\cdot) - \mu_{\rm QTW}(\cdot)\|_1 \leq \epsilon} \tau.
\end{align}
Because $\|\rho_{\rm QTW}(\tau,\cdot) - \mu_{\rm QTW}(\cdot)\|_1 \leq O(\frac{1}{\Delta E \tau})$, we have $\|\rho_{\rm QTW}(\tau,\cdot) - \mu_{\rm QTW}(\cdot)\|_1 \leq \epsilon$ if $\tau = \Omega(\frac{1}{\Delta E \epsilon})$.
Therefore, $T_{\rm mix}$ could be bounded: $T_{\rm mix} = O(\frac{1}{\Delta E \epsilon})$.

For simulating a time-independent Hamiltonian, the number of queries needed are roughly proportional to the total evolution time (as demonstrated in \sec{quantumpre} or see details in \append{quantumsimulation}). The major task left is to calculate the energy gap $\Delta E$, get different evolution times and compare them with classical results.

As is specified in \append{interactionmatrix}, $0$ is the boundary of the two wells and the tunneling amplitude $\nu$ can be given by
\begin{align}
   \nu = h^2 (\Phi_{-}(0)\Phi'_{+}(0) - \Phi'_{-}(0)\Phi_{+}(0)).
   \label{eq:eq15}
\end{align}
To obtain an explicit result, we need to use the WKB approximation of the local ground states (\append{WKB}):
\begin{align}
   \Phi_{-}(x) \approx \frac{1}{ h^{1/4}} a_0(x)e^{-\frac{1}{h}\int_{-a}^{x}\sqrt{f(\xi)}\d \xi},\quad \Phi_{+}(x) = \Phi_{-}(-x),
   \label{eq:eq16}
\end{align}
where $a_0(x)$ is given by
\begin{align}
   a_0(x) = \left( \frac{\omega}{\sqrt{2} \pi}\right)^{1/4}e^{-\frac{1}{2}\int_{-a}^{x}(\frac{f'(\xi)}{2f(\xi)} - \frac{\omega}{\sqrt{2f}})\d \xi },
   \label{eq:eq17}
\end{align}
which is determined by the transport equation (Eq.~\eq{transport}) in \append{WKB}) and the normalization condition.
Substituting \eq{eq16} and \eq{eq17} to \eq{eq15}, we get
\begin{align}
   \nu = 2\sqrt{\frac{h\omega f(0)}{\sqrt{2}\pi}}e^C[1+O(h)]e^{-\frac{S_0}{h}},
\end{align}
where the constants $C$ and $S_0$ are given by
\begin{align}
   C = \int_0^a \Big(\frac{\omega}{\sqrt{2f(\xi)}} - \frac{1}{a-\xi}\Big)\d \xi , \quad S_0 = \int_{-a}^a \sqrt{f(\xi)} \d\xi.
\end{align}
Note that if $f(x)$ is quadratic, the factor $C$ will be zero, indicating that $C$ measures the deviation of the landscape from being quadratic.
The quantity $S_0$ is called the Agmon distance between two the local minima (see \append{agmon}).
Since we assumed by \eq{eq01} that $f(x)$ is almost quadratic, we have $f(0)\approx \frac{1}{2}\omega^2a^2$ and
\begin{align}
   \nu \approx h\omega\sqrt{\frac{\sqrt{2}\omega a^2}{h\pi}} e^{-\frac{S_0}{h}}.
\end{align}
As discussed above, the mixing time and hitting time could be bounded by the following characteristic time
\begin{align}
   T_c = \frac{1}{\Delta E} = \frac{1}{2\nu } \approx \frac{\sqrt{\pi}}{2a\omega \sqrt{\sqrt{2}\omega h}}e^{\frac{1}{h}\int_{-a}^a \sqrt{2f(\xi)}\d \xi}.
\end{align}

Next, we need to find how long it takes for SGD to escape from the left local minimum. Discrete-time SGD with a small learning rate $s$ can be approximated by a learning-rate dependent stochastic differential equation (lr-dependent SDE) \cite{SSJ20}:
\begin{align}
   \d X = - \nabla f(X) \d t + \sqrt{s}\d W,
\end{align}
where $W$ is a standard Brownian motion.
Before hitting $0$, the SDE is almost an Ornstein–Uhlenbeck process:
\begin{align}
   \d X = -\omega^2 (X+a) \d t + \sqrt{s}\d W.
\end{align}
The expected time for the Ornstein–Uhlenbeck process to first hit $0$ is
\begin{align}
   \mathbb{E}T_0 \approx \frac{\sqrt{\pi s}}{a \omega^3}e^{\frac{2H_f}{s}},
\end{align}
where $H_f$ is called the Morse saddle barrier and equals to $f(0) - f(-a) \approx \frac{1}{2}\omega^2 a^2$ in our case.

Although $\mathbb{E}T_0$ is not a precise classical counterpart of either quantum mixing or hitting time, it is heuristic to compare $\mathbb{E}T_0$ with $T_c$ which can both reflect the time to escape from the left well.
The forms of $\mathbb{E}T_0$ and $T_c$ are very similar.
Two major differences can be observed: 1. the exponential term
in $\mathbb{E}T_0$ is determined by the height of the barrier,
while that in $T_c$ is related to an integral of $\sqrt{f}$;
2. $T_c \propto 1/\omega^{3/2}$ but $\mathbb{E}T_0 \propto 1/\omega^{3}$, indicating that the flatness of the wells affects differently on quantum and classical methods. We will show that for landscapes with multiple wells, the distribution of wells is also an important factor. In general, QTW could be faster than SGD if the barriers between local minima are high but thin, each well is close to many other wells, and wells are flat.

The above comparison is intuitive but not rigorous. Two important technical details for comparison are needed for quantitative discussions. It is shown in \sec{quantumpre} that a super-polynomial separation between
evolution time of QTW and SGD gives rise to a super-polynomial separation
between quantum and classical queries for QTW and SGD, respectively.
Therefore, it suffice to compare the evolution time, especially the exponential term $e^{S_0/h}$ and $e^{2H_f/s}$.
The second problem is that $h$ and $s$ are not two constants but variables. The evolution times cannot be quantitatively compared if $h$ and $s$ are independent.
In \sec{standard}, we develop two natural standards to make fair comparisons between
QTW and SGD, which specifies $h$ and $s$.

\subsection{Definition of quantum tunneling walks}\label{sec:QTW}
We now formally describe the model of a quantum tunneling walk (QTW) on
a general objective function $f(x)$ satisfying assumptions in \sec{quantumpre}.
The wells are denoted by $U_j = \{x_j\}~(j=1,\ldots,N)$.
Let $\ket{j}$ be the corresponding orthonormalized local ground state of $U_j$.
It is ensured that $\{\ket{j}:j=1,\ldots,N\}$ spans a low energy subspace, $\mathcal{F}$, of the Hamiltonian $H = -h^2\Delta + f(x)$.

If one has information about one well $U_j$ and its neighborhood, the construction of the local ground state should be easy which can be close to $\ket{j}$ or at least be almost in the subspace $\mathcal{F}$. We assume the initial state $\Phi(0)$ to be in $\mathcal{F}$. The evolution is determined by the Schr\"odinger equation,
\begin{align}
   i \frac{\d}{\d t} \ip{x}{\Phi(t)} = \bra{x} H \ket{\Phi(t)},
\end{align}
where $|\ip{x}{\Phi(t)}|^2$ is the probability distribution of finding the walker at $x$. The Schr\"odinger equation indicates the phenomenon of quantum tunneling since it can be rewritten as
\begin{align}
   i \frac{\d}{\d t} \ip{j}{\Phi(t)} = \sum_{j'} \bra{j} H_{|\mathcal{F}} \ket{j'} \ip{j'}{\Phi(t)},~\mathrm{for~any}~j=1,\ldots,N,
   \label{eq:qtwsch}
\end{align}
given that $\ket{\Phi(0)}$ is in the subspace $\mathcal{F}$.
Here, $(\z j|H_{|\mathcal{F}}|j^{\prime} \y)$ is called as the \emph{interaction matrix} and is calculated by \prop{spePF} and \ref{prop:speW}.
Once we get $\ip{j}{\Phi(t)}$ and $\ip{x}{j}$ for all $j$, we can obtain the probability distribution $|\ip{x}{\Phi(t)}|^2$.
The overlap $\ip{x}{j}$ is invariant with respect to $t$. So, we may focus on \eq{qtwsch} to investigate the time evolution.

As is shown by \eq{qtwsch}, restricted in the low energy subspace $\mathcal{F}$,
the quantum evolution is similar to that of a quantum walk on a graph.
The wells correspond to vertices of the graph, and the quantum tunneling effects between wells
determine the graph connectivity.
QTW walks among different wells by quantum tunneling, helping to find all other local minima.

Finally, according to \lem{quansimu} in \append{quantumsimulation}, the quantum query complexity of simulating the Schr\"odinger equation is directly linked to the evolution time $t$ and is bounded by
\begin{align}
   O\left(\|f\|_{L^{\infty}(\Omega)} t \frac{\log (\|f\|_{L^{\infty}(\Omega)} t/\epsilon)}{\log \log (\|f\|_{L^{\infty}(\Omega)} t/\epsilon)}\right),
\end{align}
where $\Omega$ is a large region containing all minima of interest and
$\epsilon$ is the precision quantified by the $L^2$ norm between
the target and the obtained wave functions.
Loosely speaking, we need $\tilde{O}(t)$ quantum queries if the evolution time is $t$.
For SGD, the number of queries needed is at least $\Omega(t/s)$ for time $t$.
Thus, as long as there is a super-polynomial separation between QTW and SGD evolution time, there is a super-polynomial separation between
quantum queries and classical queries for QTW and SGD, respectively.
Conclusions on speedups are essentially based on comparisons of query complexity. However, based on this relationship between evolution time and query complexity, we can focus on comparisons of time.

To sum up, QTW is quantum simulation with the system Hamiltonian being $H = -h^2\Delta +f(x)$ and the initial state being in a low energy subspace of $H$, where $f(x)$ is the potential function of a type of benign landscapes (\assum{quantum1}, \ref{assum:quantum2}, and \ref{assum:quantum3}). QTW can be efficiently implemented on quantum computers.

\subsection{Mixing time}\label{sec:Qmixing}
For a given landscape $f(x)$, the complexity of Hamiltonian simulation mainly depends on the evolution time (see~\append{quantumsimulation}). In this section, we focus on the evolution time needed for fulfilling the tasks of finding all minima.
Since quantum evolutions are unitary, different from SGD, QTW never converges.
In this case, after running QTW for some time $t$, we measure the position of the walker.
Similar to the quantum walks in \cite{CCD+03}, the evolution time $t$ can be chosen uniformly in $[0,\tau]$. Later, we will prove that under sufficiently large $\tau$, QTW can find other wells with probability larger than some constant.
Note that there can be better strategies to determine the time for measurement $t$ than uniformly sampling from an interval $[0,\tau]$ \cite{AC21}.
For simplicity, we only analyze the original strategy of \cite{CCD+03} in the present paper.

As is mentioned earlier, we initialize at a state $\ket{\Phi(0)}\in \mathcal{F}$. In later sections, we may specify $\ket{\Phi(0)}$ to be one of the local ground states.
Let the spectral decomposition of $H_{|\F}$ to be $H_{|\F} = \sum_{k=1}^N E_k \ket{E_k}\bra{E_k}$.
Simulating the system for a time $t$ chosen uniformly in $[0,\tau]$,
one can obtain the probability density of finding the walker at $x$\footnote{The distribution $\rho_{\rm QTW}(\tau,x)$ depends on the initial state $\ket{\Phi(0)}$. A more rigorous notation is $\rho_{\rm QTW}(\tau,x|\Phi(0))$. When there is no confusion, we omit the initial state for simplicity.}
\begin{align}
    \rho_{\rm QTW}(\tau,x)  :=& \frac{1}{\tau}\int_0^{\tau} \d t |\z x|e^{-i H_{|\F} t}\ket{\Phi(0)}|^2 \nonumber \\
   =& \sum_{E_k=E_{k^{\prime}}} \z x|E_k\y \z E_k |\Phi(0)\y \z \Phi(0) |E_{k^{\prime}}\y \z E_{k^{\prime}}| x\y
   \nonumber\\
   &~+\sum_{E_k \neq E_{k^{\prime}}} \frac{1-e^{-i(E_k - E_{k^{\prime}})\tau}}{i(E_k - E_{k^{\prime}})\tau}
   \z x|E_k\y \z E_k |\Phi(0)\y \z \Phi(0) |E_{k^{\prime}}\y \z E_{k^{\prime}}| x\y.
\end{align}
The time-averaged probability density leads to a limiting distribution when $\tau \to \infty$:
\begin{equation}
    \mu_{\rm QTW} :=
    \sum_{E_k=E_{k^{\prime}}} \z x|E_k\y \z E_k |\Phi(0)\y \z \Phi(0) |E_{k^{\prime}}\y \z E_{k^{\prime}}| x\y.
\end{equation}
With the strategy of measuring at $t$ randomly chosen from $[0,\tau]$,
QTW can output a distribution $\rho_{\rm QTW}(\tau,x)$ with limit. Such a process is regarded as mixing.
Quantum mixing time evaluates how fast $\rho_{\rm QTW}(\tau,x)$ converges to $\mu_{\rm QTW}(x)$, and is rigorously defined as:
\begin{definition}[Mixing time of QTW]\label{defn:mixingqtw}
$T_{\rm mix}$ is called the $\epsilon$-close mixing time, iff
for any $\tau\geq T_{\rm mix}$, we have
\begin{align}
    \|\rho_{\rm QTW}(\tau,\cdot) - \mu_{\rm QTW}(\cdot)\|_1 \leq \epsilon.
    \label{eq:mixingqtw}
\end{align}
\end{definition}
The following lemma provides a general bound for the QTW mixing time whose proof is postponed to \append{prooflem2}.
\begin{lemma}[Upper bound for QTW mixing time]\label{lem:qtwmixingtime}
The condition \eq{mixingqtw} can be satisfied if
\begin{align}
   \tau \geq \frac{2}{\epsilon} \sum_{E_k\neq E_{k^{\prime}}}\frac{|\ip{E_k}{\Phi(0)} \ip{\Phi(0)}{E_{k'}}|}{|E_k - E_{k'}|}[1+(N-1)|O(h^{\infty})|],
\end{align}
and this implies
\begin{align}
 \hspace{-2mm}  T_{\rm mix} &= O \bigg(\frac{1}{\epsilon} \sum_{E_k\neq E_{k^{\prime}}}\frac{|\ip{E_k}{\Phi(0)} \ip{\Phi(0)}{E_{k'}}|}{|E_k - E_{k'}|}[1+(N-1)|O(h^{\infty})|] \bigg) \nonumber\\
   &\leq O \bigg(\frac{N}{\epsilon \Delta E} [1+(N-1)|O(h^{\infty})|] \bigg),
   \label{eq:qtwmixing}
\end{align}
where $\Delta E := \min_{E_k\neq E_{k^{\prime}}}|E_k - E_{k'}|$ is referred to as the minimal gap of $H_{|\F}$.
\end{lemma}

The term $O(h^{\infty})$ in \eq{qtwmixing}
originates from integrals $\int |\ip{x}{j} \ip{j'}{x}| \d x,~j\neq j'$.
Intuitively, states $\ket{j}$ and $\ket{j'}$ localize in different wells, such that $\ip{x}{j} \ip{j'}{x}$ is exponentially small with respect to $h$
for any $x$.
\lem{qtwmixingtime} highlights the dependence of the mixing time on the initial state $\ket{\Phi(0)}$ and the eigenvalue gaps of $H_{|\F}$.
Concrete examples will be given in \sec{comparison},
where we further illustrate \eq{qtwmixing} and compare QTW with SGD.

As is mentioned, \eq{qtwsch} indicates a quantum walk: a well $U_j$ can be seen as a vertex of a graph and $H_{|\F}$ implies graph connectivity (interaction between wells) similar to the graph Laplacian.
The connection between QTW and quantum walks is helpful to simplify the physical picture of QTW.
However, we also address the difference between quantum walks and QTW.
For quantum walks, we only consider the probabilities of finding the walker at
vertices, that is,\footnote{Similar to $\rho(\tau,x)$, the probability $p(\tau,j)$ depends on the initial state $\ket{\Phi(0)}$ and a more rigorous notation is $p(\tau,j|\Phi(0))$. When there is no confusion, we omit the initial state for simplicity.}
\begin{align}
   p(\tau,j) &:= \frac{1}{\tau}\int_0^{\tau} \d t |\z j|e^{-i H_{|\F} t}| \Phi(0) \y|^2 \nonumber \\
   & = \sum_{E_k=E_{k^{\prime}}} \z j|E_k\y \z E_k |\Phi(0)\y \z \Phi(0) |E_{k^{\prime}}\y \z E_{k^{\prime}}| j\y
   \nonumber\\
   &~~+\sum_{E_k \neq E_{k^{\prime}}} \frac{1-e^{-i(E_k - E_{k^{\prime}})\tau}}{i(E_k - E_{k^{\prime}})\tau}
   \z j|E_k\y \z E_k |\Phi(0)\y \z \Phi(0) |E_{k^{\prime}}\y \z E_{k^{\prime}}| j\y.
\end{align}
When $\tau \to \infty$, $p(\tau,j)$ also converges to a limit
\begin{align}
   p(\infty,j) :=
    \sum_{E_k=E_{k^{\prime}}} \z j|E_k\y \z E_k |\Phi(0)\y \z \Phi(0) |E_{k^{\prime}}\y \z E_{k^{\prime}}| j\y.
\end{align}
Following results, \lem{limitdis} and \lem{mixingqw}, show the connection and difference between QTW and quantum walks in a more quantitative way (detailed proofs can be found in \append{prooflem3} and \append{prooflem4}).
\begin{lemma}[Limit distributions]\label{lem:limitdis}
Limit distributions of the QTW and the quantum walk satisfy the following:
\begin{align}
     \mu_{\rm QTW}(x)
     = \sum_{j}
    p(\infty, j) |\z x|j\y|^2 + O(h^{\infty}).
\end{align}
\end{lemma}

\begin{definition}[Mixing time of quantum walks \cite{CLR20}]\label{defn:mixingqw}
$t_{\rm mix}$ is called the $\epsilon$-close mixing time of the quantum walk, iff
for any $\tau \geq t_{\rm mix}$,
\begin{align}
    \sum_{j=1}^N |p(\tau,j) - p(\infty, j)| \leq \epsilon.
    \label{eq:disofp}
\end{align}
\end{definition}
\begin{lemma}[Upper bound for the mixing time of quantum walks]\label{lem:mixingqw}
The condition \eq{disofp} is satisfied if
\begin{align}
   \tau \geq \frac{2}{\epsilon} \sum_{E_k\neq E_{k^{\prime}}}\frac{|\ip{E_k}{\Phi(0)} \ip{\Phi(0)}{E_{k'}}|}{|E_k - E_{k'}|},
\end{align}
and we have
\begin{align}
   t_{\rm mix} = O \bigg(\frac{1}{\epsilon} \sum_{E_k\neq E_{k^{\prime}}}\frac{|\ip{E_k}{\Phi(0)} \ip{\Phi(0)}{E_{k'}}|}{|E_k - E_{k'}|} \bigg)
   \leq O \bigg(\frac{N}{\epsilon \Delta E} \bigg).
   \label{eq:qwmixing}
\end{align}
\end{lemma}
By \lem{limitdis} and the comparison of \lem{qtwmixingtime} and \lem{mixingqw},
we know that for sufficiently small $h$, which indicates that local ground states localize sufficiently near their respective wells, QTW can be well characterized by a quantum walk.
On a higher level of speaking, QTW generalizes quantum walks from walking on discrete graphs to propagating on continuous functions. And QTW may enable new phenomenons not shown in quantum walks when states $\ket{j}$ are poorly localized near $U_j$.
On the other hand, QTW under proper conditions can be used to implement quantum walks.

\subsection{Hitting time}\label{sec:Qhitting}
If we aim at finding one particular
well (the one with global minimum or the one with the best generalization properties), hitting time instead of mixing time should be of interest.
Classically, the hitting time is
the expected time required to find some target region or point.
For quantum algorithms, we cannot output the position of the walker at all times
and the system state would be destroyed by measuring its position.
Thus, the definition of the hitting time for quantum algorithms is slightly different.
We first see how previous literature defines the quantum walk hitting time:

\begin{definition}[Hitting time for quantum walks \cite{AC21}] Consider a quantum walk governed by \eq{qtwsch}. Let the state $\ket{j}$ be the one of interest.
Then, starting from the initial state $\ket{\Phi(0)}$, the hitting time of the quantum walk is defined as follows:
\begin{align}
    t_{\rm hit}(j) := \inf_{\tau>0} \frac{\tau}{p(\tau,j)}.
\end{align}
\end{definition}
To understand this definition, we first refer to the process, evolving the system for time $t$ uniformly chosen from $[0,\tau]$, as one trial.
Using one trial, the probability of getting $\ket{j}$ is $p(\tau,j)$.
So, repeating the trials for $1/p(\tau,j)$ times guarantees to hit $\ket{j}$ with high probability.
In this case, the total evolution time needed is bounded by $\tau/p(\tau,j)$.
In the same spirit, we can define and bound the QTW hitting time as follows:
\begin{definition}[Hitting time of QTW] For QTW governed by \eq{qtwsch}, let the open and $C^2$-bounded region $\Omega$ be the region of interest.
Then, starting from the initial state $\ket{\Phi(0)}$, the $\Omega$-hitting time of QTW is defined as follows:
\begin{align}
    T_{\rm hit}(\Omega) := \inf_{\tau>0} \frac{\tau}{\int_{\Omega}\rho_{\rm QTW}(\tau,x)\d x}.
\end{align}
\end{definition}

Basic results about the hitting time (\lem{qwhitting} and \lem{qtwhitting}) are present as follows. The proof of \lem{qtwhitting} is in \append{prooflem5} and that of \lem{qtwhitting} is in \append{prooflem6}.
\begin{lemma}[Upper bound of the quantum walk hitting time]\label{lem:qwhitting}
The probability of finding $\ket{j}$ can be bounded as follows
\begin{align}
    p(\tau,j) \geq p(\infty,j) - \sum_{E_k \neq E_{k^{\prime}}} \frac{2}{|E_k - E_{k^{\prime}}|\tau}
   |\z j|E_k\y \z E_k |\Phi(0)\y \z \Phi(0) |E_{k^{\prime}}\y \z E_{k^{\prime}}| j\y|\nonumber\\
    \geq p(\infty,j) - \frac{2}{\Delta E \tau}.
\end{align}
As a result, for any $\epsilon< p(\infty,j)$, setting $\tau_{\epsilon} = 2/\Delta E\epsilon$, we have
\begin{align}
    t_{\rm hit}(j) \leq \frac{\tau_{\epsilon}}{p(\tau_{\epsilon},j)} \Rightarrow
    t_{\rm hit}(j) = O\bigg(\frac{1/\Delta E\epsilon}{p(\infty,j) - \epsilon}\bigg).
    \label{eq:qwhitting}
\end{align}
\end{lemma}
If $\epsilon$ in \eq{qwhitting} is small enough, $\tau_{\epsilon} = 2/\Delta E\epsilon$ permits a good mixing and
we may write
\begin{align}
    t_{\rm hit}(j) = O\bigg(\frac{1}{p(\infty,j)\Delta E\epsilon} \bigg),
\end{align}
which suggests we are using the mixing time to bound the hitting time.

\begin{lemma}[Upper bound of the QTW hitting time]\label{lem:qtwhitting}
Consider an bounded open set $\Omega_{j}$ containing only one well $U_j$,\footnote{To be rigorous, satisfy \eq{georelation-elliptic} in \append{interactionmatrix}.} we have
\begin{align}
   \hspace{-1mm}  \int_{\Omega_j}\rho_{\rm QTW}(\tau,x)\d x
     &\geq \int_{\Omega_j}\mu_{\rm QTW}(x)\d x - \frac{2}{\Delta E \tau}(1+|O(h^{\infty})|)\nonumber\\
     &= p(\infty, j) + O(h^{\infty}) - \frac{2}{\Delta E \tau}(1+|O(h^{\infty})|).
     \label{eq:qtwhitting}
\end{align}
For any $\epsilon< \int_{\Omega_j}\mu_{\rm QTW}(x)\d x$, let $\tau_{\epsilon} = 2 (1+|O(h^{\infty})|)/\Delta E\epsilon$, we have
\begin{align}
    T_{\rm hit}(\Omega_j) \leq \frac{\tau_{\epsilon}}{\int_{\Omega_j}\rho_{\rm QTW}(\tau_{\epsilon},x)\d x} \Rightarrow
    T_{\rm hit}(\Omega_j) = O\bigg(\frac{1}{\Delta E\epsilon} \frac{1+|O(h^{\infty})|}{\int_{\Omega_j}\mu_{\rm QTW}(x)\d x - \epsilon}\bigg).
\end{align}
\end{lemma}

The upper bounds we have obtained on mixing and hitting time are
still not explicit, as $H_{|\F}$ is not given.
Next, we establish
relationships between an objective landscape, the corresponding interaction matrix $H_{|\F}$, and the time cost of the QTW algorithm. This is a main task of the present paper.
In later sections, we figure out major geometric properties that affect $H_{|\F}$ on specific landscapes.

\subsection{Application: Tensor decomposition}\label{sec:tendecom}
After giving the definition of QTW and studying its mixing and hitting time, now we use QTW to solve a practical problem, orthogonal tensor decomposition, which is a central problem in learning many latent variable models \cite{AGH+14}.
Specifically, we consider a fourth-order tensor $T\in \mathbb{R}^{d^4}$ that has orthogonal decomposition:
\begin{equation}
    T = \sum_{i=1}^d a_j^{\otimes 4},
    \label{eq:4-tensor}
\end{equation}
where the components $\{a_j \}$ form an orthonormal basis of a $d$-dimensional space ($a_j\trans a_j = \delta_{ij}$). The goal of orthogonal tensor decomposition
is to find all components $\{a_j \}$.

Following previous popular methods \cite{CLX+09,Hyv99},
we try to find all components by a single optimization problem.
Concretely, we consider the following landscape \cite{FJK96}:\footnote{The original objective function used in previous papers including \cite{FJK96} is $T(u,u,u,u)$. The function in \eq{TenDecObj} is designed such that $\min f = 0$.}
\begin{equation}
    f(u) = 1-T(u,u,u,u) = 1- \sum_{i=1}^d (u\trans a_j)^4,~\|u\|_2^2 = 1.
    \label{eq:TenDecObj}
\end{equation}
Without loss of generality, we work in the coordinate system specified by $\{a_j\}_{j=1}^{d}$. In particular, let $u = \sum_j^d x_j a_j$ and $x= (x_1,\ldots,x_d)$, we obtain $f(x) = 1-\sum_{i=1}^d x_i^4$. Later, we also use $a_j$ to denote the vector $(0,\ldots,1,\ldots,0)$ where the only nonzero coordinate with value 1 appears in the $j$th entry.
$f(x)$ has $2d$ local minima $\pm a_1,\ldots\pm a_d$ uniformly distributed on the $d$-dimensional sphere $\mathbb{S}^{d-1}$.\footnote{Here, we need to consider quantum simulation on the manifold $\mathbb{S}^{d-1}$ which should cost the same
quantum queries as quantum simulation on $\mathbb{R}^{d}$ under the same
evolution time
(see discussions in \append{quantumsimulation}).}
Therefore, finding all the minima solves the orthogonal tensor decomposition problem.

\begin{figure}
  \centerline{
  \includegraphics[width=0.4\textwidth]{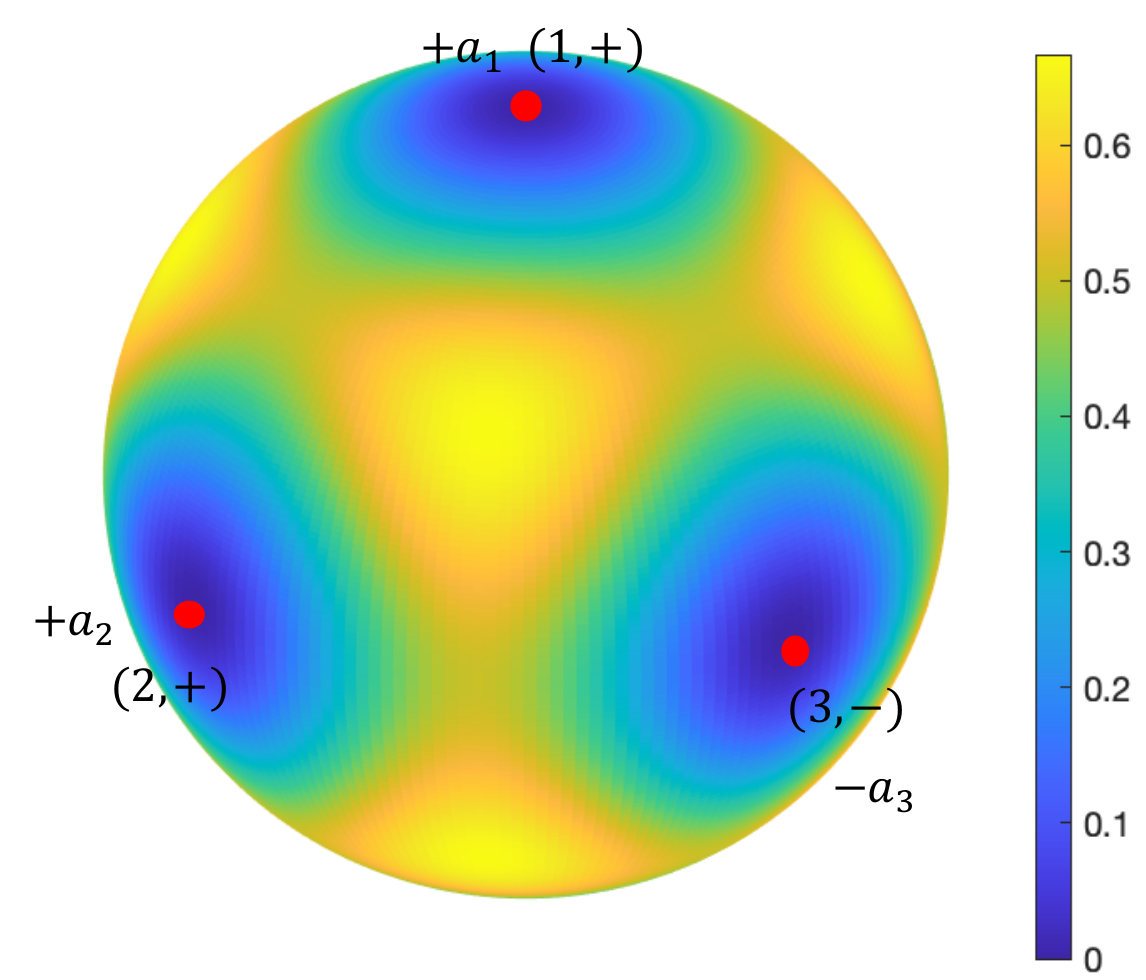}}
    \caption{The landscape given by \eq{TenDecObj} for dimension $d=3$:
     local minima $a_1$, $a_2$, and $-a_3$ are highlighted by red points ({\color{red}$\bullet$}) and corresponding labels $\alpha = (j(\alpha), \pi_{\alpha})$ are shown.
    }
\label{fig:TenDeclabel}
\end{figure}
The problem of tensor decomposition has notable symmetries. As a result, the objective function $f$ \eq{TenDecObj} is nonconvex, and we can apply QTW to such a landscape.
We can use the pair $\alpha = (j,\pi_{\alpha})$ to denote the local minima and corresponding wells, where $j = j(\alpha)$ (i.e. $j(\cdot)$ is a function) refers to $a_j$ and $\pi_{\alpha} \in \{\pm 1\}$ specifies whether it
is $+ a_j$ or $- a_j$.
\fig{TenDeclabel} shows the landscape $f$ for $d =3$, where some minima are labeled.
The local ground state of the well $(j,\pi_{\alpha})$ is denoted by $|j,\pi_{\alpha}\y$. In the basis $\{|1,+\y,\ldots,|d,+\y,|1,-\y,\ldots,|d,-\y\}$ which spans a low-energy subspace $\F$, the interaction matrix modulus an exponential error has the form
\begin{equation}
    H_{|\F} = \left(
    \begin{array}{cccc|cccc}
        \mu & w  & \cdots & w & 0 & w  & \cdots & w \\
         w & \mu & \ddots & \vdots & w & 0 & \ddots & \vdots \\
         \vdots  &\ddots  & \ddots& w  & \vdots & \ddots &\ddots &w \\
         w  & \cdots& w & \mu  & w  & \cdots& w & 0 \\
         \hline
         0 & w  & \cdots & w & \mu & w  & \cdots & w\\
         w & 0 &\ddots & \vdots & w & \mu & \ddots & \vdots\\
         \vdots & \ddots&\ddots & w & \vdots & \ddots & \ddots& w\\
          w  & \cdots& w & 0  & w  & \cdots& w & \mu \\
    \end{array}
    \right).
    \label{eq:TenDec}
\end{equation}
Here the quantity $\mu$ stands for the energy of local ground states, and $w$ is the tunneling amplitude quantifying the interactions between wells.
To understand \eq{TenDec}, for any $j$, imagine a sphere where $(j,+)$ is the north pole and $(j,-)$ the south pole, then for $j'\neq j$, $(j',\pm)$ are evenly distributed on the equator.
The energy of all local ground states are the same because of the symmetry.
So, diagonal elements of $H_{|\F}$ are all $\mu$.
The interactions between $(j,+)$ and $(j',\pm)$ for all $j'\neq j$ should be the same as well. However, the interaction between $(j,+)$ and $(j,-)$ is exponentially weaker due to the longer distance between $(j,+)$ and $(j,-)$.
As a result, we write $\bra{j,+}H_{|\F}\ket{j,-} = 0$ modulus an exponential error
and all other off-diagonal elements as $w$.

As is demonstrated in previous sections, the time cost of QTW highly depends on spectral gaps of $H_{|\F}$. The following lemma studies eigenstates and eigenvalues of $H_{|\F}$ (see proof in \append{prooflem7}).
\begin{lemma}\label{lem:TenDecEig}
The eigenstates and corresponding eigenvalues of $H_{|\F}$ in \eq{TenDec} are given by
\begin{align}
    |E_k\y &= \frac{1}{\sqrt{2d}}\sum_j e^{i\frac{2\pi}{d}kj}(|j,+\y+ |j,-\y),~k=1,\ldots,d \\
    |E_k\y &= \frac{1}{\sqrt{2}}(|k-d,+\y - |k-d,-\y),~k=d+1,\ldots,2d
\end{align}
where
\begin{align}
    E_k &= \mu - 2w,~k=1,\ldots,d-1 \\
    E_d &= \mu + 2w(d-1), \\
    E_k &= \mu,~k=d+1,\ldots,2d.
\end{align}
\end{lemma}

Evolving the system for at least the mixing time, the measured results would be subject to the limit distribution $\mu_{\rm QTW}$.
Since $\mu_{\rm QTW}$ would concentrate near all minima, we are able to find all components. Combined with results of \sec{Qmixing}, we can get the mixing time of the QTW as follows (the proof is postponed to \append{prooflem8}):
\begin{lemma}\label{lem:TenDecmix}
For the landscape \eq{TenDecObj}, starting from a local ground state $|\alpha\y$, the distribution $\rho_{\rm QTW}(\tau,x)$ converges to the limiting distribution $\mu_{\rm QTW}$ obeying the following relation:
\begin{align}
    \|\rho_{\rm QTW} - \mu_{\rm QTW}\|_{L^1(\mathbb{S}^{d-1})}
    \leq \frac{1}{|w| \tau}(\Theta(1/2)+ O(h^{\infty})).
    \label{eq:TenDecmix}
\end{align}
The $\epsilon$-close mixing time is subsequently bounded as
\begin{align}
    T_{\rm mix} = O\bigg(\frac{1}{|w|\epsilon}(1+O(h^{\infty}))\bigg).
\end{align}
\end{lemma}

To determine the total evolution time for finding all components $\{a_j\}$ (i.e., all global minima), we need to calculate $\int_{\Omega_{\beta}}\mu_{\rm QTW} \d x$, where $\Omega_{\beta}$ is an open set containing the minimum $\beta$.
According to \lem{qtwhitting}, $\int_{\Omega_{\beta}}\mu_{\rm QTW} \d x$ is the probability of finding the particle in a neighborhood of $\pi_{\beta} a_{j(\beta)}$ and can be captured by
the probability of finding the system at the state $\ket{\beta}$.
Starting from a local state $|\alpha\y$, the probability of hitting $\ket{\beta}$ is given by the following lemma and the proof can be found in \append{prooflem9}.
\begin{lemma}\label{lem:TenDecprobability}
Initiating at a local state $|\alpha\y$ where $\alpha = (j(\alpha), \pi_{\alpha})$, after simulating for a time $t$ which is chosen uniformly from $[0,\tau],~\tau \to \infty$, the limiting distribution represented by the probability of tunneling to a local state $|\beta\y$ is given by
\begin{equation}
    p(\infty,\beta | \alpha )
    =\left\{\begin{array}{ll}
         \frac{1}{2d^2},~j(\alpha) \neq j(\beta),\\[3pt]
         \frac{1}{2}-\frac{(d-1)}{2d^2},~j(\alpha) = j(\beta).
    \end{array} \right.
\end{equation}
\end{lemma}

Note that the two minima $\pm a_{j}$ are equivalent, representing one component.
Thus, starting from $\ket{\alpha}$, we are able to find a component different from $\pm a_{j(\alpha)}$ if the measured result is in a well $\beta$ where $j(\beta) \neq j(\alpha)$.
We can define the probability for a successful trial as $p_{\rm suc} = \sum_{j(\beta)\neq j(\alpha)} p(\infty,\beta | \alpha ) = \frac{d-1}{d^2}$. That is, evolving for time $T_{\rm mix}$ as described by \lem{TenDecmix}, we are able to approximately sample from the limiting distribution $\mu_{\rm QTW}$ and then get to another component with probability near $p_{\rm suc}$.
The number of trials needed for finding another component is approximately $1/p_{\rm suc}$. And the time needed for finding another component from a known component is approximately $T_{\rm mix}/p_{\rm suc}$. Repeating the procedure of looking for one component that is different from a known one, we can obtain all orthogonal components with total time\footnote{The term $O(d\log d)$ appears because our procedure is equivalent to the Coupon Collector's Problem.}
\begin{align}
    T_{\rm tot} = O(d\log d) T_{\rm mix}/p_{\rm suc} = \tilde{O}(d^2)\frac{1}{\epsilon|w|}.
    \label{eq:Ttot1}
\end{align}
To determine the time specifically, it remains to determine $1/|w|$ which
depends exponentially on $d$ and $h$. We can obtain:
\begin{lemma}\label{lem:TenDecw}
For sufficiently small $h$, the tunneling amplitude $w$ in the interaction matrix \eq{TenDec} satisfies
\begin{align}
    w = - \sqrt{h}(C_1 C_2^{d-1} + O(h))e^{-\frac{\sqrt{2}}{2h}},
\end{align}
where $C_1$ and $C_2$ are constants depending only on the landscape and are independent of the dimension $d$.
\end{lemma}

The proof of \lem{TenDecw} is in \append{prooflem10}.
It is intuitive to see from \lem{TenDecw} that the smaller $h$ is, the longer time it takes to find all components.
However, small $h$ permits more accurate measurement results.
A successful tunneling means we can find a point near a new component, but this point may not be the actual minimum. We add a constraint that the expected risk is $\delta$ (i.e., $\mathbb{E}_{x\sim \mu_{\rm QTW}}f(x) -\min f = \delta$). Subsequently, $h$ can be bounded using $\delta$ and we can have
the following proposition:
\begin{proposition}\label{prop:TenDecTtot}
For sufficiently small $\epsilon$ (such that the measured positions nearly obey $\mu_{\rm QTW}$) and sufficiently small expected risk $\delta$ (such that $h$ can be estimated by $\delta$ and \lem{TenDecw} is valid),
we have $h = \sqrt{2}\delta/(d-1 + o_{\delta}(1))$ and
the total time for finding all orthogonal components of $T$ in \eq{4-tensor} by QTW satisfies
\begin{align}
    T_{\rm tot} = O(\mathrm{poly}(1/\delta, e^d, 1/\epsilon)) e^{\frac{(d-1) +o_{\delta}(1)}{2\delta}}.
    \label{eq:TenDecTtot}
\end{align}
\end{proposition}
\begin{remark}
The strategy we adopt here, which is equivalent to repeating sampling from $\mu_{\rm QTW}$, is straightforward but may not be the optimal one under the framework of QTW. In other words, \prop{TenDecTtot} provides a general upper bound on the total evolution time needed. However, the term $e^{\frac{\sqrt{2}}{2h}}$ which gives the term $e^{\frac{d-1}{2\delta} + o_{\delta}(1)}$ in \eq{TenDecTtot} describes essential difficulty for tunneling through a barrier and would not disappear as long as we use quantum tunneling.
\end{remark}

To sum up, we provide a scenario that QTW can be used to solve orthogonal tensor decomposition problems.
For a practical landscape, the spectrum of the interaction matrix and the mixing time of QTW is explicitly calculated.
Running QTW for some time (bounded by the mixing time) repeatedly,
we can sample points from a distribution near the limiting distribution and find all tensor components, and an upper bound on the total running time for QTW is derived.


\section{Comparison Between Quantum Tunneling Walks and Classical Algorithms}\label{sec:comparison}
In this section, we use comparisons between QTW and SGD to explain the
advantages of quantum tunneling, resulting in our \slog.
Because of distinctions between quantum and classical algorithms, preparations (i.e., standards for comparisons) in \sec{standard} are needed before specific comparisons in \sec{illustration}.
Having such general understanding of QTW, in \sec{separation},
we further make use of the fact that quantum evolution is essentially global but classical algorithms rely on local queries, so that a hitting problem cannot be solved efficiently by classical algorithms can be tackled by QTW within polynomial queries when given reasonable initial states.

\subsection{Criteria of fair comparison}\label{sec:standard}
Through out \sec{comparison}, we adopt assumptions in both \sec{classicalpre} and \sec{quantumpre}
for the objective landscape $f(x)$ of interest.
We still use $U_j = \{x_j \}~(j=1,\ldots,N)$ to denote the wells and
$\ket{j}$ the corresponding orthonormalized local ground states.
The interaction matrix is $H_{|\F}$, where $H=-h^2\Delta + f(x)$ is the Hamiltonian and $\F$ the low energy subspace spanned by
$\{ \ket{j}: j=1,\ldots,N\}$.

As shown in \sec{onedimexp},
the hitting time of SGD is determined by the
landscape and an adjustable learning rate $s$.
Similarly, we can also adjust $h$ in Hamiltonian simulation.
Therefore, we need to determine the relationship between $h$ and $s$ for the comparison between the time cost of QTW and SGD.

Note that both QTW and SGD have limit distributions, namely, $\mu_{\rm QTW}$ and $\mu_{\rm SGD}$,
respectively (see \sec{classicalpre} and \lem{limitdis} for details).
If $h$ (or $s$) becomes smaller, $\mu_{\rm QTW}$ ($\mu_{\rm SGD}$) will concentrate more closely to global minima, giving more accurate outputs, whereas it would take more time for the QTW (SGD) to converge.
Comparing the running time without specifying accuracy is not fair.

In order to establish an relationship between $h$ and $s$, as well as to compare QTW and SGD fairly, we specify some kind of accuracy of the limit distributions. The two variables, $h$ and $s$, will be solved from the demand of accuracy. Hence, the time cost of different algorithms
are only related to the accuracy, the dimension, and some geometric properties of the landscapes.

There are different measures of accuracy we can choose depending on
the tasks faced.
Here, we introduce two kinds of measures along with the corresponding standards of comparison.

\begin{standard}[Risk accuracy]\label{stand:risk}
 Let $\mu_{\rm QTW}$ be the limit distribution of QTW, and $\mu_{\rm SGD}$ the invariant Gibbs distribution of SGD. Two distributions are demanded to be $\delta$-risk-accurate:
 \begin{align}
   \mathbb{E}_{x\sim \mu_{\rm QTW}}f(x) - \min f = \mathbb{E}_{x\sim \mu_{\rm SGD}}f(x) - \min f = \delta.
   \label{eq:stand1}
\end{align}
\end{standard}
\stand{risk} ensures that two limit distributions yield the same expected
risk. Then, it is natural to compare how fast QTW and SGD would converge.
The algorithm spending less time is more efficient on finding any one global minimum.
Sometimes, the task is to find some target minima or one special
minimum.
In this case, using risk accuracy cannot emphasize the particularity
of the minima of interest and we may need the following standard:
\begin{standard}[Distance accuracy]\label{stand:distance}
Let $\mu_{\rm QTW}$ be the limit distribution of QTW, and $\mu_{\rm SGD}$ be the invariant Gibbs distribution of SGD. The minima of interest are $x_{j_k},~k = 1,\ldots,m,~j_k \in \{1,\ldots,N\}$. Let $D(\cdot, \cdot)$ be any
distance function.
Two distributions are demanded to be $\delta$-distance-accurate with respect to $x_{j_k}$ and $D(\cdot, \cdot)$:
 \begin{align}
   \mathbb{E}_{x\sim \mu_{\rm QTW}}\sum_k D(x,x_{j_k}) = \mathbb{E}_{x\sim \mu_{\rm SGD}}\sum_k D(x,x_{j_k}) = \delta.
   \label{eq:stand2}
\end{align}
\end{standard}

Conditions \eq{stand1} and \eq{stand2} can specify $h$ and $s$.
To see this, we first study the expected risk for quadratic functions:
\begin{lemma}\label{lem:stand1-quad}
Assume the objective function $f\colon \mathbb{R}^d \to \mathbb{R}$ is quadratic and
\begin{equation}
    f(0) = 0,~\nabla f(0) = 0,~\nabla^2 f(0) > 0,
\end{equation}
where the last inequality means the Hessian $\nabla^2 f(0)$ is positive definite. Then, we have
\begin{equation}
    \mathbb{E}_{x\sim \mu_{\rm QTW}}f(x) = \frac{\sqrt{2}h}{4}\tr \sqrt{\nabla^2 f(0)},
\qquad \mathbb{E}_{x\sim \mu_{\rm SGD}}f(x) = \frac{sd}{4}.
\end{equation}
\end{lemma}

\lem{stand1-quad} calculates the expected risks for a landscape with only one minimum whose proof is in \append{prooflem11}.
For landscapes with multiple minima, the limit distributions concentrate near the global minima and the objective function in a small neighborhood of any minimum can be approximated by a quadratic
function based on the assumptions.
Hence, we can obtain the following general estimations (the proof is postponed to \append{prooflem12}).
\begin{lemma}\label{lem:expectedrisk}
If $\delta$ is sufficiently small and the objective function $f\colon  \mathbb{R}^d \to \mathbb{R}$
satisfies assumptions in \sec{classicalpre} and \sec{quantumpre}, then \stand{risk} gives
\begin{equation}
    h = \frac{\delta}{\frac{\sqrt{2}}{4}\sum_{j=1}^N p(\infty,j)\tr\big(\sqrt{\nabla^2f(x_j)}\big)+ o_{\delta}(1)},
    \label{eq:hstand1}
\end{equation}
\begin{equation}
    s = \frac{\delta}{\frac{d}{4}(1+ o_{\delta}(1))}.
\end{equation}
\end{lemma}

That is, we establish a relationship between $h$ and $s$ by \stand{risk}.
Similarly, for \stand{distance}, we can have the following result:
\begin{lemma}\label{lem:expecteddis}
Assume the objective function $f\colon \mathbb{R}^d \to \mathbb{R}$ is quadratic and
\begin{equation}
    f(0) = 0,~\nabla f(0) = 0,~\nabla^2 f(0) > 0,
\end{equation}
where the last inequality means the Hessian $\nabla^2 f(0)$ is positive definite. We choose the distance function $D(x,y):=\|x-y\|^2_2,~\forall x,y\in \mathbb{R}^d$. Then, we have
\begin{align}
    \mathbb{E}_{x\sim \mu_{\rm QTW}}D(x,0) &= \frac{\sqrt{2}h}{2}\tr (\nabla^2 f(0))^{-1/2}; \\
\mathbb{E}_{x\sim \mu_{\rm SGD}}D(x,0) &=\frac{s}{2}\tr(\nabla^2 f(0))^{-1}.
\end{align}
\end{lemma}

The proof of \lem{expecteddis} is shown in \append{prooflem13}.
Similar to the process from \lem{stand1-quad} to \lem{expectedrisk},
\lem{expecteddis} may be generalized to general landscapes.
However, the generalization of \lem{expecteddis} is quite
complicated as the distance function and the wells of interest
are arbitrary.
So, we stop at \lem{expecteddis}.
Regardless of different standards,
\lem{expectedrisk} and \lem{expecteddis} present some similar intuition:
the dependence of $h$ on the flatness of wells are different from that of $s$, which is going to be shown in the following section as a source of quantum speedups.\footnote{Here, we use the Hessian matrix of $f$ at minima to quantify the concept ``flatness".}

\subsection{Illustrating advantages of quantum tunneling}\label{sec:illustration}

In this subsection, we compare QTW with SGD for several special
landscapes. The goal is to explore geometric properties of the landscapes that affect relative efficiencies of QTW and SGD.
Heuristically, the comparison reveals when quantum tunneling can be faster than thermal climbing (climbing over barriers between minima by stochastic motions), which are the two mechanisms behind QTW and many classical algorithms.

For simplicity, we focus on the following kind of landscapes:
\begin{definition}[One-dimensional partially periodic functions]\label{defn:funcforillus}
A function $f\colon\mathbb{R}\to \mathbb{R}$ is partially periodic if it satisfies the assumptions in \sec{classicalpre} and \sec{quantumpre}, and all minima $\{x_j:j=1,\ldots,N\}$ are in a bounded interval which is a period of $f$.
\end{definition}

\begin{figure}
  \centerline{
  \includegraphics[width=0.8\textwidth]{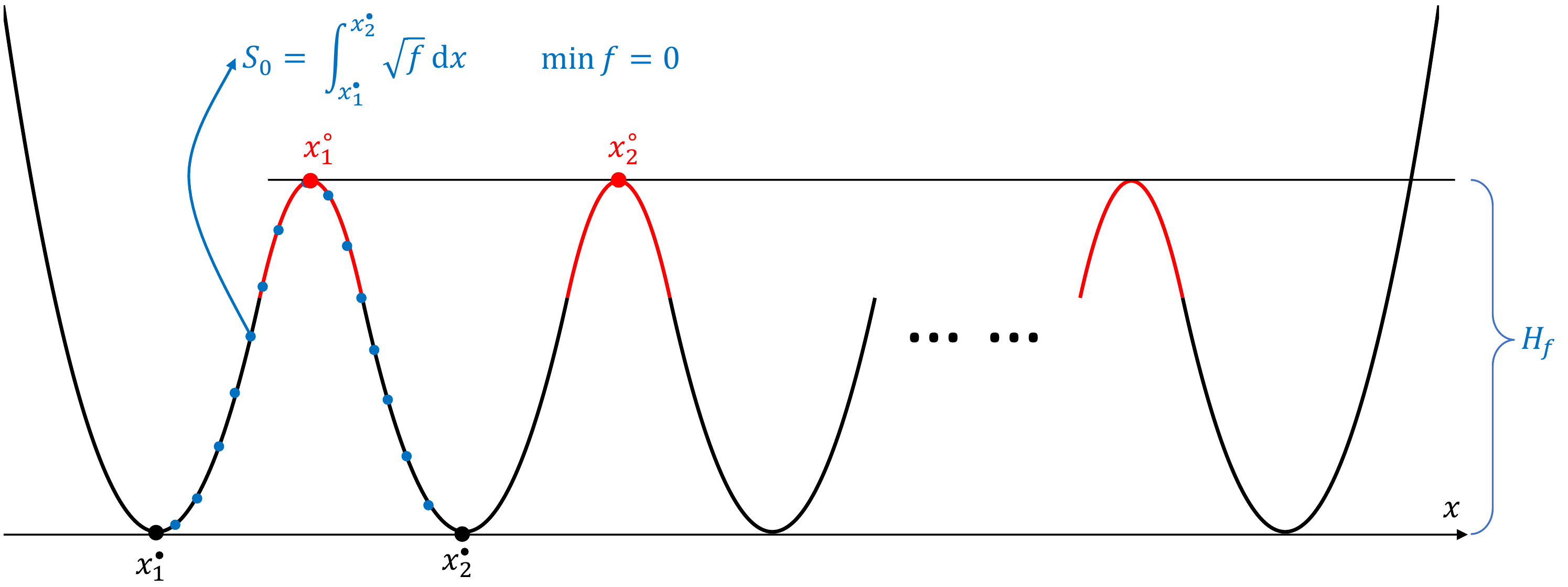}}
    \caption{A one-dimensional partially periodic function.
    }
\label{fig:oneperiodic}
\end{figure}

A sketch of functions in \defn{funcforillus} is shown in \fig{oneperiodic}. Neglect an exponentially small error and note the symmetry of the one-dimensional partially periodic function $f$, the interaction matrix under $\{ \ket{j}:j=1,\ldots,N\}$
should be given by
\begin{equation}
    H_{|\F} = \left(
    \begin{array}{cccccc}
        \mu & w &  &  &  & \\
         w & \mu & w & &  & \\
           & w & \mu & w &  & \\
           &   &  \ddots& \ddots  & \ddots &  \\
           &   &  & w & \mu &  w\\
           &   &  &  & w &  \mu\\
    \end{array}
    \right),
    \label{eq:Hforline}
\end{equation}
where $\mu$ is the energy of one local ground state and $w$ quantifies the tunneling effect between two adjacent wells.
Eigenstates and eigenvalues of $H_{|\F}$ can be given by the following lemma.
\begin{lemma}\label{lem:IlluEig}
The eigenstates and corresponding eigenvalues of the Hamiltonian \eq{Hforline} are given by
\begin{align}
    |E_k\y &= \sqrt{\frac{2}{N+1}}\sum_{j=1}^{N}\sin\big(\frac{jk\pi}{N+1}\big)\ket{j},~k=1,2,\ldots,N; \\
    E_k &= \mu + 2w \cos \frac{k\pi}{N+1},~k=1,2,\ldots,N.
\end{align}
\end{lemma}
To describe $w$ in detail, as shown in \fig{oneperiodic}, we introduce new notations $\{x^{\bullet}_j: j=1,\ldots,N\}$ and $\{x^{\circ}_j: j=1,\ldots,N-1\}$ to denote minima and saddle points, respectively.
A more general labeling of local minima and saddle points can be found in \append{W-L}.
The Morse saddle barrier reflecting height of the barrier in the present case can be given by $H_f = f(x^{\circ}_1) - f(x^{\bullet}_1)$.

Using results in \append{interactionmatrix}, we have:
\begin{lemma}[Tunneling amplitude]
The tunneling amplitude for the one-dimensional partially periodic function $f$ is given by
\begin{align}
   w = -2\sqrt{\frac{h f''(x^{\bullet}_1) H_f}{\sqrt{2}\pi}}
    e^{\int_{x_1^{\bullet}}^{x_1^{\circ}} (\sqrt{\frac{ f''(x^{\bullet}_1)}{2f(\xi)}} - \frac{1}{\xi- x_1^{\bullet}})\d \xi } e^{-\frac{S_0}{h}},~\mathrm{where}~S_0 = \int_{x_1^{\bullet}}^{x_2^{\bullet}} \sqrt{f(\xi)} \d\xi.
\end{align}
\end{lemma}
Now, we can obtain the spectrum of $H_{|\F}$ explicitly, and proceed by using \lem{qtwmixingtime} to get the quantum mixing time.
\begin{lemma}[Quantum mixing time]\label{lem:illuqtwmixing}
Staring from one local ground state of one minimum, the $\epsilon$-close mixing time of QTW is given by
\begin{align}
   T_{\rm mix}^{\rm QTW} = O \bigg(\frac{N^3}{\epsilon |w|} [1+(N-1)|O(h^{\infty})|] \bigg) = O(\mathrm{poly}(N,1/h,1/\epsilon))e^{\frac{S_0}{h}}.
\end{align}
\end{lemma}

Regarding SGD, we use the results introduced in \sec{classicalpre} to estimate the classical mixing time. First,
\begin{lemma}[Exponential decay constant]
In \prop{mic19}, let $\lambda = \delta_{s,1}/2s$ we have
\begin{align}
   \lambda = \bigg(\frac{\sqrt{f''(x^{\circ})f''(x^{\bullet})}}{2\pi} + o(s)\bigg) e^{-\frac{2 H_f}{s}}.
\end{align}
\end{lemma}
Then, by \cor{sgdmixing}, the following lemma holds.
\begin{lemma}[Classical mixing time]\label{lem:illusgdmixing}
Let $T_{\rm mix}^{\rm SGD}$ be the SGD $\epsilon$-close mixing time which is the minimum time enabling $\|\rho_{\rm SGD} (t,\cdot) - \mu_{\rm SGD}\|_{\mu_{\rm SGD}^{-1}}<\epsilon$, we have
\begin{align}
   T_{\rm mix}^{\rm SGD} = O \bigg(\frac{1}{\lambda} \ln\frac{\|\rho(0,\cdot) - \mu_{\rm SGD}(\cdot) \|_{\mu^{-1}_{\rm SGD}}}{\epsilon} \bigg) =  O(\mathrm{poly}(1/s,\ln(1/\epsilon)))e^{\frac{2H_f}{s}}.
\end{align}
\end{lemma}
Later, we do not focus on the dependence of the mixing time on $\epsilon$, as the norms ($L^1$ norm for QTW and $L^2(\mu^{-1}_{\rm SGD})$ for SGD) used to capture convergence are different.\footnote{
In terms of $\epsilon$, the same argument in \cite{AC21} but with evolution time $t$ of QTW chosen as a sum of some random variables instead of chosen uniformly in an interval, $T_{\rm mix}^{\rm QTW}$ can also achieve $\ln(1/\epsilon)$ dependence  instead of $1/\epsilon$.}
The dominant terms affecting running time of QTW and SGD are $e^{\frac{S_0}{h}}$ and
$e^{\frac{2H_f}{s}}$.

\begin{lemma}[Comparison on one-dimensional periodic landscapes]
Under \stand{risk}, let QTW and SGD be both $\delta$-accurate. For sufficiently small $\delta$, the QTW mixing time and SGD mixing time are dominated by
\begin{align}
   LT_{\rm mix}^{\rm QTW}:=e^{\frac{\sqrt{2}S_0f''(x^{\bullet})}{4\delta}}~\mathrm{and}~LT_{\rm mix}^{\rm SGD}:=e^{\frac{H_f}{2\delta}}, \textrm{ respectively.}
\end{align}
\end{lemma}

\begin{figure}
  \centerline{
  \includegraphics[width=0.5\textwidth]{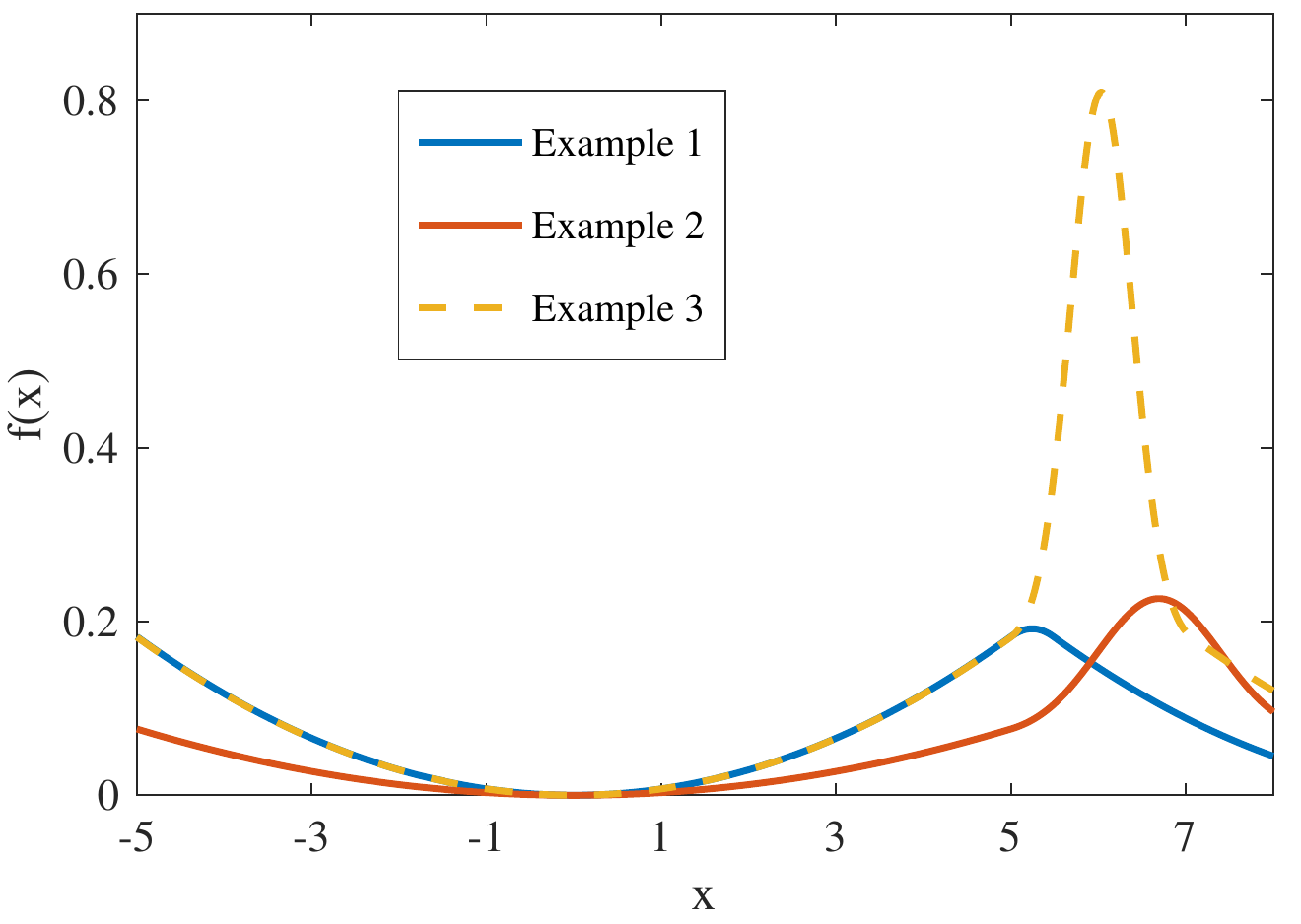}}
    \caption{The landscapes in \examp{critical}, \examp{flatness}, and \examp{sharp} (corresponding to Example 1, 2, and 3 in the figure, respectively) for illustrating the comparison between QTW and SGD.
    }
\label{fig:illuexamples}
\end{figure}

As concrete examples, we present several specific functions to illustrate the
advantages of quantum tunneling.
Since the function in the region of our interest is periodic,
we only need to specify the function value within one period to
construct a concrete example.
Without loss of generality, we set the interval $[-a,a+2b]$ to be
one period, where $[-a,a]$ is called the well region and $[a,a+2b]$ the barrier region.
The constructed landscape in $[-a,a+2b]$ is given by
\begin{align}
    f(x) = \left\{
    \begin{array}{ll}
        \frac{1}{2}k x^2\quad x\in [-a,a],\\[3pt]
         \frac{1}{2\pi \sigma^2}\exp\big(-\frac{(x-a-b)^2}{\sigma^2}\big) + \frac{1}{2}k a^2 - \varepsilon\quad x \in (a,a+2b].
    \end{array}
    \right.
    \label{eq:illuVx}
\end{align}
Here, to reduce free parameters, we make $f(x)$ differentiable at $a$, the boundary of the well, and the barrier, such that
\begin{align}
    b = \sigma \sqrt{\ln \frac{1}{2\pi \sigma^2 \varepsilon}}, \quad
    k = \frac{2\varepsilon}{a\sigma}\sqrt{\ln \frac{1}{2\pi \sigma^2 \varepsilon}}.
    \label{eq:illucondition}
\end{align}
\begin{remark}
Note that the function in \eq{illuVx} is not smooth.
We need to use the mollifier function $m_{r}$ (see detials in \append{mollified}) to smooth it such that
assumptions in \sec{classicalpre} and \sec{quantumpre} are
satisfied. Note that if $r\to 0$, the smoothed function will tend to be
$f$,
following results can be seen as arbitrarily accurate for a
smooth function arbitrarily close to \eq{illuVx}.
\end{remark}
By giving specific $a$, $\sigma$, and $\varepsilon$ in \eq{illuVx},
we can design landscapes with different properties.
Detailed variables, discussions and comparisons are given below.

\begin{example}[Critical case]\label{examp:critical}
For \eq{illuVx}, we set $a = 5.0$, $\sigma = 1.0$ and $\varepsilon=0.15$, and obtain $b\approx 0.243$ and $k\approx 0.0146$ by \eq{illucondition}.
\end{example}

\begin{example}[Flatness of minima]\label{examp:flatness}
For \eq{illuVx}, we set $a = 5.0$, $\sigma = 1.0$ and $\varepsilon=0.009$, and obtain $b\approx 1.69$ and $k\approx 0.00610$ by \eq{illucondition}.
\end{example}

\begin{example}[Sharpness of barriers]\label{examp:sharp}
For \eq{illuVx}, we set $a = 5.0$, $\sigma = 0.5$ and $\varepsilon=0.0088$, and obtain $b\approx 1.03$ and $k\approx 0.0146$ by \eq{illucondition}.
\end{example}

\fig{illuexamples} explicitly shows the shapes of above examples.
The barrier region in \examp{critical} is small and most of
the function in one period is quadratic, which is
similar to the case introduced in \sec{onedimexp}.
The Morse saddle barrier $H_f$ of \examp{flatness}
is approximately equal to that of \examp{critical}, whereas, in \examp{flatness},
the well is more flat and the barrier is thicker.
\examp{sharp} has almost the same well as \examp{critical} but is equipped with a much higher barrier.

We call \examp{critical} as the critical case because
QTW and SGD perform nearly the same on it in terms of
the leading terms $LT_{\rm mix}^{\rm QTW}$ and $LT_{\rm mix}^{\rm SGD}$:
\begin{lemma}
\examp{critical} satisfies $b\ll a$ and $\frac{1}{2}ka^2 \approx H_f$. For such a landscape, we have
\begin{align}
  \ln LT_{\rm mix}^{\rm QTW} &= \frac{H_f}{2\delta}\left[1+ 2b/a+ o(b/a) \right] + o(\delta),\\
  \ln LT_{\rm mix}^{\rm SGD} &= \frac{H_f}{2\delta}(1+o(\delta)).
\end{align}
\end{lemma}

QTW mixes faster on both \examp{flatness} and \examp{sharp} for sufficiently small $\delta$. Specifically, we have
\begin{lemma}
For \examp{flatness} and \examp{sharp}, the following holds
\begin{align}
  \ln LT_{\rm mix}^{\rm QTW} &< \frac{k}{4\delta}a^2+ \frac{k}{2\delta}ab + \frac{\sqrt{2k}}{4\delta} + o(\delta),\\
  \ln LT_{\rm mix}^{\rm SGD} &= \frac{1}{4\pi \sigma^2\delta}+\frac{k}{4\delta}a^2-\frac{\epsilon}{2\delta} + o(\delta).
\end{align}
Substituting the parameters, it is true for both \examp{flatness} and \examp{sharp} that
\begin{align}
   \frac{k}{4}a^2+ \frac{k}{2}ab + \frac{\sqrt{2k}}{4}
   < \frac{1}{4\pi\sigma^2}+\frac{k}{4}a^2-\frac{\epsilon}{2}.
\end{align}
\end{lemma}

Comparing to the critical case \examp{critical},
\examp{flatness} has a thicker barrier, which increases $S_0$ and causes difficulty for QTW.
However, QTW can perform better in \examp{flatness}. This is mainly
due to the more flat well of \examp{flatness}.
Recall that by \stand{risk}, to ensure $\delta$-risk-accuracy, $h$ and $s$ should be
\begin{align}
    h= \frac{\delta}{\frac{\sqrt{2k}}{4} + o_{\delta}(1)}
    \quad\mathrm{and}\quad
    s= \frac{\delta}{\frac{1}{4}(1+o_{\delta}(1))},
\end{align}
respectively.
That is, under the same risk accuracy, $h$ can be much larger than $s$
if the well is flat ($k$ is small), making tunneling easier.
Note that there is a trade-off between accuracy and time cost:
smaller $h$ (or $s$) ensures high accuracy but make tunneling effects (or thermal diffusion) weaker; conversely, larger $h$ (or $s$) permits faster tunneling (or diffusion) but yields inaccurate results.
Discussions on quantum tunneling effects usually focus on
properties of the barrier.
In the present study, since we aim to find global minima,
the precision of results obtained is one important concern.
Therefore, the flatness of wells, which affects differently on the accuracy of
QTW and SGD, is a crucial property determining the runtime of QTW and SGD.
Loosely speaking, QTW is faster than SGD on landscapes with flat wells.

\examp{sharp} adheres to the intuition that quantum tunneling is efficient on functions with tall and
thin barriers. The wells of \examp{sharp} are almost the same as those of the critical case \examp{critical}.
QTW can be faster in \examp{sharp} because we add a sharp barrier between wells. By \lem{illusgdmixing}, a high barrier (i.e., large $H_f$) would significantly hinder thermal climbing.
However, the tall barrier is sufficiently thin, such that $S_0 = 2\int_{0}^{a+b}\sqrt{f(x)}\d x$ can still be small and by \lem{illuqtwmixing}, the tunneling effect would be strong.

Moreover, in high dimensions, the distribution of wells can be very different from being on a line. As shown in \append{welldistr},
distribution of wells can largely affect the dependence of time on $N$.
However, such relation between the distribution of wells and running time
is not explicitly shown for SGD. Therefore, the distribution of wells can also be a factor of quantum speedups.

In summary, we can conclude our \slog.

\subsection{Efficient quantum tunneling for solving a classically difficult hitting problem}\label{sec:separation}

The above examples compare QTW driven by quantum tunneling
with SGD. In this section, an exponential separation in terms of query complexity between QTW given initial states and classical algorithms knowing one well will be shown for a specific hitting problem on a constructed landscape.

The landscape $f(\mathbf{x})$ we construct lives in $\mathbb{R}^d$. We use $\|\cdot\|$ to denote the $\ell_2$ norm of vectors, namely, $\|\mathbf{x}\| = \sqrt{\mathbf{x}\cdot \mathbf{x}}$.
Let $\mathbb{B}(\mathbf{x},r)$ denote a $d$-dimensional ball centered at $\mathbf{x}$ with radius $r$. A special direction $\mathbf{v}$ is randomly chosen from the $d$-dimensional unit sphere. We define two regions $W_- = \mathbb{B}(\mathbf{0},a)$ and $W_+ = \mathbb{B}(2b\mathbf{v},a)$ with $b\geq a$.
Let $R$ be sufficiently large s.t. $W_-, W_+ \subset \mathbb{B}(\mathbf{0},R)$.
We denote the region $\{\mathbf{x}\mid\mathbf{x} \in \mathbb{B}(\mathbf{0},R), ~|\mathbf{x}\cdot \mathbf{v}| \leq w\}$ by $S_{\mathbf{v}}$, where $w$ will be chosen from $[\sqrt{3}a/2,0)$. We denote
\begin{align}
B_{\mathbf{v}}:=\{\mathbf{x}\mid w< \mathbf{x}\cdot \mathbf{v} < 2b-w,~\sqrt{\|\mathbf{x}\|^2 - (\mathbf{x}\cdot \mathbf{v})^2} < \sqrt{a^2-w^2}, \bf{x}\notin W_-\cup W_+\}.
\end{align}
\fig{provable-acceleration} illustrates positions of the newly defined regions. The constructed function $f$ is given by
\begin{align}
    f(\mathbf{x}) = \left\{
    \begin{array}{ll}
     \frac{1}{2}\omega^2 \|\mathbf{x}\|^2,~\mathbf{x} \in W_-, \\[3pt]
     \frac{1}{2}\omega^2 \|\mathbf{x}-2b\mathbf{v}\|^2,~\mathbf{x} \in W_+, \\[3pt]
     H_1,~\mathbf{x} \in B_{\mathbf{v}},\\[3pt]
     H_2,~\mathrm{otherwise}.
    \end{array}
    \right.
    \label{eq:hardinstance}
\end{align}
Here, we define $H_0 = \frac{1}{2}\omega^2a^2$ and demand that $0<H_0 \sim H_1 \ll H_2$.
\begin{remark}
The landscape $f$ in \eq{hardinstance} is not smooth and should be smoothed to be $F_r$ with the help of a mollifier function $m_r$ (see details in \append{mollified}) such that assumptions in \sec{quantumpre} can be satisfied.
Because when $r\to 0$, $F_r\to f_r$, we can always find sufficiently small $r$ to make the following conclusions based on $f$ valid for $F_r$.
\end{remark}

There are two global minima, $\mathbf{0}$ and $2b\mathbf{v}$, of the function $f$. Given that we know $\mathbf{0}$ is a minimum, our goal is to find the other one. To avoid
complicated justifications, we deal with a simpler problem:
\begin{problem}\label{prb:provable}
    For the $f$ in \eq{hardinstance}, given that we only know $\mathbf{0}$ is a global minimum, find any point in  $W_+$.
\end{problem}

\subsubsection{Classical lower bound}\label{sec:clb}
Due to the concentration of measure, for any point $\mathbf{x} \in \mathbb{B}(\mathbf{0},R)$, the probability of $\mathbf{x}\in S_{\mathbf{v}}$ is given by
\begin{align}
    P(\mathbf{x} \in S_{\mathbf{v}}) \geq 1 - O \big(e^{-\frac{dw^2}{2R^2}}\big).
    \label{eq:PxinSv}
\end{align}
Intuitively, restricted in $\mathbb{B}(\mathbf{0},R)$, any classical algorithm cannot escape from $S_{\mathbf{v}}$
efficiently. In $\mathbb{R}^d$, queries out of $\mathbb{B}(\mathbf{0},R)$ provide no information about the landscape inside $\mathbb{B}(\mathbf{0},R)$ and are unable to help to escape from $S_{\mathbf{v}}$. Therefore, classical algorithms cannot solve \prb{provable} efficiently with or without being constrained in $\mathbb{B}(\mathbf{0},R)$. To rigorously prove above intuitions, we first introduce a mathematical result indicating \eq{PxinSv}:
\begin{lemma}[Measure concentration for the sphere]\label{lem:mconcentration}
Let $\mathbb{S}^{d-1} = \{\mathbf{x}:\|\mathbf{x}\| =1 \}$ be the unit sphere in $\mathbb{R}^d$. Let $\mathrm{Cap}(\epsilon)$ denote the spherical cap of height
$\epsilon$ above the origin (see the left part of \fig{capandcone}). We have
\begin{align}
    \frac{\mathrm{Area}(\mathrm{Cap}(\epsilon))}{\mathrm{Area}(\mathbb{S}^{d-1})} \leq e^{-d\epsilon^2/2}.
\end{align}
\end{lemma}
\begin{figure}
  \centerline{
  \includegraphics[width=0.5\textwidth]{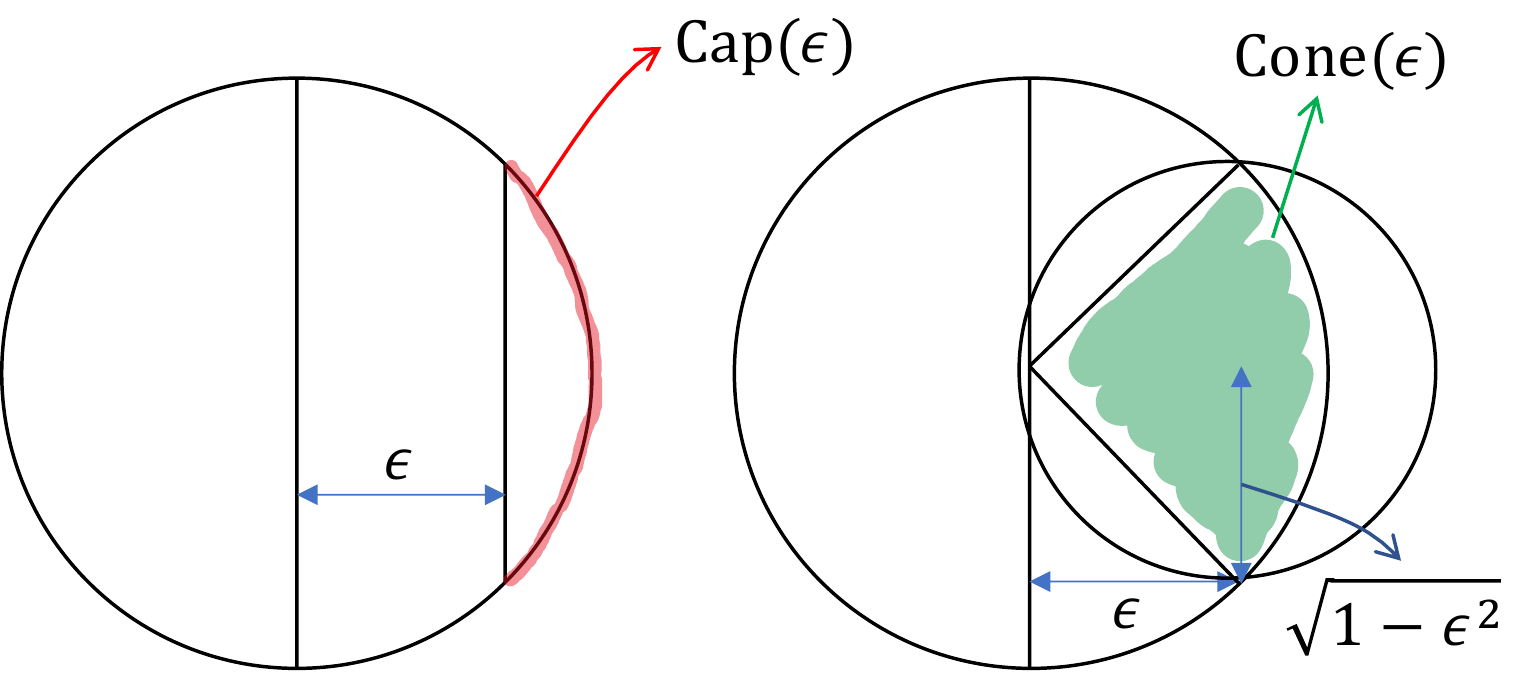}}
    \caption{Estimating the area of a spherical cap.
    }
\label{fig:capandcone}
\end{figure}

The estimation details are presented in \append{prooflem22}.
Subsequently, it is readily to have (see details in \append{prooflem23}):
\begin{lemma}\label{lem:provable-onep}
For any randomly chosen point $\mathbf{x}\in \mathbb{B}(\mathbf{0},R)$,
the probability of $\mathbf{x}\notin S_{\mathbf{v}}$ is $P(\mathbf{x}\notin S_{\mathbf{v}})\leq 2 e^{-\frac{dw^2}{2R^2}}$.
\end{lemma}

Recall that $w\in [a/2,a)$
and $R$ are independent of $d$, the measure of the region in $\mathbb{B}(\mathbf{0},R)$ and outside
$S_{\mathbf{v}}$ is exponentially small with respect to the dimension $d$. By \defn{localquery}, classical algorithms depend on
an adaptive sequence of points.
we now need to demonstrate that it is difficult for the points to hit regions beyond $S_{\mathbf{v}}$.

\begin{lemma}\label{lem:provable-manyp}
For any classical algorithm (see \defn{localquery}), after running $T$ times, we get a sequence of points and corresponding queries $\{\mathbf{x}_i, q(\mathbf{x}_i)\}_{i=1}^T$. Restricted in $\mathbb{B}(\mathbf{0},R)$, as long as any $q(\mathbf{x}_i)~(\mathbf{x}\in S_{\mathbf{v}})$ is independent of $\mathbf{v}$, the probability $P(\exists t\leq T:\mathbf{x}_t\notin S_{\mathbf{v}})\leq 2T e^{-\frac{dw^2}{2R^2}}$.
\end{lemma}

We prove \lem{provable-manyp} in \append{prooflem24}.
Now, we can prove that if the number of points and queries is small,
with high probability, any classical algorithm cannot escape from $S_{\mathbf{v}}$.
Rigorously, we have
\begin{proposition}[Classical lower bound]\label{prop:provable-exp}
Any classical algorithm (\defn{localquery}) will fail, with high probability, to solve \prb{provable} given only $o(e^{\frac{dw^2}{4R^2}})$ local queries with or without being restricted in $\mathbb{B}(\mathbf{0},R)$.
\end{proposition}
The proof sketch of \prop{provable-exp} goes as follows (see proof details in \append{proofprop4}).
By \lem{provable-manyp}, it suffices to demonstrate that restricted in the ball $\mathbb{B}$, classical algorithms
cannot escape from $W_-$ and hit $W_+$ efficiently.
The left thing is to show that queries outside $\mathbb{B}(\mathbf{0},R)$ provide no information about $\mathbb{B}(\mathbf{0},R)$.
And thus, without being restricted in $\mathbb{B}(\mathbf{0},R)$, classical algorithms still cannot hit $W_+$ by subexponential queries with high probability.

\subsubsection{Quantum upper bound}\label{sec:qub}
We now focus on the time needed for quantum tunneling
to solve \prb{provable}. The landscape \eq{hardinstance} satisfies \assum{3}
\begin{align}
    0 = \min f < \lim_{\|\mathbf{x}\|\to \infty} f = H_2,\quad  f^{-1}(0) = \{ \mathbf{0}\} \cup \{2b\mathbf{v}\},
\end{align}
where $U_-:=\{ \mathbf{0}\}$ and $U_+ := \{2b\mathbf{v}\}$ are called as wells by definition.
The neighborhoods of the two wells are quadratic, enabling the wells and corresponding local ground states to satisfy \eq{sufficient-hypo1} and \eq{sufficient-hypo2}.
Moreover, due to the symmetry of the function \eq{hardinstance}, the local
ground states are also symmetric. Therefore, \assum{TunEnCon} can be satisfied.
To use \assum{3}, we only need to verify the conditions in \assum{4}, leading to the following lemma.

\begin{lemma}\label{lem:provableS0}
There exists a unique Agmon geodesic, denoted $\boldsymbol{\gamma}_{-+}: \mathbb{R} \to \mathbb{R}^d$, which links $U_-$ and $U_+$:
\begin{equation}
    \boldsymbol{\gamma}_{-+}(s) = s \mathbf{v},\quad s\in [0,2b].
    \label{eq:provable-geodesic}
\end{equation}
And the Agmon distance $S_0 := d(U_-, U_+)$ is
\begin{equation}
    S_0 = \int_{-1/2}^{1/2} \sqrt{f(\boldsymbol{\gamma}_{-+}(s))} \d s =
    \frac{1}{\sqrt{2}}\omega a^2 + 2(b-a)\sqrt{H_1}.
    \label{eq:provable-S0}
\end{equation}
\end{lemma}

The calculation details of \lem{provableS0} are presented in \append{prooflem25}.
We are now ready to calculate the interaction matrix explicitly (see details in \append{prooflem26}):
\begin{lemma}\label{lem:provablenu}
Under the two orthonormalized local ground states,
The interaction matrix is of the form
\begin{equation}
    P = \left(
    \begin{array}{cc}
        \mu & \nu\\
         \nu & \mu
    \end{array}
    \label{eq:provable-matrix}
    \right),
\end{equation}
and the next-to-leading order formula of $w$ is given by
\begin{align}
    \nu = -\sqrt{\frac{2h}{\pi}}
    \sqrt{\frac{H_1(\sqrt{2}\omega)^d}{(4\sqrt{H_1}/b)^{d-1}}}\exp\left(-\frac{S_0}{h} + \frac{\omega d(b-a)}{\sqrt{2H_1}} - 2d \ln \frac{b}{a}\right).
    \label{eq:provable-w}
\end{align}
\end{lemma}

Using the explicit tunneling amplitude, we can estimate the time needed for quantum tunneling.
\begin{proposition}[Quantum upper bound]\label{prop:provable-poly}
For any dimension $d$,
we can always choose appropriate $h$, $\omega$, $a$, $b$, $H_1$, $H_2$, and $w$ satisfying previous restrictions, such that, given the local ground state associated to $W_-$ under the choosing $h$ as initial state, QTW can solve \prb{provable} with high probability $1-(1-C)^n$ using only $nO(\mathrm{poly}(d))$ queries, where $0<C<1$ is a constant independent of $d$.
\end{proposition}
\begin{remark}
In \prop{provable-poly}, the constant $C$ can be understood as the probability of successful hitting in one trial and $n$ the number of trails. To reach a high probability of success, say $99$\%, the number of trials needed, $M$, enabling $1-(1-C)^M\geq 99$\%, is a constant independent of $d$. Since one trial needs only one initial state, only a constant number of copies (e.g., $M$ copies) of the local ground state are needed.
\end{remark}

The proof of \prop{provable-poly} is postponed to \append{proofprop5} which is explained briefly as follows.
We take all the adjustable parameters as functions of $d$ and
discuss the evolution time as a function of $d$.
First, we have $h= \Theta(1/d)$ for sufficiently large $d$,
which can eliminate the negative effects of measure concentration brought by increasing dimension, and on the other hand we prove that our theory on quantum tunneling walks is still valid.
Thus, the quantum wave distributes near $W_-$ or $W_+$, and the limit distribution $\mu_{\rm QTW}$ permits a probability of finding the particle in $W_+$ larger than some constant independent of $d$.
Then, based on the results of semi-classical analysis, we can tune the function values in $W_-$, $W_+$, and $B_{\mathbf{v}}$ such that the time needed for tunneling is a polynomial of $d$.
As a result, the last three conditions at the end of \sec{intro} can be satisfied, and the first and third conditions suggest that with high probability, QTW can hit $W_+$ with queries polynomial in $d$.
Finally, given the fact that we can use $\tilde{O}(t)$ quantum queries to evolve QTW for time $t$,
with high probability QTW can hit $W_+$ with queries polynomial in $d$.

Combining the results of \prop{provable-poly} and \prop{provable-exp},
we can obtain \thm{provableinformal}, which is restated in a more rigorous way as follows.
\begin{theorem}\label{thm:thm2re}
For any dimension $d$, there exists a landscape with the form \eq{hardinstance} such that with a high probability $1-(1-A)^n$, QTW can solve \prb{provable} with $nO(\mathrm{poly}(d))$ queries given the local ground state associated to $U_-$, but with a high probability $1-e^{-dB}$,
no classical algorithm (\defn{localquery}) can solve the same problem
for the same landscape $f$ with $o(e^{dB})$ queries,
where $0<A<1$ and $B>0$ are two constants independent of $d$.
\end{theorem}

\subsubsection{The significance of proper initial states}\label{sec:no-exp-sep-init}

The hardness of \prb{provable} can be abstracted as that of finding an exponentially small cone on a landscape which is isotropic outside the special cone.\footnote{Specifically, this cone is a region associated to the special direction $\bf{v}$, $\{\bf{x}:\bf{x}\cdot \bf{v}/\|\bf{x}\|<C \}$ for some constant $C$, which contains the parts $B_{\bf{v}}$ and $W_+$.} There can be exponentially many such cones disjoint with each other.
Therefore, it can be proved that by solving \prb{provable} in $\mathbb{R}^d$, we can solve an unstructured search problem with a size exponential in $d$.

To show this, we first introduce unstructured search. Say, we are given $N$ data points, only one of which is assigned the value $1$ and all other points are assigned a value $0$. The goal is to find the point assigned $1$ with an oracle outputting the assigned value of the input point. Intuitively, each data point can be mapped to a unique cone in $\mathbb{R}^d$ and the point assigned $1$ should correspond to the cone containing $B_{\bf{v}}$ and $W_+$.
In this case, solving \prb{provable} can lead to the data point we want to find. Precisely speaking, if there is a quantum algorithm that can solve \prb{provable} with queries polynomial in $d$, it can solve an unstructured search whose size $N$ is exponential in $d$ within queries polynomial in $d$. That is, we can solve an $N$-size unstructured search within $O(\mathrm{poly}(\log N))$ queries with the help of the efficient algorithm for \prb{provable}.

However, it is well known that quantum algorithms have a query complexity lower bound $\Omega(\sqrt{N})$ in solving unstructured search with $N$ data points \cite{BBBV97}. Therefore, we can conclude
\begin{proposition}\label{prop:nospeedup}
No quantum algorithm can solve \prb{provable} within queries polynomial in $d$.
\end{proposition}
We prove \prop{nospeedup} and related claims rigorously in \append{nospeedup}.
It seems that \prop{nospeedup} contradicts with \prop{provable-poly}.
But there is actually no paradox as in \prop{provable-poly} QTW does not solve \prb{provable} faithfully.
The local ground state $\ket{\Phi_-}$ associated to $U_-$ under proper quantum learning rate $h$ is given to QTW as prior knowledge.
To establish polynomial decay tunneling effect, the state $\ket{\Phi_-}$ has non-vanishing probability (maybe an inverse polynomial of $d$) in $W_+$.
The state $\ket{\Phi_-}$ indicates a lot about the special direction $\bf{v}$ for QTW, such that what QTW does cannot be equivalent to unstructured search.
Indeed, by the same spirit of \prop{nospeedup}, the state $\ket{\Phi_-}$ cannot be prepared within polynomial queries, or we can reach $W_+$ efficiently by measuring $\ket{\Phi_-}$ repeatedly.

We admit that \thm{thm2re} requires the initial quantum state.
Note that QTW only uses $M$ copies of the local ground state $\ket{\Phi_-}$ to hit $W_+$ with high probability in polynomial time, where $M$ is a number independent of $d$. If the possibility of learning about $\bf{v}$ from sampling tends to $0$ when $d\to \infty$, which is likely to be true, the expected queries needed by classical algorithms to hit $W_+$ cannot be subexponential in $d$.
In this case, we have an exponential quantum-classical separation in evaluation queries even classical algorithms are given a  constant number of samples from the initial distribution $|\ip{\bf{x}}{\Phi_-}|^2$. Essentially, this is because no classical evolution can make good use of the samples of the initial state.

\section{Numerical Experiments}\label{sec:num}
We conduct numerical experiments to examine our theory.
All results and plots are obtained by simulations on a
classical computer (Dual-Core Intel Core i5 Processor, 16GB memory) via MATLAB 2020b.
Details of all numerical settings can be found in \append{fullnum}.
And code is avaliable at \href{https://github.com/liuyz0/Quantum-tunneling}{https://github.com/liuyz0/Quantum-tunneling}.

QTW is simulated by solving the Schr\"odinger equation by numerical methods,
and SGD is performed with first-order queries and the noise of each step follows the standard Gaussian distribution.

\paragraph{Quantum-classical comparisons.}
To corroborate our \slog,
we numerically study the performance of QTW and SGD on concrete examples (see details in \append{numcomparison}).
The quantum learning rate $h$ and the classical learning rate $s$ are determined under \stand{risk} which equalizes expected risks yielded by QTW and SGD.

The task is to hit a target neighborhood of one minimum beginning at another
designated minimum.
In \fig{histogramsExp123}, results on classical and quantum hitting time are shown.
We examine QTW and SGD on three landscapes, Example 1, 2, and 3 in \fig{histogramsExp123} which correspond to concrete functions given by
\examp{critical}, \examp{flatness}, and \examp{sharp} in \sec{illustration}, respectively.
We conduct 1000 experiments for QTW and SGD on each example.
For QTW, we use an \emph{experiment} to denote a process repeating trials until successfully hitting, where each \emph{trial} initiates QTW once and
measures the position at $t$ randomly chosen from $[0,\tau]$. 
For SGD, an experiment begins at a designated minimum and stops until SGD hits the target region.
The evolution time of an experiment is the sum of evolution time of the trials the experiment contains.

We use $T^{\rm QTW}_{\rm hit}$ and $T^{\rm SGD}_{\rm hit}$
to denote the evolution time of one experiment for QTW and SGD, respectively.
In \fig{histogramsExp123}, the histograms compare $T^{\rm QTW}_{\rm hit}$ with $T^{\rm SGD}_{\rm hit}/10$, and
all presented examples demonstrate that QTW is faster.
The number of quantum queries is approximately
$\tilde{O}(\|f\|_{\infty} T^{\rm QTW}_{\rm hit})$ and the number of classical queries is $\Omega(T^{\rm SGD}_{\rm hit}/s)$.
In addition, in the three examples $\|f\|_{\infty} \leq 0.85$ and $s<0.2$, and
quantum advantage exists in terms of query complexity.

This result matches our theory at large.
For Example 1, we make direct comparison between the exponential terms $e^{S_0/h}$ and $e^{2H_f/s}$, and to remove the coefficients in front of them,
we divide $T^{\rm SGD}_{\rm hit}$ by $10$ such that $T^{\rm SGD}_{\rm hit}/10$ has similar distribution t $T^{\rm QTW}_{\rm hit}$ for Example 1.
In this way, we observe that whether $T^{\rm SGD}_{\rm hit}/10$ is relatively larger than $T^{\rm QTW}_{\rm hit}$ is determined only by $e^{S_0/h}$ and $e^{2H_f/s}$.

For Example 2, $T^{\rm QTW}_{\rm hit}$ is not much smaller than $T^{\rm SGD}_{\rm hit}/10$, which is not completely coherent with our theory.
This result can be explained as that for Example 2, the quantum learning rate $h$ is not small enough such that the initial state prepared does not well stay near a low energy subspace.
Experiments on Example 2 suggest that higher energy may not be able to help quantum tunneling to run faster.

For Example 3, the $h$ chosen is small enough (see details in \append{numcomparison})
such that a significant quantum speedup is achieved as expected. In \fig{histogramsExp123}, $T^{\rm SGD}_{\rm hit}/10$ is even several orders of magnitude larger than $T^{\rm QTW}_{\rm hit}$.

\begin{figure}
\centering
\includegraphics[width=\linewidth]{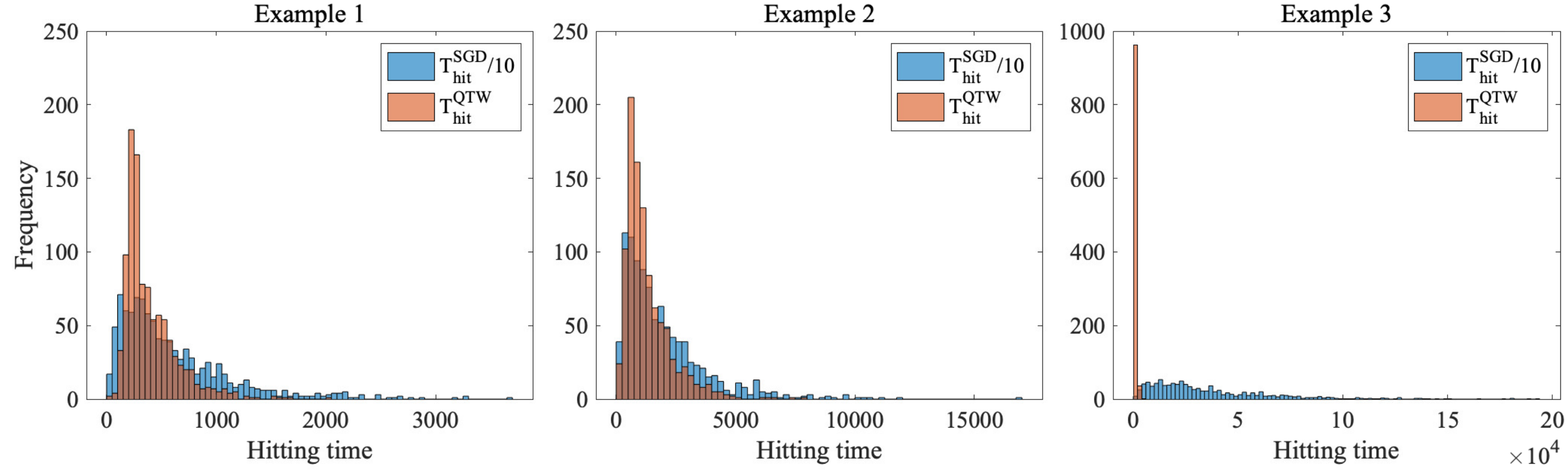}
\caption{Quantum-classical comparison between SGD and QTW on three landscapes. Example 1 is the critical case where the exponential terms in QTW and SGD evolution time are equal for sufficiently small accuracy $\delta$.
Example 2 has flatter minima but similar barriers compared to Example 1, enabling QTW to be faster.
Example 3 possesses the same flatness of minima as Example 1 but is equipped with sharp but thin barriers, enabling larger quantum speedups. We take $\tau=288, 800, 600$ in the three examples, respectively.}
\label{fig:histogramsExp123}
\end{figure}

\paragraph{Dimension dependence.}
Due to the limitations of solving the Schr\"odinger equation on classical
computers, QTW is simulated only in low dimensions (i.e., $d=1$ and $d =2$).
Here we examine \thm{provableinformal} by testing SGD and its the classical lower bound.

The classical lower bound in \thm{provableinformal} ensures that
for any $s$, SGD cannot cannot escape from $S_{\mathbf{v}}$ with subexponential queries with high probability.
Based on the constructed landscape with parameters specified in \append{numdimdep}, we test SGD with different learning rates ($s\in [0.1,1]$)
in various dimensions ($d \in [15,95]$).
For each dimension and each $s$, 1000 experiments are conducted.
The number of steps spent to escaping from $S_{\mathbf{v}}$ in one experiment is denoted as $Q_{\rm esc}$.
Here, we present the relationship between average $Q_{\rm esc}$ and the dimension $d$ in \fig{dimdepenmean} (more details are deferred to \append{numdimdep}).

For each fixed learning rate $s$, we observe that with the increase of $d$, the average $Q_{\rm esc}$ remains constant initially and then increase exponentially with respect to $d$.
Increasing $s$ yields a smaller initial constant but larger exponential rate.
Nevertheless, for all $s$, $Q_{\rm esc}$ eventually increases
exponentially with respect to $d$ (the triangle in \fig{dimdepenmean} shows the slope $1/256$ corresponding to the exponential function $e^{d/256}$ which is a lower bound of the average $Q_{\rm esc}$), supporting our prediction.

\begin{figure}
\centering
\begin{minipage}[t]{0.46\linewidth}
\centering
\includegraphics[width = \linewidth]{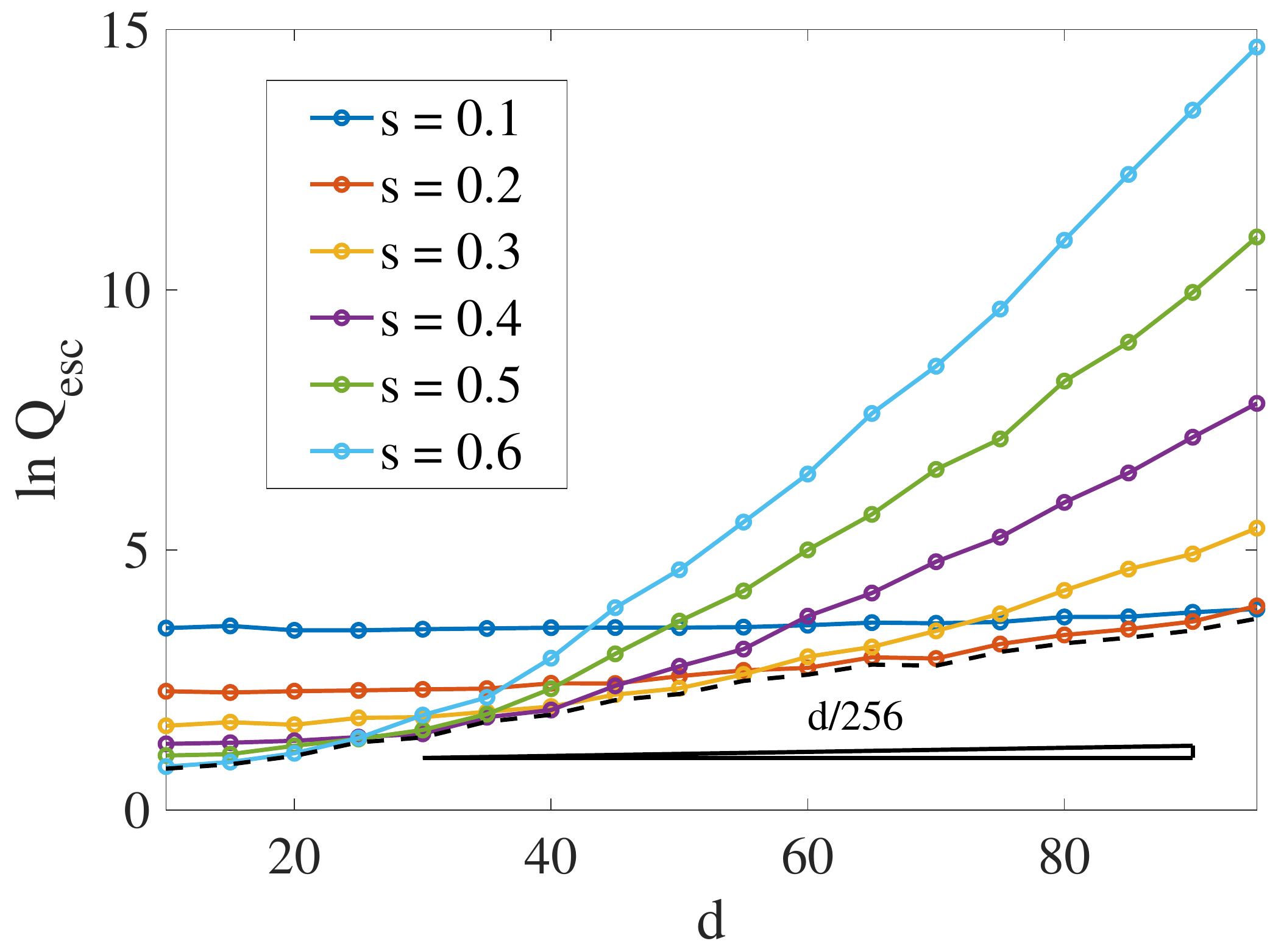}
\caption{Relationship between the average $Q_{\rm esc}$ and the dimension $d$ under different $s$. The dashed line captures the lower bound of the average $Q_{\rm esc}$.}
\label{fig:dimdepenmean}
\end{minipage}
\hspace{6mm}
\begin{minipage}[t]{0.46\linewidth}
\centering
\includegraphics[width =\linewidth]{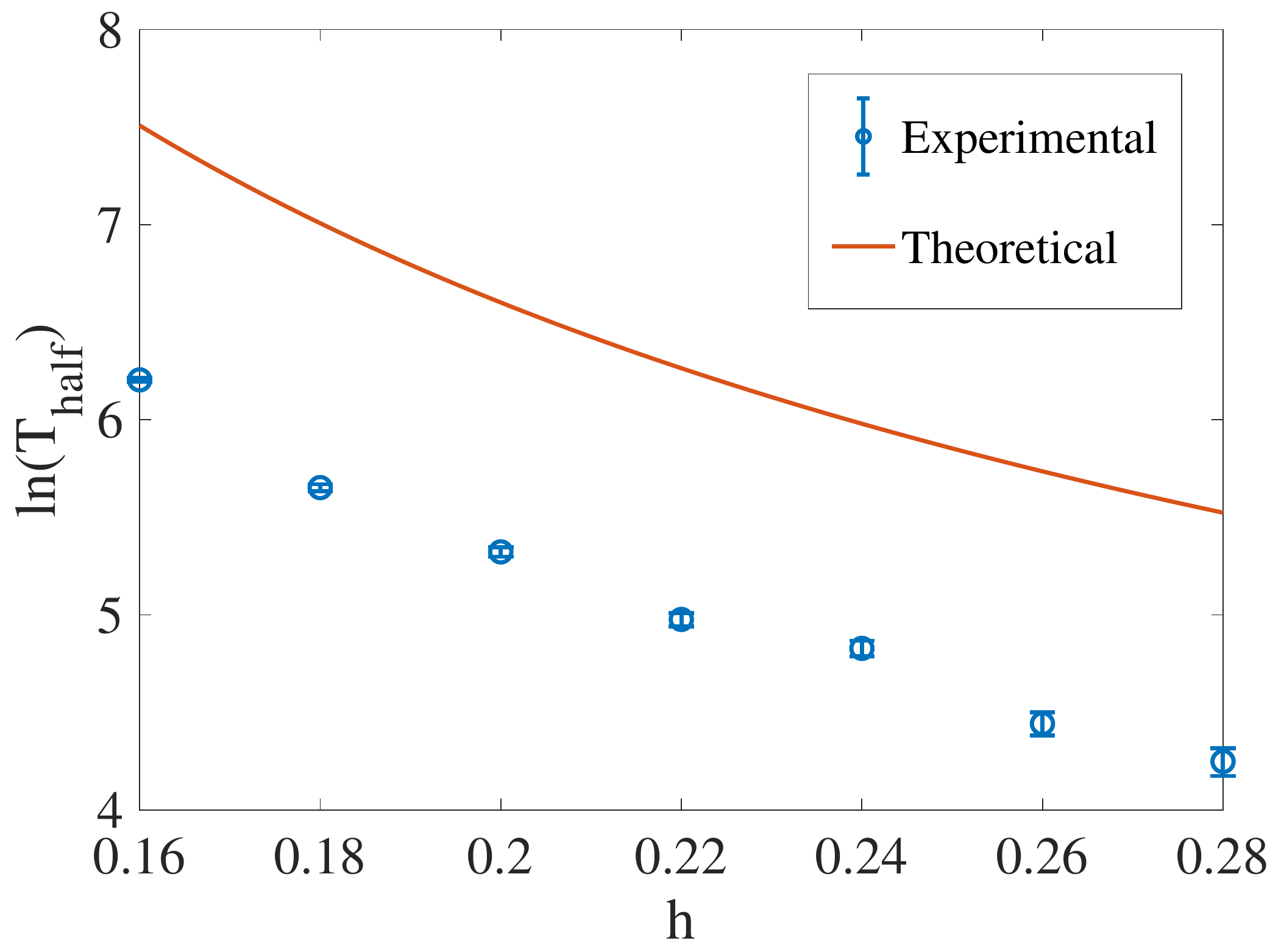}
\caption{$T_{\rm half}$ with respect to $h$ on a fixed landscape: theoretical prediction (red solid line) and time in experiments (blue circles).}
\label{fig:qlrdepend}
\end{minipage}
\end{figure}

\paragraph{Dependence on the quantum learning rate.}
In QTW, the quantum learning rate $h$ is one of the most important variables.
\thm{informalQTW} gives a general relationship between $h$ and the evolution time of QTW.
We further test the relationship on the landscape constructed in \thm{provableinformal} (dimension $d=2$) with specified parameters given in \append{numqlrdep}.
Since the landscape has two symmetric wells, the time for tunneling from one well to the other, $T_{\rm half}$, is explicitly linked to $\Delta E$ (i.e., $T_{\rm half} = \pi/\Delta E$).
On this concrete landscape, $\Delta E$ can be predicted, giving that
\begin{align}
    \ln T_{\rm half} = \frac{S_0}{h} - \frac{1}{2}\ln h + C_f,
\end{align}
where $C_f$ is a constant depending on $f$ and can be explicitly calculated.
Starting from one well, we stop when the probability of finding
the other well exceeds 90\% and record the evolution time as $T_{\rm half}$.
Experiments on $T_{\rm half}$ is shown in \fig{qlrdepend}. The results match our theory
except a constant difference between the predicted and experimental $\ln T_{\rm half}$, indicating the correctness of $\frac{S_0}{h} - \frac{1}{2}\ln h$.
The constant difference emerges because we stop evolution when the probability of tunneling exceeds 90\%, while theoretical $T_{\rm half}$ takes the time when the probability is nearly 100\%.

To conclude, several aspects of the present theory are well supported by numerical experiments.

\section{Discussion}\label{sec:discussion}
In this paper, we explore quantum speedups for nonconvex optimization by quantum tunneling. In particular, we introduce the quantum tunneling walk (QTW) and apply it to nonconvex problems where local minima are approximately global minima. We show that QTW achieves quantum speedup over classical stochastic gradient descents (SGD) when the barriers between different local minima are high but thin and the minima are flat. Moreover, we construct a specific nonconvex landscape where QTW given proper initial states is exponentially faster than classical algorithms taking local queries for hitting the neighborhood of a target global minimum. Finally, we conduct numerical experiments to corroborate our theoretical results.

We expect our results to have further impacts both in physics and optimization. In optimization theory, previous work has studied several physics-motivated optimization algorithms, including Nesterov's momentum method~\cite{su2016differential,wibisono2016variational,shi2021understanding}, stochastic gradient descents~\cite{SSJ20}, symplectic optimization~\cite{betancourt2018symplectic,jordan2018dynamical}, etc. We believe that our work can further inspire the design of optimization algorithms with physics intuitions. From theory to practice, in this work we analyzed the performance of QTW on tensor decomposition, and we expect QTW to also have decent performance on other practical problems with benign landscapes.

In quantum computing, on the one hand, previous work on continuous optimization only studies convex optimization~\cite{vanApeldoorn2020optimization,chakrabarti2020optimization} or local properties such as escaping from saddle points~\cite{zhang2021quantum}, and our work significantly extends the range of problems which quantum computers can efficiently solve to global problems in nonconvex optimization. On the other hand, we point out that QTW has the potential to be implemented on near-term quantum computers. In fact, current quantum computers have implemented both quantum simulation~\cite{arute2020hartree,ebadi2021quantum} and quantum walks~\cite{tang2018experimental,gong2021quantum} to decent scales. We deem QTW as a potential proposal for demonstrating quantum advantages in near term.

Our paper also leaves several technical questions for future investigation:
\begin{itemize}

\item What is the performance of QTW on more general landscapes? For instance, a wide range of deep neural networks~\cite{KHK19} has some (but probably not all) local minima which are approximately global. Future work on weakening the assumptions on landscapes for QTW is preferred.

\item Are there more examples with exponential quantum-classical separation? Our construction leverages a special kind of locally non-informative landscapes, and exponential quantum-classical separation can potentially be observed on other landscapes, such as nonsmooth landscapes and landscapes with negative curvature.

\item QTW simulates the Schr\"odinger equation whose potential is set to be the optimization function, and this QTW can be efficiently simulated on quantum computers. In general, are there better PDEs which are more efficient for optimization and can still be efficiently simulated on quantum computers?

\end{itemize}


\section*{Acknowledgements}
YL thanks Frédéric Hérau and Michael Hitrik for helping with understanding spectral theory, Rong Ge and Chenyi Zhang for inspiring discussions on tensor decomposition, and Jiaqi Leng for the guidance on numerical methods for quantum simulation. TL thanks Eric R. Anschuetz for general suggestions on a preliminary version of this paper, and Andrew M. Childs for helpful discussions that inspire \sec{no-exp-sep-init}.

YL was funded by SURF from Tsien Excellence in Engineering Program, Tsinghua University. 
WJS was funded by an Alfred Sloan
Research Fellowship and the Wharton Dean’s Research Fund.
TL was funded by a startup fund from Peking University, and the Advanced Institute of Information Technology, Peking University.





\newpage
\appendix

\section{Auxiliary Mathematical Tools on Spectrum Properties of Operators}\label{append:auxiliary-classical}
In this section, we introduce necessary mathematical tools on spectrum properties of operators. In \append{W-L}, we introduce the eigenvalues of the Witten-Laplacian operator. In \append{theory-Schrodinger}, we cover spectrum properties of the  Schr{\"o}dinger operator, including the WKB approximation, the Agmon distance and decay of eigenfunctions, and the interaction matrix of quantum tunneling.

\subsection{Eigenvalues of the Witten-Laplacian}\label{append:W-L}
\sec{classicalpre} shows that the discrete algorithm SGD
can be analyzed by a lr-dependent SDE \eq{lrsde}.
And the rate of convergence is bounded by the smallest non-zero eigenvalue of a Witten-Laplacian (see \prop{mic19}).
We present, in the following, a geometric
connection between the function $f$ and the spectrum of the Witten-Laplacian $\Delta^s_f$ defined by \eq{defnWL}.

The objective function is assumed to satisfy assumptions of \sec{classicalpre}.
First, we need to find a special type of saddle points at the top of barriers separating wells. To this end, we define:
\begin{definition}[Index-1 saddle points]
A critical point $x$ (i.e., $\nabla f(x) = 0$) is an index-1 saddle point, if the
Hessian at $x$, $\nabla^2 f(x)$, has exactly one negative eigenvalue.
\end{definition}
The index-1 saddle points being the bottlenecks of paths connecting two local minima are what we want.
Intuitively, for going from one local minimum to another, function values of such saddle points characterizes the smallest height needing climbed over.
Let $\mathcal{K}_{v} := \{x\in\mathbb{R}^d : f(x)<v \}$ be the sublevel set at level $v$.
For any index-1 saddle point $x$, if $r$ is sufficiently small, the set $\mathcal{K}_{f(x)} \cap \{y:\|y-x\| <x \}$ can be partitioned into two connected components, say $C_1(x,r)$ and $C_2(x,r)$.
With the help of the notations, we can formally define:
\begin{definition}[Index-1 separating saddle points or SSPs]\label{defn:SSP}
An index-1 saddle point $x$ is said to be an index-1 separating saddle point (SSP),
if for sufficiently small $r$, $C_1(x,r)$ and $C_2(x,r)$ are contained in different connected components of the
sublevel set $\mathcal{K}_{f(x)}$.
\end{definition}

Next, we aim to relate function values of SSPs and local minima to the eigenvalues of $\Delta^s_f$.
Let $\mathcal{X}^{\circ}$ and $\mathcal{X}^{\bullet}$ be the sets of SSPs and local minima, respectively. For distinguishing, we also use superscripts $\circ$ and $\bullet$ to denote SSPs and local minima, respectively, that is, $x^{\circ} \in \mathcal{X}^{\circ}$ and $x^{\bullet} \in \mathcal{X}^{\bullet}$.
The number of critical points are finite for a Morse function satisfying both the confining and the
Villani conditions.
Therefore, let $n^{\circ}$ and $n^{\bullet}$ denote numbers of SSPs and local minima, respectively.
A standard labeling of SSPs and local minima would establish a correspondence
between the two kinds of critical points and then help to rigorously define the
heights of barriers.
Denote the set of function values of SSPs by $f(\mathcal{X}^{\circ})$,
say the cardinality of this set is $I$, we can write $f(\mathcal{X}^{\circ}) = \{ v_1,v_2,\ldots,v_I \}$ and, without loss of generality, $+\infty = v_0 > v_1 > v_2 >\cdots>v_I$.
Use the values $\{v_j:j=1,\ldots,I\}$, we can define certain connected components of sublevel sets.
\begin{definition}[Critical
component]
A connected component $E$ of the sublevel set $\mathcal{K}_{v}$ for some $v\in f(\mathcal{X}^{\circ})$ is called a critical
component if either $\partial E \cap \mathcal{X}^{\circ} \neq \varnothing$ or $E = \mathbb{R}^d$.
\end{definition}
Obviously, $v_0 = \infty$ corresponds to the critical
component $E = \mathbb{R}^d$.
The labeling of SSPs and local minima goes as follows \cite{HHS11,SSJ20}:
\begin{itemize}
\item{Step 1}. Let $E^0_1 := \mathbb{R}^d$. Choose the (one)\footnote{If there are multiple global minima, choose one of them. Choosing different global minimum will result in different labeling.} global minimum in $E^0_1$ to be denoted by $x_0^{\bullet}$.
Define the set $\mathcal{X}^{\bullet}_0 :=\{x_0^{\bullet}\}$.
\item{Step 2}. Let $E^1_j$ for $j = 1,2,\ldots,m_1$ be the critical components of the sublevel set $\mathcal{K}_{v_1}$. The labeling should obey that $x_0^{\bullet}\in E^1_{m_1}$. We select $x_{1,j}^{\bullet}$ ($j=1,2,\ldots,m_1-1$) to be the (one) global minimum of $f$ restricted in $E^1_j$. Then, define $\mathcal{X}^{\bullet}_1 :=\{x_{1,1}^{\bullet}, \ldots, x_{1,m_1 -1}^{\bullet}\}$.
\item{Step 3}. For $l = 2,\ldots,I$, let $E_j^l$ where $j=1,\ldots,m_l$ be the critical components of the sublevel set $\mathcal{K}_{v_l}$.
We can always order the critical components such that there exists a an integer $k_l \leq m_l$ and
\begin{align}
    \left(\bigcup_{j=1}^{k_l} E^l_j \right) \bigcap \left(\bigcup_{\ell = 0}^{l-1} \mathcal{X}^{\bullet}_{\ell}\right) = \varnothing,
\end{align}
\begin{align}
    E^l_j \bigcap \left(\bigcup_{\ell = 0}^{l-1} \mathcal{X}^{\bullet}_{\ell}\right) \neq \varnothing,~\forall j = k_l+1,\ldots,m_l.
\end{align}
Label the (one) global minimum of $f$ restricted in $E^l_j$ by $x^{\bullet}_{l,j}$ and then define $\mathcal{X}^{\bullet}_l :=\{x_{l,1}^{\bullet}, \ldots, x_{l,k_l}^{\bullet}\}$.
It can be shown that all local minima of $f$ in $\mathbb{R}^d$ are labeled by this procedure.
\item{Step 4}. Choose the (one) point in $E^1_1 \cap \mathcal{X}^{\circ}$ whose function value is the largest among $E^1_1 \cap \mathcal{X}^{\circ}$
to be $x^{\circ}_{1,1}$. Subsequently, the point $x^{\circ}_{1,j}$ ($j = 2,\ldots,k_1$) is chosen from points with the largest function value in
\begin{align}
    E^1_j \bigcap \Big(\mathcal{X}^{\circ}\big\backslash \bigcup_{\nu< j} \{x^{\circ}_{1,\nu}\}\Big).
\end{align}
For $l=2,\ldots,I$,
the point $x^{\circ}_{l,j}$ ($j = 1,\ldots,k_l$) is chosen from points with the largest function value in
\begin{align}
    E^l_j \bigcap \Big(\mathcal{X}^{\circ}\big\backslash \bigcup_{\mu\leq l,\nu< j} \{x^{\circ}_{\mu,\nu}\}\Big).
\end{align}
It can be proved that $\mathcal{X}^{\circ} = \bigcup_{\mu,\nu} \{x^{\circ}_{\mu,\nu}\}$.
\end{itemize}

As shown in the labeling process, we can relate any local minimum $x^{\bullet}_{l,j}$ to a SSP $x^{\circ}_{l,j}$.
Note that there is no SSP corresponds to the point $x^{\bullet}_{0}$,
the number of local minima is always larger than that of SPPs, namely,
$n^{\circ} = n^{\bullet} -1$.

According to the above correspondence relation, we can relabel the SSPs and local minima by renaming the pair $(x^{\circ}_{l,j},x^{\bullet}_{l,j})$ as $x^{\circ}_{\ell},x^{\bullet}_{\ell}$ for some $\ell = 1,2,\ldots,n^{\circ}$,
such that
\begin{align}
    f(x^{\circ}_{1}) - f(x^{\bullet}_{1}) \geq V(x^{\circ}_{2}) - f(x^{\bullet}_{2}) \geq \cdots \geq f(x^{\circ}_{n^{\circ}}) - f(x^{\bullet}_{n^{\circ}}).
\end{align}
Define $x^{\circ}_{0}=\infty$ corresponding to the (a) global minimum $x^{\bullet}_{0}$, we have $f(x^{\circ}_{0}) - f(x^{\bullet}_{0}) = + \infty$.

Note that the labels of SSPs and local minima may not be unique following the above labeling procedure. Hence, the correspondence between SSPs and local minima is not unique.
However, the uniqueness of function values can be maintained. More precisely,
regardless of different labeling results on critical points,
we can always obtain unique pairs $(f(x^{\circ}_{\ell}),x^{\bullet}_{\ell})$ obeying
\begin{align}
    +\infty =f(x^{\circ}_{0}) - f(x^{\bullet}_{0})> f(x^{\circ}_{1}) - f(x^{\bullet}_{1}) \geq f(x^{\circ}_{2}) - f(x^{\bullet}_{2}) \geq \cdots \geq f(x^{\circ}_{n^{\circ}}) - f(x^{\bullet}_{n^{\circ}}).
\end{align}

The values $f(x^{\circ}_{\ell}) - f(x^{\bullet}_{\ell})$ are heights of barriers. We can now readily introduce the following fundamental result:
\begin{proposition}[Theorem 2.8 in \cite{Mic19}]\label{prop:Mic19Thm}
Under assumptions of \sec{classicalpre}, there exists $s>0$ such that for any $s\in (0,s_0]$,
the first $n^{\circ}$ smallest positive eigenvalues of the Witten-Laplacian
$\Delta_f^s$ satisfy
\begin{align}
    \delta_{s,\ell} = s (\gamma_{\ell} + o(s)) e^{-\frac{2H_{f,\ell}}{s}},
\end{align}
where $f(x^{\circ}_{\ell})-f(x^{\bullet}_{\ell})\leq H_{f,\ell} \leq f(x^{\circ}_1) - f(x^{\bullet}_{0})$ for $\ell = 1,2,\ldots,n^{\circ}$. The constants $H_{f,\ell}$ and $\gamma_{\ell}$ depend only on the function $f$.
\end{proposition}

To gain further insights, especially intuitions on $\gamma_{\ell}$, we may add a stronger assumption that is beneficial to explicit analysis:
\begin{assumption}[Hypothesis 5.1 in \cite{HHS11}]\label{assum:uniqueness}
For every critical component $E^l_j$ introduced in the labeling
process, we we assume that
\begin{itemize}
    \item There is only one global minimum of $f(x)$ restricted in $E^l_j$.
    \item If $E^l_j\cap \mathcal{X}^{\circ} \neq \varnothing$, there is a unique SSP $x^{\circ}_{l,j}$ such that $f(x^{\circ}_{l,j}) = \max_{x\in E^l_j \cap \mathcal{X}^{\circ}} f(x)$. In
    particular, $E^l_j\cap \mathcal{K}_{f(x^{\circ}_{l,j})}$ is the union of two distinct critical components.
\end{itemize}
\end{assumption}
Then, \prop{Mic19Thm} can be specified as
\begin{proposition}[Theorem 1.2 in \cite{HHS11}]
Assume the assumptions of \sec{classicalpre} and \assum{uniqueness}
are satisfied, there exists $s_0>0$ such that for any $s\in (0,s_0]$, the smallest $n^{\circ}$ non-zero eigenvalues of the Witten-Laplacian $\Delta_V^s$ associated with $f$ satisfy
\begin{align}
    \delta_{s,\ell} = s (\gamma_{\ell} + o(s)) e^{-\frac{2(f(x^{\circ}_{\ell}) - f(^{\bullet}_{\ell}))}{s}},
\end{align}
for $\ell = 1,2,\ldots,n^{\circ}$, where
\begin{align}
    \gamma_{\ell} = \frac{|\eta(x_{\ell}^{\circ})|}{\pi}\left( \frac{\mathrm{det}(\nabla^2 f(x_{\ell}^{\bullet}) )}{ -\mathrm{det}(\nabla^2 f(x_{\ell}^{\circ}) ) }\right)^{\frac{1}{2}},
    \label{eq:classicalpolycoe}
\end{align}
and $\eta(x_{\ell}^{\circ})$ is the unique negative eigenvalue of $\nabla^2 f(x_{\ell}^{\circ})$.
\end{proposition}

Under \assum{uniqueness}, we see that $H_{f,\ell}$ in \prop{Mic19Thm} can be specified as $f(x_{\ell}^{\circ}) - f(x_{\ell}^{\bullet})$.
The constants $H_{f,\ell}$ can be intuitively regarded as the heights of barriers, among which, in particular, $H_{f,1}$ is the
largest one.
Thus, $H_{f,1}$ should take the longest time for SGD to climbe over and can be used to bound the mixing time.
We formally define the Morse saddle barrier as
\begin{definition}[Morse saddle barrier]\label{defn:morse-saddle}
Under assumptions of \sec{classicalpre}, we call $H_f:=H_{f,1}$ (see \prop{Mic19Thm}) as the Morse saddle barrier, which can be specified as $f(x_{1}^{\circ}) - f(x_{1}^{\bullet})$ if \assum{uniqueness}
is also satisfied.
\end{definition}
Obviously, the smallest positive eigenvalue of the Witten-Laplacian $\Delta_f^s$
satisfy $\delta_{s,1} = s(\gamma_1+o(s))e^{-\frac{2H_f}{s}}$,
leading to \prop{mic19}, the proposition bounding the convergence rate of SGD.

\subsection{Tunneling Effect of the Schr\"odinger Operator}\label{append:theory-Schrodinger}
In this section, we introduce necessary mathematical tools from Refs.~\cite{Hel88,DS99} to study the Schr\"odinger operator of the form
\begin{equation}\label{eq:Schrodinger-appendix}
    P = -h^2\Delta + f(x).
\end{equation}
We restrict our discussions on the $d$-dimensional $C^{\infty}$ Riemannian complete manifold $M$ ($M = \mathbb{R}^d$ or $M$ is compact), and all results are obtained in the semi-classical limit referring to small $h$.

\subsubsection{WKB approximation}\label{append:WKB}
WKB approximation is a method to approximately solve differential equations.
We mainly use the WKB method to estimate local eigenfunctions.
For the Schr\"odinger equation \eq{Schrodinger-appendix}, it is accurate if the potential $f$ is quadratic and the equation constitutes a harmonic oscillator, i.e.,
\begin{align}\label{eq:P-h}
    P_h = -h^2 \Delta + \frac{1}{2}\sum_{j=1}^d \omega_j^2 x_j^2.
\end{align}
For the operator $-\frac{\d^2}{\d x^2} + x^2$, the first eigenfunction is
\begin{align}
    u_0 = \frac{1}{\pi^{1/4}}e^{-\frac{x^2}{2}},
\end{align}
whose eigenvalue is $1$. For each $j\in\N$, the $(j+1)$th eigenfunction is given by
\begin{align}
    u_j = C_{j}\left(-\frac{\d}{\d x} + x\right)^j e^{-\frac{x^2}{2}} \equiv
    p_j(x)e^{-\frac{x^2}{2}}
\end{align}
with eigenvalue $(2j+1)$. Here, $C_j$ is a constant chosen to normalize $u_j$
and $p_j(x)$ is a polynomial with degree $j$.
Define $\alpha = (\alpha_1,\ldots, \alpha_d) \in \mathbb{N}^d$, and $\|\alpha\|_1 = \sum_j^d \alpha_j$.
For the operator $P_1$ (i.e., $h=1$), we have eigenfunctions
\begin{align}
    u_{\alpha}(x) =
    \frac{\prod_j^d\omega_j^{1/4} p_{\alpha_j}(\omega_j^{1/2}x_j/2^{1/4})}{2^{d/8}}e^{-\frac{\sum_j^d \omega_j x_j^2}{2\sqrt{2}}},
\end{align}
whose corresponding eigenvalues are $\sum_j \omega_j(2\alpha_j+1)/\sqrt{2}$.
Note that
\begin{align}
    P_h =-h^2 \sum_{i=1}^d \frac{\partial^2}{\partial x_i^2} + \frac{1}{2}\sum_{j=1}^d \omega_j^2 x_j^2 = h\left(-h \sum_{i=1}^d \frac{\partial^2}{\partial x_i^2} + \frac{1}{2}\sum_{j=1}^d \omega_j^2 \frac{x_j^2}{h}\right) = h\left(- \sum_{i=1}^d \frac{\partial^2}{\partial y_i^2} + \frac{1}{2}\sum_{j=1}^d \omega_j^2y_j^2\right)
\end{align}
where $y_j = x_j/\sqrt{h}$,
the eigenfunctions and corresponding eigenvalues of $P_h$ should be given by
\begin{align}
    u_{\alpha}(x,h) =
    \frac{\prod_j\omega_j^{1/4} p_{\alpha_j}(\omega_j^{1/2}x_j/h^{1/2}2^{1/4})}{2^{d/8}h^{d/4}}e^{-\frac{\sum_{j=1}^d \omega_j x_j^2}{2\sqrt{2}h}}
    =\frac{\prod_j\omega_j^{1/4} p_{\alpha_j}(\omega_j^{1/2}x_j/h^{1/2})}{h^{d/4}}e^{-\frac{\sum_{j=1}^d \omega_j x_j^2}{2h}},
    \label{eq:wkb-quadratic}
\end{align}
and
\begin{align}
    \lambda_{\alpha}(h) = h\sum_{j=1}^d \omega_j(2\alpha_j+1)/\sqrt{2}=h\sum_{j=1}^d \omega_j(\alpha_j+1/2).
\end{align}
The eigenfunction $u_{\alpha}(x,h)$ can be written in the form $u_{\alpha}(x,h) = h^{-d/4}a_{\alpha}(x,h)e^{-\varphi(x)/h}$,
where we have $\varphi=\sum_{j=1}^d \omega_j x_j^2/2\sqrt{2}$. To understand $\varphi(x)$, first recall that any wave function should be proportional to a term $e^{i\hat{\varphi}(x)/h}$. For plane waves, $\hat{\varphi}(x) = p \cdot x$ with $p$ being the momentum of the wave. It is generally true that $\nabla \hat{\varphi}(x)$ can be seen as momentum. The function $\hat{\varphi}(x)$ is then determined by the basic law ``kinetic energy $=$ total energy $-$ potential energy", where $|\hat{\varphi}(x)|^2$ represents the kinetic energy.
For classical forbidden regions (total energy $<$ potential energy), the function $\hat{\varphi}(x)$ becomes imaginary. That is, in this case, $e^{i\hat{\varphi}(x)/h}$ no longer describes the change of wave phase, but the exponential decay of the wave amplitude.
For $u_{\alpha}(x,h)$ in \eq{wkb-quadratic}, we have $\hat{\varphi}(x) = i\varphi(x)$, total energy $= 0$, and potential energy $=f(x)$, leading to
\begin{align}
    \|\nabla \varphi (x)\|^2 = f(x),
    \label{eq:wkb-nablaphi}
\end{align}
which is true for $\varphi=\sum_{j=1}^d \omega_j x_j^2/2\sqrt{2}$. Eq.~\eq{wkb-nablaphi} is also known as the eikonal equation in optics.
For the local ground state, $\alpha = 0$, and the function $a_0(x,h)$ is a constant: $a_0(x,h) = \prod_j\omega_j^{1/4}/(\sqrt{2}\pi)^{d/4}$.

From now on, we study the local ground state near a local minimum of a general
$f(x)$. We assume that the original point $0$ is a local minimum, and the landscape $f(x)$ has the following property:
\begin{assumption}\label{assum:WKB}
\begin{align}
    f(0) = 0,\quad \nabla f(0)=0,\quad\nabla^2 f(0) >0, \quad \varliminf_{|x|\to \infty} f(x)>0
\end{align}
where $\nabla^2 f(0)$ is the Hessian of $f$ at $0$.
\end{assumption}
We further specify $\nabla^2 f(0)$ as
\begin{equation}
     \nabla^2 f(0) = \left(
    \begin{array}{cccc}
        \omega_1^2 &  &  &  \\
          & \omega_2^2 &  &  \\
           &  & \ddots &  \\
           &   &  & \omega_d^2 \\
    \end{array}
    \right),\quad \omega_j>0.
\end{equation}
The WKB construction aims to find the eigenfunctions with the form
$h^{-d/4}a(x,h)e^{-\varphi/h}$. That is, we approximate the local ground state starting from a harmonic oscillator. For general cases, $a(x,h)$ is not a constant.
The main result we need is the following lemma:
\begin{lemma}[Theorem 2.3.1 of \cite{Hel88}]\label{lem:wkb}
Under \assum{WKB}, we can find a function $\varphi(x)$, a formal series
\begin{align}
    E(h) = \sum_{j=1}^{\infty} E_jh^j,
\end{align}
and a formal symbol defined in a neighborhood of $0$
\begin{align}
    a(x,h) = \sum_{j=0}^{\infty} a_j(x)h^j,
\end{align}
s.t. $E_1 = \sum_{j=1}^d \omega_j/\sqrt{2}$ and
\begin{align}
    (P_h - E(h))(a(x,h)e^{-\varphi(x)/h}) = O(h^{\infty})e^{-\varphi(x)/h},
    \label{eq:eq39}
\end{align}
where $P_h = -h^2\Delta + f(x)$.
\end{lemma}
To be more specific, the function $\varphi$ is given by the equation
$\|\nabla \varphi(x)\|^2 = f(x)$ following the same intuition of \eq{wkb-nablaphi}.
Let $\varphi_0 = \sum_{j=1}^d \omega_j x_j^2/2\sqrt{2}$, clearly we have $\varphi(x)- \varphi_0(x) = O(\|x\|^3)$ near $0$.
To make the coefficient of $h$ in both sides of \eq{eq39} equal,
we have the first transport equation
\begin{align}
    2 \nabla\varphi\cdot \nabla a_0 + (\Delta\varphi - E_1)a_0 = 0.
    \label{eq:transport}
\end{align}
To be consistent with the results of quadratic potentials, we choose the initial condition
\begin{align}
    a_0(0) = \frac{\prod_j \omega_j^{1/4}}{(\sqrt{2}\pi)^{d/4}}.
    \label{eq:178}
\end{align}
If we look at the equality of the coefficients of $h^2$ in \eq{eq39},
\begin{align}
    2 \nabla\varphi\cdot \nabla a_1 + (\Delta\varphi - E_1)a_1 = -\Delta a_0 + E_2 a_0.
\end{align}
The initial condition is $a_1(0) = 0$. We can determine $E_2$ by
\begin{align}
    E_2 = \frac{\Delta a_0(0)}{a_0(0)}.
\end{align}
Following this procedure,
all other coefficients can be obtained recursively.

\subsubsection{The Agmon distance and decay of eigenfunctions}\label{append:agmon}
The Agmon distance generalizes the function $\varphi(x)$ in \lem{wkb} which characterizes the exponential decay of eigenfunctions. For a given value $E$, we are interested in the eigenvalues of the Schr\"odinger operator $P$ in a neighborhood of $E$.
\begin{assumption}
The potential should satisfy
$\min_{x\in M} f(x) < \lim_{|x|\to \infty} f(x)$ in the case $M=\mathbb{R}^d$,
or $\min_{x\in M} f(x)$ exists in the case $M$ is compact.
\end{assumption}
In the semi-classical limit, $E$ is very small and close to $\min f$.
Formally, we demand $\min f \leq E$ for compact $M$ and $\min f \leq E<\lim_{|x|\to \infty} f$ for $M=\mathbb{R}^d$.
\begin{definition}
The Agmon metric is defined as
\begin{equation}
    \d s^2 := (V-f)^+ \d x^2,
\end{equation}
where $\d x^2$ is the Riemannian metric in the manifold $M$, and $a^+ := \max\{0,a\}$ for any $a$.
\end{definition}
A natural distance $d(x,y)$ associated to the Agmon metric can be defined for $x,y \in M$ as
\begin{equation}
    d(x,y) = \inf_{\gamma} \int_{\gamma} \sqrt{(V-f)^+} \d x.
\end{equation}
Here $\gamma$ denotes piecewise $C^1$ paths connecting $x$ and $y$.
For a set $U$, we can also define $d(x,U) = \inf_{y\in U} d(x,y)$.
The following properties can be verified:
\begin{align}
    |d(x,y) - d(z,y)| \leq d(x,z),\quad x,y,z\in M,
\end{align}
\begin{align}
    |\nabla_x d(x,y)|^2 \leq (f(x)-E)^+.
    \label{eq:eq51}
\end{align}

As is mentioned, we have
\begin{lemma}
Under \assum{WKB} and $E=0$,
\begin{align}
    d(0,x) = \varphi(x),
\end{align}
in a neighbor of $0$. The function $\varphi(x)$ is defined in \lem{wkb}.
\end{lemma}

Before computing the interaction between wells, we should first understand the behavior of local eigenfunctions near one well.
We assume for the remainder of \append{agmon} that
\begin{align}
    U = \{x\mid x\in M,f(x)\leq E\}
\end{align}
is compact, whose diameter is 0 under the Agmon distance.
The set $U$ is defined as a \emph{well}. Note that the definitions of the Agmon distance and wells all depend on the choice of $E$.

Local eigenstates can be rigorously specified as the eigenstates of the Dirichlet realization of $P$ near one well. Let $\Omega$ be some bounded
open set containing $U$.
\begin{definition}\label{defn:Dirichlet}
The Dirichlet realization of the Schr\"odinger operator, $P_{\Omega}$ is defined as the restriction of $P$
on the domain $H^1_0(\Omega)\cap H^2(\Omega)$. Here, $H^k(\Omega)$ denotes the classical Sobolev space of order $k$ and $H^1_0$ is the closure of
$C_0^{\infty}(\mathring{\Omega})$ in $H^1$.
\end{definition}
let $u_h$ be an eigenfunction of $P$ with eigenvalue $E+\lambda(h)$ ($\lambda(h)^+\to 0,~h \to 0$). We assume $h$ is in $J$, a subset of $(0,h_0]$, (for some $h_0>0$) and $0$ belongs to the closure of $J$.
Then, the function $e^{d(x,U)/h}u_h$ and its gradient can be well controlled:

\begin{lemma}[Proposition 3.3.1 in \cite{Hel88}]\label{lem:exdecay}
For every $\epsilon>0$, there exists small enough $h$ depending on $\epsilon$ s.t.
\begin{align}
    \|\nabla(e^{d(x,U)/h}u_h)\|_{L^2(\Omega)}
    + \|e^{d(x,U)/h}u_h\|_{L^2(\Omega)} \leq C_{\epsilon}e^{\epsilon/h}.
    \label{eq:eq54}
\end{align}
\end{lemma}

The proof of \lem{exdecay} uses the following mathematical result:
\begin{lemma}[Lithner-Agmon estimation, Theorem 1.1 of \cite{HS84} or Proposition 6.1 of \cite{DS99}]\label{lem:agmonestimate}
Let $M$ be a $C^{\infty}$ Riemannian complete manifold and $\Omega \subset M$ be bounded with $C^2$-boundary. Let $f \in C(\overline{\Omega};\mathbb{R})$ and $\phi$ a real valued Lipschitz function on $\overline{\Omega}$. Then, the gradient $\nabla \phi$ is
well defined in $L^{\infty}(\Omega)$ as the almost everywhere limit of
$\nabla (\chi_{\epsilon} * \phi)$, when $\epsilon \to 0$, where $\chi_{\epsilon}(x) = \frac{1}{\epsilon^n}\chi(\frac{x}{\epsilon}) \in C^{\infty}_0(B(0,\epsilon))$ is a standard mollifier (This definition of gradient is generally valid and aims to include the case of $\mathbb{R}^d$).
For every $u \in C^2(\overline{\Omega})$ satisfying $u|_{\partial \Omega} = 0$, we have
\begin{equation}
    h^2 \int_{\Omega} \|\nabla(e^{\phi/h} u(x)) \|^2_x \d x +
    \int_{\Omega} (f(x) - \|\nabla \phi(x)\|^2_x) e^{2\phi/h}|u|^2 \d x
    = \mathrm{Re} \int_{\Omega} e^{2\phi/h} Pu(x) \overline{u(x)} \d x.
\end{equation}
where $\d x$ is the Riemannian volume and $\|\cdot \|_x$ is the norm in $T_x M$.
\end{lemma}

The proof of \lem{agmonestimate} is mainly based on the Green formula.
Namely, for any $C^2(\overline{\Omega})$-function $f$ satisfying $f|_{\partial \Omega} = 0$,
\begin{equation}
    \int_{\Omega} \|\nabla f\|_x^2 \d x = - \int_{\Omega} \Delta f \overline{f} \d x.
\end{equation}
Now, we turn back to the proof of \lem{exdecay}:

\begin{proof}[Proof of \lem{exdecay}]
Choose $\epsilon>0$ and use the identity in \lem{agmonestimate} with $f$ replaced by $f-(E+\lambda(h))$,
$\phi(x)=(1-\delta)d(x,U)$, $u$ replaced by $u_h$, and $P= -h^2\Delta+f-(E+\lambda(h))$:
\begin{align}
    h^2 \int_{\Omega} |\nabla(e^{\phi/h} u_h(x)) |^2 \d x +
    \int_{\Omega} (f(x) - E -\lambda(h) -|\nabla \phi(x)|^2) e^{2\phi/h}|u_h|^2 \d x
    = 0.
\end{align}
Here, $\delta$ will be chosen to be small depending on $\epsilon$.
Let
\begin{align}
    \Omega^+_{\delta} &=  \{x\mid x\in \Omega, f(x)\geq E+\delta \},\\
    \Omega^-_{\delta} &=  \{x\mid x\in \Omega, f(x)< E+\delta \},
\end{align}
we have
\begin{align}
    h^2 \int_{\Omega} |\nabla(e^{\phi/h} u_h(x)) |^2 \d x +
    \int_{\Omega^+_{\delta}} (f(x) - E -\lambda(h)- |\nabla \phi(x)|^2) e^{2\phi/h}|u_h|^2 \d x \nonumber\\
    \leq \sup_{x\in \Omega^-_{\delta}}|f(x) - E -\lambda(h)- |\nabla \phi(x)|^2| \int_{\Omega^-_{\delta}}  e^{2\phi/h}|u_h|^2 \d x.
\end{align}
We restrict $0\leq \delta \leq 1$ not to be large, and then could find some $C$ independent of $h$ and $\delta$, s.t.
\begin{align}
    h^2 \int_{\Omega} |\nabla(e^{\phi/h} u_h(x)) |^2 \d x +
    \int_{\Omega^+_{\delta}} (f(x) - E -\lambda(h)- |\nabla \phi(x)|^2) e^{2\phi/h}|u_h|^2 \d x
    \leq C \int_{\Omega^-_{\delta}}  e^{2\phi/h}|u_h|^2 \d x.
\end{align}
Use the inequality \eq{eq51}, we have on $\Omega^+_{\delta}$ that
\begin{align}
   f(x) - E -\lambda(h)- |\nabla \phi(x)|^2 \geq (1-(1-\delta)^2)(f-E) - \lambda(h)^+
   \geq \delta^2(2-\delta) - \lambda(h)^+.
\end{align}
We choose $h \in (0,h(\delta)]$ where $h(\delta)$ is determined by
\begin{align}
   \sup_{h\in (0,h(\delta)]}\lambda(h) \leq \delta^2.
\end{align}
Since $\delta$ will be controlled by $\epsilon$, $h$ depends on $\epsilon$.
Then, we obtain
\begin{align}
    h^2 \int_{\Omega} |\nabla(e^{\phi/h} u_h(x)) |^2 \d x +
    \delta^2 \int_{\Omega^+_{\delta}}  e^{2\phi/h}|u_h|^2 \d x
    \leq C \int_{\Omega^-_{\delta}}  e^{2\phi/h}|u_h|^2 \d x
\end{align}
and
\begin{align}
    h^2 \int_{\Omega} |\nabla(e^{\phi/h} u_h(x)) |^2 \d x +
    \delta^2 \int_{\Omega}  e^{2\phi/h}|u_h|^2 \d x
    \leq (C+1) \int_{\Omega^-_{\delta}}  e^{2\phi/h}|u_h|^2 \d x\leq (C+1) e^{2\sup_{x\in\Omega^-_{\delta}}\phi(x)/h}.
\end{align}
Define $k(\delta) = \sup_{x\in\Omega^-_{\delta}}\phi(x)$,
it is clear that $k(\delta) \to 0,~(\delta\to 0)$.
Now, we should replace $e^{\phi/h} u_h(x)$ by $e^{d(x,U)/h} u_h(x)$.
Observe that
\begin{align}
    h^2 \int_{\Omega} |\nabla(e^{d(x,U)/h} u_h(x)) |^2 \d x
    =h^2 \int_{\Omega} |\nabla(e^{\delta d(x,U)/h}e^{\phi/h} u_h) |^2 \d x
\end{align}
Let $K$ denote the maximum of $d(x,U)$ in $\Omega$, we can get
\begin{align}
    h^2 \int_{\Omega} |\nabla(e^{d(x,U)/h} u_h(x)) |^2 \d x
    \leq h^2 e^{2\delta K/h}\int_{\Omega} |\nabla(e^{\phi/h} u_h) |^2 \d x + \delta^2 \sup_{x\in \Omega}|f-E| e^{2\delta K/h}\int_{\Omega}  e^{2\phi/h}|u_h|^2 \d x
\end{align}
which leads to
\begin{align}
    h^2 \int_{\Omega} |\nabla(e^{d(x,U)/h} u_h(x)) |^2 \d x+\delta^2\int_{\Omega}e^{2d(x,U)/h}|u_h|^2 \d x
    \leq (1+C)(1+\sup_{x\in \Omega}|f-E|) e^{2\frac{K\delta + k(\delta)}{h}}.
\end{align}
At last, we choose $\delta$ s.t. $K\delta + k(\delta) \leq \epsilon$ and we can get \eq{eq54}.
\end{proof}
The above lemma leads to an intuitive result: the $L^2$ norm of $u_h$ concentrates in a local well. More specifically, we have
\begin{corollary}\label{cor:L2norm}
For each open neighborhood $\mathcal{V}$ of $U$ in $\Omega$,
\begin{align}
   \|u_h\|_{L^2(\mathcal{V})}=1+O(e^{-\epsilon/h}),\quad \mathrm{for} ~\epsilon>0.
\end{align}
\end{corollary}

\subsubsection{Interaction matrix}\label{append:interactionmatrix}
The main idea of this subsection is that,
modulo an exponentially small error, the spectrum of the Schr\"odinger operator in some interval $I(h)$ is the same as the spectrum of the direct sum of one-well Dirichlet realizations in $I(h)$.
In other words, low energy \emph{local} eigenfunctions
span a subspace close to that spanned by some low energy \emph{global} eigenfunctions.
And then, $P$ restricted in the low energy subspace
can be expressed as a matrix, which will be called the interaction matrix under the approximate local eigenstates.
The off-diagonal elements of the interaction matrix characterizes tunneling effects between local eigenstates of different wells.
We aim to calculate the interaction matrix explicitly. Technically, four important results will be introduced: \prop{dimEdimF} ensures that we can use local eigenstates to approximately span a low energy subspace; \prop{genPF} gives the most general formula of the interaction matrix; \prop{spePF} improve \prop{genPF} with a stronger assumption; and finally \prop{speW} provides explicitly tunneling effects between local ground states which we can use in the main text.

\begin{assumption}\label{assum:0-baseline}
Without loss of generality, we will take $E=0$ and the potential should satisfy
\begin{equation}
    \min f \leq 0 < \lim_{|x|\to \infty} f,
\end{equation}
where $0 < \lim_{|x|\to \infty} f$ for $M = \mathbb{R}^d.$
The region $V^{-1}((-\infty, 0])$ can be decomposed as
\begin{equation}
    V^{-1}((-\infty, 0]) = U_1 \cup U_2 \ldots \cup U_N,
\end{equation}
where $U_j$ are disjoint and compact.
Each set $U_j$ is called as a well.
\end{assumption}
Define $S_0$ as the minimal distance
between the different wells:
\begin{equation}
    S_0 = \min_{j\neq k} d(U_j,U_k) = \min_{j\neq k} \inf_{\substack{x_j\in U_j,\\x_k\in U_k}}d(x_j,x_k).
\end{equation}

To find the local eigenstates, we need to associate to each well a Dirichlet problem in an open set $M_j$ containing $U_j$.
For a small value $\eta >0$, let $B(U_j, \eta) = \{x \in M| d(x,U_j) \leq \eta\}$.
The open set $M_j$ is then defined as
\begin{equation}
    M_j^{(\eta)} = M\backslash \cup_{k\neq j} B(U_k,\eta),
    \label{eq:smallregion}
\end{equation}
if $M$ is compact.
For $M=\mathbb{R}^d$,
\begin{equation}
    M_j^{(\eta)} = (\mathbb{R}^d\backslash \cup_{k\neq j} B(U_k,\eta))\cap \mathring{B}(U_j,S),
\end{equation}
where $S> 2S_0$ is large, and for a set $B$, $\mathring{B}$ means the interior of $B$.
We simply use $M_j$ to denote $M_j^{(\eta)}$ when there is no ambiguation, but keep in mind that
everything depends on $\eta$.

Let $P_{M_j}$ denote the Dirichlet realization of $P$ in $M_j$ (\defn{Dirichlet}).
Consider a subset $J\subset (0,h_0]$ s.t. $0$ is an accumulation point of $J$, we set the interval $I(h)$ to be
\begin{align}
    I(h) = [\alpha(h),\beta(h)]~(h\in J),
\end{align}
and $I(h) \to \{ 0\}$ when $h\to 0$.
Let $a(h)$ to be a function of $h \in J$, and
\begin{align}
    |\log a(h)| = o(1/h).
\end{align}
\begin{definition}[Local eigenvalues and eigenstates]\label{defn:localeigen}
The eigenvalues of $P_{M_j}$ in $I(h)$ are $\mu_{j,1}, \mu_{j,2},\ldots,\mu_{j,m_j}$ whose corresponding orthonormal eigenstates (local eigenstates) are $\phi_{j,1}, \phi_{j,2},\ldots,\phi_{j,m_j}$.
\end{definition}
It is clear to use $\alpha = (j,k),~1\leq k\leq m_j$ to denote the states in different wells and put $j(\alpha) = j$ to know which well the $\alpha$ state is in.

If $\eta$ is large, the state $\phi_{\alpha}$ may have a large overlap with local eigenstates in other wells. We modify $\phi_{\alpha}$ in the following way.
\begin{definition}[Local eigenstates with cut-offs]\label{defn:localcutoff}
Let $\theta_j$ be a $C^{\infty}$ function, $\mathrm{Supp}\theta_j\subset B(U_j,2\eta)$, and $\theta_j = 1$ in $B(U_j, 3\eta/2)$. Then, we can define $\chi_j = 1-\sum_{k\neq j} \theta_k$ (for compact $M$) or $\chi_j = (1-\sum_{k\neq j} \theta_k)\chi^j_S$ where $\mathrm{Supp}\chi^j_S \subset B(U_j,S)$ and $\chi^j_S=1$ in $B(U_j,S-\eta)$ (for $M = \mathbb{R}^d$). The modified local eigenfunctions are
\begin{align}
    \psi_{\alpha} = \chi_{j(\alpha)}\phi_{\alpha}.
\end{align}
\end{definition}
For the cases where $\eta$ is small compared to $S_0$, $\psi_{\alpha}$ and $\phi_{\alpha}$ are nearly the same thing.

The following result ensures that we can use local eigenstates to approximately span a low energy subspace of $P$.

\begin{proposition}[Theorem 4.2.1 in \cite{Hel88}, Proposition 6.7 in \cite{DS99}]\label{prop:dimEdimF}
Let $\E_j$ be the subspace spanned by $\psi_{j,k}~(k=1,2,\ldots,m_j)$, $\E = \bigoplus_{j}\E_j$, and $\F$ the subspace spanned by eigenstates of $P$ whose eigenvalues are in $I(h),~h\in J$. Suppose that for $h\in J$,
$P$ and $P_{M_j} (\forall j)$ have no eigenvalue in $\big(\alpha(h)-2a(h),\alpha(h)\big) \cup \big(\beta(h), \beta(h)+2a(h)\big)$.
Then, we have $\mathrm{dim}\E = \mathrm{dim}\F$ for sufficiently small $h$. Equivalently, there exists a bijection $b$ from $\sigma(P)\cap I(h)$ onto $(\bigcup_j \sigma(P_{M_j})) \cap I(h)$. And for every $c<S_0-2\eta$,
we have $b(\lambda) - \lambda = O(e^{-c/h}),~\lambda\in \sigma(P)\cap I(h)$.
\end{proposition}

The problem now is to compute the matrix $P_{|\F}$ in a convenient basis.
This basis is heavily related to the local eigenstates, s.t. the off-diagonal entries can capture tunneling effects.
Let $\pi_{\F}$ be the orthogonal projection onto $\F$, we introduce
\begin{align}
    v_{\alpha} = \pi_{\F} \psi_{\alpha}.
\end{align}
By \prop{dimEdimF}, we can get
\begin{align}
    v_{\alpha} -\psi_{\alpha} = \hat{O}(e^{-S_0/h})~~\mathrm{in}~ L^2(M).
\end{align}
Here the $\hat{O}$ notation is defined as:
\begin{definition}
For any function $f$, $f=\hat{O}(e^{-S_0/h})$ means that
for any $\epsilon>0$, there exists $\eta_0>0$ s.t. for $0<\eta<\eta_0$,
$f=O(e^{-S_0/h+\epsilon/h})$.
\end{definition}
Note that
\begin{align}
    \ip{v_{\alpha}}{v_{\beta}} = \ip{\psi_{\alpha}}{\psi_{\beta}} - \ip{v_{\alpha}-\psi_{\alpha}}{v_{\beta}-\psi_{\beta}},
\end{align}
we have
\begin{align}
    \ip{v_{\alpha}}{v_{\beta}} = \ip{\psi_{\alpha}}{\psi_{\beta}} + \hat{O}(e^{-2S_0/h}).
\end{align}
By the definition of $\psi_{\alpha}$ and the decay of $\phi_{\alpha}$,
it can be estimated that
\begin{align}
    \ip{\psi_{\alpha}}{\psi_{\beta}} =\left\{
    \begin{array}{lll}
     1 + \hat{O}(e^{-2S_0/h}),~\alpha=\beta,\\
    \hat{O}(e^{-2S_0/h}),~\alpha \neq \beta,~j(\alpha)=j(\beta),\\
    \hat{O}(e^{-S_0/h}),~\alpha \neq \beta,~j(\alpha)\neq j(\beta),
    \end{array}
    \right.
\end{align}
We define $r_{\alpha}$ by
\begin{align}
    P\psi_{\alpha} = \mu_{\alpha} \psi_{\alpha} + r_{\alpha},
    \label{eq:Ppsi}
\end{align}
obtaining that $r_{\alpha} = \hat{O}(e^{-S_0/h})$ in $L^2(M)$.
These estimations lead to
\begin{align}
    \bra{v_{\alpha}} P \ket{v_{\beta}} = \bra{\psi_{\alpha}} P \ket{\psi_{\beta}} + \hat{O}(e^{-2S_0/h}).
\end{align}
Using the fact that $P$ is self-adjoint, we have
\begin{align}
    \bra{\psi_{\alpha}} P \ket{\psi_{\beta}} = \frac{\mu_{\alpha} + \mu_{\beta}}{2}\ip{\psi_{\alpha}}{\psi_{\beta}} + \frac{1}{2} (\ip{r_{\alpha}}{\psi_{\beta}}+ \ip{\psi_{\alpha}}{r_{\beta}}).
\end{align}
Equation \eq{Ppsi} gives that $r_{\alpha} = [P, \chi_{j(\alpha)}]\phi_{\alpha}$
and then
\begin{align}
    \ip{\psi_{\alpha}}{r_{\beta}} &= h^2\int \chi_{j(\alpha)}(\phi_{\beta}
    \nabla\phi_{\alpha} - \phi_{\alpha}\nabla\phi_{\beta})\cdot \nabla\chi_{j(\beta)} \d x + h^2\int \nabla\chi_{j(\beta)} \cdot \nabla\chi_{j(\alpha)} \phi_{\beta}\phi_{\alpha} \d x \nonumber \\
    & =h^2\int \chi_{j(\alpha)}(\phi_{\beta}
    \nabla\phi_{\alpha} - \phi_{\alpha}\nabla\phi_{\beta})\cdot \nabla\chi_{j(\beta)} \d x + \hat{O}(e^{-2S_0/h}).
\end{align}
Define the matrix $T$ as
\begin{align}
    T_{\alpha\beta} =\left\{
    \begin{array}{ll}
     0,~\alpha=\beta,\\
    \ip{\psi_{\alpha}}{\psi_{\beta}},~\alpha \neq \beta.
    \end{array}
    \right.
\end{align}
The matrix $(\bra{v_{\alpha}} P \ket{v_{\beta}})$ can be given by
\begin{align}
    (\bra{v_{\alpha}} P \ket{v_{\beta}}) = \mathrm{diag} (\mu_{\alpha}) + \frac{1}{2} T \mathrm{diag}(\mu_{\alpha}) + \frac{1}{2} \mathrm{diag}(\mu_{\alpha}) T + (\hat{W}_{\alpha \beta}) + \hat{O}(e^{-2S_0/h}),
\end{align}
where
\begin{align}
    \hat{W}_{\alpha \beta}=\frac{1}{2} (W_{\alpha \beta} + W_{\beta \alpha}),
    \label{eq:hatW}
\end{align}
and
\begin{align}
    W_{\alpha \beta} = h^2\int \chi_{j(\alpha)}(\phi_{\beta}
    \nabla\phi_{\alpha} - \phi_{\alpha}\nabla\phi_{\beta})\cdot \nabla\chi_{j(\beta)} \d x.
    \label{eq:defnW}
\end{align}
It is more natural to compute the matrix of $P_{|\F}$ in an orthonormal basis.
\begin{definition}[Orthonormalized local eigenstates]\label{defn:ortholocal}
Let $\Xi = (\z v_{\alpha}|v_{\beta} \y)$, the following basis is orthonormal in $\F$,
\begin{equation}
    e_{\alpha} := \sum_{\beta} v_{\beta} (\Xi^{-1/2})_{\beta \alpha}.
\end{equation}
\end{definition}
The states $\psi_{\alpha}$ and $e_{\alpha}$ are defined for rigorous and general discussions. In practice, when we choose sufficiently small $\eta$ and $h$, $\phi_{\alpha}$, $\psi_{\alpha}$, and $e_{\alpha}$ would be the same.

The matrix of $P_{|\F}$ with respect to $\{e_{\alpha}\}$ is
\begin{align}
    \Xi^{-1/2} (\bra{v_{\alpha}} P \ket{v_{\beta}}) \Xi^{-1/2}.
    \label{eq:XimXi}
\end{align}
It is easy to verify that $\Xi^{-1/2} = I - T/2 + \hat{O}(e^{-2S_0/h})$,
giving the matrix \eq{XimXi} as
\begin{align}
    \mathrm{diag}(\mu_{\alpha}) + (\hat{W}_{\alpha \beta}) + \hat{O}(e^{-2S_0/h}).
\end{align}

To summarize, we obtain a very general result:
\begin{proposition}[Theorem 4.3.1 of \cite{Hel88}]\label{prop:genPF}
The matrix of $P_{|\F}$ in the basis $\{e_{\alpha} \}$ is given by
\begin{equation}
    \mathrm{diag}(\mu_{\alpha}) + (\hat{W}_{\alpha \beta}) + \hat{O}(e^{-2S_0/h}),
\end{equation}
where $\hat{W}_{\alpha \beta}$ is specified by \eq{hatW} and \eq{defnW}.
\end{proposition}

The left task is to compute $W_{\alpha \beta}$ explicitly. We can easily know from \eq{defnW} that
\begin{align}
    W_{\alpha \beta} = \left\{
    \begin{array}{ll}
        \hat{O}(e^{-2S_0/h}),~j(\alpha) = j(\beta),\\
        \hat{O}(e^{-S_0/h}),~j(\alpha) \neq j(\beta).
    \end{array}
    \right.
\end{align}
Those $W_{\alpha \beta}$ when $j(\alpha) \neq j(\beta)$ characterize quantum tunneling which are our focuses now.
Let us discuss first when tunneling effects can be neglected.
For two local eigenstates $\phi_1$ and $\phi_2$ in different wells whose corresponding eigenvalues are $\mu_1$ and $\mu_2$, respectively, if
\begin{align}
    \exists \epsilon_0~(0<\epsilon_0<S_0)~\mathrm{s.t.}~|\mu_1 - \mu_2| \geq e^{-\epsilon_0/h},
\end{align}
the tunneling effect between the two states would not affect much and we may say the two states are non-resonant (see \cite{HS85} for more rigorous explanations). Intuitively, we may understand the fact through energy conservation. Initiating a system at the state $\phi_1$, the evolution governed by the Schr\"odinger equation keeps the total energy unchanged which is $\mu_1$. If $|\mu_2-\mu_1|$ is large, the system state can not overlap much with $\phi_2$.
The problem of tunneling appears when $|\mu_2-\mu_1|$
is exponentially small for the order $e^{S_0/h}$.
Thus, we may only consider $W_{\alpha\beta}$ for those $\alpha,~\beta$ satisfying the following assumption:
\begin{assumption}\label{assum:TunEnCon}
\begin{equation}
    \mu_{\alpha} - \mu_{\beta} = O(h^{\infty})e^{-\frac{S_0}{h}},
\end{equation}
\begin{equation}
    \phi_{\alpha} = O(h^{-N_0} e^{d(x,U_{j(\alpha)})/h}),
    \label{eq:phiaO}
\end{equation}
where $N_0 \in \mathbb{N}$ is a constant.
\end{assumption}
A sufficient condition to get \eq{phiaO} is that
$U_{j(\alpha)}$ is a point and
\begin{align}
    \min f = f(U_{j(\alpha)}) = 0,~\nabla f(U_{j(\alpha)}) =0,~ \nabla^2f(U_{j(\alpha)}) \geq 0,
    \label{eq:sufficient-hypo1}
\end{align}
\begin{align}
    \mu_{\alpha} \in [0, C_0 h],~\mathrm{for ~some}~C_0>0.
    \label{eq:sufficient-hypo2}
\end{align}
Under \assum{TunEnCon}, first note that
\begin{equation}
    W_{\alpha \beta} = O(h^{\infty})e^{-S_0/h},~\mathrm{if}~j(\alpha) = j(\beta)\vee d(U_{j(\alpha)},U_{j(\beta)}) > S_0,
\end{equation}
which is negligible as the principal terms are of the order $h^{-\nu} e^{-S_0/h}$ for some $\mu>0$.
The remained case to consider is
\begin{equation}
    j(\alpha) \neq j(\beta)\wedge d(U_{j(\alpha)},U_{j(\beta)}) = S_0.
\end{equation}
For simplicity, we introduce $\simeq$ as
\begin{align}
    a \simeq b \iff a - b = O(h^{\infty})e^{S_0/h}.
\end{align}

In order to calculate the integration, we need to introduce some geometry settings.
Let $E^{(a)}$ be an elliptic region, namely,
\begin{equation}
    E^{(a)} = \{x\in \mathbb{R}^d, d(U_{j(\alpha)}) + d(U_{j(\beta)}) \leq S_0 + a\}.
\end{equation}
The integration \eq{defnW} for $x\notin E^{(a)}$ gives a contribution $\simeq 0$. So later we can only consider integration near $E^{(a)}$.
The parameter $a$ is chosen to be sufficiently small so that
$E^{(a)} \subset M_{j(\alpha)} \cup M_{j(\beta)}$ and there is no other wells in $E^{(a)}$.
Then, there will always exist such an open set $\Omega$ with smooth boundary satisfying
\begin{equation}
    U_{j(\alpha)} \subset \Omega,~U_{j(\beta)}\cap \overline{\Omega} = \varnothing,~E^{(a)}\cap \overline{\Omega} \subset M_{j(\alpha)},
    ~E^{(a)}\cap \Omega^{\complement} \subset M_{j(\beta)}.
    \label{eq:georelation-elliptic}
\end{equation}
\begin{figure}
  \centerline{
  \includegraphics[width=0.6\textwidth]{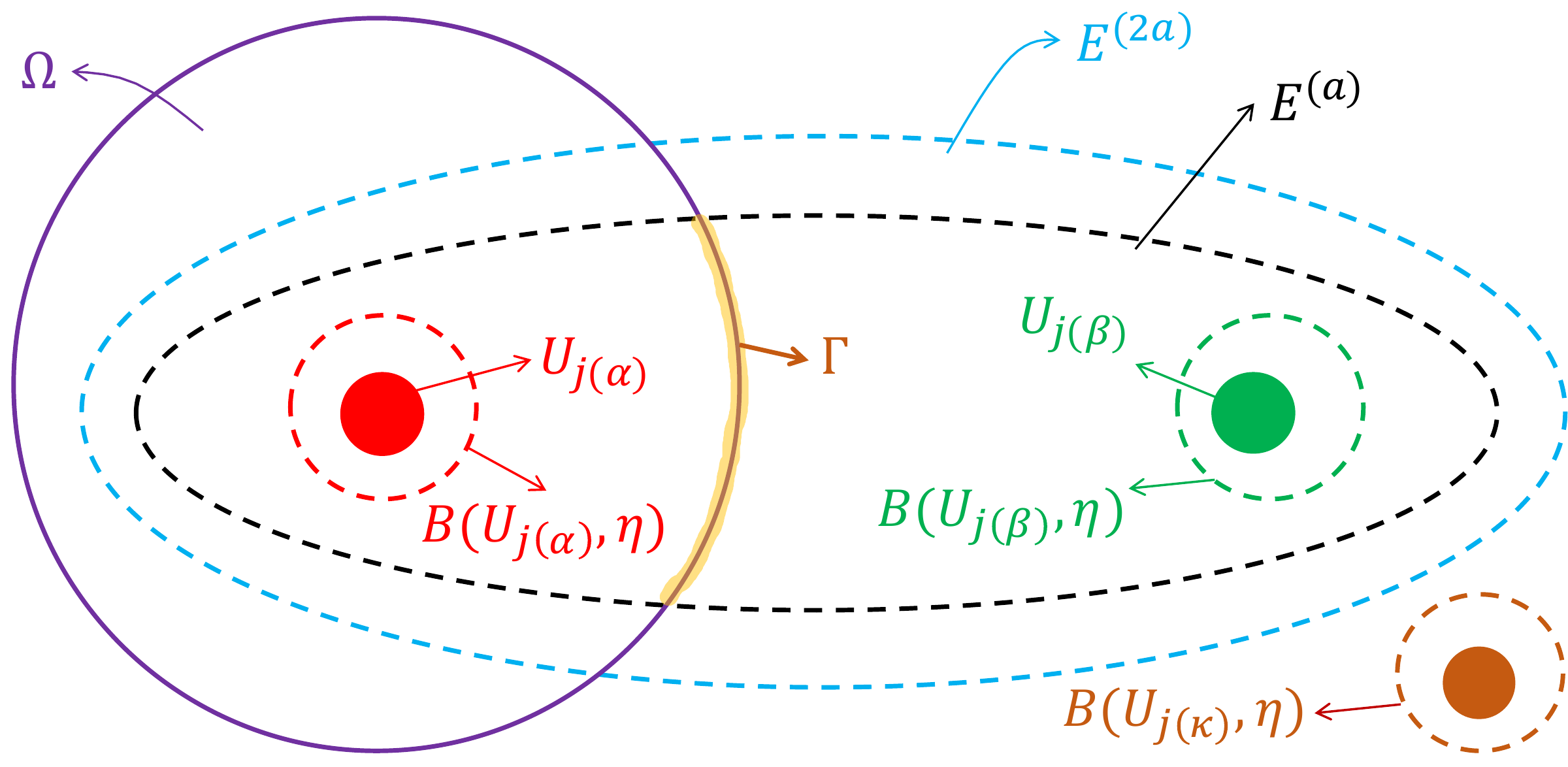}}
    \caption{Illustration of the geometric relations in Eq.~\eq{georelation-elliptic}.
    }
\label{fig:ellipticregion}
\end{figure}
Let $\Gamma = \partial \Omega \cap E^{(a)}$, which is compact in $M_{j(\alpha)}\cap M_{j(\beta)}$.
Then, we make $E^{(2a)}$ satisfy the properties as well by choosing proper $a$.
To simplify the calculation, we introduce $\chi_{E^{(a)}}\in C^{\infty}_0 (M)$ which equal to $1$ in $E^{(a)}$ and has support in
$E^{(2a)}$. The following relations hold
\begin{align}
    \mathrm{Supp} \chi_{E^{(a)}} \cap \overline{\Omega} \subset M_{j(\alpha)},~\mathrm{Supp} \chi_{E^{(a)}} \cap \Omega^{\complement} \subset M_{j(\beta)}
\end{align}
Therefore, we know $\chi_{j(\alpha)} = 1$ on $\mathrm{Supp} \chi_{E^{(a)}} \cap \overline{\Omega}$ and $\chi_{j(\beta)} = 1$ on $\mathrm{Supp} \chi_{E^{(a)}} \cap \Omega^{\complement}$, leading to
the second $\simeq$ in the following equation:
\begin{align}
    W_{\alpha\beta} &\simeq h^2\int \chi_{E^{(a)}} \chi_{j(\alpha)}(\phi_{\beta}
    \nabla\phi_{\alpha} - \phi_{\alpha}\nabla\phi_{\beta})\cdot \nabla\chi_{j(\beta)} \d x\\
    & \simeq h^2\int_{\Omega} \chi_{E^{(a)}} (\phi_{\beta}
    \nabla\phi_{\alpha} - \phi_{\alpha}\nabla\phi_{\beta})\cdot \nabla\chi_{j(\beta)} \d x.
\end{align}
Use the Green formula, we have
\begin{equation}
    W_{\alpha \beta} \simeq
    h^2\int_{\partial \Omega} \chi_{E^{(a)}}\chi_{j(\beta)} (\phi_{\beta}
    \nabla\phi_{\alpha} - \phi_{\alpha}\nabla\phi_{\beta})\cdot (-\boldsymbol{n})
    dS_{\Gamma}
     - h^2\int_{\Omega} \nabla \cdot ( \chi_{E^{(a)}} (\phi_{\beta}
    \nabla\phi_{\alpha} - \phi_{\alpha}\nabla\phi_{\beta}) ) \chi_{j(\beta)} \d x,
    \label{eq:green}
\end{equation}
where $\boldsymbol{n}$ is the unit normal on $\partial \Omega$ riented into $\Omega$ and $dS_{\Gamma}$ is the induced measure on $\partial \Omega$.
With \assum{TunEnCon}, the second term in the right hand side of \eq{green} is of $O(h^{\infty}e^{-S_0/h})$, which means
\begin{equation}
    W_{\alpha \beta} \simeq
    h^2\int_{\Gamma} \left(\phi_{\alpha}\frac{\partial \phi_{\beta}}{\partial n}
    -\phi_{\beta}\frac{\partial \phi_{\alpha}}{\partial n}\right)
    dS_{\Gamma}.
\end{equation}
And, we can have:
\begin{proposition}[Theorem 4.3.4 of \cite{Hel88}]\label{prop:spePF}
By adding \assum{TunEnCon}, \prop{genPF} can be improved to the following:
the matrix $P_{|\F}$ in the basis $\{e_{\alpha} \}$ is given by
\begin{equation}
    \mathrm{diag}(\mu_{\alpha}) + \widetilde{W}_{\alpha \beta} + O(h^{\infty})e^{S_0/h},
\end{equation}
with
\begin{align}
    \widetilde{W}_{\alpha \beta} = \left \{
    \begin{array}{ll}
        h^2\int_{\Gamma} \left(\phi_{\alpha}\frac{\partial \phi_{\beta}}{\partial n}
    -\phi_{\beta}\frac{\partial \phi_{\alpha}}{\partial n}\right)
    dS_{\Gamma},~\mathrm{if}~d(U_{j(\alpha)},U_{j(\beta)}) = S_0,\\
        0,~\mathrm{otherwise}.
    \end{array} \right.
    \label{eq:tildeW}
\end{align}
\end{proposition}

Now, we just need the knowledge of $\phi_{\alpha}$ in a neighborhood of $\Gamma$ to get the explicit form of $P_{|\F}$.
WKB method in \append{WKB} can be used to estimate $\phi_{\alpha}$ in the bottom (local ground states),
whose validity is secured by
\begin{assumption}\label{assum:3}
\begin{equation}
    0 = \min f < \lim_{|x|\to \infty} f,
\end{equation}
\begin{equation}
    f^{-1}(0) = U_1 \cup U_2 \ldots \cup U_N,
\end{equation}
\begin{equation}
    U_j = {x_j}~\mathrm{is~a~point,}~\mathrm{and}~\nabla^2V(x_j) > 0~\mathrm{for}~j=1,\ldots,N.
\end{equation}
\end{assumption}
That is, local minima are global minima.
We choose $I(h)$ to be small enough such that
there is a integer $N^+$ and for $j = 1,\ldots,N^+$,
$P_{M_j}$ has exactly one eigenvalue in $I(h)$ but for
$j = N^+ + 1,\ldots,N$, $P_{M_j}$ has no eigenvalue near
$I(h)$.
Resonance happens between wells $\{x_j\}$ for $j =1,\ldots,N^+ $.
We now just have to compute $\widetilde{W}_{jk}$ for $j\neq k$, $d(x_j,x_k) = S_0$ and $j,k\in \{1,2,\ldots,N^+ \}$.

We know from \lem{wkb} that the local eigenstates have the form
\begin{equation}
    \phi_j (x) = h^{-d/4}a^{(j)}(x)e^{-d(x,x_j)/h},
\end{equation}
and satisfy
\begin{equation}
    (P - E^{(j)})\phi_j (x) = O(h^{\infty})e^{-d(x,x_j)/h}.
\end{equation}
In the semi-classical analysis, the function $a^{(j)}(x)$ can be expanded as
\begin{equation}
    a^{(j)}(x) \approx \sum_{k=0}^{\infty} a_k^{(j)}(x) h^k.
\end{equation}
Insert this form to \eq{tildeW}, we have
\begin{equation}
    \widetilde{W}_{jk} = -h^{1-d/2}\int_{\Gamma_{jk}} \left[a^{(j)}a^{(k)}
    \Big(\frac{\partial d(x,x_k)}{\partial n}
    -\frac{\partial d(x,x_j)}{\partial n}\Big) + O(h)\right]e^{-\frac{d_{jk}(x)}{h}}
    dS_{\Gamma_{jk}},
    \label{eq:explicitW}
\end{equation}
where $d_{jk}(x) = d(x,x_j)+d(x,x_k)$.

\begin{assumption}\label{assum:4}
There are a finite number of geodesics $\gamma_{jk}^{(l)}$
of length $S_0$ joining $x_j$ and $x_k$, where $l \in \Lambda_{jk}$ is the index denoting which geodesic it is.
Let $\Gamma_{jk}^{(l)} \subset \Omega_{j}^{\complement} \cap \Omega_{k}$ be a smooth hypersurface intersecting
$\gamma_{jk}^{(l)}$ such that there would be a point $x_{jk}^{(l)} \in \gamma_{jk}^{(l)} \cap \Omega_{j} \cap \Omega_{k}$ being the only point in $\overline{\Gamma}_{jk}^{(l)} \cap \gamma_{jk}^{(l)}$.
\end{assumption}

\begin{proposition}[Theorem 4.4.6 and 4.4.7 in \cite{Hel88}]\label{prop:speW}
For $j,k\in \{1,2,\ldots,N^+ \}$ and $d(x_j,x_k) = S_0$,
there is a constant $C>0$ such that
\begin{equation}
    \frac{1}{C}h^{1/2} \leq -\widetilde{W}_{jk}e^{S_0/h}
    \leq C h^{1-d/2}.
\end{equation}
More specifically, we have
\begin{equation}
    \widetilde{W}_{jk} = -h^{1/2}
    (\sum_{m}b^{(m)}_{jk}h^m)e^{-S_0/h},
\end{equation}
with
\begin{equation}
    b^{(0)}_{jk} = 2(2\pi)^{\frac{d-1}{2}} \sum_{l\in \Lambda_{jk}}
    \sqrt{\frac{f(x_{jk}^{(l)})}{\mathrm{det}\Big(\nabla^2 d_{jk}(x_{jk}^{(l)})\Big)}}a^{(j)}_0(x_{jk}^{(l)})a^{(k)}_0(x_{jk}^{(l)}).
\end{equation}
\end{proposition}

\begin{proof}
We would use the semi-classical hypothesis to
calculate the leading term of \eq{explicitW}.
Note that only regions near $x^{(l)_{jk}},~l \in \Lambda_{jk}$ have contribution to the leading term:
\begin{align}
    \widetilde{W}_{jk} = -h^{1-d/2}e^{-\frac{S_0}{h}}\sum_{l\in \Lambda_{jk}}\int_{B_l(h)} \left[a_0^{(j)}(x_{jk}^{(l)})a^{(k)}_0(x_{jk}^{(l)})
    \Big(\frac{\partial d(x,x_k)}{\partial n}
    -\frac{\partial d(x,x_j)}{\partial n}\Big)+ O(h) \right]e^{-\frac{d_{jk}(x)-S_0}{h}}
    dS_{\Gamma_{jk}},
\end{align}
where $B_l(h) \subset \Gamma_{jk}$ is a neighborhood of $x^{(l)}_{jk}$. We can set $B_l(h) \to \{ x^{(l)}_{jk}\},~(h\to 0)$.
Recall that $x^{(l)}_{jk}$ is a minimum of $d_{jk}(x)$ in $\Gamma_{jk}$,
in the neighborhood of $x^{(l)}_{jk}$, we have
\begin{align}
   e^{-\frac{d_{jk}(x)-S_0}{h}} = e^{-\frac{X\trans\nabla^2d_{jk}(x_{jk}^{(l)})X}{2h}} + \mathrm{h.o.t},
\end{align}
where $X = \{x_1,\ldots,x_{d-1}\}$ is the local coordinate of $\Gamma_{jk}$ centered in $x_{jk}^{(l)}$ and h.o.t denotes higher order terms.
Examine the integral near $x_{jk}^{(l)}$:
\begin{align}
    \int_{B_l(h)}
    \Big(\frac{\partial d(x,x_k)}{\partial n}
    -\frac{\partial d(x,x_j)}{\partial n}\Big) e^{-\frac{d_{jk}(x)-S_0}{h}}
    dS_{\Gamma_{jk}}
    &=
    \int
    2\sqrt{f(x_{jk}^{(l)})} e^{-\frac{X\trans\nabla^2d_{jk}(x_{jk}^{(l)})X}{2h}}
    \d x_1\cdots \d x_{d-1}
    + \mathrm{h.o.t}
    \nonumber\\
    &=2\sqrt{f(x_{jk}^{(l)})} \sqrt{\frac{(2\pi h)^{d-1}}{\mathrm{det}\Big(\nabla^2d_{jk}(x_{jk}^{(l)})\Big)}}
    + \mathrm{h.o.t}.
\end{align}
So that we have
\begin{align}
    \widetilde{W}_{jk} = -h^{1/2}e^{-\frac{S_0}{h}}\sum_{l\in \Lambda_{jk}} 2(2\pi )^{\frac{d-1}{2}}a_0^{(j)}(x_{jk}^{(l)})a^{(k)}_0(x_{jk}^{(l)}) \sqrt{\frac{f(x_{jk}^{(l)})}{\mathrm{det}\Big(\nabla^2d_{jk}(x_{jk}^{(l)})\Big)}}
    + \mathrm{h.o.t}.
\end{align}
\end{proof}
\begin{theorem}[Energy gap]\label{thm:energygap}
Assume the assumptions for \prop{speW} are satisfied and $N^+ = N$.
The minimal energy gap $\Delta E$ of $P_{|\F}$, i.e., the minimal absolute difference between unequal eigenvalues of $P_{|\F}$,
is given by $\Delta E = \sqrt{h}(b + O(h))e^{-\frac{S_0}{h}}$
where $b>0$ is a constant that depends only on the potential $f$.
\end{theorem}
\begin{proof}
Because of \prop{spePF}, we have now
\begin{align}
    P_{|\F} = \mathrm{diag}(\mu_{j}) + (\widetilde{W}_{i j}) + O(h^{\infty})e^{S_0/h}.
\end{align}
According to \prop{speW}, $\widetilde{W}_{i j}$ is of the form
$-h^{1/2}(b_{ij} + O(h))e^{-S_0/h}$, where $b_{ij}$ depends only on $i,j$, and the potential $f$.
By \assum{TunEnCon}, $|\mu_j - \mu_i| = O(h^{\infty})e^{-S_0/h}$,
so we have
\begin{align}
    P_{|\F} = \mathrm{diag}(\mu)  - h^{1/2}(b_{ij} + O(h))e^{-S_0/h} + O(h^{\infty})e^{S_0/h},
\end{align}
where $\mu$ is chosen to be one $\mu_j$.
The eigenvalues of $P_{|\F}$ should be of the form
\begin{align}
    E_j = \mu + h^{1/2}(b_j + O(h))e^{-S_0/h} + O(h^{\infty})e^{S_0/h},
\end{align}
where $b_j$ are constants depending only on $f$ and $j$.
Therefore, we have the energy gap $\Delta E = \min_{E_j\neq E_i}|E_j - E_i|$ as
\begin{align}
    \Delta E =  h^{1/2}(b + O(h))e^{-S_0/h} + O(h^{\infty})e^{S_0/h} =h^{1/2}(b + O(h))e^{-S_0/h},
\end{align}
where $b$ is a constant only related to $f$.
\end{proof}

\section{Auxiliary Algorithmic Techniques for Quantum Tunneling Walks}\label{append:auxiliary-quantum}
In this section, we introduce two quantum algorithmic techniques that are necessary for obtaining our results. One is quantum simulation of the Schr{\"o}dinger equation (\append{quantumsimulation}), and the other is initial state preparation (\append{initial}).

\subsection{Quantum Simulation of the Schr\"odinger Equation}\label{append:quantumsimulation}
The Schr\"odinger equation we want to simulate is
\begin{align}
    i\frac{\partial }{\partial t}\Phi = (-h^2\Delta + f(x))\Phi.
    \label{eq:simuSchr}
\end{align}
We regard the simulation of \eq{simuSchr} as solving this special partial differential equation (PDE) by quantum algorithms.
Following \cite{CJO19,zhang2021quantum}, we work in $\mathbb{R}^d$ and use the finite difference
method which discretizes the space into
grids with side-length $a$.
In this way, any function $\phi$ can be viewed as a vector
and the Hamiltonian $H = -h^2\Delta + f$
becomes a matrix.
For instance, in the one-dimensional case,
$\phi_j$ is the function value of $\phi$ on the $j$th grid which forms a
vector, $-h^2 \Delta$ is discretized to be $- \frac{h^2}{a^2} L$ where
$L$ is the Laplacian matrix of the graph of the grids (whose off-diagonal entries are $-1$ for connected grids and zero otherwise, and diagonal entries are the degree of the grids), giving
\begin{align}
    (\Delta \phi)_j \approx \frac{1}{a^2}(L\phi)_j = \frac{\phi_{j+1} - 2\phi_{j} + \phi_{j-1}}{a^2},
    \label{eq:discretized}
\end{align}
and $f$ in the Hamiltonian $H$ is discretized to be a matrix whose $j$th diagonal entry is the value of $f$ at the $j$th grid. Here the approximation $\approx$ follows the definition in Eq.~\eqn{approx-definition} and the LHS and RHS of \eq{discretized} is equal under the limit $a \to 0$.
The step of discretization would bring errors which can be
controlled by making $a$ smaller. For convenience, we define $A := -\frac{h^2}{a^2}L_k$,
where $L_k$ is the Laplacian matrix associated to the $k$th order finite difference method \cite{CLO21}, and define $B$ the diagonal matrix discretizing $f$.
We use the matrix $\hat{H}:= A+B$ to model the Hamiltonian of the Schr\"odinger equation under discretization.
To simulate the time evolution governed by $\hat{H}$, we use the following result:
\begin{lemma}[Lemma 6 in \cite{LW19}]\label{lem:QSinter}
Let $A$,$B\in \C^{n\times n}$ be two time-independent Hermitian matrices such that $\|A\|\leq \alpha_A$ and $\|B\|\leq \alpha_B$. Then, the time evolution operator $e^{-i(A+B)t}$ can be approximated up to error $\epsilon$ with
\begin{align}
    O\left(\alpha_B t \frac{\log (\alpha_B t/\epsilon)}{\log \log (\alpha_B t/\epsilon)}\right)
\end{align}
quires to the unitary oracle $O_B$.
\end{lemma}
\begin{remark}\label{rmk:errorsimu}
The $\epsilon$ in \lem{QSinter} originally denotes the error bound quantified by
the spectral norm. Namely, let $\mathcal{T}_t$ be the real operator obtained,
the error $\|\mathcal{T}_t - e^{-i(A+B)t}\| \leq \epsilon$.
For any initial state $\ket{\Phi(0)}$, $\ket{\Phi(t)} = e^{-i(A+B)t}\ket{\Phi(0)}$
is the state we want to prepare and $|\tilde{\Phi}(t)\y = \mathcal{T}_t\ket{\Phi(0)}$ the state we can prepare.
Nevertheless, It is straightforward to have $\|\tilde{\Phi}(t) - \Phi(t) \|\leq \epsilon$, which implies $\epsilon$ that can be understood as the error under $\ell^2$ norm.
\end{remark}
\begin{remark}\label{rmk:equioracle}
Note the Lemma 6 in \cite{LW19} gives the number of queries to the oracle HAM-T.
Because the construction of HAM-T only needs one query to $O_B$, \lem{QSinter} follows.
Moreover, if we specify $B$ as the diagonal matrix discretizing the objective function $f$, one quantum query to $O_B$ can be implemented by one query to the quantum evaluation oracle $U_f$ given by \eqn{quantumquery}.
\end{remark}

Classical and quantum methods approximating \eq{simuSchr} need to discretize the time.
The time step $\Delta t$ for classical methods to maintain convergence
is of the order $1/\|\hat{H}\|$.
Therefore, steps or queries needed for evolving a time $t$ is about $t/\Delta t \sim t \|\hat{H}\|$. The term $\|\hat{H}\|$ can be understood as the maximum energy, indicating that simulating higher energy states accurately costs more.
This intuition is partially true for quantum simulation.
By simulating the Hamiltonian in the interaction picture, the query complexity of
\lem{QSinter} depends on the maximum potential energy instead of the maximum energy because $\|B\| \leq \|f\|_{\infty}$.
With \lem{QSinter} on simulating $\hat{H}$, we now determine the number of queries needed to simulate $H$.
\begin{lemma}[Slightly modified version from {Lemma 2 in \cite{zhang2021quantum}}]\label{lem:quansimu}
Let $f\colon \mathbb{R}^d \to \mathbb{R}$ be bounded in a domain $\Omega \subset \mathbb{R}^d$, consider the Schr\"odinger equation \eq{simuSchr}
restricted in the domain $\Omega$ with periodic boundary condition
or Dirichlet boundary condition (value of the wave function vanishes at the boundary).\footnote{The set $\Omega$ should be appropriate such that the Schr\"odinger equation is solvable. For the purpose of optimization, $\Omega$ should be sufficiently large to contain all minima of interest, and the wave function can hardly hit the boundary $\partial \Omega$. For consistency in mathematics, the initial state should be in $H^1_0(\Omega)\cap H^2(\Omega)$ (see \defn{Dirichlet}) if we use the Dirichlet boundary condition.}
Having the quantum evaluation oracle $U_f$ in \eqn{quantumquery}
and an initial state adapting to the boundary condition at $t=0$, the evolution
of \eq{simuSchr} for time $t$ can be
simulated up to $L^2$ norm error $\epsilon$ using
\begin{align}
    O\left(\|f\|_{L^{\infty}(\Omega)} t \frac{\log (\|f\|_{L^{\infty}(\Omega)} t/\epsilon)}{\log \log (\|f\|_{L^{\infty}(\Omega)} t/\epsilon)}\right)
    \label{eq:querynum}
\end{align}
queries to $U_f$.
\end{lemma}
\begin{proof}[Proof of \lem{quansimu}]
There are two sources that producing error: the first comes from discretizing the spatial domain (approximating the Hamiltonian $H$ by the matrix $\hat{H}$) and the second comes from simulating the Hamiltonian matrix
$\hat{H}$ with the help of \lem{QSinter}.

Let the initial wave function be $\Phi(0,x)$ and the accurate wave function after time $t$ be $\Phi(t,x) = e^{-iHt} \Phi(0,x)$.
By discretization, let $\{x_j\}_{j\in I}$ be the grid points where $I$ is the set of index,
we then obtain a vector $\hat{\Phi}(0)$ whose $j$th entry is the value $\Phi(0,x_j)$.
Denote the vector $\hat{\Phi}(t) = e^{-i\hat{H}t}\hat{\Phi}(0)$,\footnote{The boundary condition can be satisfied by
setting some entries at the edge of the matrix $A$ in $\hat{H}$ as 0 or 1/$\Delta x^2$,
which is a standard process in solving differential equations by finite difference methods and can be found in extensive classical literature.} we can rebuild a function
$\hat{\Phi}^{\vee}(t,x)$ by setting the function value on the $j$th grid (a hypercube with side-length $a$) to be the $j$th entry of $\hat{\Phi}(t)$.
By setting $1/a = \mathrm{poly}(d)\mathrm{poly}(\log (2/\epsilon))$, it can be made that $\|\hat{\Phi}^{\vee}(t,\cdot) - \Phi(t,x)\|_{L^2(\Omega)} \leq \epsilon/2$ \cite{CLO21}.
Since $\|B\|\leq \|f\|_{L^{\infty}(\Omega)}$, \lem{QSinter} permits an operator $\mathcal{T}_t$ simulating $e^{-i\hat{H}t}$ with
\begin{align}
   Q(t,\epsilon) = O\left(\|f\|_{L^{\infty}(\Omega)} t \frac{\log (2\|f\|_{L^{\infty}(\Omega)} t/\epsilon)}{\log \log (2\|f\|_{L^{\infty}(\Omega)} t/\epsilon)}\right)
\end{align}
queries to $O_B$, such that $\|\tilde{\Phi}(t) - \hat{\Phi}(t)\| \leq \epsilon/2$ where $\tilde{\Phi}(t) := \mathcal{T}_t \hat{\Phi}(0)$ (see \rmk{errorsimu}).
Similar to rebuilding $\hat{\Phi}^{\vee}(t,x)$ from the vector
$\hat{\Phi}(t)$, we can rebuild a function $\tilde{\Phi}^{\vee}(t,x)$ from $\tilde{\Phi}(t)$ and
$\|\tilde{\Phi}(t) - \hat{\Phi}(t)\| \leq \epsilon/2$ is equivalent to $\|\tilde{\Phi}^{\vee}(t,\cdot) - \hat{\Phi}^{\vee}(t,\cdot)\|_{L^2(\Omega)} \leq \epsilon/2$.
Finally, by the triangle inequality,
\begin{align}
 \|\tilde{\Phi}^{\vee}(t,\cdot) - \Phi(t,\cdot)\|_{L^2(\Omega)} \leq \|\tilde{\Phi}^{\vee}(t,\cdot) - \hat{\Phi}^{\vee}(t,\cdot)\|_{L^2(\Omega)} + \|\hat{\Phi}^{\vee}(t,\cdot) - \Phi(t,x)\|_{L^2(\Omega)}
  \leq \epsilon.
\end{align}

In the present case, a query to $O_B$ can be implemneted by
one query to the quantum evaluation oracle $U_f$ (\rmk{equioracle}). Note that
\begin{align}
   \frac{\log (2\|f\|_{L^{\infty}(\Omega)} t/\epsilon)}{\log \log (2\|f\|_{L^{\infty}(\Omega)} t/\epsilon)}
   \leq  \frac{1 + \log (\|f\|_{L^{\infty}(\Omega)} t/\epsilon)}{\log \log (\|f\|_{L^{\infty}(\Omega)} t/\epsilon)}
   \leq \frac{2\log (\|f\|_{L^{\infty}(\Omega)} t/\epsilon)}{\log \log (\|f\|_{L^{\infty}(\Omega)} t/\epsilon)},
\end{align}
we can also write
\begin{align}
   Q(t,\epsilon) = O\left(\|f\|_{L^{\infty}(\Omega)} t \frac{\log (\|f\|_{L^{\infty}(\Omega)} t/\epsilon)}{\log \log (\|f\|_{L^{\infty}(\Omega)} t/\epsilon)}\right).
\end{align}
In summary, using $Q(t,\epsilon)$ queries to $U_f$, we are able to
obtain a wave function $\tilde{\Phi}^{\vee}(t,x)$
close to the actual solution $\Phi(t,x)$ up to error $\epsilon$,
which completes the proof.
\end{proof}
\begin{remark}
The discretization method introduced indicates that every grid is a hypercube with side-length $a$.
If the boundary $\Omega$ is curved, grids fitting the boundary and the boundary condition can be irregular. In this case, discretization of the
Laplacian, or equivalently the form of the matrix $A$, will be different. We clarify that \lem{quansimu} is still true for irregular grids since
grids do not affect the fact that $\|B\| \leq \|f\|_{L^{\infty}(\Omega)}$.
\end{remark}

Up to now, we introduce quantum simulation in $\mathbb{R}^d$ or more precisely, a subset $\Omega$ of $\mathbb{R}^d$.
We further claim that for smooth (compact) manifolds that can be embedded in $\mathbb{R}^d$, the query complexity presented in \lem{quansimu} should still be valid.
On such manifolds, say $M$, we can use discretization methods in $\mathbb{R}^d$
to approximate the discretization in $M$.
Note that the Laplacian should be replaced by the Laplace–Beltrami operator, adding difficulties for construct the matrix $A$.
However, since $A$ does not affect query complexity and $B$ will always be a diagonal matrix and $\|B\| \leq \|f\|_{L^{\infty}(M)}$,
the form of \eq{querynum} should be the same.
In general, quantum simulation on manifolds are interesting topics, deserving further studies.

Note that \lem{QSinter} we use bound the error of states by bounding the error of operators.
This may not provide the possible minimal queries needed if the initial state is in a low energy subspace which is exactly the case for QTW. Quantum simulation in a low energy subspace intuitively lowers the query complexity bound \cite{csahinouglu2021hamiltonian},
indicating larger speedups of QTW.

\subsection{Initial State Preparation}\label{append:initial}
QTW is a quantum simulation algorithm making use of quantum tunneling. To fulfill the tunneling phenomenon, not only should the potential function $f(x)$ satisfy
some assumptions (as is presented in \sec{quantumpre}), but also the
initial state should be in a low-energy subspace. For completeness, we introduce in this appendix possible ways for preparing initial states under different scenarios.

For a landscape $f(x)$ satisfying assumptions in \sec{quantumpre},
\prop{dimEdimF} in \append{interactionmatrix}
guarantees that for sufficiently small $h$, local eigenstates with cut-offs (see \defn{localcutoff})
can span a space $\E$ that is close to $\F$, a low-energy subspace of the Hamiltonian $H = -h^2\Delta + f$.
Moreover, orthonormalized local eigenstates (see \defn{ortholocal}) can span the space $\F$.
Note that we only consider tunneling between local ground states,
hence the above claim implies that orthonormalized local ground states can span the space $\F$.
\assum{quantum1} supposes that there are $N$ global minima (wells) of interest.
Therefore, we can label the minima by $j\in \{1,\ldots,N\}$ and denote
the orthonormalized local ground state corresponding to the $j$th minima by $\ket{j}$.
Let $\ket{E_j}$ ($j=1,2,\ldots$) be eigenstates of $H$ with eigenvalues $E_j$, respectively.
Mathematically, we can write
\begin{align}
    \F = \mathrm{span}\{\ket{j}\}_{j=1}^N = \mathrm{span}\{\ket{E_j}\}_{j=1}^N.
\end{align}

The goal is to prepare a state (nearly) in the subspace $\F$ such that QTW can be initiated.

First, we explore a simple scenario where we have previous knowledge of one minimum.

\paragraph{Local ground state preparation.}
With the knowledge of one minimum, say the $j$th one, an efficient way of state preparation is to
prepare or approximate the local ground state corresponding to the $j$th minimum.
In practice, $\ket{j}$ has no explicit form and we need to
prepare the original local ground state $\ket{\phi_j}$ (see \defn{localeigen}).
Restating \defn{localeigen} in the present case,
$M_j$ (rigorously given by \eq{smallregion})
is a bounded region that contains only the $j$th minimum
and $\ket{\phi_j}$ is the ground state of $H_{M_j}$ which
is the Dirichlet realization of $H$ in $M_j$ (see \defn{Dirichlet}).
Regardless of the rigorous mathematical definition,
the function $\phi_j = \ip{x}{\phi_j} \colon M_j \to \C$
could be understood as the function $\phi$ satisfying $(-h^2\Delta+f)\phi = E \phi$ and $\phi_{|\partial M_j} = 0$ whose $E$ is the smallest.

One way to prepare (approximate) $\ket{\phi_j}$ is through adiabatic quantum computing restricted in
$M_j$ (some set $\Omega_j \subset M_j$).
Adiabatic quantum computing presents such a picture, for a time-dependent Hamiltonian $H(t)$,
if the initial state $\psi(0)$ is the ground state of $H(0)$ and $H(t)$ changes sufficiently slow, the system will approximately stay at the ground state of $H(t)$ at $t$.
For the purpose of preparing ground state of $H_{M_j}$ ($H_{\Omega_j}$),
we can let $H(0)$ be some known Hamiltonian and let $H(T) = H_{M_j}$ ($H_{\Omega_j}$).
Sufficiently slowly changing $H(t)$ is equivalent to a sufficiently large $T$.
Let $G(t)$ be the fundamental gap of $H(t)$, i.e., the difference between the second smallest and smallest eigenvalues of $H(t)$,
time for state preparation can be bounded as follows.

\begin{lemma}[Informal quantum adiabatic theorem~\cite{Dua20}]\label{lem:adiabatic}
Let $\phi(t)$ be be the unit eigenstate of $H(t)$ and $\ip{\partial_t \phi(t)}{\phi(t)} = 0$. Simulating the time dependent Hamiltonian $H(t)$ when $\psi(0) = \phi(0)$,
we can obtain $\psi(T)$ such that $\|\psi(T) - \phi(t)\|\leq \epsilon$ if
\begin{align}
    T \geq \frac{1}{\epsilon}\left(\frac{4\|\partial_t H\| + 2 \| \partial_{t t} H\|}{G_m^2} + \frac{20\|\partial_{t} H \|^2}{G_m^3}\right),
\end{align}
where $G_m = \inf_{0 \leq t \leq T} G(t)$.
\end{lemma}
Intuitively, the time cost is bounded by the inverse of the minimal fundamental gap.
For simplicity, we study $H_{\Omega_j}$ for small convex $\Omega_j$
where $f_{|\Omega_j}$ is also weakly convex.
The ground state of $H_{\Omega_j}$ can be close to that of $H_{M_j}$ if $h$ is small.
By \cite{AC11}, fundamental gaps of Hamiltonians restricted in $\Omega_j$
is $\Omega(h^2/D^2)$ where $D$ is the diameter of $\Omega_j$.
As is shown by \thm{informalQTW}, time cost of
QTW is dominated by $1/\Delta E$ where $\Delta E = O(e^{-S_0/h}) = O(h^{\infty})$ is also an energy gap.
Therefore, for small $h$ that ensuring tunneling and high accuracy,
$1/\Delta E$ is much larger than the inverse of any fundamental gaps of Hamiltonians restricted in $\Omega_j$, giving that local ground state preparation is much easier than QTW.

Classical methods could also help to prepare local ground states.
Still given the $j$th minimum $x_j$, we can use classical sampling to
learn the Hessian matrix at $x_j$, $\nabla^2f(x_j)$.
Then, we can directly put the following wave function centered at $x_j$.
\begin{align}
    \phi(x) = \frac{\big(\det\sqrt{\nabla^2 f(x_j)}\big)^{1/4}}{(\sqrt{2}\pi h)^{d/4}}
    e^{-\frac{(x-x_j)\trans\sqrt{\nabla^2 f(x_j)}(x-x_j)}{2\sqrt{2}h}}.
    \label{eq:samplingLG}
\end{align}
Note that if $f$ is quadratic, the $\phi(x)$ in \eq{samplingLG} will be the ground state precisely.
Using the ground state under quadratic landscape to approximate the actual ground state is actually the idea of WKB approximation in \append{WKB}.
The smaller $h$ is, the more accurate the approximation of \eq{samplingLG} can be.
In addition, Gaussian wave packets can be efficiently sampled on quantum computers~\cite{kitaev2008wavefunction,rattew2021efficient}.
It is still true that QTW takes more time than initial state preparation.

Local ground state preparation is convenient because it suffices to focus on a neighborhood of one minimum. However, it also has drawback as we cannot ensure that the state we obtain is in $\F$. Detailed error estimation is needed in future work.\footnote{To be specific, $\ket{j}$ is in $\F$ but we cannot prepare; what we can approximately prepare is $\ket{\phi_j}$. \defn{localeigen}, \defn{localcutoff}, and \defn{ortholocal} together give an estimation that $\ip{x}{j} = \phi_j + O(e^{-S_0/h})$.}

\paragraph{General initial state preparation.}
Next, we consider the general scenario where we do not have any priori knowledge on any minimum.

A straightforward way is to use classical algorithms (e.g., SGD) or quantum algorithms to find one minimum first and then apply above procedure of preparing local ground states.
Besides, we present a new idea that can directly yield a
state near the subspace $\F$.
The method is based on a generalized version of the quantum adiabatic theorem. Define $ G_{\F} := E_{N+1} - E_N$ as the fundamental gap of space $\F$.
The proof of \lem{adiabatic} suggests that
as long as the total time $T$
is long enough respective to the inverse of the fundamental gap of space $\F$, the system can always be near the space $\F$. See also~\cite{boixo2010necessary}.
The state $\ket{E_{N+1}}$ is orthogonal to $\F$ as $\F = \mathrm{span}\{E_j\}_{j=1}^N$.
Hence, $\ket{E_{N+1}}$ is at least a superposition of local first excited states
such that $G_{\F}$ has a similar value to the gap between local ground energy and local first excited energy.
Then, for any landscape $f$ satisfying assumptions in \sec{quantumpre},
we can prepare states in the low-energy subspace $\F$ with time $\sim \mathrm{poly}(1/G_{\F})$.

In summary, we provide several ways of preparing the desired initial states under different conditions. In principle, discussions on initial state preparation also suggest potential hybrid quantum-classical algorithms and potential quantum simulation algorithms containing both adiabatic quantum computing and QTW, which are of independent interest.

\section{Technical Details for Quantum Tunneling Walks}\label{append:section3}

\subsection{Details of quantum mixing time}
\subsubsection{Proof of \texorpdfstring{\lem{qtwmixingtime}}{Lemma 2}}\label{append:prooflem2}
\lem{qtwmixingtime} aims to bound the minimum time $\tau$ enabling $\|\rho_{\rm QTW}(\tau,x) -\mu_{\rm QTW}(x)\|_1 \leq \epsilon$ (i.e., the mixing time).
\begin{proof}
Straightforward calculations yield
\begin{align}
    \|\rho_{\rm QTW}(\tau,x) -\mu_{\rm QTW}(x)\|_1 &=
    \int \left|\sum_{E_k \neq E_{k^{\prime}}} \frac{1-e^{-i(E_k - E_{k^{\prime}})\tau}}{i(E_k - E_{k^{\prime}})\tau}
   \z x|E_k\y \z E_k |\Phi(0)\y \z \Phi(0) |E_{k^{\prime}}\y \z E_{k^{\prime}}| x\y\right| \d x \\
   & \leq \int \sum_{E_k \neq E_{k^{\prime}}}\frac{2}{|E_k - E_{k^{\prime}}|\tau}\left|
   \z x|E_k\y \z E_k |\Phi(0)\y \z \Phi(0) |E_{k^{\prime}}\y \z E_{k^{\prime}}| x\y\right| \d x \\
   &\leq \sum_{E_k \neq E_{k^{\prime}}}\frac{2}{|E_k - E_{k^{\prime}}|\tau} \sum_{j,j'}|
   \z j|E_k\y \z E_k |\Phi(0)\y \z \Phi(0) |E_{k^{\prime}}\y \z E_{k^{\prime}}| j'\y | \nonumber \\
    &\qquad \times \int |\ip{x}{j} \ip{j'}{x}| \d x.
\end{align}
Because of \assum{quantum1} in \sec{prelim} which can lead to \assum{TunEnCon}, for $j \neq j^{\prime}$, we have
\begin{equation}
    \z x|j\y \z j^{\prime}| x\y = O(h^{-N'} e^{-\frac{d(x_j,x)+d(x_{j^{\prime}},x)}{h}})
    \leq O(h^{-N'}e^{-\frac{S_0}{h}}) = O(h^{\infty}),
    \label{eq:normoverlap}
\end{equation}
where $N'>0$ is some constant.
Therefore,
\begin{equation}
    \int |\z x|j\y \z j^{\prime}| x\y| \d x =
    \left\{\begin{array}{ll}
        1,\quad j=j', \\
        O(h^{\infty}), \quad j\neq j'.
    \end{array}
    \right.
\end{equation}
By the Cauchy-Schwarz inequality,
\begin{align}
   |\z j|E_k\y \z E_{k^{\prime}}| j'\y| \leq \frac{1}{2}(|\z j|E_k\y|^2 + |\z E_{k^{\prime}}| j'\y|^2),
\end{align}
such that,
\begin{align}
    \|\rho_{\rm QTW}(\tau,x) -\mu_{\rm QTW}(x)\|_1
   &\leq \sum_{E_k \neq E_{k^{\prime}}}\frac{2}{|E_k - E_{k^{\prime}}|\tau} \sum_{j,j'}|
   \z j|E_k\y \z E_k |\Phi(0)\y \z \Phi(0) |E_{k^{\prime}}\y \z E_{k^{\prime}}| j'\y | \nonumber \\
    &\qquad \times \int |\ip{x}{j} \ip{j'}{x}| \d x\\
   &\leq\sum_{E_k \neq E_{k^{\prime}}}\frac{2}{|E_k - E_{k^{\prime}}|\tau} \sum_{j,j'}|
    \z E_k |\Phi(0)\y \z \Phi(0) |E_{k^{\prime}}\y|\frac{1}{2}(|\z j|E_k\y|^2 + |\z E_{k^{\prime}}| j'\y|^2) \nonumber \\
    &\qquad \times \int |\ip{x}{j} \ip{j'}{x}| \d x.
\end{align}
Splitting the sum on $j$ and $j'$, we have
\begin{align}
    \|\rho_{\rm QTW}(\tau,x) -\mu_{\rm QTW}(x)\|_1&\leq\sum_{E_k \neq E_{k^{\prime}}}\frac{2}{|E_k - E_{k^{\prime}}|\tau} \sum_{j=1}^n|
    \z E_k |\Phi(0)\y \z \Phi(0) |E_{k^{\prime}}\y|\frac{1}{2}(|\z j|E_k\y|^2 + |\z E_{k^{\prime}}| j\y|^2) \nonumber \\
    &+\sum_{E_k \neq E_{k^{\prime}}}\frac{2}{|E_k - E_{k^{\prime}}|\tau} \sum_{j\neq j'}|
    \z E_k |\Phi(0)\y \z \Phi(0) |E_{k^{\prime}}\y|\frac{1}{2}(|\z j|E_k\y|^2 + |\z E_{k^{\prime}}| j'\y|^2) |O(h^{\infty})| \\
    &\leq \sum_{E_k \neq E_{k^{\prime}}}\frac{2}{|E_k - E_{k^{\prime}}|\tau} |\z E_k |\Phi(0)\y \z \Phi(0) |E_{k^{\prime}}\y| \nonumber \\
    &+\sum_{E_k \neq E_{k^{\prime}}}\frac{2}{|E_k - E_{k^{\prime}}|\tau}|
    \z E_k |\Phi(0)\y \z \Phi(0) |E_{k^{\prime}}\y|(N-1) |O(h^{\infty})|\\
    & = \sum_{E_k \neq E_{k^{\prime}}}\frac{2}{|E_k - E_{k^{\prime}}|\tau} |\z E_k |\Phi(0)\y \z \Phi(0) |E_{k^{\prime}}\y| [1 + (N-1)|O(h^{\infty})|].
\end{align}
One of the sufficient conditions to obtain \eq{mixingqtw} is that
\begin{align}
   \tau \geq \frac{2}{\epsilon}\sum_{E_k \neq E_{k^{\prime}}} \frac{| \z E_k |\Phi(0)\y \z \Phi(0) |E_{k^{\prime}}\y|}{|E_k - E_{k^{\prime}}|}[1 + (N-1)|O(h^{\infty})|].
\end{align}
Since the mixing time is the minimum $\tau$ for \eq{mixingqtw} to be valid, an upper bound of the mixing time is given by the above equation. In other words,
\begin{align}
   T_{\rm mix} = O\bigg(\frac{1}{\epsilon}\sum_{E_k \neq E_{k^{\prime}}} \frac{| \z E_k |\Phi(0)\y \z \Phi(0) |E_{k^{\prime}}\y|}{|E_k - E_{k^{\prime}}|}[1 + (N-1)|O(h^{\infty})|]\bigg)
\end{align}
Note the definition of $\Delta E$ and use the Cauchy-Schwarz inequality, we have
\begin{align}
   \sum_{E_k \neq E_{k^{\prime}}} \frac{| \z E_k |\Phi(0)\y \z \Phi(0) |E_{k^{\prime}}\y|}{|E_k - E_{k^{\prime}}|}
   &\leq \frac{1}{\Delta E} \sum_{E_k \neq E_{k^{\prime}}} | \z E_k |\Phi(0)\y \z \Phi(0) |E_{k^{\prime}}\y|\nonumber\\
   & \leq \frac{1}{ \Delta E} \sum_{k, k^{\prime}} \frac{1}{2}(| \z E_k |\Phi(0)\y |^2 + |\z \Phi(0) |E_{k^{\prime}}\y|^2)
   = \frac{N}{\Delta E},
   \label{eq:n/E}
\end{align}
which complete the proof of \lem{qtwmixingtime}.
\end{proof}

\subsubsection{Proof of \texorpdfstring{\lem{limitdis}}{Lemma 3}}\label{append:prooflem3}
\lem{limitdis} reveals the relation between probabilities $p(\infty,j)$ for $j=1,\ldots,N$ and the limit distribution $\mu_{\rm QTW}$.
\begin{proof}
We further estimate $\mu_{\rm QTW}$ as follows
\begin{align}
     \mu_{\rm QTW}(x) &=
    \sum_{j,j^{\prime}}
    \sum_{E_k=E_{k^{\prime}}} \z x|j\y \z j |E_k\y \z E_k |\Phi(0)\y \z \Phi(0) |E_{k^{\prime}}\y \z E_{k^{\prime}}| j^{\prime}\y \z j^{\prime}| x\y \nonumber\\
    & = \sum_{j}
    \sum_{E_k=E_{k^{\prime}}} |\z x|j\y|^2 \z j |E_k\y \z E_k |\Phi(0)\y \z \Phi(0) |E_{k^{\prime}}\y \z E_{k^{\prime}}| j\y \nonumber\\
    &~~~+ \sum_{j\neq j^{\prime}}
    \sum_{E_k=E_{k^{\prime}}} \z j |E_k\y \z E_k |\Phi(0)\y \z \Phi(0) |E_{k^{\prime}}\y \z E_{k^{\prime}}| j^{\prime}\y \z j^{\prime}| x\y\z x|j\y \nonumber\\
    & = \sum_{j}
    p(\infty,j) |\z x|j\y|^2 + O(h^{\infty}),
\end{align}
where the last equality uses \eq{normoverlap} and
\begin{align}
    p(\infty,j) =
    \sum_{E_k=E_{k^{\prime}}}\z j |E_k\y \z E_k |\Phi(0)\y \z \Phi(0) |E_{k^{\prime}}\y \z E_{k^{\prime}}| j\y.
\end{align}
Because $|\z x|j\y|^2$ is the position distribution of the state $|j \y$, $\mu_{\rm QTW}$ can be roughly seen as a weighted sum of the distributions of local states when $h$ is small enough, which is intuitive.
\end{proof}

\subsubsection{Proof of \texorpdfstring{\lem{mixingqw}}{Lemma 4}}\label{append:prooflem4}
\lem{mixingqw} bounds the mixing time of quantum walks whose evolution is governed by the interaction matrix of QTW.
\begin{proof}
By the definition of $p(\tau,j)$, we have
\begin{align}
   \sum_{j=1}^N |p(\tau,j) - p(\infty,j)| &=
   \sum_{j=1}^N \left|\sum_{E_k \neq E_{k^{\prime}}} \frac{1-e^{-i(E_k - E_{k^{\prime}})\tau}}{i(E_k - E_{k^{\prime}})\tau}
   \z j|E_k\y \z E_k |\Phi(0)\y \z \Phi(0) |E_{k^{\prime}}\y \z E_{k^{\prime}}| j\y\right| \nonumber\\
   & \leq \sum_{j=1}^N \sum_{E_k \neq E_{k^{\prime}}} \frac{2}{|E_k - E_{k^{\prime}}|\tau}
   |\z j|E_k\y \z E_k |\Phi(0)\y \z \Phi(0) |E_{k^{\prime}}\y \z E_{k^{\prime}}| j\y|.
\end{align}
The Cauchy-Schwarz inequality gives
\begin{align}
   |\z j|E_k\y \z E_{k^{\prime}}| j\y| \leq \frac{1}{2}(|\z j|E_k\y|^2 + |\z E_{k^{\prime}}| j\y|^2),
\end{align}
and then,
\begin{align}
   \sum_{j=1}^N |p(\tau,j) - p(\infty,j)|
   &\leq \sum_{E_k \neq E_{k^{\prime}}}\sum_{j=1}^N \frac{2}{|E_k - E_{k^{\prime}}|\tau}
   | \z E_k |\Phi(0)\y \z \Phi(0) |E_{k^{\prime}}\y|\frac{1}{2}(|\z j|E_k\y|^2 + |\z E_{k^{\prime}}| j\y|^2) \nonumber\\
   &=\sum_{E_k \neq E_{k^{\prime}}} \frac{2}{|E_k - E_{k^{\prime}}|\tau}
   | \z E_k |\Phi(0)\y \z \Phi(0) |E_{k^{\prime}}\y|.
\end{align}
The condition \eq{disofp} is satisfied if
\begin{align}
   \tau \geq \frac{2}{\epsilon}\sum_{E_k \neq E_{k^{\prime}}} \frac{| \z E_k |\Phi(0)\y \z \Phi(0) |E_{k^{\prime}}\y|}{|E_k - E_{k^{\prime}}|},
\end{align}
Therefore, an upper bound of the quantum walk mixing time $t_{\rm mix}$ is
\begin{align}
   \frac{2}{\epsilon}\sum_{E_k \neq E_{k^{\prime}}} \frac{| \z E_k |\Phi(0)\y \z \Phi(0) |E_{k^{\prime}}\y|}{|E_k - E_{k^{\prime}}|} \leq \frac{2N}{\epsilon \Delta E}.
\end{align}
The last inequality uses \eq{n/E}, which finishes the proof.
\end{proof}

\subsection{Details of quantum hitting time}
\subsubsection{Proof of \texorpdfstring{\lem{qwhitting}}{Lemma 5}}\label{append:prooflem5}
We introduce by \lem{qwhitting} the typical upper bound of hitting time for quantum walks \cite{CCD+03}.

\begin{proof}
The definition of $p(\tau,j)$ gives
\begin{align}
   |p(\tau,j) - p(\infty,j)| &=
   \left|\sum_{E_k \neq E_{k^{\prime}}} \frac{1-e^{-i(E_k - E_{k^{\prime}})\tau}}{i(E_k - E_{k^{\prime}})\tau}
   \z j|E_k\y \z E_k |\Phi(0)\y \z \Phi(0) |E_{k^{\prime}}\y \z E_{k^{\prime}}| j\y\right| \nonumber\\
   & \leq \sum_{E_k \neq E_{k^{\prime}}} \frac{2}{|E_k - E_{k^{\prime}}|\tau}
   |\z j|E_k\y \z E_k |\Phi(0)\y \z \Phi(0) |E_{k^{\prime}}\y \z E_{k^{\prime}}| j\y|,
\end{align}
which naturally leads to
\begin{align}
   p(\tau,j)
    \geq  p(\infty,j) - \sum_{E_k \neq E_{k^{\prime}}} \frac{2}{|E_k - E_{k^{\prime}}|\tau}
   |\z j|E_k\y \z E_k |\Phi(0)\y \z \Phi(0) |E_{k^{\prime}}\y \z E_{k^{\prime}}| j\y|.
\end{align}
The second term of the right hand side of the above equation can be bounded as
\begin{align}
  \sum_{E_k \neq E_{k^{\prime}}}&\frac{2}{|E_k - E_{k^{\prime}}|\tau}
   |\z j|E_k\y \z E_k |\Phi(0)\y \z \Phi(0) |E_{k^{\prime}}\y \z E_{k^{\prime}}| j\y| \nonumber\\
   & \leq \frac{2}{\Delta E \tau}\sum_{E_k, E_{k^{\prime}}}
   |\z j|E_k\y \z E_k |\Phi(0)\y \z \Phi(0) |E_{k^{\prime}}\y \z E_{k^{\prime}}| j\y| \\
   & \leq \frac{1}{\Delta E \tau}\sum_{E_k, E_{k^{\prime}}}(|\z j|E_k\y \z E_{k^{\prime}}| j\y|^2 + |\z \Phi(0) |E_{k^{\prime}}\y\z E_k |\Phi(0)\y|^2 )
    = \frac{2}{\Delta E \tau}.
\end{align}
Here, the second inequality uses the Cauchy-Schwarz inequality.
Therefore, we have $p(\tau, j)\geq p(\infty,j)-2/\Delta E \tau$.
Because $t_{\rm hit} (j)$ is defined as the lower bound of $\tau/p(\tau,j)$ for all $\tau>0$. Define $\tau_{\epsilon}= 2/\Delta E \epsilon$ for $\epsilon<p(\infty,j)$, we must have
\begin{align}
    t_{\rm hit} (j) \leq \frac{\tau_{\epsilon}}{p(\tau_{\epsilon},j)} \leq \frac{2/\Delta E \epsilon}{p(\infty,j) - \epsilon}.
\end{align}
\end{proof}

\subsubsection{Proof of \texorpdfstring{\lem{qtwhitting}}{Lemma 6}}\label{append:prooflem6}
For QTW, the hitting time is estimated by \lem{qtwhitting} similar to \lem{qwhitting}.
\begin{proof}
We first estimate the quantity $\int_{\Omega_j}\rho_{\rm QTW}(\tau,x)\d x$:
\begin{align}
    \big| \int_{\Omega_j} (\rho_{\rm QTW}(\tau,x) -\mu_{\rm QTW}(x))\d x \big|
   & \leq \int_{\Omega_j} \sum_{E_k \neq E_{k^{\prime}}}\frac{2}{|E_k - E_{k^{\prime}}|\tau}\left|
   \z x|E_k\y \z E_k |\Phi(0)\y \z \Phi(0) |E_{k^{\prime}}\y \z E_{k^{\prime}}| x\y\right| \d x \\
   &\leq \sum_{E_k \neq E_{k^{\prime}}}\frac{2}{|E_k - E_{k^{\prime}}|\tau} \sum_{l,j'}|
   \z l|E_k\y \z E_k |\Phi(0)\y \z \Phi(0) |E_{k^{\prime}}\y \z E_{k^{\prime}}| j'\y | \nonumber \\
    &\qquad \times \int_{\Omega_j} |\ip{x}{l} \ip{j'}{x}| \d x.
\end{align}
Note that \assum{TunEnCon} is satisfied as assumed, for $l \neq j^{\prime}$, we have
\begin{equation}
    \z x|l\y \z j^{\prime}| x\y = O(h^{-N'} e^{-\frac{d(x_j,x)+d(x_{j^{\prime}},x)}{h}})
    \leq O(h^{-N'}e^{-\frac{S_0}{h}}) = O(h^{\infty}),
\end{equation}
where $N'>0$ is some constant.
Then, it is readily to have
\begin{align}
    \big| \int_{\Omega_j} (\rho_{\rm QTW}(\tau,x) -\mu_{\rm QTW}(x))\d x \big|
   &\leq \frac{2}{\Delta E\tau} \big[\sum_{E_k \neq E_{k^{\prime}}} \sum_{l}|
   \z l|E_k\y \z E_k |\Phi(0)\y \z \Phi(0) |E_{k^{\prime}}\y \z E_{k^{\prime}}| l\y | \nonumber \\
    &\qquad \times \int_{\Omega_j} |\ip{x}{l} \ip{l}{x}| \d x + |O(h^{\infty})| \big].
\end{align}
According to \cor{L2norm}, the integrals $\int_{\Omega_j} |\ip{x}{j} \ip{j}{x}| \d x = 1+ O(h^{\infty})$ and $\forall l\neq j,~\int_{\Omega_l} |\ip{x}{l} \ip{l}{x}| \d x = O(h^{\infty})$, which can yield
\begin{align}
    \big| \int_{\Omega_j} (\rho_{\rm QTW}(\tau,x) -\mu_{\rm QTW}(x))\d x \big|
   \leq \frac{2}{\Delta E\tau} \big[\sum_{E_k \neq E_{k^{\prime}}} |
   \z j|E_k\y \z E_k |\Phi(0)\y \z \Phi(0) |E_{k^{\prime}}\y \z E_{k^{\prime}}| j\y | + |O(h^{\infty})| \big].
\end{align}
Apply the Cauchy-Schwarz inequality, we have
\begin{align}
    \sum_{E_k \neq E_{k^{\prime}}} |
   \z j|E_k\y \z E_k |\Phi(0)\y \z \Phi(0) |E_{k^{\prime}}\y \z E_{k^{\prime}}| j\y |  &\leq \sum_{E_k, E_{k^{\prime}}} |
   \z j|E_k\y \z E_k |\Phi(0)\y \z \Phi(0) |E_{k^{\prime}}\y \z E_{k^{\prime}}| j\y | \\
   &\leq \frac{1}{2}\sum_{E_k, E_{k^{\prime}}}(|\z j|E_k\y \z E_{k^{\prime}}| j\y|^2 + |\z \Phi(0) |E_{k^{\prime}}\y\z E_k |\Phi(0)\y|^2 ) = 1,
\end{align}
and then
\begin{align}
    \big| \int_{\Omega_j} (\rho_{\rm QTW}(\tau,x) -\mu_{\rm QTW}(x))\d x \big|
   \leq \frac{2}{\Delta E\tau}(1 + |O(h^{\infty})|).
   \label{eq:rho-mudx}
\end{align}
So, $\int_{\Omega_j} \rho_{\rm QTW}(\tau,x) \d x$ can be bounded as
\begin{align}
    \int_{\Omega_j}\rho_{\rm QTW}(\tau,x) \d x \geq \int_{\Omega_j}\mu_{\rm QTW}(x) \d x - \frac{2}{\Delta E\tau}(1 + |O(h^{\infty})|).
\end{align}
Use \cor{L2norm} again, we get $\int_{\Omega_j}\mu_{\rm QTW}(x) \d x = p(\infty,j) + O(\infty)$ which completes the proof of \eq{qtwhitting}.

Because $T_{\rm hit} (\Omega_j)$ is the lower bound of $\tau/\int_{\Omega_j}\rho_{\rm QTW}(\tau,x)\d x$ for all $\tau>0$. Given that $\tau_{\epsilon}= 2(1+|O(h^{\infty})|)/\Delta E \epsilon$ (the term $O(h^{\infty})$ should be the same as that in \eq{rho-mudx}) for $\epsilon<\int_{\Omega_j}\mu_{\rm QTW}(x)\d x$, we must have
\begin{align}
    T_{\rm hit} (\Omega_j) \leq \frac{\tau_{\epsilon}}{\int_{\Omega_j}\rho_{\rm QTW}(\tau,x)\d x} \leq \frac{2}{\Delta E \epsilon} \frac{1+|O(h^{\infty})|}{\int_{\Omega_j}\mu_{\rm QTW}(x)\d x - \epsilon}.
\end{align}
\end{proof}

\subsection{Details about tensor decomposition}
\subsubsection{Proof of \texorpdfstring{\lem{TenDecEig}}{Lemma 7}}\label{append:prooflem7}
\lem{TenDecEig} explores the eigenstates and eigenvalues of $H_{|\F}$ given by \eq{TenDec}.To prove \lem{TenDecEig}, we first study a simpler matrix.
\begin{lemma}\label{lem:mostsymmetric}
Consider matrices of the following form:
\begin{equation}
    A = \left(
    \begin{array}{cccc}
        \mu & w & \cdots &  w \\
         w & \mu & w & \vdots \\
          \vdots & w & \ddots & w \\
          w & \cdots  &  w& \mu  \\
    \end{array}
    \right).
\end{equation}
The eigenvalues are
\begin{equation}
    E_1 = \cdots = E_{n-1} = \mu - w,~E_n = \mu + (n-1)w,
\end{equation}
And the corresponding eigenstates can be given by
\begin{equation}
    |E_k\y  = \frac{1}{\sqrt{n}}\sum_{j}e^{i\frac{2\pi}{n}k j}|j\y,
\end{equation}
\end{lemma}
\lem{mostsymmetric} can be easily verified. To prove it, we use the standard way: eigenvalues are solved from $\mathrm{det}(E_i I - A) = 0$, and then a orthonormal set of eigenstates can be solve from $(E_i I -A)\ket{E_i} = 0$.
Now, we turn to the proof of \lem{TenDecEig}:
\begin{proof}[Proof of \lem{TenDecEig}]
By symmetry, we first introduce the new basis,
\begin{align}
    |j,S\y &= \frac{1}{\sqrt{2}}(|j,+\y+ | j,-\y),\nonumber\\
    |j,A\y &= \frac{1}{\sqrt{2}}(|j,+\y- | j,-\y).
\end{align}
Here, $S$ and $A$ refer to symmetric and anti-symmetric, respectively. In the basis
$\{|1,S\y,\ldots,|d,S\y,$ $
|1,A\y,\ldots,|d,A\y\}$,
$H_{|\F}$ should be changed to
\begin{equation}
    \left(
    \begin{array}{cccc|cccc}
        \mu & 2w  & \cdots & 2w &  &   &  &  \\
         2w & \mu & \ddots & \vdots &  &  &  &\\
         \vdots  &\ddots  & \ddots& 2w  &  &  & & \\
         2w  & \cdots& 2w & \mu  &  & &  &  \\
         \hline
          &   &  &  & \mu &   &  & \\
          &  & &  &  & \mu &  & \\
          & & &  &  &  & \ddots& \\
          & &  &  &  & &  & \mu \\
    \end{array}
    \right),
\end{equation}
where the upper left block is $d\times d$ and is of the case described by \lem{mostsymmetric}. So, we know that $d$ eigenstates of $H_{|\F}$ is given by
\begin{equation}
    |E_k\y = \frac{1}{\sqrt{d}}\sum_j e^{i\frac{2\pi}{d}kj}|j,S\y,~k=1,\ldots,d
\end{equation}
And the corresponding eigenvalues are
\begin{equation}
    E_1 = \cdots = E_{d-1} = \mu - 2w,~E_d = \mu + 2w(n-1).
\end{equation}
Clearly, $|j,A\y,~(j=1,\ldots,d)$ form the left $d$ eigenstates with the eigenvalue $\mu$.
\end{proof}

\subsubsection{Proof of \texorpdfstring{\lem{TenDecmix}}{Lemma 8}}\label{append:prooflem8}
\lem{TenDecmix} estimates the mixing time of QTW on the landscape \eq{TenDec}.
\begin{proof}
First recall that
\begin{align}
    \rho_{\rm QTW} - \mu_{\rm QTW} &= \sum_{E_k \neq E_{k^{\prime}}} \frac{1-e^{-i(E_k - E_{k^{\prime}})\tau}}{i(E_k - E_{k^{\prime}})\tau/}
   \z x|E_k\y \z E_k |\alpha\y \z \alpha |E_{k^{\prime}}\y \z E_{k^{\prime}}| x\y \nonumber\\
   &= \sum_{\beta} \sum_{E_k \neq E_{k^{\prime}}} \frac{1-e^{-i(E_k - E_{k^{\prime}})\tau}}{i(E_k - E_{k^{\prime}})\tau}
   \z\beta |E_k\y \z E_k |\alpha\y \z \alpha |E_{k^{\prime}}\y \z E_{k^{\prime}}| \beta \y |\z x | \beta \y|^2 + \frac{1}{w \tau}O(h^{\infty})e^{-\frac{S_0}{h}} \nonumber\\
   &= \sum_{\beta} \sum_{k^{\prime}\in \{1,\ldots,d-1\}}
   \frac{1-e^{-i(2dw)\tau}}{i2dw\tau}
   \z\beta |E_d\y \z E_d |\alpha\y \z \alpha |E_{k^{\prime}}\y \z E_{k^{\prime}}| \beta \y |\z x | \beta \y|^2 + \mathrm{c.c.}\nonumber\\
   &\quad+ \sum_{\beta} \sum_{\substack{k\in \{d+1,\ldots,2d\}\\k^{\prime}\in \{1,\ldots,d-1\}}}
   \frac{1-e^{-i(2w)\tau}}{i2w\tau}
   \z\beta |E_k\y \z E_k |\alpha\y \z \alpha |E_{k^{\prime}}\y \z E_{k^{\prime}}| \beta \y |\z x | \beta \y|^2 + \mathrm{c.c.}
   \nonumber\\
   &\quad+ \sum_{\beta} \sum_{k^{\prime}\in \{d+1,\ldots,2d\}}
   \frac{1-e^{-i(2(d-1)w)\tau}}{i2(d-1)w\tau}
   \z\beta |E_d\y \z E_d |\alpha\y \z \alpha |E_{k^{\prime}}\y \z E_{k^{\prime}}| \beta \y |\z x | \beta \y|^2 + \mathrm{c.c.}
   \nonumber\\
   &\quad+\frac{1}{w \tau}O(h^{\infty})e^{-\frac{S_0}{h}}.
   \label{eq:TenDecrpm}
\end{align}
Here, the second equality is true because
\begin{equation}
    \z x | \mu \y \z \nu |z \y \propto e^{-\frac{d_{x,\pi_{\mu}a_{j(\mu)}}+ d_{x,\pi_{\nu}a_{j(\nu)}}}{h}} \leq e^{-\frac{S_0}{h}},
\end{equation}
and the third equality uses the fact
\begin{equation}
    f(E_k, E_{k^{\prime}}, \alpha,\beta) = f(E_{k^{\prime}}, E_k,\alpha,\beta)^{*},
\end{equation}
where
\begin{equation}
    f(E_k, E_{k^{\prime}},\alpha,\beta) \equiv \frac{1-e^{-i(E_k - E_{k^{\prime}})\tau}}{i(E_k - E_{k^{\prime}})\tau}
   \z\beta |E_k\y \z E_k |\alpha\y \z \alpha |E_{k^{\prime}}\y \z E_{k^{\prime}}| \beta \y |\z x | \beta \y|^2.
\end{equation}
and c.c. denote ``complex conjugate." Now, we calculate
different terms separately using \lem{TenDecEig}:
\begin{align}
    \sum_{k^{\prime}\in \{1,\ldots,d-1\}} f(E_d,E_k^{\prime},\alpha,\beta)&=
   \sum_{k^{\prime}\in \{1,\ldots,d-1\}}\frac{1-e^{-i(2dw)\tau}}{i2dw\tau}
    \frac{1}{(2d)^2}e^{i\frac{2\pi}{d}k^{\prime}(j(\alpha)-j(\beta))} |\z x | \beta \y|^2 \nonumber\\
   &= \left\{\begin{array}{c}
         \frac{1-e^{-i2dw\tau}}{i2dw\tau}\frac{d-1}{(2d)^2}|\z x | \beta \y|^2,~j(\alpha) = j(\beta), \\
         \frac{e^{-i2dw\tau}-1}{i2dw\tau}\frac{1}{(2d)^2}|\z x | \beta \y|^2,~j(\alpha) \neq j(\beta),
    \end{array} \right.
\end{align}
\begin{align}
   \sum_{\substack{k\in \{d+1,\ldots,2d\}\\k^{\prime}\in \{1,\ldots,d-1\}}}  f(E_k,E_k^{\prime},\alpha,\beta)&=
   \sum_{\substack{k\in \{d+1,\ldots,2d\}\\k^{\prime}\in \{1,\ldots,d-1\}}}
   \frac{1-e^{-i2w\tau}}{i2w\tau}
    \frac{\pi_{\alpha}\pi_{\beta}\delta_{j(\alpha),j(\beta)}\delta_{j(\alpha),k-d}}{4d}e^{i\frac{2\pi}{d}k^{\prime}(j(\alpha)-j(\beta))} |\z x | \beta \y|^2 \nonumber\\
   &= \delta_{j(\alpha),j(\beta)}\frac{1-e^{-i2w\tau}}{i2w\tau}\pi_{\alpha}\pi_{\beta}
    \frac{d-1}{4d} |\z x | \beta \y|^2,
\end{align}
\begin{align}
    \sum_{k^{\prime}\in \{d+1,\ldots,2d\}} f(E_d,E_k^{\prime},\alpha,\beta)&=
   \sum_{k^{\prime}\in \{1,\ldots,d-1\}}\frac{1-e^{-i2(d-1)w\tau}}{i2(d-1)w\tau}
    \frac{1}{4d}\delta_{j(\alpha),j(\beta)}\delta_{j(\alpha),k^{\prime}-d}\pi_{\alpha}\pi_{\beta} |\z x | \beta \y|^2 \nonumber\\
   &= \frac{1-e^{-i2(d-1)w\tau}}{i2(d-1)w\tau}
    \frac{1}{4d}\delta_{j(\alpha),j(\beta)}\pi_{\alpha}\pi_{\beta} |\z x | \beta \y|^2.
\end{align}
Substituting the above results in \eq{TenDecrpm} gives
\begin{align}
    \rho_{\rm QTW} - \mu_{\rm QTW} &=
    \sum_{j(\beta) \neq j(\alpha)} -\frac{\sin(2dw\tau)}{dw\tau}\frac{1}{(2d)^2}
    |\z x | \beta \y|^2 +\sum_{j(\beta) = j(\alpha)}
    \frac{\sin(2dw\tau)}{dw\tau}\frac{d-1}{(2d)^2}
    |\z x | \beta \y|^2
    \nonumber\\
   &\quad+ \frac{\sin(2w\tau )}{w\tau }\frac{d-1}{4d}(|\z x|j(\alpha),\pi_{\alpha} \y|^2-|\z x|j(\alpha),-\pi_{\alpha} \y|^2)
   \nonumber\\
   &\quad+ \frac{\sin(2(d-1)w\tau )}{(d-1)w\tau }\frac{1}{4d}(|\z x|j(\alpha),\pi_{\alpha} \y|^2-|\z x|j(\alpha),-\pi_{\alpha} \y|^2)
   +\frac{1}{w \tau}O(h^{\infty})e^{-\frac{S_0}{h}}.
   \label{eq:TenDecrmmr}
\end{align}
The 1-norm here refers in particular to the integral
on $\mathbb{S}^{d-1}$, i.e., $\|\cdot\|_1 = \int_{\mathbb{S}^{d-1}}|\cdot|\d x$. We use $\d x$
to denote the volume element of $\mathbb{S}^{d-1}$ induced from $\mathbb{R}^d$.
As $\sin(\cdot)\leq 1$, it is easily estimated from \eq{TenDecrmmr} that
\begin{align}
    |\rho_{\rm QTW} - \mu_{\rm QTW}|
   &\leq \sum_{j(\beta) \neq j(\alpha)} \frac{1}{d|w|\tau}\frac{1}{(2d)^2}
    |\z x | \beta \y|^2 +\sum_{j(\beta) = j(\alpha)}
    \frac{1}{d|w|\tau}\frac{d-1}{(2d)^2}
    |\z x | \beta \y|^2
    \nonumber\\
   &\quad+ \frac{1}{|w|\tau}\frac{d-1}{4d}(|\z x|j(\alpha),\pi_{\alpha} \y|^2+|\z x|j(\alpha),-\pi_{\alpha} \y|^2)
   \nonumber\\
   &\quad+ \frac{1}{(d-1)|w|\tau}\frac{1}{4d}(|\z x|j(\alpha),\pi_{\alpha} \y|^2+|\z x|j(\alpha),-\pi_{\alpha} \y|^2)
   +\frac{1}{|w| \tau}O(h^{\infty})e^{-\frac{S_0}{h}},
\end{align}
Note that for any $\beta$, $\int_{\mathbb{S}^{d-1}}|\z x|\beta\y|^2\d x=1$, so we can obtain
\begin{align}
    \|\rho_{\rm QTW} - \mu_{\rm QTW}\|_1 & \leq \frac{1}{|w| \tau}\Big(\frac{d-1}{d^3} + \frac{d-1}{2d} + \frac{1}{2d(d-1)}+ O(h^{\infty})\Big),
\end{align}
which gives \eq{TenDecmix}.
\end{proof}

\subsubsection{Proof of \texorpdfstring{\lem{TenDecprobability}}{Lemma 9}}\label{append:prooflem9}

\begin{proof}
\begin{align}
    p(\infty,\beta|\alpha) &= \sum_{E_k = E_{k^{\prime}}}
    \z \beta | E_k \y\z E_k |\alpha \y\z \alpha | E_{k^{\prime}}
    \y\z E_{k^{\prime}} | \beta\y \nonumber\\
    &= \sum_{k,k^{\prime}\in \{1,\ldots,d-1\}} +
    \sum_{k = k^{\prime} = d} +
    \sum_{k,k^{\prime}\in \{d+1,\ldots,2d\}}
    \z \beta | E_k \y\z E_k |\alpha \y\z \alpha | E_{k^{\prime}}
    \y\z E_{k^{\prime}} | \beta\y
\end{align}
Apply the results from \lem{TenDecEig}, we have
\begin{align}
    \sum_{k,k^{\prime}\in \{1,\ldots,d-1\}}\z \beta | E_k \y\z E_k |\alpha \y\z \alpha | E_{k^{\prime}}
    \y\z E_{k^{\prime}} | \beta\y &=
    \sum_{k,k^{\prime} \in \{1,\ldots,d-1\}} \frac{1}{(2d)^2}
    e^{i\frac{2\pi}{d}[(j(\alpha) - j(\beta)) (k^{\prime} - k ) ]}
     \nonumber\\
    &= \left\{\begin{array}{c}
         \frac{(d-1)^2}{(2d)^2},~j(\alpha) = j(\beta), \\
         \frac{1}{(2d)^2},~j(\alpha) \neq j(\beta),
    \end{array} \right.
\end{align}
\begin{equation}
    \z \beta | E_d \y\z E_d |\alpha \y\z \alpha | E_d
    \y\z E_d| \beta\y = \frac{1}{(2d)^2},
\end{equation}
\begin{align}
    \sum_{k,k^{\prime}\in \{d+1,\ldots,2d\}}\z \beta | E_k \y\z E_k |\alpha \y\z \alpha | E_{k^{\prime}}
    \y\z E_{k^{\prime}} | \beta\y
    &=
    \sum_{k,k^{\prime}\in \{d+1,\ldots,2d\}} \frac{1}{4}
    \delta_{k-d,j(\alpha)}\delta_{k-d,j(\beta)}
    \delta_{k^{\prime}-d,j(\alpha)}\delta_{k^{\prime}-d,j(\beta)}
     \nonumber\\
     &=
    \sum_{k,k^{\prime}\in \{d+1,\ldots,2d\}} \frac{1}{4}
    \delta_{k,k^{\prime}}\delta_{j(\alpha),j(\beta)}
    \delta_{k-d,j(\alpha)}
     \nonumber\\
    &= \left\{\begin{array}{c}
         \frac{1}{4},~j(\alpha) = j(\beta), \\
         0,~j(\alpha) \neq j(\beta),
    \end{array} \right.
\end{align}
which complete the proof.
\end{proof}

\subsubsection{Proof of \texorpdfstring{\lem{TenDecw}}{Lemma 10}}\label{append:prooflem10}
The proof of \lem{TenDecw} makes use of \prop{speW}.
So, we should first study the Agmon geodesics and Agmon distances between minima.
The following Lemma and Corollary give the explicit value of the shortest Agmon distance $S_0 = \min_{\alpha\neq \beta} d(\alpha,\beta)$ whose proof indicates there is at most one geodesic with the Agmon length $S_0$ linking two different minima.
\begin{lemma}\label{lem:TenDecagmondis}
Constrained on $\mathbb{S}^{d-1}$, for two wells $\alpha \neq \beta$,
\begin{equation}
    d( \pi_{\alpha}a_{j(\alpha)}, \pi_{\beta}a_{j(\beta)})
    =\left\{\begin{array}{c}
         \sqrt{2},~j(\alpha) = j(\beta), \\
         \sqrt{2}/2,~j(\alpha) \neq j(\beta),
    \end{array} \right.
\end{equation}
where, $d(\cdot ,\cdot)$ refers to the Agmon distance (see \defn{Agmon-distance} or \append{agmon} for more details).
\end{lemma}

\begin{corollary}
$S_0 = \sqrt{2}/2$.
\end{corollary}

\begin{proof}[Proof of \lem{TenDecagmondis}]
First, we prove that for $j(\alpha) \neq j(\beta)$, there exists an Agmon geodesic linking local minima $\pi_{\alpha}a_{j(\alpha)}$ and $\pi_{\beta}a_{j(\beta)}$ which is in the plane spanned by
$a_{j(\alpha)}$ and $a_{j(\beta)}$. Also, this geodesic will also pass through $-\pi_{\alpha} a_{j(\alpha)}$, joining $\pm \pi_{\alpha} a_{j(\alpha)}$.
Because $\{a_j \}$ symmetrically distribute on $\mathbb{S}^{d-1}$ and are orthonormal, the proof can be completed if we prove there exists such a geodesic between $a_1$ and $a_2$.
We work in the coordinate $\psi^{\mu} = \{\psi^1,\ldots,\psi^{d-1}\}$ which can represent $\{x^1,\ldots,x^d\}$, the coordinate in $\mathbb{R}^d$, as follows:
\begin{align}
    x^1 &= \cos\psi^1 \nonumber\\
    x^2 &= \sin\psi^1 \cos\psi^2 \nonumber\\
    x^3 &= \sin\psi^1 \sin\psi^2 \cos\psi^3 \nonumber\\
    &\cdots \nonumber\\
    x^{d-1} &= \sin\psi^1 \sin\psi^2\cdots \sin\psi^{d-2}\cos\psi^{d-1} \nonumber\\
    x^d &= \sin\psi^1 \sin\psi^2\cdots \sin\psi^{d-2}\sin\psi^{d-1},
    \label{eq:coordinate}
\end{align}
where, $\psi^{\mu} \in [0,\pi]\times [0,\pi] \times \cdots \times [0,\pi] \times [0,2\pi]$.
The objective function $f(x) = 1 - \sum_{i}(x^i)^4$ can be rewritten as
\begin{equation}
    f(\psi^{\mu}) = 1 - \cos^4\psi^1 - \sin^4\psi^1 \cos^4\psi^2 -\cdots - \sin^4\psi^1 \sin^4\psi^2\cdots \sin^4\psi^{d-2}\sin^4\psi^{d-1}.
\end{equation}
The metric of $\mathbb{S}^{d-1}$, which is induced by the Euclidean metric of $\mathbb{R}^d$, is given by
\begin{equation}
    d x^2 = (\d \psi^1)^2 + \sin^2\psi^1 (\d \psi^2)^2 +
    \sin^2\psi^1\sin^2\psi^2 (\d \psi^3)^2 + \cdots + \prod_{\nu=1}^{d-2}\sin^2\psi^{\nu} (\d \psi^{d-1})^2,
\end{equation}
or equivalently,
\begin{equation}
    g_{\mu \nu} = \left(
    \begin{array}{ccccc}
       1 & & & & \\
        &\sin^2\psi^1 & & & \\
        & & \sin^2\psi^1\sin^2\psi^2& & \\
        & & &\ddots & \\
        & & & & \prod_{k=1}^{d-2}\sin^2\psi^{k}\\
    \end{array}
    \right)
\end{equation}
with $d x^2 = g_{\mu \nu} \d\psi^{\mu} \d \psi^{\nu}$.
Through out this proof, we use Einstein's summation convention.
The Agmon metric is given by $\d s^2 = V \d x^2$, indicating that
\begin{equation}
    d s^2 = G_{\mu \nu} \d\psi^{\mu} \d \psi^{\nu},\quad G_{\mu \nu}
    =f g_{\mu \nu}.
\end{equation}
A geodesic is also a curve which can be described by a real parameter (i.e., $\psi^{\mu} =\psi^{\mu}(\tau)$, $\tau \in \mathbb{R}$).
Geodesics associated to the Agmon metric should satisfy geodesic equations:
\begin{equation}
    \frac{\d^2 \psi^{\mu}}{\d \tau^2} + \Gamma^{\mu}_{\rho \sigma} \frac{\d \psi^{\rho}}{\d \tau}\frac{\d \psi^{\sigma}}{\d \tau}=0,
    \label{eq:geodesic}
\end{equation}
where $\Gamma^{\mu}_{\rho \sigma}$ are the Christoffel symbols associated to the Agmon metric.
For $G_{\mu \nu}$, which is a diagonal metric, the Christoffel symbols are given by
\begin{align}
    \Gamma^{\lambda}_{\mu \nu} &= 0 \nonumber\\
    \Gamma^{\lambda}_{\mu \mu} &= -\frac{1}{2}(G_{\lambda \lambda})^{-1}\partial_{\lambda}G_{\mu \mu} \nonumber\\
    \Gamma^{\lambda}_{\lambda \mu} &= \partial_{\mu} \big(\ln \sqrt{|G_{\lambda \lambda}|}\big)\nonumber\\
    \Gamma^{\lambda}_{\lambda \lambda} &= \partial_{\lambda} \big(\ln \sqrt{|G_{\lambda \lambda}|}\big),
\end{align}
where, in these expressions, $\mu \neq \nu \neq \lambda$, and repeated indices are not summed over.
Recall that $\Gamma^{\lambda}_{\mu \nu} = \Gamma^{\lambda}_{\nu \mu}$, all the Christoffel symbols have been given by the upper four identities.
In the plane spanned by $a_1$ and $a_2$, we can have $\psi^k = 0~(k\geq 2)$, which satisfies \eq{geodesic} for $\mu \geq 2$ naturally. And for $\mu = 1$, \eq{geodesic} reduces to
\begin{equation}
    \frac{\d^2 \psi^{1}}{\d \tau^2} + \frac{2\sin4\psi^1}{1-\cos4\psi^1}\left(\frac{\d\psi^1 }{\d \tau}\right)^2 = 0.
\end{equation}
Let $y = \frac{\d \psi^1}{\d \tau}$, we have
\begin{align}
    \frac{\d y}{\d \tau} = -\frac{2\sin4\psi^1}{1-\cos4\psi^1} y^2 &\Rightarrow \d\left(\frac{1}{y}\right) = \frac{2\sin4\psi^1}{1-\cos4\psi^1}\d \tau  \nonumber\\
    &\Rightarrow \frac{\d}{\d \psi^1}\left(\frac{\d \tau}{\d \psi^1}\right) = \frac{2\sin4\psi^1}{1-\cos4\psi^1}\frac{\d \tau}{\d \psi^1} \nonumber\\
    &\Rightarrow \left(\frac{\d \tau}{\d \psi^1}\right)\Big/\sqrt{1-\cos4\psi^1} = \rm const.
\end{align}
Given initial conditions, this ordinary differential equation clearly has a solution. Here, we set $\tau$ to be the Agmon length of the curve and $\tau = 0$ when $\psi^1 = 0$. Then, the solution is
given by
\begin{equation}
    \tau = \int_{0}^{\psi^1}\frac{1}{2}\sqrt{1-\cos4\psi} \d \psi \equiv \gamma(\psi^1).
\end{equation}
So, we prove the existence of a geodesic which can be described by
\begin{equation}
    \psi^1 = \gamma^{-1}(\tau),~\psi^k = 0~(k\geq 2).
    \label{eq:geodesicRes}
\end{equation}
This geodesic associated to the Agmon metric
is the intersection of $\mathbb{S}^{d-1}$ and the plane spanned by $a_1$ and $a_2$, which is an one-dimensional circle, namely, a geodesic associated to the metric $g_{\mu \nu}$ (shown in \fig{geodesic}). We can easily see that the shorter Agmon geodesic linking $a_1$ and $a_2$ in the plane spanned by $a_1$ and $a_2$ has the Agmon length $\gamma(\pi/2)$. And, there are two paths with identical Agmon length $\gamma(\pi)$ joining $a_1$ and $-a_1$ which are Agmon geodesics.
The integrals of the Agmon lengths are given by \begin{align}
    \gamma(\pi/2) &= \int_{0}^{\pi/2}\frac{1}{2}\sqrt{1-\cos4\psi} \d \psi \nonumber\\
    & = 2\int_{0}^{\pi/4}\frac{1}{2}\sqrt{1-\cos4\psi} \d \psi \nonumber\\
    &=2\int_{0}^{2}\frac{1}{8}\frac{1}{\sqrt{2-t}} \d t \quad (\mathrm{let}~t = 1-\cos4\psi)\nonumber\\
    &=- \frac{1}{2} \sqrt{2-t}|_0^2 = \frac{\sqrt{2}}{2},
\end{align}
\begin{equation}
    \gamma(\pi) = 2\gamma(\pi/2) = \sqrt{2}.
\end{equation}
\begin{figure}
  \centerline{
  \includegraphics[width=0.8\textwidth]{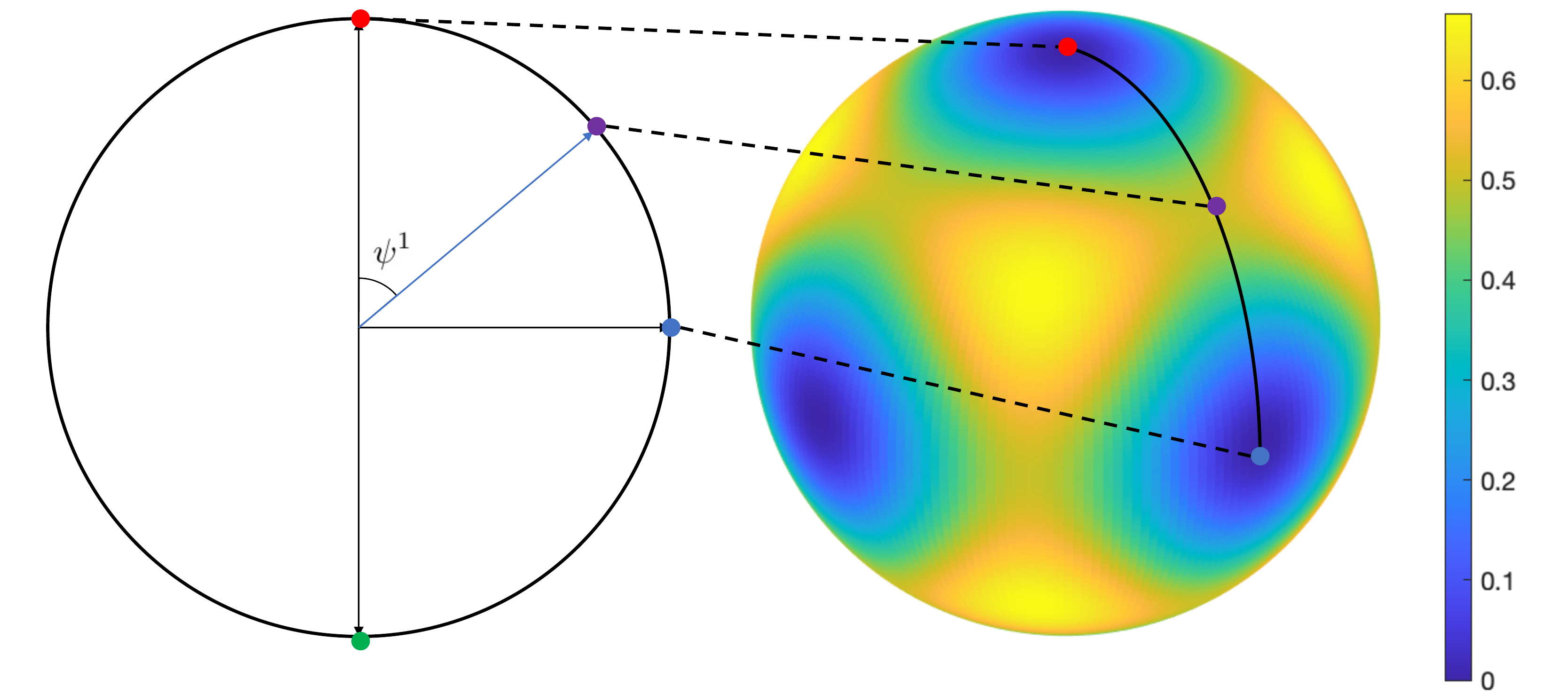}}
    \caption{Sketch of the geodesic associated to the Agmon metric in the plane spanned by $a_1$ and $a_2$: red point refers to $a_1$, blue point $a_2$, green point $-a_1$, and the purple point the separating saddle point (see \defn{SSP}) between $a_1$ and $a_2$, respectively.
    }
\label{fig:geodesic}
\end{figure}

Recall that $d(\cdot,\cdot)$ is the Agmon distance between two points or sets. From now on, we will prove that
for $j(\alpha) \neq j(\beta)$, $d(\pi_{\alpha}a_{j(\alpha)}, \pi_{\beta}a_{j(\beta)}) = \sqrt{2}/2$, and for any $\alpha$,
 $d(\pi_{\alpha}a_{j(\alpha)}, -\pi_{\alpha}a_{j(\alpha)}) = \sqrt{2}$. Because of the landscape symmetry, we only need to
 prove that $d(a_1,a_2) = \sqrt{2}/2$ and $d(a_1,-a_1) = \sqrt{2}$, which is equivalent to say that the Agmon geodesic
 described by \eq{geodesicRes} with $0\leq \tau \leq \sqrt{2}/2$ ($0\leq \tau \leq \sqrt{2}$) is the shortest joining $a_1$ and $a_2$ ($a_1$ and $-a_1$).

First note that for $\sum_{i} (x^i)^2 =1$ or $\sum_{i=2}^d (x^i)^2 =1 - (x^1)^2$, equivalently,
\begin{equation}
    \sum_{i} (x^i)^4 = (x^1)^4 + \sum_{i=2}^d (x^i)^4 \leq
    (x^1)^4 + (1-(x^1)^2)^2.
\end{equation}
Therefore,
\begin{equation}
    f(\psi^{\mu}) = f(\psi^{1}, \psi^{2},\cdots, \psi^{d})
    = 1 - \sum_{i} (x^i)^4 \geq 1- (x^1)^4 - (1-(x^1)^2)^2
    = f(\psi^{1}, 0,\cdots, 0),
\end{equation}
where the coordinates $\psi^{\mu}$ and $x^{\nu}$ obey the transformation relation \eq{coordinate}.
For any curve whose parameter is $\tau$, we have
\begin{equation}
    \d x = \sqrt{g_{\mu \nu}\frac{\partial \psi^{\mu}}{\partial \tau}
    \frac{\partial \psi^{\nu}}{\partial \tau}} \d \tau.
\end{equation}
And if the parameter $\tau$ can be set to be $\psi^1$,
\begin{align}
    \d x = \sqrt{g_{\mu \nu}\frac{\partial \psi^{\mu}}{\partial \psi^1}
    \frac{\partial \psi^{\nu}}{\partial \psi^1}} \d \psi^1
     = \sqrt{\sum_i g_{i i}\left(\frac{\partial \psi^{i}}{\partial \psi^1}\right)^2
    } \d \psi^1
    \geq \sqrt{g_{1 1}\left(\frac{\partial \psi^{1}}{\partial \psi^1}\right)^2
    } \d \psi^1 = \d \psi^1.
\end{align}
For any curve $C$ with finite length that can be described by $\psi^1$ (i.e., $\psi^{\mu} = \psi^{\mu}(\psi^1)$ and $\psi^1 \in [a,b]$),
we can conclude that
\begin{equation}
    \int_C \sqrt{f} \d x = \int_{a}^{b}  \sqrt{f g_{\mu \nu}\frac{\partial \psi^{\mu}}{\partial \psi^1}
    \frac{\partial \psi^{\nu}}{\partial \psi^1}} \d \psi^1
    \geq \int_{a}^{b}  \sqrt{f(\psi^1,0,\cdots,0)} \d \psi^1
    = \gamma(b) - \gamma(a).
    \label{eq:agmoninequality}
\end{equation}

Let $\psi^i = l^i(\tau)~(0\leq \tau \leq M_l)$ be an arbitrary smooth path joining $a_1$ and $a_2$ ($\tau = 0$ and $\tau = M_l$ correspond to $a_1$ and $a_2$, respectively). If this path can be described by $\psi^1$ (i.e., $\psi^1 = l^1 (\tau)$ is a bijection and $\psi^i = l^i((l^1)^{-1}(\psi^1))~(0\leq \psi^1 \leq \pi/2)$), the Agmon length between $a_1$ and $a_2$ along this path is given by
\begin{equation}
    l(a_1,a_2) \equiv \int_{0}^{M_l}  \sqrt{f g_{\mu \nu}\frac{\partial \psi^{\mu}}{\partial \tau}
    \frac{\partial \psi^{\nu}}{\partial \tau}} \d \tau
    = \int_{0}^{\pi/2}  \sqrt{f g_{\mu \nu}\frac{\partial \psi^{\mu}}{\partial \psi^1}
    \frac{\partial \psi^{\nu}}{\partial \psi^1}} \d \psi^1 \geq \gamma(\pi/2).
\end{equation}
However, if $\psi^1 = l^1 (\tau)$ is not a bijection, or there is no one-to-one correspondence between $\psi^1$ and the points in the path, the second equality in the above equation will be invalid.
In this case, we should restrict the domain on $\mathcal{D} = \bigcup_i [\tau_i^-,\tau_i^+] \subset [0,M_l]$ ($\tau_i^+ \leq \tau_{i+1}^+$) such that $l^1: \mathcal{D} \to [0,\pi/2]$ is a bijection and $[0,M_l]-\mathcal{D}$ is a zero measure set, which can always be done as the path is smooth. Then, we have
\begin{align}
    l(a_1,a_2) &\geq \sum_i\int_{\tau_i^-}^{\tau_i^+}  \sqrt{f g_{\mu \nu}\frac{\partial \psi^{\mu}}{\partial \tau}
    \frac{\partial \psi^{\nu}}{\partial \tau}} \d \tau \nonumber\\
    &= \sum_i\int_{l^1(\tau_i^-)}^{l^1(\tau_i^+)}  \sqrt{f g_{\mu \nu}\frac{\partial \psi^{\mu}}{\partial \psi^1}
    \frac{\partial \psi^{\nu}}{\partial \psi^1}} \d \psi^1 \nonumber\\
    &\geq \sum_i [\gamma(l^1(\tau_i^+)) -\gamma(l^1(\tau_i^-))] =\gamma(\pi/2).
\end{align}

To sum up, for paths joining $a_1$ and $a_2$, $\gamma(\pi/2)$ is the shortest Agmon length, and thus, the Agmon distance between $a_1$ and $a_2$. For the same reason \eq{agmoninequality}, $\gamma(\pi)$ can be proved to be the Agmon distance between $a_1$ and $-a_1$ (more specifically, the proof is the above paragraph with $a_2$ and $\pi/2$ replaced by $-a_1$ and $\pi$, respectively), which completes the proof of \lem{TenDecagmondis}.
\end{proof}

Now, we are able to prove \lem{TenDecw} as follows.
\begin{proof}[Proof of \lem{TenDecw}]
Because of the symmetry, we can calculate $w$ focusing on the interaction between $a_1$ and $a_2$.
The point on the geodesic (see \fig{geodesic}) separating
the two wells containing $a_1$ and $a_2$ should be $x^{\circ}_{12}:=(a_1+a_2)/\sqrt{2}$ which is also a saddle point.
Restate \prop{speW} in the present case,
we have
\begin{align}
w = -\sqrt{h} (b + O(h))e^{-\frac{S_0}{h}},
\end{align}
where
\begin{align}
b = 2(2\pi)^{\frac{d-2}{2}}\sqrt{\frac{1/2}{\mathrm{det}' \nabla^2 ( d(x^{\circ}_{12},a_1)+d(x^{\circ}_{12},a_2))}}a_0^{(1)}(x^{\circ}_{12})a_0^{(2)}(x^{\circ}_{12}),
\end{align}
and $a_0^{(1)}(x^{\circ}_{12}) = a_0^{(2)}(x^{\circ}_{12})$ (due to the symmetry) are factions in the WKB estimation.

We first evaluate the function $a^{(1)}_0(x)$.
The initial condition \eq{178} should be
\begin{align}
    a^{(1)}_0(a_1) = \left(\frac{\det \sqrt{\nabla^2 f(a_1)}}{(\sqrt{2}\pi)^{d-1}}\right)^{\frac{1}{4}}.
\end{align}
The transport equation \eq{transport} in this case should be
\begin{align}
    2\nabla d(x, a_1) \cdot \nabla a^{(1)}_0(x) + (\Delta d(x,a_1) - \mu)a^{(1)}_0(x) = 0.
\end{align}
Note that all the differential operator $\nabla$ here adapts to the manifold $\mathbb{S}^{d-1}$.
Because $\mathbb{S}^{d-1}$ can be embedded in $\mathbb{R}^d$, we may use the differential operator $\nabla_{\mathbb{R}^d}$ to calculate the Hessian matrix $\nabla^2 f(a_1)$:
\begin{align}
    \nabla^2 f(a_1) = \mathbf{P}_{a_1^{\perp}} [\nabla^2_{\mathbb{R}^d} f(a_1) - (\nabla f(a_1))\trans a_1 \mathbf{I} ]\mathbf{P}_{a_1^{\perp}},
\end{align}
where $\mathbf{I}$ is the identity matrix in $\mathbb{R}^d$, $\mathbf{P}_{a_1^{\perp}} = \mathbf{I} - a_1 a_1\trans $ is the orthogonal projector onto the tangent space $T_{a_1}\mathbb{S}^{d-1}$,
$\nabla^2_{\mathbb{R}^d} f(a_1)$ is the Hessian matrix in $\mathbb{R}^d$, and $(\nabla f(a_1))\trans a_1 \mathbf{I}$ reflects the curvature of the sphere.
Under the coordinate $\{a_2,a_3,\ldots,a_d \}$ which can also be seen
as a local orthonormal coordinate of $T_{a_1}\mathbb{S}^{d-1}$,
we can explicitly write $\nabla^2 f(a_1) = \mathrm{diag}(4,\ldots4)$.
So, we have $\det \sqrt{\nabla^2 f(a_1)} = 2^{d-1}$.
To get $x^{\circ}_{12}$, we should use the transport equation.
For any point $x$ on the geodesic linking $a_1$ and $a_2$,
in the tangent space $T_{x}\mathbb{S}^{d-1}$,
we set $\frac{\partial}{\partial x^1}$ to be the tangent vector of the geodesic.
Then, because of the symmetry,
\begin{align}
    \nabla d(x, a_1) \cdot \nabla a^{(1)}_0(x) = \frac{\partial d(x, a_1) }{\partial x^1} \frac{\partial a^{(1)}_0(x) }{\partial x^1}.
\end{align}
Similarly, $\Delta d(x,a_1)$ should be of the form
\begin{align}
    \Delta d(x,a_1) = f_1(x) + (d-2)f_2(x),
\end{align}
for some functions $f_1$ and $f_2$.
For sufficiently small $h$, the energy of the ground state $\mu$ should be
\begin{align}
  \mu  =  h\tr \sqrt{\nabla^2 f(a_1)}\propto (d-1).
\end{align}
At last, the transport equation restricted on the geodesic linking $a_1$ and $a_2$ can be changed to
\begin{align}
    \frac{\partial a^{(1)}_0(x(\tau)) }{\partial \tau} = (d-1)f_3(\tau) a^{(1)}_0(x(\tau)),
\end{align}
for some function $f_3(\tau)$. Here, $\tau$ is the parameter of the geodesic.
In this way, we know
\begin{align}
   \ln \frac{a^{(1)}_0(x^{\circ}_{12})}{a^{(1)}_0(a_1)} \propto d-1,
\end{align}
and the coefficient is independent of $d$.
As a result, we know that $\ln a^{(1)}_0(x^{\circ}_{12}) \propto d-1$
and the coefficient is independent of $d$.

Note that although $d(x,a_1)+d(x,a_2))$
is a complicated function, it is symmetric in the $d-2$-dimension subspace of $T_{x^{\circ}_{12}}\mathbb{S}^{d-1}$ which is orthogonal to the tangent vector of the geodesic linking $a_1$ and $s_2$.
Thus, we claim that, similar to the Hessian matrix of $f$,
$\nabla^2 (d(x,a_1)+d(x,a_2)))$ has $d-2$ entries of the same value which are independent of $d$.
This fact give that $\ln \mathrm{det}' \nabla^2 ( d(x^{\circ}_{12},a_1)+d(x^{\circ}_{12},a_2)) \propto d-2$.

Combining all the results above, we find that $b$ is of the form $C_1C_2^{d-1}$ where $C_1$ and $C_2$ are two constants independent of $d$.

Given that $S_0 = \sqrt{2}/2$, we can have
\begin{align}
    w = -\sqrt{h} (C_1C_2^{d-1} + O(h))e^{-\frac{\sqrt{2}}{2h}}
\end{align}
as expected.
\end{proof}

\subsubsection{Proof of \texorpdfstring{\prop{TenDecTtot}}{Proposition 2}}\label{append:proofprop2}
\prop{TenDecTtot} presents a general upper bound for the total time needed for finding all orthogonal components by QTW.
Constrain on the expected risk, $\mathrm{E}_{x\sim\mu_{\rm QTW}}f(x) = \delta$, control the accuracy of QTW and will give $h$.
\begin{proof}
We first estimate the expected risk, $\mathbb{E}_{x\sim\mu_{\rm QTW}}f(x)$.
For sufficiently small $h$, $\mu_{\rm QTW}$ localizes near minima.
So, only regions near minima can contribute to the integral
\begin{align}
    \mathbb{E}_{x\sim\mu_{\rm QTW}}f(x) = \int_{\mathbb{S}^{d-1}} f(x)\mu_{\rm QTW}(x) \d x.
\end{align}
Namely,
\begin{align}
    \mathbb{E}_{x\sim\mu_{\rm QTW}}f(x) = \sum_{\alpha}\int_{\Omega_{\alpha}} f(x) \mu_{\rm QTW}(x) \d x
    + O(h^{\infty}),
\end{align}
where $\Omega_{\alpha}$ is a small neighborhood of the minima labeled by $\alpha$.
The term $O(h^{\infty})$ appears because $\mu_{\rm QTW}$ is a mixing of $|\ip{x}{\beta}|^2$ which decay exponentially with respect to $d(x,\beta)/h$ (or, apply \cor{L2norm}).
Still, we choose $h$ to be enough small, such that $\Omega_{\alpha}$
can be seen as a subset of $\mathbb{R}^d$.
Then, with the same Gaussian integral in \lem{expectedrisk},
we have
\begin{equation}
    h = \frac{\delta}{\frac{\sqrt{2}}{4}\sum_{\beta} p(\infty,\beta|\alpha)\tr\big(\sqrt{\nabla^2f(\beta)}\big)+ o_{\delta}(1)} = \frac{\delta}{\frac{\sqrt{2}}{4}\tr\big(\sqrt{\nabla^2f(a_1)}\big)+ o_{\delta}(1)}.
\end{equation}
Here, $\nabla$ should still be the differential operator adapted to the sphere manifold and the second equality is valid because $\nabla^2f(\beta)$ should be the same regardless of $\beta$.
As is calculated in the proof of \lem{TenDecw},
the matrix form of $\sqrt{\nabla^2f(a_1)}$ can be $\mathrm{diag}(2,\ldots,2)$. Thus, $\tr \sqrt{\nabla^2f(a_1)} = 2(d-1)$, which gives
\begin{align}
    h = \frac{\delta}{\frac{\sqrt{2}}{2}(d-1)+ o_{\delta}(1)}.
\end{align}
Substituting this result in $w$ and note the relationship between
$T_{\rm tot}$ and $w$ in \eq{Ttot1}, we have
\begin{align}
    T_{\rm tot} = \tilde{O}(d^2)\frac{1}{\epsilon (\sqrt{\delta}C_1C_2^{d-1} + O(\delta))} e^{\frac{d-1 + o_{\delta}(1)}{2\delta}}.
\end{align}
Note that $C_2$ should be less $1$ as the tunneling effect $w$ would not increase exponentially with respect to $d$,
for small $\delta$, we can have
\begin{align}
    T_{\rm tot} = O(\mathrm{poly}(1/\delta, e^d, 1/\epsilon)) e^{\frac{d-1 + o_{\delta}(1)}{2\delta}},
\end{align}
which completes the proof.
\end{proof}

\section{Technical Details for Quantum-Classical Comparisons}\label{append:section4}

\subsection{Details of comparison standards}
\subsubsection{Proof of \texorpdfstring{\lem{stand1-quad}}{Lemma 11}}\label{append:prooflem11}
In \lem{stand1-quad}, we estimate the expected risk yielded by QTW and SGD for a quadratic function.
\begin{proof}
Firstly, as there is only one (local) minimum, $\mu_{\rm QTW}$ is the distribution of the ground state $|E_0\y$,
namely,
\begin{equation}
    \mu_{\rm QTW} = |\z x|E_0\y|^2 = \frac{\big(\det\sqrt{\nabla^2 f(0)}\big)^{1/2}}{(\sqrt{2}\pi h)^{d/2}}
    e^{-\frac{x\trans\sqrt{\nabla^2 f(0)}x}{\sqrt{2}h}}
\end{equation}
Let $\{x_i\}$ be a coordinate where $\nabla^2 f(0)$ is a
diagonal matrix denoted by $\mathrm{diag}(\lambda_1,..,\lambda_d)$.
For simplicity, set $\omega_i = \sqrt{\lambda_i}$,
$\mathbb{E}_{x\sim \mu_{\rm QTW}}f(x)$ is then given by
\begin{align}
    \mathbb{E}_{x\sim \mu_{\rm QTW}}f(x) &=
    \int \frac{(\prod_{i} \omega_i )^{1/2}}{(\sqrt{2}\pi h)^{d/2}}
    e^{-\frac{\sum_i \omega_i x_i^2}{\sqrt{2}h}} \frac{1}{2} \sum_i \lambda_i x_i^2 \d x_1\cdots \d x_d \nonumber\\
    & =\frac{1}{2}\frac{(\prod_{i} \omega_i )^{1/2}}{(\sqrt{2}\pi h)^{d/2}} \sum_j \lambda_j \int e^{-\frac{\sum_i \omega_i x_i^2}{\sqrt{2}h}} x_j^2 \d x_1\cdots \d x_d \nonumber\\
    & = \frac{\sqrt{2}h}{4}\sum_j \omega_j,
\end{align}
where the last equality uses the facts
\begin{equation}
    \int_{-\infty}^{\infty} e^{-ax^2}\d x = \sqrt{\frac{\pi}{a}},
\end{equation}
\begin{equation}
    \int_{-\infty}^{\infty} x^2 e^{-ax^2}\d x = \frac{1}{2}\sqrt{\frac{\pi}{a^3}}.
\end{equation}
Similarly, we have
\begin{equation}
    \int e^{\frac{-2f}{s}} \d x = \int e^{\frac{-\sum_i \lambda_i x_i^2}{s}}\d x_1 \cdots \d x_d = \frac{(\pi s)^{d/2}}{\sqrt{\prod_i \lambda_i }},
\end{equation}
and
\begin{equation}
    \int e^{\frac{-2f}{s}}f(x) \d x = \frac{1}{2}\sum_j \lambda_j \int e^{\frac{-\sum_i \lambda_i x_i^2}{s}}x_j^2 \d x_1 \cdots \d x_d = \frac{sd}{4}\frac{(\pi s)^{d/2}}{\sqrt{\prod_i \lambda_i }},
\end{equation}
which complete the proof.
\end{proof}

\subsubsection{Proof of \texorpdfstring{\lem{expectedrisk}}{Lemma 12}}\label{append:prooflem12}

With the help of \lem{stand1-quad}, \lem{expectedrisk} captures $h$ and $s$ by \stand{risk} for general landscapes satisfying assumptions in \sec{classicalpre} and \sec{quantumpre}. Note that in \sec{standard},
normal lower-case letters, $x$, $y$,\ldots, are used to denote vectors without ambiguity.
However, in the proof, we will use coordinates and components.
Therefore, bold lower-case letters $\bf{x}$, $\bf{y}$,\ldots, are used to denote vectors or points in $\mathbb{R}^d$.

\begin{proof}
For a general landscape $f$ satisfying assumptions in \sec{classicalpre} and \sec{quantumpre},
we can find a
sufficiently large but bounded set $\Omega$ containing all the
minima $\{\mathbf{x}_j:j=1,\ldots,N\}$.
First recall \lem{limitdis} which gives
\begin{align}
     \mu_{\rm QTW}(\mathbf{x})
     = \sum_{j}
    p(\infty, j) |\z \mathbf{x}|j\y|^2 + O(h^{\infty}),
\end{align}
where $p(\infty, j)$ is the probability of finding the system at state $\ket{j}$ for $\tau \to \infty$.
Every local ground state decreases exponentially with respect to $1/h$
out of $\Omega$, giving that
\begin{align}
     \int_{\mathbb{R}^d\backslash \Omega}\mu_{\rm QTW}(\mathbf{x})\d \mathbf{x} = O(h^{\infty}).
\end{align}

In the region $\Omega$, we first study the integral near one minima $\mathbf{x}_j$. Without loss of generality, we can translate the coordinate to $\{ x_k, k=1,\ldots,d\}$ such that $\mathbf{x}_j = 0$. By properly rescaling the coordinate, we can find
a coordinate $\{y_k:k=1,\ldots,d\}$ and let $\mathbf{y} = (y_1,y_2,\ldots,y_d)\trans$ such that $\sqrt{\nabla^2_\mathbf{y} f(0)} = I$ where $I$
is the $d\times d$ identity matrix.
Under the coordinate $\{y_k:k=1,\ldots,d\}$, we can always find a fixed ball $\Omega_j$ only containing one minimum $\mathbf{x}_j=0$ and for any $\mathbf{y} \in \Omega_j$, $\nabla^2_\mathbf{y} f(\mathbf{y})>0$ (because $f$ is smooth).
The local ground state satisfies
\begin{equation}
    |\z \mathbf{y} | j \y|^2 = \frac{1}{(\sqrt{2}\pi h)^{d/2}}(1+ O(\|\mathbf{y}\|))
    e^{-\frac{\mathbf{y}\trans \mathbf{y} + O(\|\mathbf{y}\|^3)}{\sqrt{2}h}},
\end{equation}
by WKB approximation, and the function can be rewritten as
\begin{align}
    f(\mathbf{y}) = \frac{1}{2}\mathbf{y}\trans \sqrt{\nabla^2_\mathbf{x} f(0)} \mathbf{y} + O(\|\mathbf{y}\|^3).
\end{align}
Then, we define a hypercube $\Omega^{(h)}_j$ whose edge length (under the coordinate of $\mathbf{y}$) is $\Theta(h^{5/12})$.
Evaluate the integral:
\begin{align}
    \int_{\Omega^{(h)}_j} f(\mathbf{x})|\z \mathbf{x} | j \y|^2 \d \mathbf{x} &= \int_{\Omega^{(h)}_j} f(\mathbf{y})|\z \mathbf{y} | j \y|^2 \d \mathbf{y} \\
    & = \frac{1}{2}\int_{\Omega^{(h)}_j} \frac{1+ O(h^{5/12})}{(\sqrt{2}\pi h)^{d/2}}
    e^{-\frac{\mathbf{y}\trans \mathbf{y}}{\sqrt{2}h}}(1+ O(h^{1/4}))(\mathbf{y}\trans \sqrt{\nabla^2_\mathbf{x} f(0)} \mathbf{y} + O(h^{5/4})) \d \mathbf{y} \\
    & = \frac{1}{2}\int_{\Omega^{(h)}_j} \frac{1}{(\sqrt{2}\pi h)^{d/2}}
    e^{-\frac{\mathbf{y}\trans\mathbf{y}}{\sqrt{2}h}}(\mathbf{y}\trans \sqrt{\nabla^2_\mathbf{x} f(0)} \mathbf{y}) \d \mathbf{y} + o(h). \\
    & = \frac{1}{2}\int_{\Omega^{(h)}_j}\frac{\big(\det\sqrt{\nabla^2_{\mathbf{x}} f(0)}\big)^{1/2}}{(\sqrt{2}\pi h)^{d/2}}
    e^{-\frac{\mathbf{x}\trans\sqrt{\nabla^2 f(0)}\mathbf{x}}{\sqrt{2}h}}(\mathbf{x}\trans \nabla^2_{\mathbf{x}} f(0) \mathbf{x} ) \d \mathbf{x} + o(h)
    \\
    & = \frac{\sqrt{2}h}{4} \tr\sqrt{\nabla^2_{\mathbf{x}}f(\mathbf{x}_j)} + o(h).
\end{align}
The last equality is valid because $h^{5/12}\gg \sqrt{h}$ for small $h$ and $\sqrt{h}$ is the scale of the standard deviation of the Gaussian distribution.
Calculations of the last integral in can be seen in the proof of \lem{stand1-quad}.
Define $f_j = \max_{\mathbf{x}\in \Omega_j}f(\mathbf{x})$,
we have
\begin{align}
    \int_{\Omega_j \backslash\Omega^{(h)}_j} f(\mathbf{x})|\z \mathbf{x} | j \y|^2 \d \mathbf{x} \leq  \int_{\Omega_j \backslash\Omega^{(h)}_j} f_j |\z \mathbf{x} | j \y|^2 \d \mathbf{x} = O(h^{\infty}).
\end{align}
The last equality relies on two facts: $\Omega_j$ is a bounded region and $\sqrt{h} = o(h^{5/12})$. The first fact ensures that the difference between $|\z \mathbf{x} | j \y|^2$ and the Gaussian distribution can be bounded by multiplying some constant, and the second fact make the integral of the Gaussian distribution be of $O(h^{\infty})$.
We now have
\begin{align}
    \int_{\Omega_j } f(\mathbf{x})|\z \mathbf{x} | j \y|^2 \d \mathbf{x} = \frac{\sqrt{2}h}{4} \tr\sqrt{\nabla^2_{\mathbf{x}}f(\mathbf{x}_j)} + o(h).
\end{align}
In the region $\Omega\backslash \Omega_j$,
$f$ is bounded and by \cor{L2norm}, the integral of distributions of local ground states should be of $O(h^{\infty})$.
Therefore, we claim that $\int_{\Omega\backslash  \Omega_j} \mu_{\rm QTW}f \d \mathbf{x} = O(h^{\infty})$.

Combine all the results above, we can have
\begin{align}
    \mathbb{E}_{\mathbf{x}\sim\mu_{\rm QTW}}f(\mathbf{x}) = \sum_{j=1}^N \frac{\sqrt{2}h}{4} \tr\sqrt{\nabla^2_{\mathbf{x}}f(\mathbf{x}_j)} + o(h).
\end{align}
By setting $\mathbb{E}_{\mathbf{x}\sim\mu_{\rm QTW}}f(\mathbf{x}) = \delta$,
\begin{align}
    h = \frac{\delta}{\frac{\sqrt{2}}{4}\sum_{j=1}^N p(\infty,j)\tr\sqrt{\nabla^2_{\mathbf{x}}f(\mathbf{x}_j)}+ o_{h}(1)}.
\end{align}
And because for sufficiently small $h$, we have $h\propto \delta$ from the above equation, the term $o_{h}(1)$ in the above equation can be replaced by $o_{\delta}(1)$, giving \eq{hstand1}.

For $\mu_{\rm SGD}$, we still begin with a large region $\Omega$ containing all the minima $\{ \mathbf{x}_j :j=1,\ldots,N\}$.
Note that
\begin{align}
    \mu_{\rm SGD}(x) = \frac{e^{-\frac{2f(\mathbf{x})}{s}}}{\int_{\mathbb{R}^d}e^{-\frac{2f(\mathbf{x})}{s}} \d \mathbf{x}},
\end{align}
we first evaluate
\begin{align}
    \int_{\mathbb{R}^d\backslash \Omega} \mu_{\rm SGD}(\mathbf{x})f(\mathbf{x})\d \mathbf{x}
    = \frac{\int_{\mathbb{R}^d\backslash \Omega} e^{-\frac{2f(\mathbf{x})}{s}} f(\mathbf{x})\d \mathbf{x} }{\int_{\mathbb{R}^d}e^{-\frac{2f(\mathbf{x})}{s}} \d \mathbf{x}}.
\end{align}
Without loss of generality, we set $f(\partial \Omega) = \{ C \}$ and $f(\mathbf{x})\leq C$ for $\mathbf{x}\in \Omega$.
Define $g = \max(f-C,0)$, we have
\begin{align}
    \frac{\int_{\mathbb{R}^d\backslash \Omega} e^{-\frac{2f(\mathbf{x})}{s}} f(\mathbf{x})\d \mathbf{x} }{\int_{\mathbb{R}^d}e^{-\frac{2f(\mathbf{x})}{s}} \d \mathbf{x}}
    &= \frac{e^{-\frac{2C}{s}}\int_{\mathbb{R}^d\backslash \Omega} e^{-\frac{2g(\mathbf{x})}{s}} g(\mathbf{x})\d \mathbf{x} }{e^{-\frac{2C}{s}}\int_{\mathbb{R}^d}e^{-\frac{2(f(\mathbf{x})-C)}{s}} \d \mathbf{x}} \\
    &= \frac{\int_{\mathbb{R}^d} e^{-\frac{2g(\mathbf{x})}{s}} g(\mathbf{x})\d \mathbf{x} }{\int_{\mathbb{R}^d}e^{-\frac{2g}{s}} \d \mathbf{x}+\int_{\Omega}(e^{-\frac{2(f(\mathbf{x})-C)}{s}}-1) \d \mathbf{x}}.
\end{align}
Note that $I_g(s): = \int_{\mathbb{R}^d} e^{-\frac{2g(\mathbf{x})}{s}} g(\mathbf{x})\d \mathbf{x}$ decreases when $s$ decreases.
Let $\Omega':=\{ \mathbf{x}\mid f(\mathbf{x})\leq C/2\}$, we have
\begin{align}
    \int_{\Omega}(e^{-\frac{2(f(\mathbf{x})-C)}{s}}-1) \d \mathbf{x} &\geq \int_{\Omega'}(e^{-\frac{2(f(\mathbf{x})-C)}{s}}-1) \d \mathbf{x} \\
    &\geq \int_{\Omega'}(e^{\frac{C}{s}}-1) \d \mathbf{x} = (e^{\frac{C}{s}}-1) \mathrm{Volume}(\Omega').
\end{align}
Therefore,
\begin{align}
    \frac{\int_{\mathbb{R}^d} e^{-\frac{2g(\mathbf{x})}{s}} g(\mathbf{x})\d \mathbf{x} }{\int_{\mathbb{R}^d}e^{-\frac{2g}{s}} \d \mathbf{x}+\int_{\Omega}(e^{-\frac{2(f(\mathbf{x})-C)}{s}}-1) \d \mathbf{x}}
    \leq \frac{I_g(s)}{(e^{\frac{C}{s}}-1) \mathrm{Volume}(\Omega')}
     = O(s^{\infty}).
\end{align}
That is,
\begin{align}
    \int_{\mathbb{R}^d\backslash \Omega} \mu_{\rm SGD}(\mathbf{x})f(\mathbf{x})\d \mathbf{x}
    = O(s^{\infty}).
\end{align}

Now, we consider the small neighborhood $\Omega_j$ only containing one minimum $\mathbf{x}_j$. Similar to the quantum case, we let $\mathbf{x}_j= 0$ and introduce a rescaled coordinate $\{y_k,k=1,\ldots,d\}$ and $\mathbf{y} = (y_1,\ldots,y_d)\trans$ such that $\nabla^2_{\mathbf{y}}f(0) = I$.
Also define the hypercube $\Omega^{s}_j$ with edge length (under the coordinate of $\mathbf{y}$) $\Theta(s^{5/12})$,
apply the same argument as the quantum case,
we have
\begin{align}
    \int_{\Omega_j^s} e^{-\frac{2f}{s}} \d \mathbf{x} = \frac{(\pi s)^{d/2}}{\sqrt{\det \nabla^2_{\mathbf{x}}f(0)}}(1+ O(s^{1/4})) = \frac{(\pi s)^{d/2}}{\sqrt{\det \nabla^2_{\mathbf{x}}f(0)}}(1+ o_s(1)),
\end{align}
\begin{align}
    \int_{\Omega_j^s} e^{-\frac{2f}{s}} f \d \mathbf{x} = \frac{sd}{4}\frac{(\pi s)^{d/2}}{\sqrt{\det \nabla^2_{\mathbf{x}}f(0)}}(1+ o_s(1)),
\end{align}
\begin{align}
    \int_{\Omega_j\backslash \Omega_j^s} e^{-\frac{2f}{s}} \d \mathbf{x} = O(s^{\infty})
\end{align}
\begin{align}
    \int_{\Omega_j\backslash \Omega_j^s} e^{-\frac{2f}{s}} f \d \mathbf{x} = O(s^{\infty}).
\end{align}
In the region $\Omega\backslash (\bigcup_{j=1}^N \Omega_j )$, $f$ is bounded from below and both the integrals of $e^{-\frac{2f}{s}}$ and $e^{-\frac{2f}{s}} f$ should be of $O(s^{\infty})$.

Finally, we can obtain that
\begin{align}
    \mathbb{E}_{\mathbf{x}\sim \mu_{\rm SGD}}f(\mathbf{x})  = \int_{\mathbb{R}^d} \mu_{\rm SGD}(\mathbf{x})f(\mathbf{x})\d \mathbf{x}
    &= \frac{\int_{\mathbb{R}^d} e^{-\frac{2f(\mathbf{x})}{s}} f(\mathbf{x})\d \mathbf{x} }{\int_{\mathbb{R}^d}e^{-\frac{2f(\mathbf{x})}{s}} \d \mathbf{x}}
    = \frac{\frac{sd}{4}(1+o_s(1))}{1+ O(s^{\infty})} + O(s^{\infty})\nonumber\\
    &=\frac{sd}{4}(1+o_s(1)),
\end{align}
and then
\begin{align}
   s = \frac{\delta}{\frac{d}{4}(1+o_s(1))}.
\end{align}
Since $s\propto \delta$ for small $\delta$ or $s$, we can replace $o_s(1)$ in the above equation to $o_{\delta}(1)$.

The two equations in \lem{expectedrisk} are obtained finishing the proof.
\end{proof}

\subsubsection{Proof of \texorpdfstring{\lem{expecteddis}}{Lemma 13}}\label{append:prooflem13}
We estimate the expected distances from the minimum and establish an asymptotic relation between $h$ and $s$ by \stand{distance} in \lem{expecteddis}.

\begin{proof}
Note that
\begin{equation}
    \mu_{\rm QTW} = \frac{\big(\det\sqrt{\nabla^2 f(0)}\big)^{1/2}}{(\sqrt{2}\pi h)^{d/2}}
    e^{-\frac{x\trans\sqrt{\nabla^2 f(0)}x}{\sqrt{2}h}}
\end{equation}
We choose $\{x_i\}$ be a coordinate where $\nabla^2 f(0)$ is a
diagonal matrix denoted by $\mathrm{diag}(\lambda_1,..,\lambda_d)$.
Define $\omega_i = \sqrt{\lambda_i}$,
$\mathbb{E}_{x\sim \mu_{\rm QTW}}D(x,0)$ is then given by
\begin{align}
    \mathbb{E}_{x\sim \mu_{\rm QTW}}D(x,0) &=
    \int \frac{(\prod_{i} \omega_i )^{1/2}}{(\sqrt{2}\pi h)^{d/2}}
    e^{-\frac{\sum_i \omega_i x_i^2}{\sqrt{2}h}}\sum_i  x_i^2 \d x_1\cdots \d x_d \nonumber\\
    & = \frac{\sqrt{2}h}{2}\sum_j \frac{1}{\omega_j},
\end{align}
where the last equality uses the facts
\begin{equation}
    \int_{-\infty}^{\infty} e^{-ax^2}\d x = \sqrt{\frac{\pi}{a}},
\end{equation}
\begin{equation}
    \int_{-\infty}^{\infty} x^2 e^{-ax^2}\d x = \frac{1}{2}\sqrt{\frac{\pi}{a^3}}.
\end{equation}
Still use the above two Gaussian integrals, we have
\begin{equation}
    \int e^{\frac{-2f}{s}} \d x = \int e^{\frac{-\sum_i \lambda_i x_i^2}{s}}\d x_1 \cdots \d x_d = \frac{(\pi s)^{d/2}}{\sqrt{\prod_i \lambda_i }},
\end{equation}
and
\begin{equation}
    \int e^{\frac{-2f}{s}}D(x,0) \d x = \sum_j \int e^{\frac{-\sum_i \lambda_i x_i^2}{s}}x_j^2 \d x_1 \cdots \d x_d = \frac{s}{2}\frac{(\pi s)^{d/2}}{\sqrt{\prod_i \lambda_i }}\sum_i\frac{1}{\lambda_i},
\end{equation}
which can readily give \lem{expecteddis}.
\end{proof}

\subsection{Details about the exponential quantum speedup}
\subsubsection{Proof of \texorpdfstring{\lem{mconcentration}}{Lemma 22}}\label{append:prooflem22}
\begin{proof}
As shown in \fig{capandcone}, we use $\mathrm{Cone}(\epsilon)$ to denote the spherical cone subtended at one end by $\mathrm{Cap}(\epsilon)$. Let $\mathbb{B}^d$ be the $d$-dimensional unit ball, $\mathrm{Cone}(\epsilon)$ can be enclosed in a ball with radius $\sqrt{1-\epsilon^2}$. Therefore,
\begin{align}
    \frac{\mathrm{Area}(\mathrm{Cap}(\epsilon))}{\mathrm{Area}(\mathbb{S}^{d-1})} = \frac{\mathrm{Volume}(\mathrm{Cone}(\epsilon))}{\mathrm{Volume}(\mathbb{B}^{d})} \leq \frac{\mathrm{Volume}(\sqrt{1-\epsilon^2}\mathbb{B}^{d})}{\mathrm{Volume}(\mathbb{B}^{d})} = (1-\epsilon^2)^{d/2} \leq e^{-d\epsilon^2/2}.
\end{align}
\end{proof}

\subsubsection{Proof of \texorpdfstring{\lem{provable-onep}}{Lemma 23}}\label{append:prooflem23}
\begin{proof}
\begin{align}
    P(\mathbf{x} \notin S_{\mathbf{v}}) &= P(\mathbf{x}\in \mathbb{B}(\mathbf{0},R),~|\mathbf{x}\cdot\mathbf{v}|\leq w) \\
    &=
    P(\mathbf{y}\in \mathbb{B}(\mathbf{0},1),~|\mathbf{y}\cdot\mathbf{v}|\leq w/R)\quad \mathrm{where}~\mathbf{y}=\mathbf{x}/R\\
    &\leq P(\mathbf{z}\in \mathbb{S}^{d-1},~|\mathbf{z}\cdot\mathbf{v}|\leq w/R)\\
    &= 2\frac{\mathrm{Area}(\mathrm{Cap}(w/R))}{\mathrm{Area}(\mathbb{S}^{d-1})} \leq 2 e^{-\frac{dw^2}{2R^2}}.
\end{align}
\end{proof}

\subsubsection{Proof of \texorpdfstring{\lem{provable-manyp}}{Lemma 24}}\label{append:prooflem24}
\begin{proof}
For the algorithm, the point $\mathbf{x}_i$ it queried is fixed while
the direction $\mathbf{v}$ is unknown. Define $C_i = \{\mathbf{v}\in \mathbb{S}^{d-1}:|\mathbf{x}_i\cdot \mathbf{v}|>w\}$, we have $\mathbf{v}\notin C_i \iff \mathbf{x}_i\in S_{\mathbf{v}}$. Similar to \lem{provable-onep},
\begin{align}
    \frac{\mathrm{Area}(C_i)}{\mathrm{Area}(\mathbb{S}^{d-1})}
    =2\frac{\mathrm{Area}(\mathrm{Cap}(w/\|\mathbf{x}_1\|))}{\mathrm{Area}(\mathbb{S}^{d-1})}
    \leq 2\frac{\mathrm{Area}(\mathrm{Cap}(w/R))}{\mathrm{Area}(\mathbb{S}^{d-1})}
   \leq 2e^{-\frac{dw^2}{2R^2}}.
   \label{eq:provable-area}
\end{align}
Now, we need to relate the area ratios to the probability we want to calculate. Firstly, examine the term $P(\mathbf{x}_t\in S_{\mathbf{v}}\mid \forall \tau < t : \mathbf{x}_{\tau} \in S_{\mathbf{v}})$. Given that $\{ \forall \tau < t : \mathbf{x}_{\tau} \in S_{\mathbf{v}} \}$, we know $\mathbf{v} \in \mathbb{S}^{d-1} - \cup_{i=1}^{t-1}C_i$. Because $\forall \tau<t,~q(\mathbf{x}_{\tau})$ is independent of $\mathbf{v}$ conditioned on $\{ \forall \tau < t : \mathbf{x}_{\tau} \in S_{\mathbf{v}} \}$, we have no information about $\mathbf{v}$ except $\mathbf{v} \in \mathbb{S}^{d-1} - \cup_{i=1}^{t-1}C_i$ when determining $\mathbf{x}_t$. In this case, $\mathbf{v}$ should be uniformly
distributed over $\mathbb{S}^{d-1} - \cup_{i=1}^{t-1}C_i$, and we have
\begin{align}
    P(\mathbf{x}_t\in S_{\mathbf{v}}\mid \forall \tau < t : \mathbf{x}_{\tau} \in S_{\mathbf{v}}) = \frac{\mathrm{Area}(\mathbb{S}^{d-1} - \cup_{i=1}^{t}C_i)}{\mathrm{Area}(\mathbb{S}^{d-1} - \cup_{i=1}^{t-1}C_i)}.
    \label{eq:provable-oncondition}
\end{align}
By the product rule, we can get
\begin{align}
    P(\forall t\leq T:\mathbf{x}_t\in S_{\mathbf{v}}) = \prod_{t=1}^T P(\mathbf{x}_t\in S_{\mathbf{v}}\mid\forall \tau < t : \mathbf{x}_{\tau} \in S_{\mathbf{v}})
    =\frac{\mathrm{Area}(\mathbb{S}^{d-1} - \cup_{i=1}^{T}C_i)}{\mathrm{Area}(\mathbb{S}^{d-1})},
\end{align}
which leads to
\begin{align}
    P(\exists t\leq T:\mathbf{x}_t\notin S_{\mathbf{v}}) &= 1-P(\forall t\leq T:\mathbf{x}_t\in S_{\mathbf{v}})\\
    &= \frac{\mathrm{Area}( \cup_{i=1}^{T}C_i)}{\mathrm{Area}(\mathbb{S}^{d-1})}
    \leq \sum_i\frac{\mathrm{Area}(C_i)}{\mathrm{Area}(\mathbb{S}^{d-1})}
    \leq 2Te^{-\frac{dw^2}{2R^2}}.
\end{align}
\end{proof}

\subsubsection{Proof of \texorpdfstring{\prop{provable-exp}}{Proposition 4}}\label{append:proofprop4}
\begin{proof}
On the condition that $\mathbf{x} \in S_{\mathbf{v}}$,
\begin{align}
    f(\mathbf{x}) = \left\{
    \begin{array}{ll}
    \frac{1}{2}\omega^2 \|x\|^2,~\|x\| \leq a \\[3pt]
     H_2,~\|x\| > a,
    \end{array}
    \right.
\end{align}
which is independent of $\mathbf{v}$. Any local query $q(\mathbf{x})$ is then independent of $\mathbf{v}$.

By \lem{provable-manyp}, restricted in $\mathbb{B}(\mathbf{0},R)$, for any classical algorithm with only $T = o(e^{\frac{dw^2}{4R^2}})$ local queries,
\begin{align}
    P(\exists t\leq T: \mathbf{x}_t \notin S_{\mathbf{v}}) \leq 2Te^{-\frac{dw^2}{2R^2}} \leq e^{-\frac{dw^2}{4R^2}}.
\end{align}
Because $W_+ \cap S_{\mathbf{v}} = \varnothing$, with high probability, namely, $1-e^{-\frac{dw^2}{4R^2}}$, any classical algorithm cannot land in $W_+$.

Next, we consider \prb{provable} in $\mathbb{R}^d$. Let $F = \mathbb{R}^d\backslash \mathbb{B}(\mathbf{0},R)$ and $U_{\mathbf{v}} = \mathbb{B}(\mathbf{0},R)\backslash S_{\mathbf{v}}$.
The probabilities $P(\mathbf{x}_1\in U_{\mathbf{v}}\mid \mathbf{x}_1 \in \mathbb{B}(\mathbf{0},R)) \leq 2e^{-\frac{dw^2}{2R^2}}$ and $P(\mathbf{x}_1\in U_{\mathbf{v}}\mid \mathbf{x}_1 \in F ) =0$, such
that $P(\forall t\leq T: \mathbf{x}_t\notin U_{\mathbf{v}}) \geq 1 - 2Te^{-\frac{dw^2}{2R^2}}$ is valid for $T=1$.
Assume that $P(\forall t\leq T-1: \mathbf{x}_t\notin U_{\mathbf{v}}) \geq 1- 2(T-1)e^{-\frac{dw^2}{2R^2}}$ is true. We should examine the term $P(\mathbf{x}_T \notin U_{\mathbf{v}} \mid \forall t< T: \mathbf{x}_t\notin U_{\mathbf{v}})$.
Given that $\{\mathbf{x}_T \in F\}$, $P(\mathbf{x}_T \notin U_{\mathbf{v}} \mid \mathbf{x}_T \in F,~ \forall t< T: \mathbf{x}_t\notin U_{\mathbf{v}}) = 1$.
And we have
\begin{align}
    P(\mathbf{x}_T \notin U_{\mathbf{v}} \mid \forall t< T: \mathbf{x}_t\notin U_{\mathbf{v}}) \geq P(\mathbf{x}_T \notin U_{\mathbf{v}} \mid \mathbf{x}_T \in \mathbb{B}(\mathbf{0},R),~ \forall t< T: \mathbf{x}_t\notin U_{\mathbf{v}}).
\end{align}
In the queried $T-1$ pints, assume that there are $T_1 - 1$ points in $S_{\mathbf{v}}$.
All other $T-T_1$ points and queries in $F$ provide no information inside $\mathbb{B}(\mathbf{0},R)$ and are independent of $\mathbf{v}$. Therefore,
\begin{align}
\hspace{-2mm}     P(\mathbf{x}_T \notin U_{\mathbf{v}} \mid \mathbf{x}_T \in \mathbb{B}(\mathbf{0},R),~ \forall t< T: \mathbf{x}_t\notin U_{\mathbf{v}})
     & =P(\mathbf{x}_{T_1} \in S_{\mathbf{v}} \mid \forall t< T_1: \mathbf{x}_t\in S_{\mathbf{v}})\\
     & = \frac{\mathrm{Area}(\mathbb{S}^{d-1} - \cup_{i=1}^{T_1} C_i)}{\mathrm{Area}(\mathbb{S}^{d-1} - \cup_{i=1}^{T_1 - 1} C_i)}
     \geq 1 - \frac{\mathrm{Area}(C_{T_1})}{\mathrm{Area}(\mathbb{S}^{d-1} - \cup_{i=1}^{T_1 - 1} C_i)}\\
     &\geq 1 - \frac{2e^{-\frac{dw^2}{2R^2}}}{1-2(T_1-1)e^{-\frac{dw^2}{2R^2}}}
     \geq 1 - \frac{2e^{-\frac{dw^2}{2R^2}}}{1-2(T-1)e^{-\frac{dw^2}{2R^2}}}.
\end{align}
In the second equality, we use \eq{provable-oncondition} as queries in $S_{\mathbf{v}}$ are independent of $\mathbf{v}$. The second inequality is deduced from \eq{provable-area}.
Finally, we have
\begin{align}
     P(\forall t\leq T: \mathbf{x}_t\notin U_{\mathbf{v}}) = P(\mathbf{x}_T \notin U_{\mathbf{v}} \mid \forall t< T: \mathbf{x}_t\notin U_{\mathbf{v}})P(\forall t\leq T-1: \mathbf{x}_t\notin U_{\mathbf{v}}) \geq 1- 2Te^{-\frac{dw^2}{2R^2}}.
\end{align}
By mathematical induction, $P(\forall t\leq T: \mathbf{x}_t\notin U_{\mathbf{v}})\geq 1- 2Te^{-\frac{dw^2}{2R^2}}$ is true for any $T\in \mathbb{N}_+$. In $\mathbb{R}^d$, for any classical algorithm given only $T = o(e^{-\frac{dw^2}{4R^2}})$ queries, the probability of finding any point in $W_+$ is given by
\begin{align}
     P(\exists t\leq T: \mathbf{x}_t\in W_+) \leq P(\exists t\leq T: \mathbf{x}_t\in U_{\mathbf{v}}) \leq 2Te^{-\frac{dw^2}{2R^2}}\leq e^{-\frac{dw^2}{4R^2}},
\end{align}
which completes the proof.
\end{proof}

\subsubsection{Proof of \texorpdfstring{\lem{provableS0}}{Lemma 25}}\label{append:prooflem25}
\begin{proof}
First, the integral in \eq{provable-S0} can be easily verified. We need to prove that $\frac{1}{\sqrt{2}}\omega a^2 + 2(b-a)\sqrt{H_1}$ is the shortest Agmon length of paths linking $U_-$ and $U_+$.

Without loss of generality, we choose an orthonormal basis $\{\mathbf{e}_1,\ldots,\mathbf{e}_d\}$, and write any vector $\mathbf{x} = \sum_i x^i\mathbf{e}_i$.
We set $\mathbf{e}_1 = \mathbf{v}$.
Let $x^i = l^i(\tau)~(0\leq \tau \leq M_l)$ be an arbitrary smooth path joining $U_-$ and $U_+$, where $\tau$ is the parameter of the curve, $l^i(0) = 0$, and $l^i(M_l) = \delta_{i1} 2b$.
If $x^1$ can be the parameter of this path (i.e., $x^1 = l^1 (\tau)$ is a bijection and $x^i = l^i((l^1)^{-1}(x^1))~(0\leq x^1 \leq 2b)$), the Agmon length between $U_-$ and $U_+$ along this path is given by
\begin{align}
    l(U_-,U_+) &= \int_{0}^{M_l}  \sqrt{f \delta_{\mu \nu}\frac{\partial x^{\mu}}{\partial \tau}
    \frac{\partial x^{\nu}}{\partial \tau}} \d \tau
    = \int_{0}^{2b}  \sqrt{f \delta_{\mu \nu}\frac{\partial x^{\mu}}{\partial x^1}
    \frac{\partial x^{\nu}}{\partial x^1}} \d x^1 \\
    & \geq \int_{0}^{2b}  \sqrt{f(x^1\mathbf{e}_1)} \d x^1 = \frac{1}{\sqrt{2}}\omega a^2 + 2(b-a)\sqrt{H_1}.
    \label{eq:64}
\end{align}
If $\psi^1 = l^1 (\tau)$ is not a bijection the second equality in the above equation will be invalid.
In this case, we should restrict the domain on $\mathcal{D} = \bigcup_i [\tau_i^-,\tau_i^+] \subset [0,M_l]$ ($\tau_i^+ \leq \tau_{i+1}^-$) s.t. $l^1: \mathcal{D} \to [0,2b]$ is a bijection and $[0,M_l] - \mathcal{D}$ is a zero measure set. Then, we have
\begin{align}
    l(U_-,U_+) &\geq \sum_i\int_{\tau_i^-}^{\tau_i^+}  \sqrt{f \delta_{\mu \nu}\frac{\partial x^{\mu}}{\partial \tau}
    \frac{\partial x^{\nu}}{\partial \tau}} \d \tau \\
    &= \sum_i\int_{l^1(\tau_i^-)}^{l^1(\tau_i^+)}  \sqrt{f
    \delta_{\mu \nu}\frac{\partial x^{\mu}}{\partial x^1}
    \frac{\partial x^{\nu}}{\partial x^1}} \d x^1\\
    & \geq \sum_i \int_{l^1(\tau_i^-)}^{l^1(\tau_i^+)}  \sqrt{f(x^1\mathbf{e}_1)
    } \d x^1 = \frac{1}{\sqrt{2}}\omega a^2 + 2(b-a)\sqrt{H_1}.
    \label{eq:67}
\end{align}
This implies the correctness of \eq{provable-S0}.

Next, observe that $(\partial x^i/\partial x^1)^2>0$, if any path $x^i = l^i(\tau)~(0\leq \tau \leq M_l)$ is not the one given by \eq{provable-geodesic}, inequalities in \eq{64} and \eq{67} will be strict. Therefore, we prove that \eq{provable-geodesic} is the only geodesic linking $U_-$ and $U_+$ with the Agmon length $S_0$.
\end{proof}

\subsubsection{Proof of \texorpdfstring{\lem{provablenu}}{Lemma 26}}\label{append:prooflem26}
\begin{proof}
Because of the symmetry of the two local ground states,
we can naturally have \eq{provable-matrix}.
Since the constructed landscape satisfies assumptions needed for
\prop{speW}, we can use \prop{speW} to calculate the tunneling amplitude
$\nu$ directly. By the geometrical symmetry of the two wells,
the special point on the geodesic $\boldsymbol{\gamma}_{-+}$ should be
$b\mathbf{v}$.
The next-to-leading term of $\nu$ is then given by
\begin{align}
    \nu = -h^{1/2}
    2(2\pi)^{\frac{d-1}{2}}
    \sqrt{\frac{f(b\mathbf{v})}{\mathrm{det}'\Big(\nabla^2 d_{-+}(b\mathbf{v})\Big)}}a^{(-)}_0(b\mathbf{v})a^{(+)}_0(b\mathbf{v})e^{-S_0/h},
    \label{eq:raww}
\end{align}
where $\mathrm{det}'$ denotes the usual determinant with the
zero mode removed, and $a^{(\pm)}_0$ are the leading terms in the WKB approximations of the local ground states.

Choose an orthonormal basis $\{\mathbf{e}_1, \ldots,\mathbf{e}_d\}$ in $\mathbb{R}^d$ where $\mathbf{e}_1 = \mathbf{v}$, we have $\mathbf{x} = \sum_{i=1}^d x_i \mathbf{e}_i$.
Since along the geodesic $\boldsymbol{\gamma}_{-+}$, $d_{-+} \equiv S_0$,
we obtain that $\partial d_{-+}/\partial x_1 = 0$ and $\partial^2 d_{-+}/\partial x_1^2 = 0$. That is, $\forall j,~(\nabla^2d_{-+}(b\mathbf{v}))_{ij} = 0$.
Constrained on the hypersurface perpendicular to $\mathbf{e}_1$, consider an infinitesimal $\delta x = \sum_{i=2}^d \delta x_i$ and we have
\begin{align}
    d(b\mathbf{v} + \delta \mathbf{x},\mathbf{0}) - d(b\mathbf{v},\mathbf{0})
    = \frac{\|\delta x\|^2}{b}\sqrt{H_1} + o(\|\delta x\|^2).
\end{align}
Given that
\begin{align}
    d(b\mathbf{v} + \delta \mathbf{x},\mathbf{0}) - d(b\mathbf{v},\mathbf{0})
    = \sum_{i,j=2}^d \left.\frac{\partial^2 d(\mathbf{x},\mathbf{0})}{\partial x_i \partial x_j}\right|_{\mathbf{x}=b\mathbf{v}}\delta x_i \delta x_j,
\end{align}
it is readily to get $(\nabla^2 d(\mathbf{x},\mathbf{0})|_{\mathbf{x}= b\mathbf{v}})_{i,j} = \delta_{ij}2\sqrt{H_1}/b~(i,j\geq 2)$.
Similarly, by symmetry, we can obtain $(\nabla^2 d(\mathbf{x},2b\mathbf{v})|_{\mathbf{x}= b\mathbf{v}})_{i,j} = \delta_{ij}2\sqrt{H_1}/b~(i,j\geq 2)$.
Use the fact $d_{-+}(\mathbf{x}) = d(\mathbf{x},\mathbf{0}) + d(\mathbf{x}, 2b\mathbf{v})$,
we deduce the following
\begin{align}
    \nabla^2 d_{-+}(b\mathbf{v}) = \mathrm{diag}\left(0,\frac{4\sqrt{H_1}}{b},\ldots,\frac{4\sqrt{H_1}}{b}\right)~\mathrm{under~the~basis}~\{\mathbf{e}_1, \ldots,\mathbf{e}_d\},
\end{align}
which leads to
\begin{align}
    \mathrm{det}' \big(\nabla^2 d_{-+}(b\mathbf{v})\big) = (4\sqrt{H_1}/b)^{d-1}.
    \label{eq:detbv}
\end{align}

Next, we need to calculate $a^{(\pm)}_0(b\mathbf{v})$. Using the symmetry between the two local states, we know $a^{(-)}_0(b\mathbf{v}) = a^{(+)}_0(b\mathbf{v})$.
So, we can only focus on the calculation of $a^{(-)}_0(\mathbf{x})$, which is the leading term in WKB approximation of the local ground state of the well $U_-$.
By \lem{wkb}, $a^{(-)}_0(\mathbf{x})$ should satisfy the first transport equation
\begin{align}
    2 \nabla d(\mathbf{x},\mathbf{0}) \cdot \nabla a^{(-)}_0 + (\Delta d(\mathbf{x},\mathbf{0}) - E_1)a^{(-)}_0 = 0,
\end{align}
where $E_1 = \omega d /\sqrt{2}$ as the neighborhood of the well is quadratic. On any point $\mathbf{x} = x_1 \mathbf{v} \in \boldsymbol{\gamma}_{-+}$, the transport equation is reduced to
\begin{align}
    \nabla a^{(-)}_0 = 0,\quad x_1\leq a,
\end{align}
\begin{align}
    \frac{\partial a^{(-)}_0}{\partial x_1} + d \left(\frac{d-1}{x_1d} -
    \frac{\omega}{2\sqrt{2H_1}}\right)a^{(-)}_0 = 0,\quad a< x_1 \leq b.
\end{align}
With the initial condition $a^{(-)}_0(\mathbf{0}) = (\frac{\omega}{\sqrt{2}\pi})^{d/4}$, $a^{(-)}_0(b\mathbf{v})$ is determined as
\begin{align}
    a^{(-)}_0(b\mathbf{v}) = (\frac{\omega}{\sqrt{2}\pi})^{d/4} \exp\left(\frac{\omega d(b-a)}{2\sqrt{2H_1}} - d \ln \frac{b}{a}\right).
    \label{eq:a0bv}
\end{align}
Substitute the results \eq{detbv} and \eq{a0bv} into \eq{raww}, we can get \eq{provable-w}, which finished the proof.
\end{proof}

\subsubsection{Proof of \texorpdfstring{\prop{provable-poly}}{Proposition 5}}\label{append:proofprop5}
\begin{proof}
First of all, $H_2$ is a constant independent of $d$.
We set $h$ to be of the form $h = \frac{4\delta}{\sqrt{2}\omega d}$ and $\delta \in (0,\delta_0]$ is bounded where $\delta_0 = O(H_0)$.
To understand this setting, we evaluate the energy of the local ground state $\ket{\Phi_-}$ corresponding to the well $U_-$.
Recall that $H_0\ll H_2$, if $\delta$ is small enough to let
$\ket{\Phi_-}$ localize in the quadratic region $W_-$,
the energy of $\ket{\Phi_-}$ should be approximately $h\tr \nabla^2 f(\mathbf{0})/\sqrt{2} = hd\omega/\sqrt{2}=2\delta \ll H_2$.
If $\delta$ is of the order $H_0$, a conservative estimation is that
the potential energy $\bra{\Phi_-} f \ket{\Phi_-} = O(H_0)$ and the
kinetic energy is $\bra{\Phi_-} -h^2\Delta \ket{\Phi_-} = O( d h^2(\frac{2\pi}{D})^2)$ where $D = \Omega(a)$ can be seen as the
diameter of the region where $\ket{\Phi_-}$ is localized.
Therefore, the energy of $\ket{\Phi_-}$ can be always far less than $H_2$, and $\ket{\Phi_-}$ should be localized in the region with
low potential energy $W_-\cup B_{\mathbf{v}} \cup W_+$.
A more quantitative way of seeing $h = \frac{4\delta}{\sqrt{2}\omega d}$ is that this setting solves the problem associated to measure concentration.

Together with \lem{provableS0} and \lem{provablenu},
we have
\begin{align}
    \nu = - \sqrt{\frac{8hH_1\sqrt{H_1}}{\pi b}}\exp\left(\Gamma\Big(\frac{H_0}{\delta}, \frac{b}{a}, \frac{H_0}{H_1}\Big) d \right),
\end{align}
where
\begin{align}
    \Gamma(x,y,z) := -\frac{1}{2}x + (y-1)\sqrt{z}\big(1-\frac{x}{z}\big) - \frac{3}{2} \ln y - \ln \sqrt{2} + \frac{1}{4}\ln z.
\end{align}
Let $\ket{\Phi_+}$ be the local ground state corresponding to $U_+$.
The key to make our analysis on $\nu$ valid is that orthonormalized local ground states $\ket{\Phi_-}$ and $\ket{\Phi_+}$ span the 2-dimensional subspace spanned by $\ket{E_0}$ and $\ket{E_1}$. Here, $\ket{E_0}$ and $\ket{E_1}$ are the ground state and the first excited state, respectively.
To be more specific,
let $E_0$ and $E_1$ be the energy of $\ket{E_0}$ and $\ket{E_1}$, respectively, then
\begin{align}
    \ket{E_0} = \frac{1}{\sqrt{2}} (\ket{\Phi_-} + \ket{\Phi_+}),\quad \ket{E_1} = \frac{1}{\sqrt{2}} (\ket{\Phi_-} - \ket{\Phi_+}),
\end{align}
and
\begin{align}
    E_0 = \mu + \nu,\quad E_1= \mu - \nu,
\end{align}
where $\mu$ is the energy of $\ket{\Phi_-}$ (or $\ket{\Phi_+}$) and without loss of generality, inner products  $\ip{\mathbf{x}}{\Phi_-}$ and $\ip{\mathbf{x}}{\Phi_+}$ are real and positive.
Only studying $\nu$, it seems we can let $|\nu|$ increase exponentially with respect to $d$, but this is infeasible in certain.
To make this scenario valid, we should have $|\nu| \ll \mu$ and $|\nu| \ll E_2 - E_1$, where $E_2$ is the energy of the second excited state.
Next, we show these constraints can be satisfied while $|\nu|$ is not too small.
Let $\mu'$ be the energy of the local first excited states,\footnote{In general, there may be many local eigenstates with the same energy $\mu'$. This phenomenon is called the degeneracy of energy levels.} we set $|\nu'| = \mu'- E_2$.
Then, $|\nu| \ll E_2 - E_1$ is equivalent to $|\nu| + |\nu'| \ll \mu' - \mu$.
Since $g_1:=\mu' - \mu$ is the energy gap between the local ground state and local first excited states which all concentrate near the region $W_-$ (or $W_+$),\footnote{Because of the symmetry of $W_-$ and $W_+$, it suffices to study one region $W_-$ and the conclusion is also true for $W_+$.}
its value satisfies $g_1 =\Omega(h^2) = \Omega(1/\mathrm{poly}(d))$ \cite{AC11}.
Let $g_2 = 1/\mathrm{poly}(d)$ be a function of $d$ such that $g_2 = o_d(g_1)$, we further demand that
\begin{align}
    \frac{1}{\sqrt{d}}\exp\left(\Gamma\Big(\frac{H_0}{\delta}, \frac{b}{a}, \frac{H_0}{H_1}\Big) d \right) = g_2(d),
\end{align}
which gives
\begin{align}
    \Gamma\Big(\frac{H_0}{\delta}, \frac{b}{a}, \frac{H_0}{H_1}\Big) = -\frac{1}{d}\ln \mathrm{poly}(d) \to 0,~(d \to \infty).
    \label{eq:provable-valid1}
\end{align}
Eq. \eq{provable-valid1} can always be valid by properly choosing the parameters $(\frac{H_0}{\delta}, \frac{b}{a}, \frac{H_0}{H_1})$.
If we demand $H_1$ and $b$ to be two constants independent of $d$,
$(\frac{H_0}{\delta}, \frac{b}{a}, \frac{H_0}{H_1})$ are still three free variables.
Adding the restriction \eq{provable-valid1}, there are two free variables in $(\frac{H_0}{\delta}, \frac{b}{a}, \frac{H_0}{H_1})$ left and
we can obtain $|\nu| = o_d(g_1)$.
Although $|\nu'|$ which captures tunneling effects between
local excited states is difficult to be calculated explicitly,
we claim that $|\nu'|$, as another tunneling amplitude, still depends exponentially on a function of
$(\frac{H_0}{\delta}, \frac{b}{a}, \frac{H_0}{H_1})$.
Therefore, by the same procedure of restricting $|\nu|$, $|\nu'|$ can be set to be $o_d(g_1)$.
Now, there is only one free variable left in $(\frac{H_0}{\delta}, \frac{b}{a}, \frac{H_0}{H_1})$.
Making use of the left one free variable, we can demand $(\frac{H_0}{\delta}, \frac{b}{a}, \frac{H_0}{H_1})$ to converge to non-zero constants when $d\to \infty$.
In this case, $a$, $b$, $H_0$, $H_1$, and $\delta$ all have limits and
previous constraints (independent of determining $\nu$ and $\nu'$) on these variables can be satisfied.
Since $\delta$ has a non-zero limit, and $|\nu| = \Theta(1/\mathrm{poly}(d))$ the condition $|\nu| \ll \mu = \Omega(\delta)$ can be satisfied
for sufficiently large $d$.

In the double-well case, initiating QTW from the given local ground state $\ket{\Phi_-}$, \lem{qtwmixingtime} can be simplified.
Note that the $O(h^{\infty})$ term in \lem{qtwmixingtime} is due to the integral
$\int |\ip{\Phi_-}{\mathbf{x}}\ip{\mathbf{x}}{\Phi_+} | \d \mathbf{x}$
which can be re-estimated as
\begin{align}
    \int |\ip{\Phi_-}{\mathbf{x}}\ip{\mathbf{x}}{\Phi_+} | \d \mathbf{x} \leq \Big(\int |\ip{\Phi_-}{\mathbf{x}}|^2 \d \mathbf{x}\Big)^{1/2}\Big(\int |\ip{\Phi_+}{\mathbf{x}}|^2 \d \mathbf{x}\Big)^{1/2} = 1.
\end{align}
Thus, for the present case, we can write $T_{\rm mix} = O(1/\epsilon |\nu|) = O(\mathrm{poly}(d)/\epsilon)$ (note that $\Delta E = 2|\nu|$ here).
Under $T_{\rm mix} = O(1/(\epsilon/2)|\nu|) = O(1/\epsilon|\nu|)$, we have
\begin{align}
    \|\rho_{\rm QTW}(T_{\rm mix}, \cdot) - \mu_{\rm QTW}(\cdot) \|_1 \leq \frac{\epsilon}{2}.
\end{align}
Let $\Phi(t)=e^{-iH t}\Phi(0)$ with the initial state $\Phi(0)$.
By quantum simulation, we cannot get exactly the state $\Phi(t)$
but an approximation $\tilde{\Phi}(t)$.
Set $\tilde{\rho}_{\rm QTW}(T_{\rm mix},x) = \int_{0}^{T_{\rm mix}} |\bra{x} \tilde{\Phi}(t)\y|^2 \d t/T_{\rm mix}$,
the condition
\begin{align}
   \| \tilde{\Phi}(t) - \Phi(t)\| \leq \frac{\epsilon}{2}~(t\leq T_{\rm mix})
\end{align}
gives
\begin{align}
   \|\rho_{\rm QTW}(T_{\rm mix}, \cdot) - \tilde{\rho}_{\rm QTW}(T_{\rm mix}, \cdot)\|_1 \leq \frac{\epsilon}{2},
\end{align}
which requires the number of quantum queries to be
\begin{align}
    O\left(\|f\|_{L^{\infty}} t \frac{\log (2\|f\|_{L^{\infty}} t/\epsilon)}{\log \log (2\|f\|_{L^{\infty}} t/\epsilon)}\right)
    = O\left(H_2 t \frac{\log (2 H_2 t/\epsilon)}{\log \log (2H_2 t/\epsilon)}\right)
    \label{eq:428}
\end{align}
(see \lem{quansimu} in \append{quantumsimulation}).
After simulating for some time bounded by $T_{\rm mix}$, measurements are equivalent to sampling from $\tilde{\rho}_{\rm QTW}(T_{\rm mix}, \cdot)$,
and we have
\begin{align}
   \|\tilde{\rho}_{\rm QTW}(T_{\rm mix}, \cdot)-\mu_{\rm QTW}(\cdot)\|_1 \leq \epsilon.
\end{align}

Next, we need to analyze $\mu_{\rm QTW}$ to determine the number of iterations and then the number of total queries.
In our case, the probability of finding the particle at $W_+$ is given by
\begin{align}
    \int_{W_+} \mu_{\rm QTW} \d \mathbf{x}
    &= \int_{W_+} [\frac{1}{2}(|\ip{\mathbf{x}}{\Phi_+}|^2 + |\ip{\mathbf{x}}{\Phi_-}|^2 )
    + 2\mathrm{Re}(\ip{\Phi_-}{\Phi(0)}\ip{\Phi(0)}{\Phi_+})\ip{\Phi_-}{\mathbf{x}}\ip{\mathbf{x}}{\Phi_+} ] \d \mathbf{x}\\
    &\geq \frac{1}{2} \int_{W_+}|\ip{\mathbf{x}}{\Phi_+}|^2 \d \mathbf{x} - \int_{W_+} |\ip{\Phi_-}{\mathbf{x}}\ip{\mathbf{x}}{\Phi_+}| \d \mathbf{x}.
\end{align}
For $\mathbf{x}\in W_+$, we have $\mathbf{x}-2(\mathbf{x}\trans\mathbf{v}) \mathbf{v}+2\mathbf{b} \in W_-$, and thus the following holds because of the symmetry between $\Phi_+$ and $\Phi_-$:
\begin{align}
    \frac{|\ip{\Phi_-}{\mathbf{x}}|}{|\ip{\Phi_+}{\mathbf{x}}|} = \frac{|\ip{\Phi_-}{\mathbf{x}}|}{|\ip{\Phi_-}{\mathbf{x}-2(\mathbf{x}\trans\mathbf{v}) \mathbf{v}+2\mathbf{b}}|},~\mathrm{for}~\mathbf{x}\in W_+.
\end{align}
Under previously determined parameters, $\Phi_-$ and $\Phi_+$ concentrate near $W_-$ and $W_+$ respectively, giving $\frac{|\ip{\Phi_-}{\mathbf{x}}|}{|\ip{\Phi_+}{\mathbf{x}}|}\leq 1$ for any $\mathbf{x}\in W_+$. Manipulating the parameters, for example, changing $S_0/h$ to $S_0/h+C$ for some constant $C>0$ independent of $d$, $1/\nu$ will be a larger polynomial and the ratio $\frac{|\ip{\Phi_-}{\mathbf{x}}|}{|\ip{\Phi_-}{\mathbf{x}-2(\mathbf{x}\trans\mathbf{v}) \mathbf{v}+2\mathbf{b}}|}~(\mathbf{x}\in W_+)$ at least reduces by constant times. So, it can always be done in making $\frac{|\ip{\Phi_-}{\mathbf{x}}|}{|\ip{\Phi_+}{\mathbf{x}}|}\leq \frac{1}{4}~(\mathbf{x}\in W_+)$ while maintaining $1/\nu$ as a polynomial in $d$.
And we have
\begin{align}
    \int_{W_+} \mu_{\rm QTW} \d \mathbf{x}
   \geq \frac{1}{4} \int_{W_+}|\ip{\mathbf{x}}{\Phi_+}|^2 \d \mathbf{x}.
\end{align}
Next, we should estimate $\int_{W_+}|\ip{\mathbf{x}}{\Phi_+}|^2 \d \mathbf{x}$.
Note that
\begin{align}
     1 = \int |\ip{\mathbf{x}}{\Phi_+}|^2 = (\int_{\mathbb{R}^d\backslash(W_-\cup B_{\mathbf{v}} \cup W_+)} +\int_{W_-} + \int_{B_{\mathbf{v}}} + \int_{W_+})|\ip{\mathbf{x}}{\Phi_+}|^2 \d \mathbf{x},
\end{align}
we estimate the integrals separately.
The potential energy of $\ket{\Phi_+}$ should be of $O(H_0)$, so that
\begin{align}
    \int_{\mathbb{R}^d\backslash(W_-\cup B_{\mathbf{v}} \cup W_+)}H_2 |\ip{\mathbf{x}}{\Phi_+}|^2 = O(H_0),
\end{align}
giving that
\begin{align}
      \int_{\mathbb{R}^d\backslash(W_-\cup B_{\mathbf{v}} \cup W_+)} |\ip{\mathbf{x}}{\Phi_+}|^2 \d \mathbf{x} = O(H_0/H_2).
\end{align}
So, we know $(\int_{W_-} + \int_{B_{\mathbf{v}}} + \int_{W_+})|\ip{\mathbf{x}}{\Phi_+}|^2 \d \mathbf{x}$ is larger than some constant $C_1$ independent of $d$.
Because $\ket{\Phi_+}$ is the local ground state corresponding to $U_+$, it concentrates in $U_{+}$ and promises that $\int_{W_+}|\ip{\mathbf{x}}{\Phi_+}|^2 \d \mathbf{x} \geq \int_{W_-}|\ip{\mathbf{x}}{\Phi_+}|^2 \d \mathbf{x}$.
Note that the volume of $B_{\mathbf{v}}$ is exponentially smaller than that of $W_+$, that is,
\begin{align}
    \frac{\mathrm{Volume(B_{\mathbf{v}})}}{\mathrm{Volume}(W_+)}
    \leq \frac{(\sqrt{3}a/2)^{d-1} \mathrm{Area}(\mathbb{S}^{d-1})(2b-a)}{a^d\mathrm{Volume}(\mathbb{S}^{d})} = O(\mathrm{poly}(e^{-d})).
\end{align}
After steps of determining $\nu$ and $\nu'$, we still have a free variable, say, $H_0/H_1$.
We set $H_0/H_1 =1$,
and the energy of a particle restricted in $B_{\mathbf{v}}$ will significantly larger than that restricted in $W_+$.
Therefore, for the local ground state $\ket{\Phi_+}$, we can also have $\int_{W_+}|\ip{\mathbf{x}}{\Phi_+}|^2 \d \mathbf{x} \geq \int_{B_{\mathbf{v}}}|\ip{\mathbf{x}}{\Phi_+}|^2 \d \mathbf{x}$.
In all,
\begin{align}
    C_1 \leq \Big(\int_{W_-} + \int_{B_{\mathbf{v}}} + \int_{W_+}\Big)|\ip{\mathbf{x}}{\Phi_+}|^2 \d \mathbf{x} \leq 3\int_{W_+}|\ip{\mathbf{x}}{\Phi_+}|^2 \d \mathbf{x},
\end{align}
and then
\begin{align}
    \int_{W_+} \mu_{\rm QTW} \d \mathbf{x} \geq C_2,
\end{align}
where $C_2 := C_1/12$ is a constant independent of $d$.
Choose $\epsilon = C_2/2$ to be a constant independent of $d$,
we have
\begin{align}
    \int_{W_+}\tilde{\rho}_{\rm QTW}(T_{\rm mix}, \mathbf{x}) \d \mathbf{x}
    &\geq \int_{W_+} \mu_{\rm QTW}(\mathbf{x}) \d \mathbf{x} - \int_{W_+}|\tilde{\rho}_{\rm QTW}(T_{\rm mix},\mathbf{x}) - \mu_{\rm QTW}(\mathbf{x})| \d \mathbf{x}\\
    &\geq C_2 - \|\tilde{\rho}_{\rm QTW}(T_{\rm mix}, \cdot) - \mu_{\rm QTW}(\cdot) \|_1
    \geq  C_2/2.
\end{align}
Let $C_3 := C_2/2$ be also a constant independent of
$d$.
After evolving the system for at most $T_{\rm mix}$ once
the probability of hitting $W_+$ is at least $C_3$.
Then, using $n$ interactions which takes a total time bounded by $T_{\rm tot} = n T_{\rm mix} = n O(\mathrm{poly}(d))$,
\footnote{Here, $\epsilon$ has been set to be a constant and can be omitted.} the probability of successfully hitting $W_+$ is $p_{\rm suc} \geq 1 - (1-C_3 )^n$.\footnote{In the proof, we have determined $H_2$, $h$, $H_1$, $H_0$, $a$ and $b$; the only remained parameter is $w$. We demand $w$ to be a constant independent of $d$. In this case, as long as $h = 4\delta/\sqrt{2}\omega d$, the Gaussian integral used in \prop{speW} which gives \lem{provablenu} is accurate enough for large $d$, and our conclusion based on \lem{provablenu} is not affected.}
Because of \eq{428} and note that $H_2$ and $\epsilon$ are constants independent of $d$,
the number of total quantum quires is $n\tilde{O}(T_{\rm mix})$. Since $T_{\rm mix}$ is polynomial in $d$, the number of total quantum quires needed is polynomial in $d$.
\end{proof}

\subsubsection{Proof of \texorpdfstring{\prop{nospeedup}}{Proposition 6}}\label{append:nospeedup}
Let $\bf{n}$ be a unit vector, we define the cone $\mathcal{C}_{q}(\bf{n})$ as
\begin{align}
    \mathcal{C}_{q}(\bf{n}) := \{\bf{x}: \bf{x}\cdot \bf{n}/\|\bf{x}\| > q\}.
\end{align}
For $f$ defined by \eq{hardinstance}, the regions
$B_{\bf{v}}$ and $W_+$ can be in a cone $\mathcal{C}_{q}(\bf{v})$ for some constant $q>\sqrt{3}/2$.

Suppose there are $N$ directions $\bf{n}_j~(j=1,...,N)$ such that all $\mathcal{C}_{q}(\bf{n}_j)~(j=1,...,N)$ are disjoint.
We first show that an unstructured search with datasize $N$ can be solved by solving \prb{provable}.
Let the datapoints be labeled by $\alpha \in \{1,...,N\}$.
In the unstructured search problem, we are given an oracle $\mathcal{O}$ giving $\mathcal{O}(\alpha)$ the value assigned to $\alpha$. Recall that there are only one point is assigned $1$ and others are assigned $0$. Our goal is to find the only point, say $\alpha_s$, such that $\mathcal{O}(\alpha_s)=1$.
Every point in unstructured search, $\alpha$, can be made uniquely correspond to a cone $\mathcal{C}_{q}(\bf{n}_{\alpha})$.

Let $f_0$ be a function satisfying \eq{hardinstance} whose $\bf{v} := \bf{n}_0$. Using the oracle $\mathcal{O}$ and $f_0$, we can construct a new function $f$. If $\|\bf{x}\|\leq a$, $f(\bf{x}) = f_0(\bf{x})$, For any $\bf{x}> a$, first determine whether $\bf{x}\in \mathcal{C}_{q}(\bf{n}_j)$ for some $j$. If so, $f(\bf{x}) = f_0(R(j,0)\bf{x})$ for $\mathcal{O}(j) = 1$ and $f(\bf{x}) = H_2$ for $\mathcal{O}(j) = 0$, where $R(j,0)$ is a rotation mapping $\bf{n}_j$ to $\bf{n}_0$. Otherwise, $f(\bf{x}) = H_2$.
From the above construction, we know that $f$ satisfies the following properties:
\begin{itemize}
\item This $f$ is a function of the form \eq{hardinstance} whose $\bf{v} = \bf{n}_{\alpha_s}$.
\item One query to the function value of $f$ need at most one query to the oracle $\mathcal{O}$.
\end{itemize}
Therefore, if we can solve \prb{provable} on this $f$ within queries (to $f$) polynomial in $d$, we can have an algorithm for the $N$-size unstructured search finding $\alpha_s$ within queries (to $\mathcal{O}$) polynomial in $d$. With the fact $N$ is exponential in $d$, assuming \prb{provable} can be solved within queries polynomial in $d$ will violate the quantum lower bound for unstructured search, i.e., $\Omega(\sqrt{N})$ \cite{BBBV97}.

Next, to prove \prop{nospeedup} by contradiction, we just need to show that $N$ can be exponential in $d$.
The only constrain on $N$ is that $\mathcal{C}_{q}(\bf{n}_j)\cap \mathcal{C}_{q}(\bf{n}_k) = \varnothing$ for $j\neq k$, which is equivalent to $\|\bf{n}_j - \bf{n}_k\|\geq \varepsilon$ for $j\neq k$ and some $\varepsilon<1$.
The following Lemma provides a desired estimation of $N$.
\begin{lemma}
There exists a set of $N$ unit vectors in $\mathbb{R}^d$ (or $N$ points on $\mathbb{S}^{d-1}$) such that
\begin{itemize}
\item $\forall \bf{x} \neq \bf{y}$ in the set, $\|\bf{x}-\bf{y}\|\geq \varepsilon$;
\item $N \geq (\frac{1}{2\varepsilon}+ \frac{1}{2})^d - (\frac{1}{2\varepsilon} - \frac{1}{2})^d$.
\end{itemize}
\end{lemma}
\begin{proof}
Let the set of points be $A$, we consider an argument that as long as there is a point $\bf{x}$ satisfying $\min_{\bf{y}\in A}\|\bf{x} - \bf{y}\|\geq \varepsilon$, we add $\bf{x}$ to $A$.

If we cannot add any more point to $A$, the $N$ is possibly the largest. In this case, consider any point $\bf{z}$ with $1- \varepsilon \leq \|\bf{z}\|\leq 1+ \varepsilon$, there must exists $\bf{x}_{\bf{z}}\in A$ such that $\|\bf{x}_{\bf{z}} - \bf{z}/\|\bf{z}\|\| < \varepsilon$.
Note that $\|\bf{z} - \bf{z}/\|\bf{z}\|\| < \varepsilon$,
by the triangle inequality, we obtain $\|\bf{x} - \bf{z}\|\leq 2\varepsilon$.
This means that the balls $\{\mathbb{B}(\bf{x},2\varepsilon)\mid \bf{x}\in A\}$ cover the region $\{\bf{z}: 1- \varepsilon \leq \|\bf{z}\|\leq 1+ \varepsilon\}$, giving
\begin{align}
    N(2\varepsilon)^d \geq (1+\varepsilon)^d - (1-\varepsilon)^d,
\end{align}
which completes the proof.
\end{proof}

\subsection{Mollified functions}\label{append:mollified}
For any $r>0$, let $m_r\colon \mathbb{R}^d \to \mathbb{R}$ to be the mollifier function of width $r$, which is given by
\begin{align}
    m_r(x) = \left\{
    \begin{array}{l}
        \frac{1}{I_d} \exp\left( -\frac{1}{1-\|x/r\|^2 }\right),~x\in \mathbb{B}(0,r)  \\
          0,~\mathrm{otherwise},
    \end{array}
    \right.
\end{align}
where $I_d$ is the normalize constant such that $\int_{\mathbb{B}(0,r)}m_r(x)\d x = 1 $.
The mollification of $f$ by $m_r(x)$ is defined as
\begin{align}
    F_r(x) := (f*m_r)(x) = \int f(x-y)m_r(y)\d x.
\end{align}
The mollified functions satisfy the following desired property:
\begin{lemma}\label{lem:mollify}
Let $f\colon \mathbb{R}^d \to \mathbb{R}$ and $F_r = f*m_r$, then $F_r$ is infinitely differentiable and $F_r\to f_r~(r\to0)$.
\end{lemma}
\lem{mollify} follows from the properties of convolution and the fact $m_r$ is infinitely differentiable.

\subsection{Distribution of wells}\label{append:welldistr}
Here we show that in high dimensions, the distribution of wells
will affect coefficient of the exponential term $e^{S_0/h}$ in the
evolution time of QTW.
We assume that tunneling amplitude $w$ is the same and all local ground
states have the energy $\mu$. Thus, the only variable is the ``distribution".

Let us see the example that the wells form a ring. The interaction matrix can be written as
\begin{equation}
    H_{|\F} = \left(
    \begin{array}{cccccc}
        \mu & w &  &  &  & w \\
         w & \mu & w & &  & \\
           & w & \mu & w &  & \\
           &   &  \ddots& \ddots  & \ddots &  \\
           &   &  & w & \mu &  w\\
           w  & &  &  & w &  \mu\\
    \end{array}
    \right).
\end{equation}
We have
\begin{equation}
    E_k = 2|w|\cos\frac{k\pi}{N+1} + \mu,~k=1,2,\ldots,N.
\end{equation}
Therefore,
\begin{equation}
    \sum_{E_k=E_{k^{\prime}}} \z j|E_k\y \z E_k |i\y \z i |E_{k^{\prime}}\y \z E_{k^{\prime}}| j\y = \sum_{k}
   |\z E_k| i\y|^2 |\z E_k| j\y|^2 \geq \frac{1}{2N}.
\end{equation}
The last equality is secured by the fact that
\begin{equation}
    \z E_k |j\y \sim \frac{1}{\sqrt{N}}e^{-\frac{i k j \pi}{N+1}}.
\end{equation}
The energy gap can be calculated as
\begin{equation}
   \Delta E = 2|w|\left(\cos\frac{\pi}{N+1}-\cos\frac{2\pi}{N+1} \right) = \Omega\left(\frac{1}{N^2}\right)|w|.
\end{equation}
If we set $\tau > \frac{N}{\epsilon \Delta E}$ for a small constant $\epsilon>0$, then
\begin{equation}
   p(\tau,j|i)\geq \frac{1}{2N}\left(1-\frac{3N}{\tau \Delta E} \right) \geq \frac{1}{2N}\left(1-3\epsilon\right).
\end{equation}
And the time for one trial can be
\begin{equation}
   \frac{N}{\epsilon \Delta E} = O\left(\frac{N^3}{\epsilon} \right)\frac{1}{|w|}.
\end{equation}
Repeating trials for $1/p(\tau,j|i)$ times, we can hit the $j$th well with probability $\Omega(1)$, and the total time can be bounded by
\begin{equation}
   \frac{N}{\epsilon \Delta E p(\tau,j|i)} = O\left(\frac{N^4}{\epsilon} \right)\frac{1}{|w|}.
   \label{eq:442}
\end{equation}

Consider the most symmetric case, where the interaction matrix is given by
\begin{equation}
    H_{|\F} = \left(
    \begin{array}{cccc}
        \mu & w & \cdots &  w \\
         w & \mu & w & \vdots \\
          \vdots & w & \ddots & w \\
          w & \cdots  &  w& \mu  \\
    \end{array}
    \right).
    \label{eq:Hamsym}
\end{equation}
The eigenvalues are
\begin{equation}
    E_1 = \cdots = E_{N-1} = \mu - w,~E_N = \mu + (N-1)w,
\end{equation}
giving the energy gap $\Delta E = N|w|$.
The corresponding eigenstates can be given by
\begin{equation}
    |E_k\y  = \frac{1}{\sqrt{N}}\sum_{j}e^{i\frac{2\pi}{N}k j}|j\y,
\end{equation}
which gives
\begin{equation}
    \sum_{E_k=E_{k^{\prime}}} \z j|E_k\y \z E_k |i\y \z i |E_{k^{\prime}}\y \z E_{k^{\prime}}| j\y
    =\left\{\begin{array}{c}
         \frac{2}{N^2},~i \neq j,\\
         1-\frac{2(N-1)}{N^2},~i=j,
    \end{array} \right.
\end{equation}
and
\begin{equation}
   \sum_{k}
   |\z E_k| i\y|^2 |\z E_k| j\y|^2 = \frac{1}{N}.
\end{equation}
For $i\neq j$, we will have that
\begin{equation}
   p(\tau,j|i) \geq \frac{2}{N^2} -  \frac{4}{\tau \Delta E}\frac{N-1}{N^2}\geq \frac{2}{N^2}\left(1-2\epsilon\right),
\end{equation}
if $\tau > \frac{N}{\epsilon \Delta E}$ where $\epsilon$ is a small constant.
The time needed for one trial can be bounded by
\begin{align}
    \frac{N}{\epsilon \Delta E} = O(1/\epsilon)\frac{1}{|w|},
\end{align}
which is independent of $N$.
Repeating $1/p(\tau,j|i)$ times, we are able to hit the $j$th well with high probability and the total time is bounded by
\begin{align}
    \frac{N}{\epsilon \Delta E p(\tau,j|i)} = O(N^2/\epsilon)\frac{1}{|w|}.
    \label{eq:449}
\end{align}

Comparing \eq{449} with \eq{442}, the distribution of well
can reduce the coefficient before $1/|w|$ from $O(N^4/\epsilon)$ to $O(N^2/\epsilon)$,
which is meaningful especially when $N$ is large.

\section{Full Numerical Experiments}\label{append:fullnum}
In this appendix, we present detailed numerical results testing our theoretical analysis and demonstrating our quantum speedup claims.
All experiments are performed on
a classical computer (with Dual-Core Intel
Core i5 Processor, 16GB memory) using MATLAB 2020b.

In \append{numcomparison}, we run QTW and SGD on concrete examples
provided by \sec{illustration}. Initial state preparation is numerically explored and discussed. Subsequently, QTW is implemented
and compared with SGD, demonstrating the power of quantum tunneling
as expected by our theory.
In \append{numdimdep},
classical lower bound proved in \sec{separation} (i.e., \prop{provable-exp}) is tested for a specific classical algorithm, SGD.
At last, the dependence of QTW running time on quantum learning rate
$h$ is highlighted by a experiment in \append{numqlrdep},
supporting our result, \thm{informalQTW}.
For experiments involving QTW,
we only deal with low-dimensional landscapes due to the limitation of solving the Schr\"odinger equation by classical numerical methods.

For numerical integration of the Schr\"odinger equation,
we follow \cite{zhang2021quantum} and \cite{GM96}. We first discretize the space such that the Schr\"odinger equation is approximated by a linear system (see \append{quantumsimulation} for details of discretization),
\begin{align}
    i \frac{\d \Phi}{\d t} = \hat{H} \Phi,
\end{align}
where $\hat{H}$ is a matrix approximating the Hamiltonian and $\Phi$ is a complex valued vector simulating the wave function.
Then, we write $\Phi(t) = Q(t) + i P(t)$ for $Q$ and $P$ being the real and imaginary part of $\Phi$, respectively.
The discretized Schr\"odinger equation is equivalent to
a separable Hamiltonian system:
\begin{align}
    \frac{\d}{\d t}Q = \hat{H} P,\quad \frac{\d}{\d t}P = - \hat{H} Q.
\end{align}
In practice, we use a symplectic leap frog scheme to solve
this Hamiltonian system \cite{Fra20}. This also has connection to recent literature on symplectic optimization~\cite{betancourt2018symplectic,jordan2018dynamical}.

SGD refers to the iterative algorithm $x_{k+1} = x_k - s\nabla f(x_k) -s \xi_k$, where $f$ is the objective function and $\xi_k$ is the noise of the $k$th step.
Queries to the gradients of $f$ are permitted and $\xi_k$
is normally distributed with unit variance.

\subsection{Tests on quantum-classical comparisons}\label{append:numcomparison}
Our first set of experiments are conducted on one-dimensional
examples in \sec{illustration}: the critical case (\examp{critical}),
the case stressing flatness of minima (\examp{flatness}), and the case
stressing sharpness of barriers (\examp{sharp}).
Parameters of the landscapes used in \append{numcomparison} are defined in \sec{illustration}.

To make the experiments non-trivial and
still can be efficiently executed on our computer,
we choose the number of minima $N=3$.
In solving the discretized Schr\"odinger equation, we fix the spatial domain to be $\Omega = [0, 6a + 4b]$ and the mesh number to be 512, where $a$, $3a+2b$, $5a + 4b$ are the three minima.
The problem to be solved is specified as hitting $[4a + 4b, 6a + 4b]$ beginning at one minimum $a$.

With the knowledge of one minimum $a$,
we approximate the local ground state corresponding to $a$.
Let $\Omega_1$ be a neighborhood of $a$, the Hamiltonian restricted in the domain $\Omega_1$ is also discretized as a matrix $\hat{H}_{\Omega_1}$.
We solve the eigenvectors of $\hat{H}_{\Omega_1}$ and $\hat{H}$ in MATLAB.
By our theory, the ideal initial state $\ket{\Phi(0)}$ prepared should be in
the low energy subspace $\F$ spanned by the three eigenstates of $\hat{H}$ with smallest eigenvalues, i.e., $\ket{E_j},~j=1,2,3$.
Whether the ground state of $\hat{H}_{\Omega_1}$ is nearly in
$\F$ depends on the quantum learning rate $h$ and the region $\Omega_1$.
We numerically discuss error estimation of initial state preparation (local ground state preparation)
here as a supplement to \append{initial}.
To this end, we prepare the ground state of $\hat{H}_{\Omega_1}$
for small local region (i.e., $\Omega_1 = [0,2a]$), middle local region (i.e., $\Omega_1 = [0,2a+b]$), and large local region (i.e., $\Omega_1 = [0,2a+2b]$) under various $h$. The probability or overlap of $\ket{\Phi(0)}$ in $\F$ given by
\begin{align}
    |\ip{\Phi(0)}{E_1}|^2 + |\ip{\Phi(0)}{E_2}|^2 + |\ip{\Phi(0)}{E_3}|^2
\end{align}
is numerically calculated whose results are shown in \fig{overlapinF}.
\begin{figure}
\centering
\includegraphics[width=0.32\linewidth]{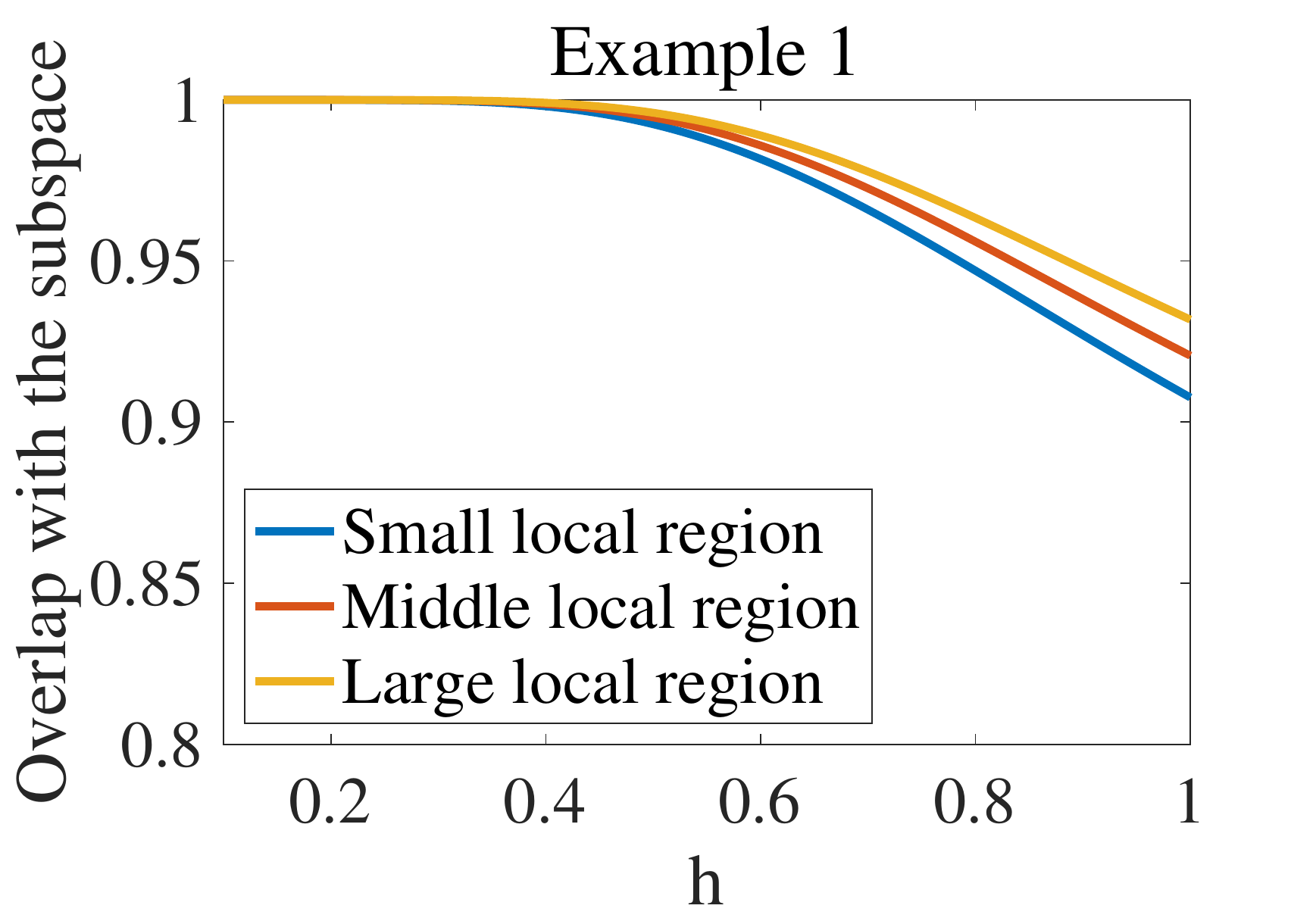}
\includegraphics[width=0.32\linewidth]{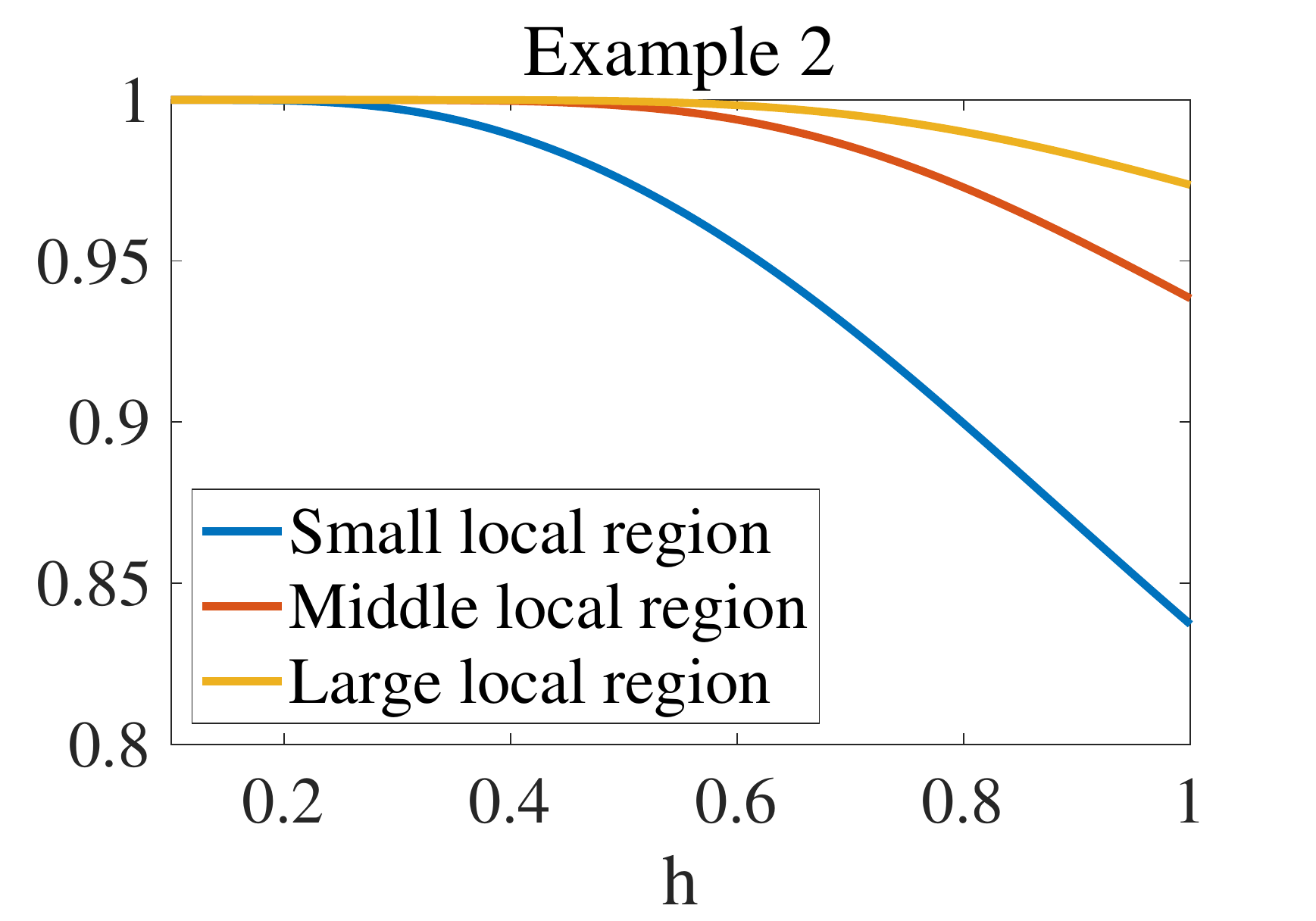}
\includegraphics[width=0.32\linewidth]{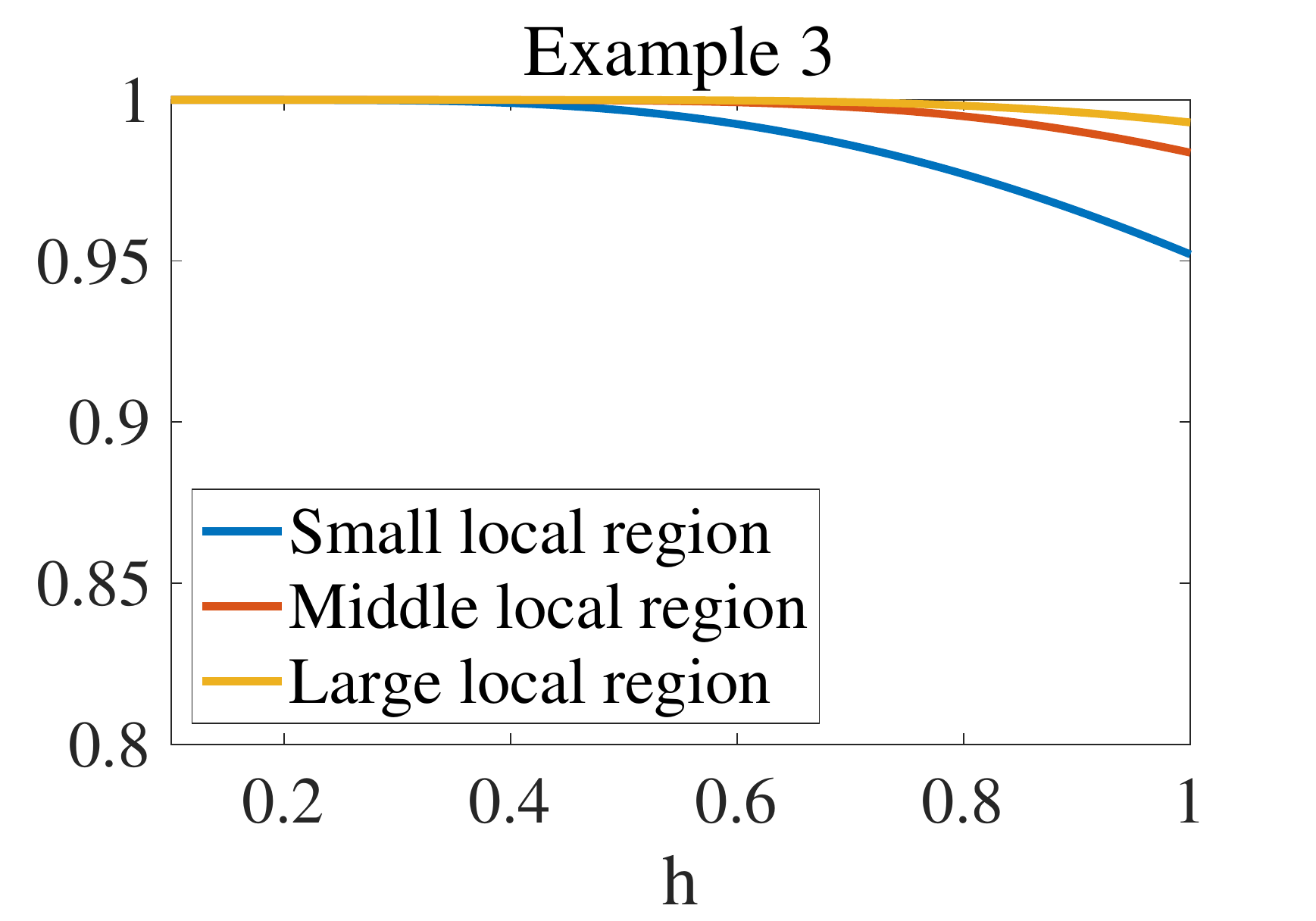}
\caption{Overlap between the local ground state and the low energy subspace $\F$: Example 1, 2, and 3 correspond to \examp{critical}, \examp{flatness}, and \examp{sharp}, respectively.}
\label{fig:overlapinF}
\end{figure}
For all examples, the larger $h$ is, the more unlikely the local state prepared $\Phi(0)$ is to be in $\F$.
For \examp{critical}, small, middle, and large local regions have similar results as $b\ll a$ in \examp{critical}.
Comparing \examp{critical} with \examp{flatness},
a thicker barrier (a larger $b$) yields lager differences
between small, middle, and large local regions.
However, if the barrier is high, as shown by \examp{sharp},
results for the three local regions can be close to each other again.
This may be accounted for by that wave functions decay rapidly in the barrier region, $[2a,2a+2b]$, and the local ground state concentrate in $[0,2a]$ regardless of whether $\Omega_1 = [0,2a]$ or $[0,2a+b]$ or $[a,2a+2b]$.

Due to the large running time cost of solving the Schr\"odinger equation on classical computers,\footnote{In this section, we use \emph{running time} to refer to the actual time spent by our classical computer to solve the problem, and we use \emph{evolution time} to denote the time variable in simulating QTW and SGD.} we
simulate QTW for each example choosing one $h$ and one initial state.
In practice, preparing the initial state in a smaller region may be
more convenient while producing a larger error.
We choose the small local region for each example to prepare the initial state, and the chosen $h$ for different examples are shown in \tab{choiceh}.
On the one hand, we need $h$ to be small to obtain an initial state largely staying in $\F$.
On the other hand, $h$ should not be too small, such that the running time is tractable on classical computers. As shown in \tab{choiceh}, all chosen $h$ manage to make at least 90\% of the initial state
be in $\F$.
\begin{table}
    \centering
    \begin{tabular}{c|c c c}
         & \examp{critical} & \examp{flatness} & \examp{sharp} \\\hline
       $h$  & 0.8 & 0.8 & 1.0 \\
       Overlap  & $\approx$ 0.95 & $\approx$ 0.90 & $\approx$ 0.95 \\
    \end{tabular}
    \caption{Choices of $h$ for different examples. Overlap denotes the overlap of the initial state prepared
    in the small local region with the subspace $\F$.}
    \label{tab:choiceh}
\end{table}

Having obtained $h$ and the prepared initial state $\Phi(0)$ for each
example, we need to determine the time $\tau$ for QTW.
Recall that the evolution time of one QTW trial, $t$, is uniformly chosen from
$[0,\tau]$ (see \sec{QTW}), and we repeat the trials until success.
As long as $\tau$ is large enough, the probability of success in one trial will be larger than some constant (see \sec{Qmixing} and \sec{Qhitting}).
Practically, we determine $\tau$ by testing a geometric series (this is a common trick in randomized algorithms), i.e., $\tau=1,2,4,8,\ldots$ and repeat the experiments for many times. If the success probability is too small, we just double $\tau$ until we reach a value where the success probability is decent.
For simplicity, we discuss and determine $\tau$ here by directly observe
the evolution of $\Phi(t)$.
\fig{illuevoExp1} illustrates
a typical quantum tunneling process on the landscape \examp{critical}.
The time $288$ when the wave function concentrates in the target well $[4a+ab,6a+4b]$ is defined as the characteristic time $T^{\star}_1$ for \examp{critical}.
Similarly, for \examp{flatness} and \examp{sharp}, we can find
the characteristic time $T^{\star}_2 = 800$ and $T^{\star}_3 = 600$, respectively.

\begin{figure}
    \centering
    \includegraphics[width=\linewidth]{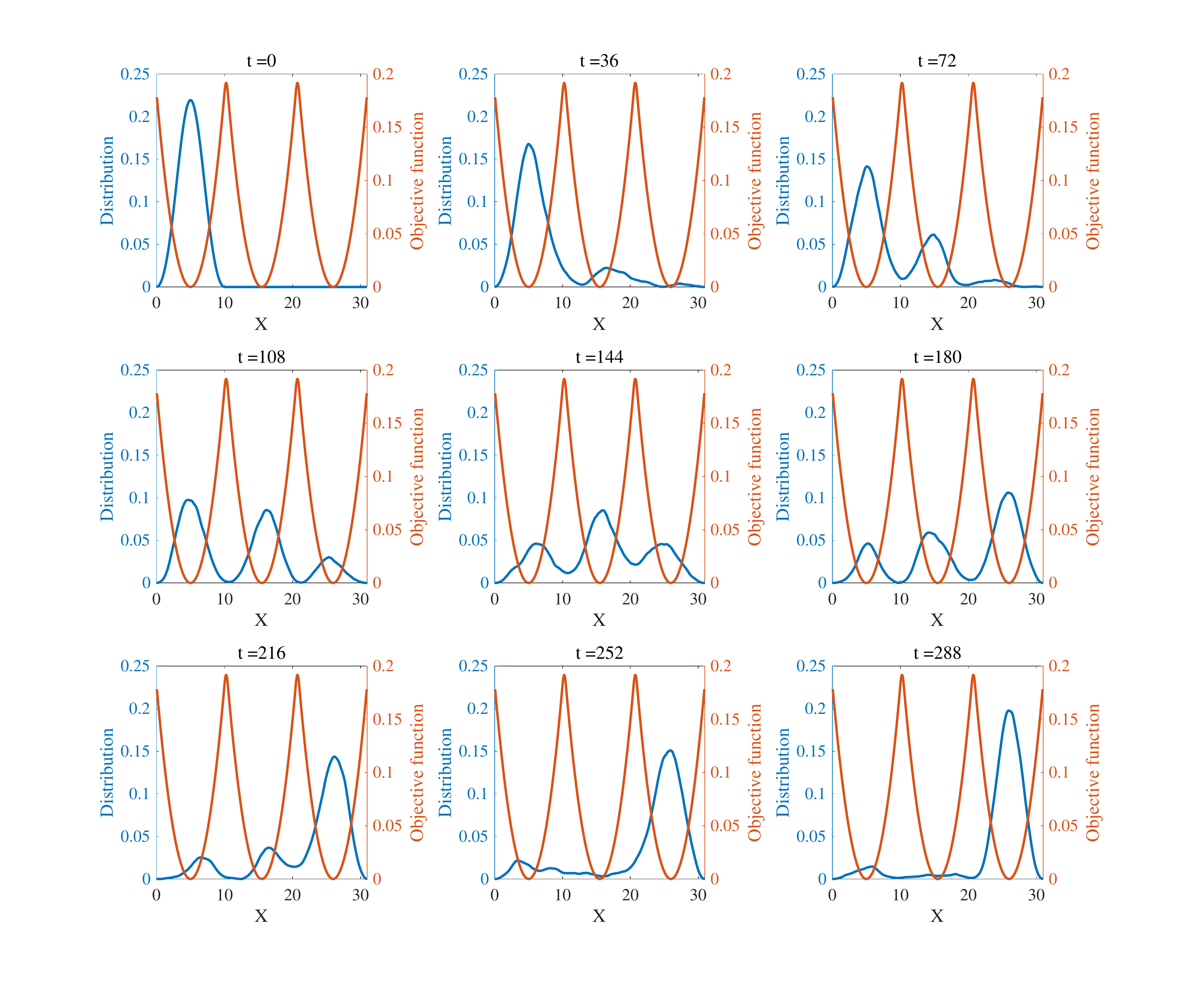}
    \caption{Evolution of the wave function on \examp{critical}:
    the blue line represents the probability distributing $|\ip{\Phi(t)}{x}|^2$
    and the red line is the landscape $f$.}
    \label{fig:illuevoExp1}
\end{figure}
Setting $\tau$ as the half characteristic time, the characteristic time, and the double characteristic time, respectively, we repeat our experiments and obtain the average hitting time.
Rigorously, we use an \emph{experiment} to denote a process repeating trials until
successfully hitting $[4a+ab,6a+4b]$.
A \emph{trial} initiates the simulation at $\Phi(0)$ and
measures the position at $t$ randomly chosen from $[0,\tau]$.
The average hitting time means the
average total evolution time of different experiments. Results of average hitting time are presented in \tab{choisetau}, each of which is obtained from 20 experiments.

\begin{figure}
\centering
\includegraphics[width=0.32\linewidth]{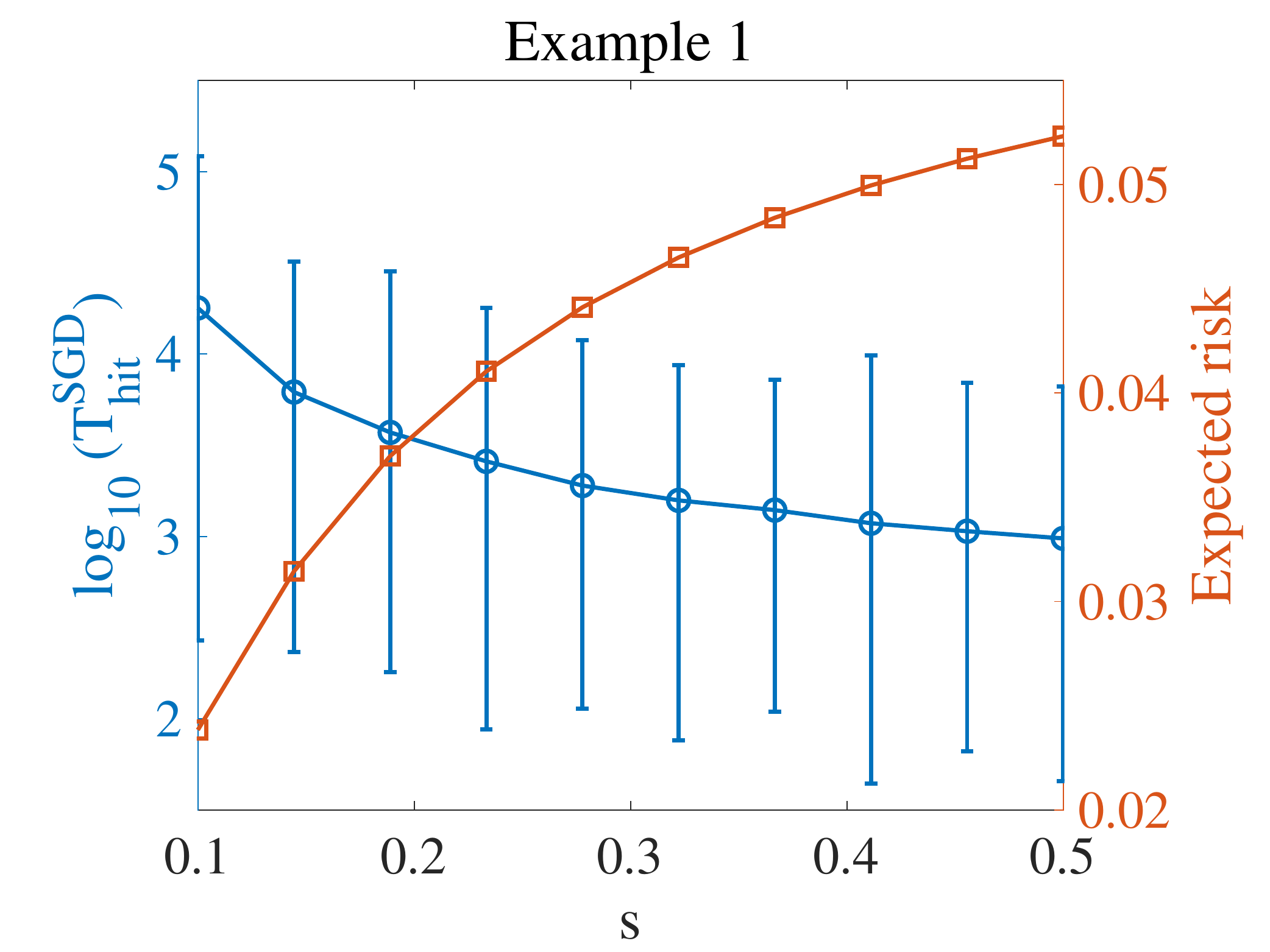}
\includegraphics[width=0.32\linewidth]{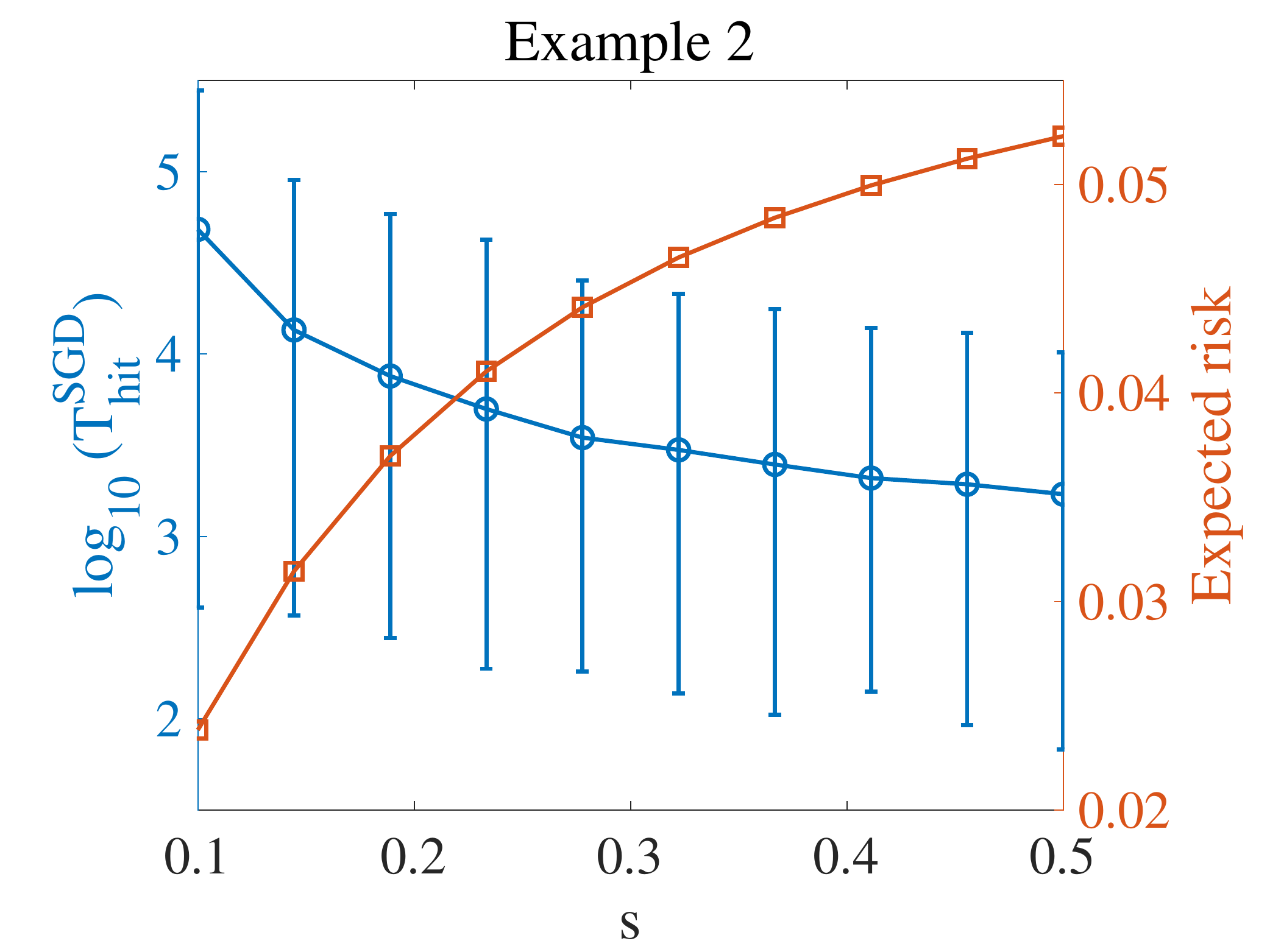}
\includegraphics[width=0.32\linewidth]{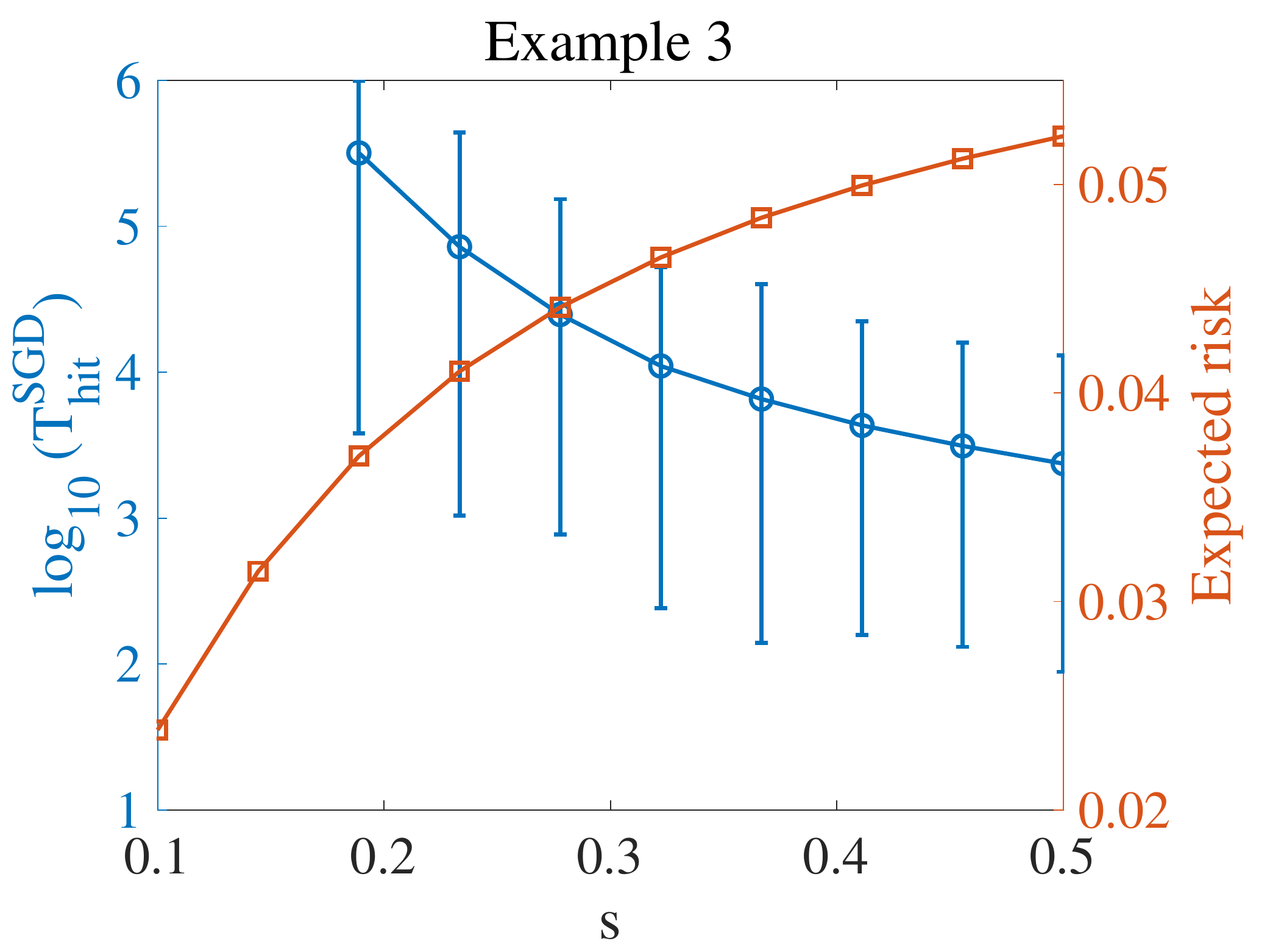}
\caption{Dependence of the SGD hitting time $T^{\rm SGD}_{\rm hit}$ and the SGD expected risk on the learning rate $s$. Each experiment is repeated for 1000 times, the (blue) circles denote the average hitting time, and one error bar represents the range of the corresponding 1000 data.}
\label{fig:lrdepenExp}
\end{figure}

\begin{table}
    \centering
    \begin{tabular}{l|c c c}
         & \examp{critical} & \examp{flatness} & \examp{sharp} \\\hline
       $\tau = T^{\star}/2$ (half characteristic time)  & $655 \pm 648$ & $1372 \pm 953$ & $708\pm 519$ \\
       $\tau = T^{\star}$ (characteristic time) & $409\pm 259$ & $1184 \pm 561$& $758 \pm 471$ \\
       $\tau = 2T^{\star}$ (double characteristic time) & $759 \pm 489$ & $1623 \pm 1010$ & $1464\pm 1067$
    \end{tabular}
    \caption{Average hitting time obtained from 20 experiments for each $\tau$ and each example.}
    \label{tab:choisetau}
\end{table}
Theoretically, when $\tau$ is large enough, the probability of successfully
hitting in one trial is near a constant $p_{\rm suc}$, so that the number of trials needed in one experiment is also approximately a constant $1/p_{\rm suc}$.
The expected time for each trial is $\tau/2$. Thus, for large enough $\tau$,
the expected time of one experiment is approximately $\tau/p_{\rm suc} \propto \tau$.
Based on results of \tab{choisetau}, the characteristic time $T^{\star}$
is large enough, and is chosen as $\tau$ in later experiments.

Having specified the values of $h$ and $\tau$,
we conduct 1000 experiments for each example and study the distribution of hitting time.
The statistics of the results are shown in \fig{histogramsExp123}.

Next, we proceed by studying the hitting time of SGD.
Taking fair comparisons into consideration, we employ \stand{risk} in
\sec{standard} which equalizes the expected risks yielded by SGD and QTW.
The expected risk of QTW is numerically determined by
calculating
\begin{align}
   \frac{1}{T}\int_{0}^T \Big(\int_{\Omega} f(x)|\ip{\Phi(t)}{x}|^2 \d x\Big)\d t
   \label{eq:numquanexprisk}
\end{align}
for $T = 1000$.
The expected risk st $t$, $\int_{\Omega} f(x)|\ip{\Phi(t)}{x}|^2 \d x$,
oscillates rapidly with respect to time but converges to a fixed value; see for instance \fig{2dimexpectedrisk} for a one-dimensional case. Thus, we pick a long enough time $T$ to calculate an expected risk for all time.
For \examp{critical}, \examp{flatness}, and \examp{sharp}, the expected risks given by \eq{numquanexprisk} are $0.0324\pm 0.0054$, $0.0222 \pm 0.0051$, and $0.0377 \pm  0.0070$, respectively.
Using the Gibbs distribution \eq{Gibbs} $\mu_{\rm SGD}$ for the calculation of
the expected risk
\begin{align}
    \int \mu_{\rm SGD}(x)f(x)\d x,
\end{align}
we can plot the relationship between the expected risk and the learning rate $s$
as shown in \fig{lrdepenExp}.

For each example, we calculate the hitting time for various $s$, which is also plotted in \fig{lrdepenExp}.
Theoretically, $\ln T^{\rm SGD}_{\rm hit} \sim 2H_f/s$ with $2H_f$ being the Morse saddle barrier of $f$, which agrees well with the results in \fig{lrdepenExp}.

To make the expected risks of QTW and SGD equal,
we set learning rates $s =0.1525$, $s = 0.129$, and $s = 0.189$ for
\examp{critical}, \examp{flatness}, and \examp{sharp}, respectively.
For each example and corresponding learning rate $s$, we
run SGD for 1000 times and draw the histogram of SGD hitting time in
\fig{histogramsExp123}.
We use $T^{\rm QTW}_{\rm hit}$ and $T^{\rm SGD}_{\rm hit}$
to denote the evolution time of one experiment for QTW and SGD, respectively.
In \fig{histogramsExp123}, the histograms compare $T^{\rm QTW}_{\rm hit}$ with $T^{\rm SGD}_{\rm hit}/10$, and
all presented examples demonstrate that QTW is faster.
The number of quantum queries is approximately
$\tilde{O}(\|f\|_{L^{\infty}_{\Omega}} T^{\rm QTW}_{\rm hit})$ and the number of classical queries is $\Omega(T^{\rm SGD}_{\rm hit}/s)$.
In addition, in the three examples $\|f\|_{L^{\infty}_{\Omega}} \leq 0.85$ and $s<0.2$, and
quantum advantage exists in terms of query complexity.

This result matches our theory at large.
For \examp{critical}, we make direct comparison between the exponential terms $e^{S_0/h}$ and $e^{2H_f/s}$,and to remove the coefficients in front of them,
we divide $T^{\rm SGD}_{\rm hit}$ by $10$ such that $T^{\rm SGD}_{\rm hit}/10$ has similar distribution t $T^{\rm QTW}_{\rm hit}$ for \examp{critical}.
In this way, we observe that whether $T^{\rm SGD}_{\rm hit}/10$ is relatively larger than $T^{\rm QTW}_{\rm hit}$ is determined only by $e^{S_0/h}$ and $e^{2H_f/s}$.

For \examp{flatness}, $T^{\rm QTW}_{\rm hit}$ is not much smaller than $T^{\rm SGD}_{\rm hit}/10$, which is not completely coherent with our theory.
This result can be explained as that for \examp{flatness}, the quantum learning rate $h$ is not small enough such that the initial state prepared does not well stay near a low energy subspace.
Specifically, \tab{choiceh} shows that initial states of both \examp{critical} and \examp{sharp} are largely
in $\F$ (more than 95\%), while about 90\% of the initial state of \examp{flatness} overlaps with
$\F$.
Experiments on \examp{flatness} suggest that higher energy may not be able to help quantum tunneling to run faster.

For \examp{sharp}, significant quantum speedup is achieved as expected. As shown by \fig{histogramsExp123}, $T^{\rm SGD}_{\rm hit}/10$ is even several orders of magnitude larger than $T^{\rm QTW}_{\rm hit}$.

\subsection{Dimension dependence}\label{append:numdimdep}
\begin{figure}
\centering
\includegraphics[width=\linewidth]{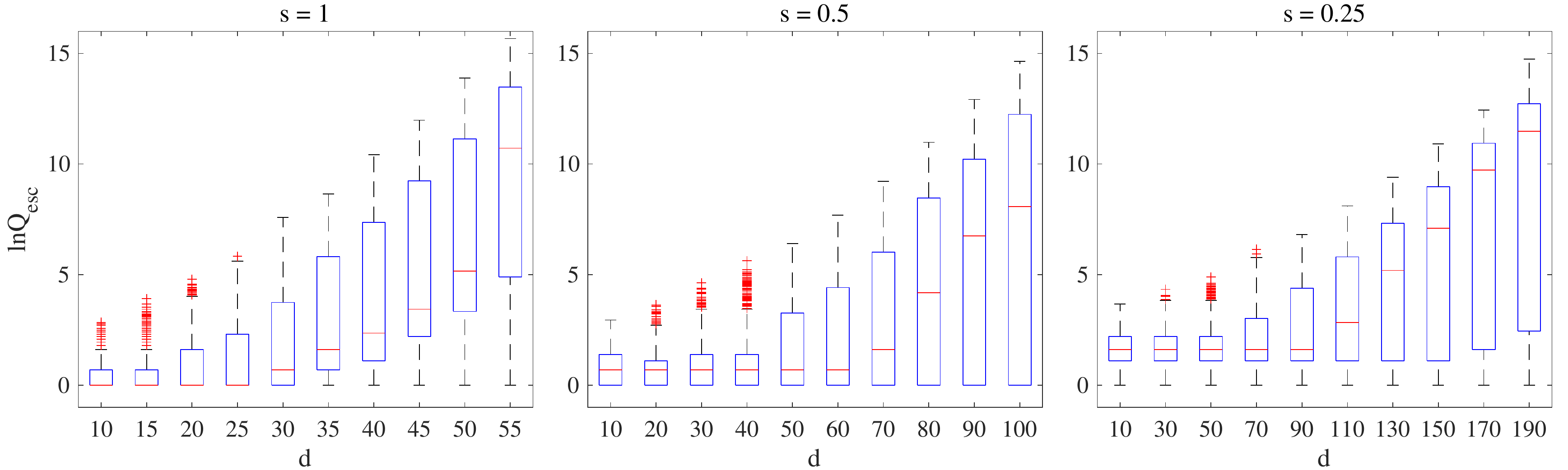}
\caption{Distribution of the number of steps $Q_{\rm esc}$ needed to escape $S_{\mathbf{v}}$: for each $s$ and each $d$, 1000 samples are gathered.}
\label{fig:boxlnT}
\end{figure}
Here we examine part of \sec{separation} by testing SGD and the classical lower bound.

The classical lower bound in \sec{separation} ensures that
for any $s$, SGD cannot cannot escape from $S_{\mathbf{v}}$ with subexponential queries with high probability.
Based on the constructed landscape \eq{hardinstance} with parameters specified by $a=1$, $R = 4\sqrt{2}a$, $w=a/2$, $\omega = 0.5$, $b = 1.4$, $H_1 = H_0$, and $H_2 = 20H_0$, we test SGD with different learning rates ($s\in [0.1,1]$)
in various dimensions ($d \in [15,95]$).
For each dimension and each $s$, 1000 experiments are conducted.
The number of steps spent to escaping from $S_{\mathbf{v}}$ in one experiment is denoted as $Q_{\rm esc}$.

The distributions of $Q_{\rm esc}$ under different dimensions for $s = 1$, 0.5, and 0.25 are shown in \fig{boxlnT}.
With the increase of $d$, all characterized values of the distribution gradually grows exponentially in terms of $d$, supporting our theory.

We also present the relationship between the average $Q_{\rm esc}$ and the dimension $d$ in \fig{dimdepenmean}.
A counterintuitive fact is that for a fixed dimension $d$, SGD with larger $s$ is more difficult to escape $S_{\mathbf{v}}$, which needs further explanation.
For each fixed learning rate $s$, we observe that with the increase of $d$, the average $Q_{\rm esc}$ remains constant initially and then increase exponentially with respect to $d$.
Increasing $s$ yields a smaller initial constant but larger exponential rate.
Nevertheless, for all $s$, $Q_{\rm esc}$ eventually increases
exponentially with respect to $d$.
With the prediction in \sec{separation}, if the number of queries is of the order $e^{\frac{dw^2}{2R^2}} = e^{d/256}$, the probability of escaping $S_{\mathbf{v}}$ will no longer exponentially small.
The lower bound of average $Q_{\rm esc}$ increases larger than $e^{d/256}$, supporting the prediction.

\subsection{Quantum tunneling and the quantum learning rate \texorpdfstring{$h$}{h}}\label{append:numqlrdep}
In QTW, the quantum learning rate $h$ is one of the most important variables.
\thm{informalQTW} gives a general relationship between $h$ and the evolution time of QTW.
We further test the relationship on the landscape constructed in \sec{separation} (dimension $d=2$) with specified parameters given in the same as \append{numdimdep}: $a=1$, $R = 4\sqrt{2}a$, $w=a/2$, $\omega = 0.5$, $b = 1.4$, $H_1 = H_0$, and $H_2 = 20H_0$. The Schr\"odinger equation is solved in the region $\Omega = \{(x,y):|x|<4, |y|<4\}$
and the mesh number is 399 on each edge.

Since the landscape has two symmetric wells, the time for tunneling from one well to the other, $T_{\rm half}$, is explicitly linked to $\Delta E$, i.e., $T_{\rm half} = \pi/\Delta E$.
On this concrete landscape, $\Delta E$ can be predicted using \lem{provablenu} and note that $\Delta E = 2|\nu|$, we have
\begin{align}
    \ln T_{\rm half} = \frac{S_0}{h} - \frac{1}{2}\ln\frac{2h}{\pi} + \ln\frac{\pi}{2} -\frac{1}{2}\ln\frac{2H_1 \omega^2}{4\sqrt{H_1}/b}
    -\frac{2\omega(b-a)}{\sqrt{2H_1}} + 4\ln\frac{b}{a}.
\end{align}

The initial state is prepared in the small region $\Omega_1 = \{(x,y):|x|<2, |y|<2\}$. Varying $h$, the overlap between the initial state and the low-energy subspace $\F$ is shown in \fig{overlapExp4}.
We choose $h\in[0.16,0.28]$ such that the overlap is approximately in $[0.85, 0.95]$.
\begin{figure}
\centering
\includegraphics[width=0.4\linewidth]{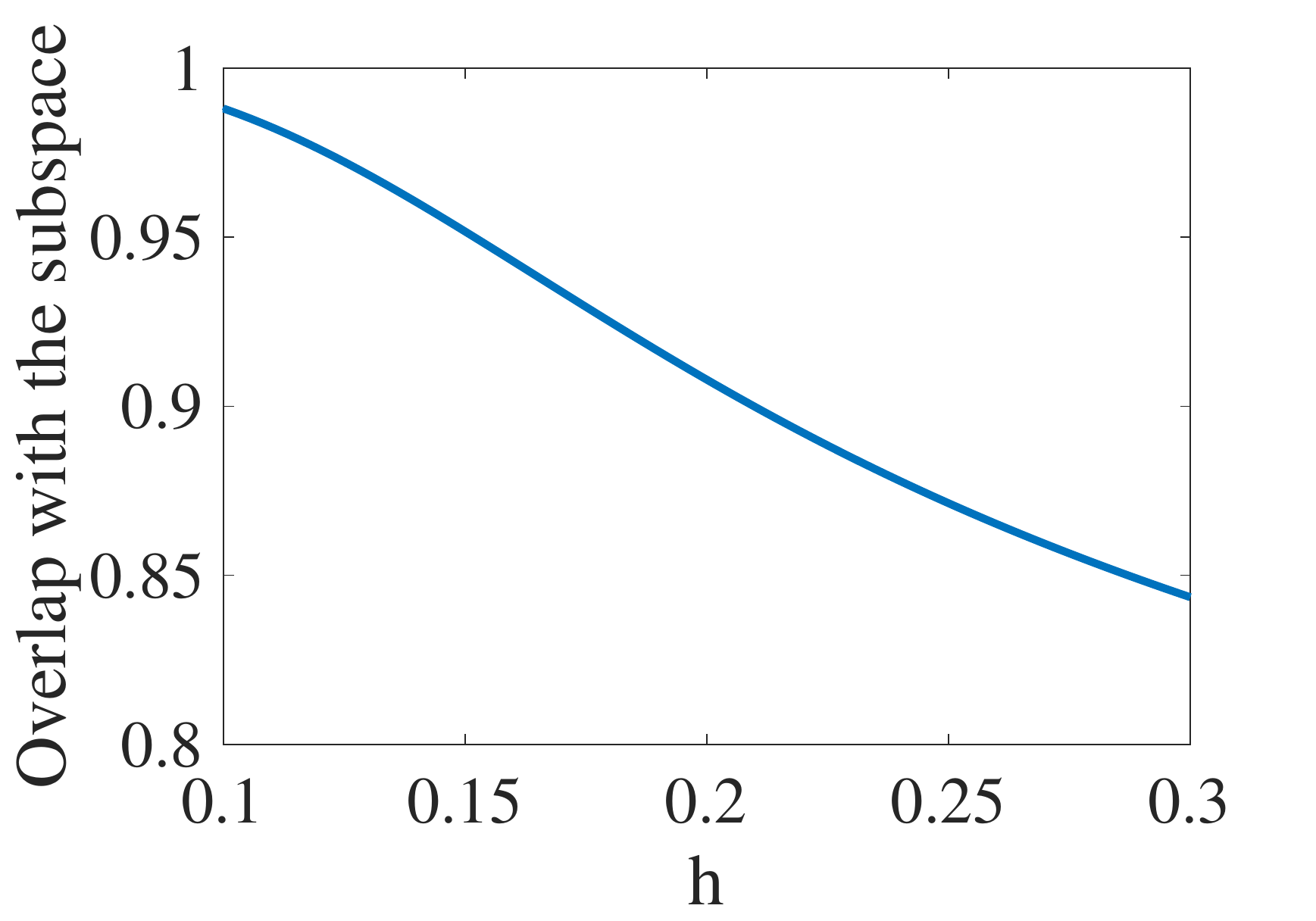}
\caption{The relationship between the overlap and $h$.}
\label{fig:overlapExp4}
\end{figure}

Starting from one well, we stop when the probability of finding
the other well exceeds 90\% and record the evolution time as $T_{\rm half}$. The probability of finding
the other well is numerically calculated based on
\begin{align}
    \int_{W_+} |\ip{x}{\Phi(t)}|^2 \d x.
\end{align}
For $h = 0.2$, $T_{\rm half}$ is approximately 200, and we show the
evolution of the wave function in \fig{h0.2t2002}.
\begin{figure}[!b]
    \centering
    \includegraphics[width=\linewidth]{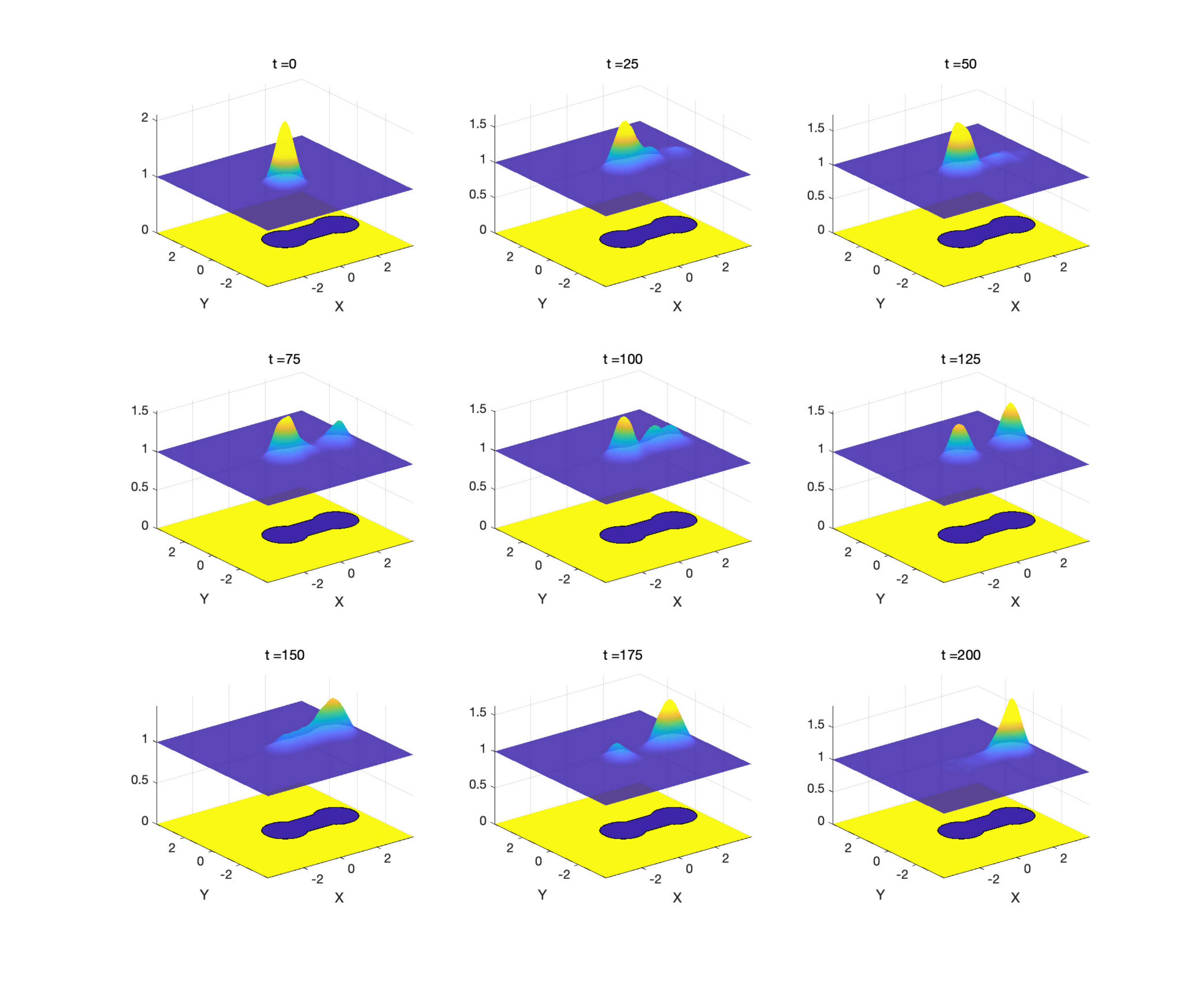}
    \caption{Wave function evolution for $h=0.2$. The upper part of all the nine figures shows the wave packet at different time and the lower part is the contour map of the landscape showing $W_-$, $B_{\mathbf{v}}$, and $W_+$.}
 \label{fig:h0.2t2002}
\end{figure}

The expect risk
\begin{align}
    \int_{\Omega} f(x) |\ip{x}{\Phi(t)}|^2 \d x
\end{align}
with respect to time is plotted in \fig{2dimexpectedrisk}, which oscillates rapidly but maintains to be near 0.065.

Numerical results on $T_{\rm half}$ are shown in \fig{qlrdepend}. The results match our theory
except a constant difference between the predicted and experimental $\ln T_{\rm half}$, indicating the correctness of $\ln T_{\rm half} = \frac{S_0}{h} - \frac{1}{2}\ln h + \cdots$.
The constant difference emerges because we stop evolution when the probability of tunneling exceeds 90\%, while the theoretical $T_{\rm half}$ takes the time when the probability is nearly 100\%.
Anyway, in studying the dependence of $T_{\rm half}$ on $h$, a constant coefficient is insignificant.

\begin{figure}[htbp]
\centering
\includegraphics[width=0.4\linewidth]{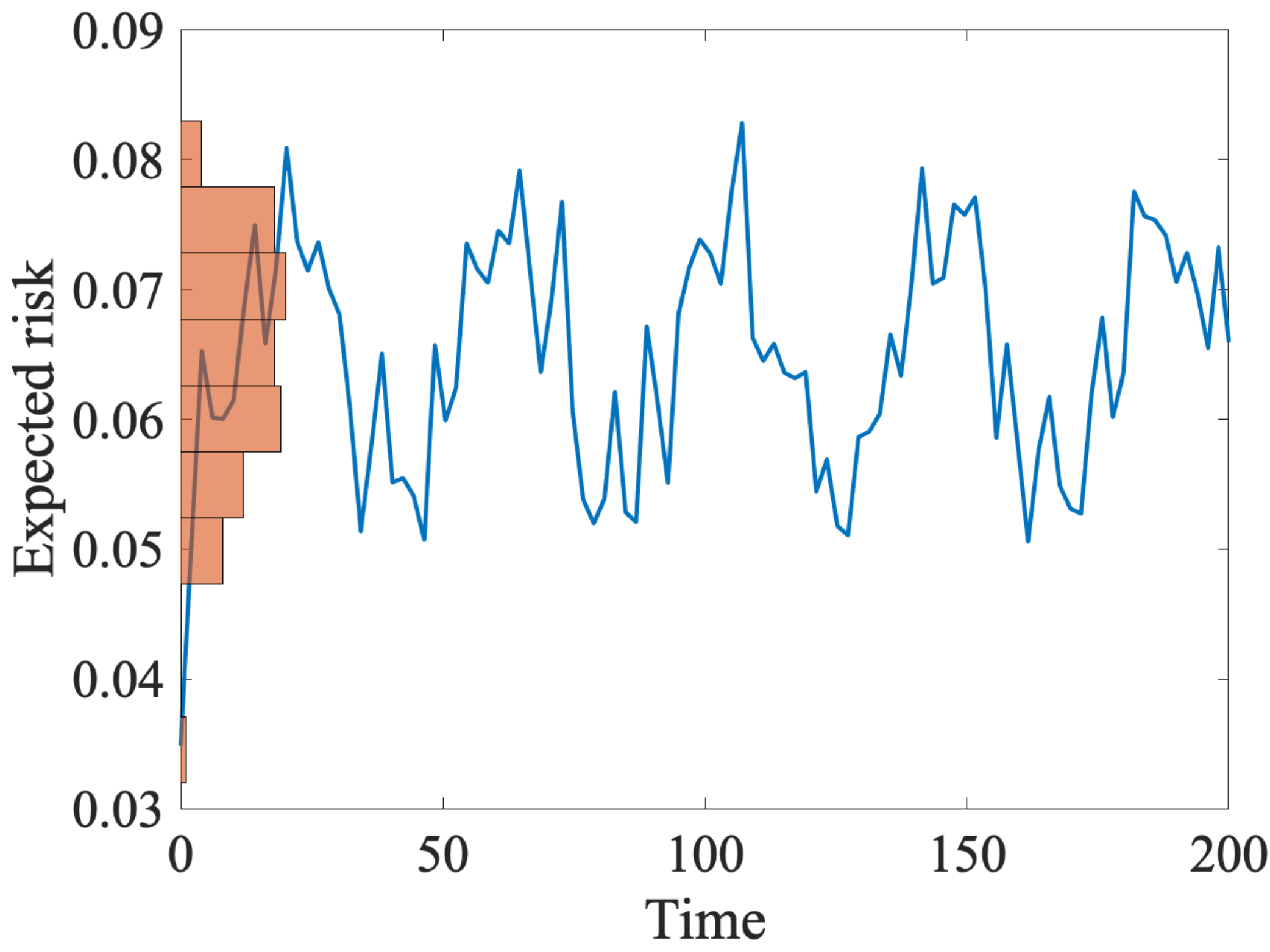}
\caption{Expected risk when $h=0.2$ with total running time 200.}
\label{fig:2dimexpectedrisk}
\end{figure}


\end{document}